\documentclass[11pt, a4paper, reqno]{amsart}
\usepackage{amsthm, amsmath, amsfonts, amssymb, appendix, dsfont, latexsym}
\usepackage{amsfonts, mathrsfs}
\usepackage{tikz}
\usepackage{braids}
\usetikzlibrary{arrows}

\makeatletter
\usepackage{delarray, a4, color}
\usepackage{euscript}
\usepackage[latin1]{inputenc}
\usepackage{enumitem}
\makeindex

\newcommand{\lect}[1]{}

\theoremstyle{plain}
\newtheorem{theorem}{Theorem}[section]
\newtheorem{lemma}[theorem]{Lemma}
\newtheorem{corollary}[theorem]{Corollary}
\newtheorem{proposition}[theorem]{Proposition}

\newtheorem{exc}{Exercise}[section]
\theoremstyle{definition}
\newtheorem{definition}[theorem]{Definition}
\newtheorem{example}[theorem]{Example}
\theoremstyle{remark}
\newtheorem{remark}[theorem]{Remark}
\newtheorem{remark*}[theorem]{Remark\textup{*}}

\numberwithin{equation}{section}

\makeatletter
\newcommand{\keyword}{\@dblarg\@keyword}
\def\@keyword[#1]#2{\textbf{#2}\index{#1}}
\makeatother

\DeclareMathAlphabet{\mathpzc}{OT1}{pzc}{m}{it}
\def\eps{\varepsilon}
\def\B {\mathbb{B}}
\def\C {\mathbb{C}}
\def\F {\mathbb{F}}
\def\N {\mathbb{N}}
\def\Q {\mathbb{Q}}
\def\R {\mathbb{R}}
\def\Rplus {\mathbb{R}_\limplus}
\def\S {\mathbb{S}}
\def\Z {\mathbb{Z}}
\def\Zplus {\mathbb{Z}_\limplus}
\def\eD {\EuScript{D}}
\def\eN {\EuScript{N}}
\def\eP {\EuScript{P}}

\newcommand{\0}{\mathbf{0}}
\newcommand\1{{\ensuremath {\mathds 1} }}

\newcommand{\bA}{\mathbf{A}}
\newcommand{\bF}{\mathbf{F}}

\newcommand{\bL}{\mathbf{L}}
\newcommand{\bP}{\mathbf{P}}
\newcommand{\bR}{\mathbf{R}}
\newcommand{\bX}{\mathbf{X}}
\newcommand{\ba}{\mathbf{a}}
\newcommand{\bb}{\mathbf{b}}

\newcommand{\bk}{\mathbf{k}}
\newcommand{\bn}{\mathbf{n}}
\newcommand{\bp}{\mathbf{p}}
\newcommand{\br}{\mathbf{r}}
\newcommand{\bv}{\mathbf{v}}
\newcommand{\bx}{\mathbf{x}}
\newcommand{\by}{\mathbf{y}}
\newcommand{\bz}{\mathbf{z}}
\newcommand{\cA}{\mathcal{A}}
\newcommand{\cB}{\mathcal{B}}
\newcommand{\cD}{\mathcal{D}}
\newcommand{\cE}{\mathcal{E}}
\newcommand{\cF}{\mathcal{F}}
\newcommand{\cG}{\mathcal{G}}
\newcommand{\cH}{\mathcal{H}}
\newcommand{\cK}{\mathcal{K}}
\newcommand{\cL}{\mathcal{L}}
\newcommand{\cN}{\mathcal{N}}
\newcommand{\cM}{\mathcal{M}}
\newcommand{\cO}{\mathcal{O}}

\newcommand{\cQ}{\mathcal{Q}}

\newcommand{\gH}{\mathfrak{H}}

\newcommand{\gh}{\mathfrak{h}}
\newcommand{\ha}{\hat{a}}
\newcommand{\hb}{\hat{b}}
\newcommand{\hf}{\hat{f}}
\newcommand{\hh}{\hat{h}}
\newcommand{\hp}{\hat{p}}
\newcommand{\hu}{\hat{u}}
\newcommand{\hx}{\hat{x}}
\newcommand{\hH}{\hat{H}}
\newcommand{\hT}{\hat{T}}
\newcommand{\hV}{\hat{V}}
\newcommand{\hW}{\hat{W}}
\newcommand{\hbx}{\hat{\mathbf{x}}}
\newcommand{\hbp}{\hat{\mathbf{p}}}
\newcommand{\hsx}{\hat{\textup{x}}}
\newcommand{\hsp}{\hat{\textup{p}}}

\newcommand{\scF}{\mathscr{F}}
\newcommand{\sA}{\textup{A}} %mathrm
\newcommand{\sC}{\textup{C}}
\newcommand{\sF}{\textup{F}}
\newcommand{\sR}{\textup{R}}
\newcommand{\sP}{\textup{P}}

\newcommand{\ssp}{\textup{p}}
\newcommand{\sw}{\textup{w}}
\newcommand{\sx}{\textup{x}}

\newcommand{\sz}{\textup{z}}
\newcommand{\phP}{\mathcal{P}}
\newcommand{\phX}{\mathcal{C}}

\DeclareMathOperator{\esssup}{\mathrm{ess\,sup}}
\DeclareMathOperator{\re}{\mathrm{Re}}

\DeclareMathOperator{\sign}{\mathrm{sign}}
\DeclareMathOperator{\Span}{\mathrm{Span}}
\DeclareMathOperator{\supp}{\mathrm{supp}}
\DeclareMathOperator{\Tr}{\mathrm{Tr}}

\newcommand{\inp}[1]{\left\langle#1\right\rangle}
\newcommand{\binp}[1]{\bigl\langle#1\bigr\rangle}
\newcommand{\norm}[1]{\left\|#1\right\|}
\newcommand{\slot}{\cdot}

\newcommand{\dist}{\mathrm{dist}}

\newcommand{\sym}{\mathrm{sym}}
\newcommand{\asym}{\mathrm{asym}}
\newcommand{\loc}{\mathrm{loc}}
\newcommand{\COM}{\mathrm{cm}} %{\mathrm{COM}}
\newcommand{\rel}{\mathrm{rel}}
\newcommand{\cl}{\mathrm{cl}}
\newcommand{\hs}{\mathrm{hs}}
\newcommand{\vphi}{\varphi}

\renewcommand{\phi}{\varphi}

\newcommand{\bDelta}{{\mbox{$\triangle$}\hspace{-8.0pt}\scalebox{0.8}{$\triangle$}}}

\newcommand{\limplus}{{\mathchoice{\vcenter{\hbox{$\scriptstyle +$}}}
  {\vcenter{\hbox{$\scriptstyle +$}}}
  {\vcenter{\hbox{$\scriptscriptstyle +$}}}
  {\vcenter{\hbox{$\scriptscriptstyle +$}}}
}}
\newcommand{\limminus}{{\mathchoice{\vcenter{\hbox{$\scriptstyle -$}}}
  {\vcenter{\hbox{$\scriptstyle -$}}}
  {\vcenter{\hbox{$\scriptscriptstyle -$}}}
  {\vcenter{\hbox{$\scriptscriptstyle -$}}}
}}

\usepackage{mathtools}
\mathtoolsset{showonlyrefs=true}

\usepackage{placeins}

\title[Uncertainty and exclusion principles in quantum mechanics]{Methods of modern mathematical physics\\\footnotesize{Uncertainty and exclusion principles in quantum mechanics}}

\author[D. Lundholm]{Douglas LUNDHOLM}
\address{Department of Mathematics, KTH Royal Institute of Technology, SE-100 44 Stockholm, Sweden}

\begin{document}

\begin{abstract}
These are lecture notes for a master-level course given at KTH, Stockholm,
in the spring of 2017, with the primary aim of proving the stability
of matter from first principles using modern 
mathematical methods in many-body quantum mechanics.
General quantitative formulations of the uncertainty and the exclusion principles of
quantum mechanics are introduced, 
such as the Hardy, Sobolev and Poincar\'e 
functional inequalities
as well as the powerful Lieb--Thirring inequality that combines these two principles.
The notes are aimed to be both fairly self-contained and 
at the same time complementary to existing literature,
also covering recent developments to prove Lieb--Thirring inequalities and stability
from general, weaker formulations of the exclusion principle.
\end{abstract}

\maketitle
\setcounter{tocdepth}{2}
\tableofcontents

\section{Introduction\lect{ [1]}}\label{sec:intro}

	Most of us take the stability of the world around us
	--- as observed to consist of atoms, molecules, and even larger lumps of 
	matter such as rocks, biological beings and entire planets ---
	for granted every day. There is nothing strange about it.
	However, proving mathematically from first principles of 
	mechanics that this is indeed so turns out to be surprisingly 
	challenging and subtle.
	It was considered to be 
	one of the great triumphs of mathematical physics 
	when this problem was solved, first by 
	Dyson and Lenard in 1967 \cite{DysLen-67},
	then in a better understood approach by 
	Lieb and Thirring in 1975 \cite{LieThi-75},
	and subsequently further details have been worked out over several decades
	by numerous other mathematicians and physicists \cite{LieSei-09}.
	Its resolution turns out to rest fundamentally on the two basic principles of
	quantum mechanics: the \keyword{uncertainty principle} and the 
	\keyword[exclusion principle]{(Pauli) exclusion principle}.
	Namely, without 
	these two concepts, i.e.\ relying strictly on the 
	(indeed very well-founded) 
	framework for mechanics which was available at the end of the 19th century
	and nowadays called classical mechanics,
	matter turns out to be \emph{unstable} because the orbit of an electron around 
	the nucleus of an atom can be made arbitrarily small and the electron 
	may thus crash into the nucleus.
	Quantum mechanics came in to resolve this puzzle by introducing the idea
	that electron orbits are quantized 
	into a discrete set of spatial probability distributions, 
	with a smallest approximate 
	radius called the Bohr radius. 
	This prevents the electron from falling further into the attractive
	and infinitely deep potential well caused by the nucleus, 
	and seemingly leads to the stability of matter.
	Indeed, most physicists are content with this answer even 
	today,
	and the typical quantum mechanics textbook digs no further into the issue.
	However, a more careful mathematical analysis of the usual argument invoked
	(known as Heisenberg's formulation of the uncertainty principle)
	leads to the realization that it remains \emph{insufficient} 
	to rigorously prove stability. 
	Stronger formulations of the uncertainty principle, such as the functional
	inequalities known as Hardy's or Sobolev's inequality, may instead be
	used to prove that an atom is indeed stable.
	
	This is not the end of the story, however, because evidently the
	world consists of many more particles than just one single atom,
	with a mix of attractive and repulsive electromagnetic forces between them,
	and it turns out to be a very subtle issue to understand why in fact taking 
	two similar 
	lumps of matter and putting them together produces 
	just twice the amount of matter
	when it comes to volume and energy, and why not some new state forms which 
	is more favorable energetically and takes much less space.
	Here the Pauli exclusion principle comes into play, which tells us that
	particles such as electrons 
	(and generally known as fermions)
	cannot all occupy the same quantum state, but must
	rather move into different configurations, such as different atom orbitals.
	This is what gives rise to the periodic table of the elements
	along with their chemical properties, and
	effectively produces larger and larger atoms and molecules, and in general, 
	matter whose energy and volume scales \emph{linearly} with the number of particles.
	The effect makes its presence all the way up to the size of stars, 
	and explains 
	for example why certain astronomical objects known as white dwarfs 
	do not collapse under their own extreme gravitation to form black holes.
	
	The story of the problem of stability of matter, from the invention
	of quantum mechanics to present times, is told as it should
	--- in the rigorous language of mathematics ---
	in the one 
	textbook on the subject, namely \cite{LieSei-09},
	to which we refer the reader for a more complete account of its background 
	and subsequent developments in various directions.
	The aim of these lecture notes is to provide 
	an as concise as possible, and at the same time rigorous and self-contained, 
	path to stability 
	via a powerful functional inequality introduced by Lieb and Thirring 
	which elegantly combines the uncertainty and exclusion principles.
	However, we will in contrast to \cite{LieSei-09} take a recently 
	developed route to proving this inequality, which actually lies
	closer in spirit to the original Dyson--Lenard approach.
	In particular, the way we incorporate the exclusion principle will
	make transparent its role in the proof of stability
	as an effective repulsion between particles, 
	and furthermore clarifies that it also extends to other particles 
	than those obeying the usual Pauli principle, 
	as long as they experience 
	a strong enough repulsive interaction.
	For completeness and
	in order to further complement \cite{LieSei-09},
	these notes also contain some mathematical preliminaries and
	some background material on classical and quantum mechanics aimed for mathematicians,
	including a general discussion on identical particles and quantum statistics.
	In parallel with our general treatment of exclusion principles
	we also discuss a wide variety of formulations of the uncertainty principle,
	both global and local with respect to the configuration space,
	though we typically focus on their conceptual content rather than the 
	most precise formulation or the optimal constants.

	This version of the lecture notes, dated May 2018, still lacks some
	of the intended topics and corrections, however they will hopefully anyway 
	find use in a wider audience.

	A brief note concerning the notation: 
	In the many-body context we usually write
	$x \in \R$ for scalars, $\bx \in \R^d$ for one-body vectors, 
	and $\sx \in \R^n$ for general or many-body vectors,
	such as $\sx = (\bx_1,\bx_2,\ldots,\bx_N) \in (\R^d)^N = \R^{dN}$.
	The letter $C$ will generally denote a constant whose exact value is unimportant
	and which may vary from one expression to another.
	Remarks with * signify that some more background (in math or physics) 
	is required.
	
	\medskip
	\noindent\textbf{Acknowledgments.} 
	I would like to thank Ari Laptev for originally introducing me to Hardy and
	Lieb--Thirring inequalities during my PhD studies, 
	and Jens Hoppe for bridging my gap to spectral theory at that time.
	Furthermore, I thank my collaborators 
	on some of the topics briefly touched upon in these notes: 
	Michele Correggi,
	Simon Larson, 
	Phan Th\`anh Nam,
	Fabian Portmann, 
	Viktor Qvarfordt,
	Nicolas Rougerie,
	Robert Seiringer and
	Jan Philip Solovej.
	I have very much enjoyed discussing the contents of these lecture notes with
	the participants of the course held at KTH in 2017, 
	who also greatly helped to improve the quality.
	Financial support from the Swedish Research Council, 
	grant no.\ {2013-4734}, is gratefully acknowledged.

% ------------------  SECTIONS  --------------------

\section{Some preliminaries and notation\lect{ [1,2]}}\label{sec:prelims}

	We assume that the reader is familiar with basic notions in real analysis
	and has already encountered for example Hilbert spaces, Lebesgue integrals, 
	as well as Fourier transforms.
	However, for convenience and for setting our notation we give here a very brief
	overview of these and a few other important concepts.
	See \cite{LieLos-01,ReeSim1,ReeSim2,Teschl-14,Thirring2} for more. 

\subsection{Hilbert spaces}\label{sec:prelims-Hilbert}

	We let $V$ denote a vector space of arbitrary dimension
	over the scalars $\F = \R$ or $\C$,
	and $z \mapsto \bar{z}$ complex conjugation.
	
\begin{definition}\label{def:forms-norms}
	A 
	\keyword{sesquilinear form} on $V$ is a map
	$V \times V \to \F$, $(u,v) \mapsto \langle u, v \rangle$
	such that for all $\alpha,\beta \in \F$, $u,v,w \in V$:
	\begin{enumerate}[label=\textup{(\roman*)}]
	\item $\inp{u,\alpha v + \beta w} = \alpha \inp{u,v} + \beta \inp{u,w}$ 
		\quad (linear in the \emph{second}\footnote{Note 
		that this convention varies in the literature;
		the one here is used in almost every physics text.} 
		argument),
	\item $\inp{\alpha u + \beta v,w} = \bar\alpha \inp{u,w} + \bar\beta \inp{v,w}$
		\quad (conjugate linear in the \emph{first} argument).
	\end{enumerate}
	A \keyword{hermitian form} on $V$ is a sesquilinear form $\inp{\slot,\slot}$
	satisfying, in addition:
	\begin{enumerate}[label=\textup{(\roman*)},resume]
	\item\label{itm:symmetry} $\inp{u,v} = \overline{\inp{v,u}}$ \quad (symmetry).
	\end{enumerate}
	An \keyword{inner product} or \keyword{scalar product} on $V$ is a hermitian 
	form $\inp{\slot,\slot}$ satisfying, in addition:
	\begin{enumerate}[label=\textup{(\roman*)},resume]
	\item $\inp{v,v} > 0$ for $v \neq 0$ \quad (positive definite).
	\end{enumerate}
	A 
	(sesqui-)\keyword{quadratic form} 
	on $V$ is a map $q\colon V \to \F$ such that for $\alpha \in \F$,
	$u,v \in V$:
	\begin{enumerate}[label=\textup{(\roman*)}]
	\item $q(\alpha v)= \bar\alpha \alpha q(v)$ \quad (scaling quadratically),
	\item $\inp{u,v}_q := \frac{1}{4}\bigl(q(u+v) - q(u-v) + iq(u-iv) - iq(u+iv)\bigr)$ 
		is sesquilinear in $u,v$.
	\end{enumerate}
	A \keyword{norm} on $V$ is a map
	$V \to \Rplus := [0,\infty)$, 
	$v \mapsto \|v\|$, 
	such that for all $\alpha \in \F$, $u,v \in V$:
	\begin{enumerate}[label=\textup{(\roman*)}]
	\item $\norm{\alpha v} = |\alpha| \norm{v}$ \quad (scaling linearly),
	\item $\norm{u+v} \le \norm{u} + \norm{v}$ \quad (triangle inequality),
	\item $\norm{v} = 0$ if and only if $v=0$ \quad (positive definite).
	\end{enumerate}
	The pair $(V,\norm{\slot})$ is called a \keyword{normed linear space},
	the pair $(V,q)$ a (sesqui-)\keyword{quadratic space},
	and the pair $(V,\inp{\slot,\slot})$ an \keyword{inner product space} 
	or a \keyword{pre-Hilbert space}.
\end{definition}
	
\begin{example}\label{exmp:Cn}
	The space $\C^n$ of $n$-tuples $\sz = (z_1,\ldots,z_n)$ 
	with the standard inner product 
	$\inp{\sz,\sw} = \sum_{j=1}^n \bar z_j w_j$
	and norm $\norm{\sz} = \sqrt{\inp{z,z}}$
	is a normed, quadratic and pre-Hilbert space.
\end{example}
	
\begin{example}\label{exmp:parallelogram-space}
	Any normed linear space $(V,\norm{\slot})$ satisfying the 
	\keyword{parallelogram identity}
	\begin{equation}\label{eq:parallelogram}
		\norm{u+v}^2 + \norm{u-v}^2 = 2\norm{u}^2 + 2\norm{v}^2
	\end{equation}
	is a quadratic space with $q(v) := \norm{v}^2$,
	and an inner product space with
	\begin{equation}\label{eq:polarization}
		\inp{u,v} := \inp{u,v}_q = \frac{1}{4}\left(
			\norm{u+v}^2 - \norm{u-v}^2 + i\norm{u-iv}^2 - i\norm{u+iv}^2
			\right)
	\end{equation}
	(this implication is known as the Jordan--von Neumann theorem).
	Conversely, any inner product space is also a normed space, as follows:
\end{example}

\begin{proposition}[\keyword{Cauchy--Schwarz inequality}]\label{prop:CS}
	Let $V$ be an inner product space. Then for every $u,v \in V$ we have
	\begin{equation}\label{eq:CS}
		|\langle u,v \rangle| \le \norm{u}\norm{v}
	\end{equation}
	with equality iff $u$ and $v$ are parallel.
\end{proposition}

	Because of the Cauchy-Schwarz inequality, the square root of the induced 
	quadratic form $\norm{u} := \sqrt{\inp{u,u}}$
	satisfies the triangle inequality, by means of
	$$
		\norm{u+v}^2 = \norm{u}^2 + 2\re \inp{u,v} + \norm{v}^2 
		\le (\norm{u}+\norm{v})^2,
	$$
	and hence becomes a norm on $V$.
	The norm induces a metric $d(u,v) := \norm{u-v}$ 
	and therefore a topology on $V$,
	with the open sets generated by balls defined using the metric,
	$$
		B_r(x) := \{ u \in V : d(x,u) < r \}.
	$$
	Recall that a \keyword{Cauchy sequence} is a sequence 
	$(v_n)_{n=1}^\infty \subset V$ such that
	$$
		\forall \eps > 0 \ \exists N \in \N : n,m > N \Rightarrow \norm{v_n-v_m} < \eps,
	$$
	and a topological space is called \keyword{complete} if every Cauchy 
	sequence converges.
	Also recall that a topological space is called \keyword{separable}
	if it contains a countable dense subset.
	
\begin{definition}
	A complete normed linear space is called a \keyword{Banach space}.
	A complete inner product (i.e. pre-Hilbert) space is called a 
	\keyword{Hilbert space}.
\end{definition}

\begin{exc}
	Prove \eqref{eq:CS} and Proposition~\ref{prop:CS},
	for example by considering the expression
	$\inp{u-\alpha v,u-\alpha v}$ with 
	$\alpha = \langle v,u\rangle/\langle v,v \rangle$.
\end{exc}

\begin{exc}\label{exc:form-properties}
	Let $V$ be a vector space, $s(u,v)$ a sesquilinear form on $V$,
	and $q(v) = s(v,v)$ the associated quadratic form.
	Prove that it satisfies the parallelogram identity
	$$
		q(u+v) + q(u-v) = 2q(u) + 2q(v)
	$$
	and the \keyword{polarization identity}
	$s(u,v) = \inp{u,v}_q$
	for all $u,v \in V$. 
	Show that $s(u,v)$ is hermitian if and only if $q$ is real-valued,
	i.e. $q\colon V \to \R$,
	and that if $q$ is non-negative, 
	i.e. $q\colon V \to \Rplus$,
	then it also satisfies the
	Cauchy-Schwarz inequality
	$|s(u,v)| \le q(u)^{1/2} q(v)^{1/2}$.
\end{exc}

\begin{exc}\label{exc:Bessels-inequality}
	Let $V$ be an inner product space and $\{u_j\}_{j=1}^n$ an 
	\keyword{orthonormal} set in $V$,
	i.e. $\inp{u_j,u_k} = \delta_{jk}$.
	Prove \keyword{Bessel's inequality}
	\begin{equation}\label{eq:Bessel}
		\sum_{j=1}^n |\langle u_j,v \rangle|^2 \le \|v\|^2
	\end{equation}
	for all $v \in V$, with equality iff $v \in \Span\{u_j\}_{j=1}^n$.
\end{exc}

\subsection{Lebesgue spaces}\label{sec:prelims-Lp}

	The typical example of a Hilbert space encountered in quantum mechanics is
	either the finite-dimensional space $\C^n$, or the infinite-dimensional
	Lebesgue space of square-integrable functions $L^2(\Omega)$. 
	Recall that for $\Omega \subseteq \R^n$ 
	(which could be replaced by some measure space $(X,\mu)$ in general)
	and for a measurable function $f\colon \Omega \to \F$,
	we define the \keyword{$L^p$-norms} as
	$$
		\norm{f}_{L^p(\Omega;\F)} := \left( \int_\Omega |f(x)|^p \,dx \right)^{1/p},
		\quad 1 \le p < \infty,
	$$
	and
	$$
		\norm{f}_{L^\infty(\Omega;\F)} := \esssup_{x \in \Omega} |f(x)| 
		:= \inf \bigl\{ K \in [0,\infty] : \text{$|f(x)| \le K$ for a.e. $x \in \Omega$} \bigr\}.
	$$
	Then the \keyword{Lebesgue spaces} are defined as
	$$
		L^p(\Omega;\F) := \bigl\{ f\colon \Omega \to \F \ : \ 
			\text{$f$ is measurable and $\norm{f}_{L^p(\Omega;\F)} < \infty$} 
			\bigr\}.
	$$
	In the above $(\F,|\slot|)$ can be taken to be any finite-dimensional 
	normed space (algebra), 
	however the typical case is $\F=\C$ for which we simply write
	$L^p(\Omega) := L^p(\Omega;\C)$.
	If the domain $\Omega$ (or space $(X,\mu)$) is also understood from context
	we could write simply $\norm{f}_{L^p} := \norm{f}_{L^p(\Omega)}$.
	We follow the standard convention that we identify two functions 
	$f=g$ iff $f(x) = g(x)$ for a.e. $x \in \Omega$.
	It is a classical fact that $L^p(\Omega)$ forms a Banach space 
	for any $p \in [1,\infty]$ and that $L^2(\Omega)$ is a separable
	Hilbert space with the standard inner product
	$$
		\inp{f,g} := \int_\Omega \overline{f(x)} g(x) \,dx.
	$$
	
\begin{proposition}[\keyword{H\"older's inequality}]\label{prop:Hoelder}
	Let $1 \le p,q,r \le \infty$. For $L^p = L^p(\Omega;\F)$ we then have%
	\footnote{We use the convention $1/\infty = 0$.}
	\begin{equation}\label{eq:Hoelder}
		\norm{fg}_{L^r} \le \norm{f}_{L^p} \norm{g}_{L^q}
		\quad \text{for} \quad \frac{1}{r} = \frac{1}{p} + \frac{1}{q},
	\end{equation}
	and for all $f \in L^p$, $g \in L^q$. 
	In particular, with $p=q=2$ and $r=1$,
	\begin{equation}\label{eq:CS-L2}
		\norm{fg}_{L^1} \le \norm{f}_{L^2} \norm{g}_{L^2}
	\end{equation}
	is the Cauchy-Schwarz inequality in $L^2$.
	Moreover, if $p,q \in (1,\infty)$ then equality holds in~\eqref{eq:Hoelder} 
	if and only if $|f|^p$ and $|g|^q$ are linearly dependent in $L^1$.
\end{proposition}
	
	There is a trick to remember the precise form of, 
	or to check the validity of, 
	inequalities or identities of the type \eqref{eq:Hoelder}. Namely,
	first note how it depends upon rescaling $f$ or $g$ with a positive
	number $\lambda>0$, i.e. linearly on both sides of the inequality, 
	since by the property of the norms
	$\norm{\lambda f}_{L^p} = \lambda \norm{f}_{L^p}$.
	Secondly, one should note how the expressions behave upon rescaling
	the argument of the functions $f$ and $g$ by a number $\mu > 0$,
	i.e. $f_\mu(x) := f(x/\mu)$,
	$$
		\norm{f_\mu}_{L^p} = \left( \int_\Omega f(x/\mu) \,dx \right)^{1/p}
		= \mu^{n/p} \norm{f}_{L^p},
	$$
	if $\Omega = \mu\Omega = \R^n$.
	Also in the case that $\Omega \subsetneq \R^n$ one may note that if
	$f$ is dimensionless then the norm $\norm{f_\mu}_{L^p}$ has the dimension
	of a volume in $\R^n$ to the power $1/p$, i.e. $n/p$.
	We then find that the l.h.s. of \eqref{eq:Hoelder} scales as $\mu^{n/r}$,
	but also the r.h.s. scales as $\mu^{n/p} \mu^{n/q} = \mu^{n/r}$.
	These two scaling principles must always be obeyed and
	can be used to check similar expressions.
	
	An application of the H\"older inequality~\eqref{eq:Hoelder}
	proves the triangle inequality on $L^p$:

\begin{proposition}[\keyword{Minkowski's inequality}]
	Let $1 \le p \le \infty$. Then
	$$
		\norm{f + g}_{L^p} \le \norm{f}_{L^p} + \norm{g}_{L^p},
	$$
	for all $f,g \in L^p$. 
\end{proposition}

	In the case that $\Omega$ is noncompact it is useful to define,
	given any function space $\scF(\Omega;\F)$ (such as $\scF=L^p$),
	the \emph{local} function space
	$$
		\scF_\loc(\Omega;\F) := \bigl\{f\colon \Omega \to \F \ : \ 
			\text{$\varphi f \in \scF(\Omega;\F)$ for any $\varphi \in C^\infty_c(\Omega;\F)$} 
			\bigr\},
	$$
	where $C_c^\infty(\Omega;\F)$ denotes the space of smooth functions 
	with compact support on $\Omega^\circ$.
	
\begin{exc}
	Prove H\"older's inequality, for example by
	\begin{enumerate}[label=\arabic*.]
	\item first reducing~\eqref{eq:Hoelder} to the case 
		$\|fg\|_1 \le \|f\|_p \|g\|_q$,
		$1 = 1/p + 1/q$, i.e. $r=1$,
	\item proving \keyword{Young's inequality}
		\begin{equation}\label{eq:Young}
			ab \le \frac{a^p}{p} + \frac{b^q}{q}
		\end{equation}
		for such $p,q$ and $a,b \ge 0$, with equality iff $a^p = b^q$,
	\item using this to prove H\"older's inequality with $r=1$ by first 
		normalizing $f$ and $g$.
	\end{enumerate}
\end{exc}

\begin{exc}
	Prove Minkowski's inequality, f.ex.\ by writing $|f+g|^p = |f+g||f+g|^{p-1}$
	and then using the triangle and H\"older inequalities.
\end{exc}

\subsubsection{Convergence properties}\label{sec:prelims-Lp-convergence}

	One has the following extremely useful convergence properties 
	of the Lebesgue integral:
	
\begin{theorem} 
	\label{thm:Lebconv}
	Let $(X,\mu)$ be a measure space with positive measure $\mu$,
	and let $(f_n)_{n=1}^\infty$ be a sequence of measurable functions that 
	converges pointwise a.e. on $X$ to a function $f$.
	\begin{enumerate}[label=\textup{(\roman*)}]
	\item\label{itm:Lebconv-mono} \keyword{Lebesgue monotone convergence theorem}.\\
		Suppose that $0 \le f_n(x) \le f_{n+1}(x)$ for a.e. $x \in X$.
		Then $f$ is measurable and 
		$\lim_{n \to \infty} \int_X f_n \,d\mu = \int_X f \,d\mu$.
	\item\label{itm:Lebconv-dom} \keyword{Lebesgue dominated convergence theorem}.\\
		Suppose that there exists $F \in L^1(X,d\mu)$ such that
		$|f_n(x)| \le |F(x)|$ for a.e. $x \in X$.
		Then $f \in L^1(X,d\mu)$ and 
		$\lim_{n \to \infty} \int_X |f_n-f| \,d\mu = 0$.
	\item\label{itm:Fubini} \keyword{Fubini's theorem}.\\
		Suppose that $X=X_1 \times X_2$, where $(X_j,\mu_j)$, $j=1,2$
		are two $\sigma$-finite measure spaces, 
		and that $f$ is a measurable function on $X$.
		If $f \ge 0$ then the following three integrals are equal:
		$$
			\int_{X_1 \times X_2} f(x,y) \,(\mu_1 \times \mu_2)(dx\,dy),
		$$
		$$
			\int_{X_1} \left( \int_{X_2} f(x,y) \,\mu_2(dy) \right) \mu_1(dx),
		$$
		$$
			\int_{X_2} \left( \int_{X_1} f(x,y) \,\mu_1(dx) \right) \mu_2(dy).
		$$
		If $f\colon X \to \C$ then the above holds if one assumes in addition
		that
		$$
			\int_{X_1 \times X_2} |f(x,y)| \,(\mu_1 \times \mu_2)(dx\,dy) < \infty.
		$$
	\end{enumerate}
\end{theorem}
	
\subsubsection{The layer-cake representation}\label{sec:prelims-Lp-layer}

	Let $(X,\mu)$ be a measure space and $f\colon X \to \C$ measurable.
	Define the function $\lambda_f\colon \Rplus \to [0,\infty]$,
	$$
		\lambda_f(t) := \mu( \{x \in X : |f(x)| > t\} ),
	$$
	and its formal differential 
	(the sign because $\lambda_f$ is decreasing)
	$$
		d\lambda_f(t) := -\mu( \{x \in X : |f(x)| = t\} ).
	$$ 
	These allow to express the properties of $f$ in terms of 
	layers of its graph or its level sets,
	namely, using that 
	$$
		|f(x)|^p = \int_{t=0}^\infty \1_{\{|f(x)| > t\}} \,d(t^p),
	$$
	with $d(t^p) = pt^{p-1}dt$,
	one has the \keyword{layer-cake representation}
	\begin{equation}\label{eq:layer-cake-p}
		\norm{f}_p^p = \int_{t=0}^\infty \lambda_f(t) \,d(t^p)
		= \int_{t=\infty}^0 t^p \,d\lambda_f(t),
	\end{equation}
	where the second identity is a formal partial integration.
	Also,
	\begin{equation}\label{eq:layer-cake-infty}
		\norm{f}_\infty = \inf \{ t \ge 0 : \lambda_f(t) = 0 \}.
	\end{equation}
	An immediate application of \eqref{eq:layer-cake-p} is 
	\keyword{Chebyshev's inequality}
	\begin{equation}\label{eq:Chebyshev}
		\norm{f}_p^p \ge t^p \lambda_f(t),
		\qquad \forall t \ge 0.
	\end{equation}

\subsection{Fourier transform}\label{sec:prelims-Fourier}

	Given a measurable function $f\colon \R^n \to \C$,
	we formally define its \keyword{Fourier transform} $\hf$ by
	\begin{equation}\label{eq:Fourier-def}
		\hf(\xi) := (\cF f)(\xi) 
		:= (2\pi)^{-n/2} \int_{\R^n} f(x) e^{-i\xi \cdot x} dx.
	\end{equation}
	It can be shown that $\cF$ is a bijective (and obviously linear)
	map from $L^2(\R^n)$ into itself, with its inverse given by
	\begin{equation}\label{eq:Fourier-inverse}
		f(x) = (\cF^{-1} \hf)(x) 
		:= (2\pi)^{-n/2} \int_{\R^n} \hf(\xi) e^{i\xi \cdot x} d\xi.
	\end{equation}
	Moreover, \keyword{Plancherel's identity}
	\begin{equation}\label{eq:Plancherel}
		\norm{\cF f}_{L^2(\R^n)} = \norm{f}_{L^2(\R^n)},
		\qquad f \in L^2(\R^n),
	\end{equation}
	shows that $\cF$ is a \keyword{unitary} map on $L^2(\R^n)$, 
	i.e. its adjoint satisfies $\cF^* = \cF^{-1}$.
	
	We also have the important relations between differentiation
	and multiplication
	\begin{equation}\label{eq:Fourier-derivatives}
		\left( \frac{\partial f}{\partial x_j} \right)^\wedge(\xi) 
			= i\xi_j \hf(\xi), 
			\qquad
		\left( x_j f \right)^\wedge(\xi) 
			= i\frac{\partial \hf}{\partial \xi_j}(\xi).
	\end{equation}

\subsection{Sobolev spaces}\label{sec:prelims-Hk}

	For any $s \ge 0$, we define the \keyword{Sobolev space} 
	to be the space of square-integrable functions having square-integrable
	derivatives to order $s$, using the Fourier transform on $\R^n$ 
	and~\eqref{eq:Fourier-derivatives} as
	\begin{equation}\label{eq:Sobolev-def}
		H^s(\R^n) := \bigl\{ f \in L^2(\R^n) : |\xi|^s \hf(\xi) \in L^2(\R^n) \bigr\}.
	\end{equation}
	The most common case is that $s=k$ is a non-negative integer, 
	$H^k(\R^n)$, while otherwise the space is called a 
	\emph{fractional} Sobolev space.
	We note that $H^s(\R^n)$ is a Hilbert space with the inner product
	\begin{equation}\label{eq:Sobolev-innerprod}
		\inp{f,g}_{H^s(\R^n)} := \int_{\R^n} 
			\overline{\hf(\xi)} \hat{g}(\xi) 
			(1 + |\xi|^2)^s d\xi.
	\end{equation}
	In particular, given a subset $\Omega \subseteq \R^n$, one may consider
	the subspace $C_c^\infty(\Omega) \subseteq H^s(\R^n)$ 
	of smooth and compactly supported functions on $\Omega$ 
	and take its closure in $H^s(\R^n)$ 
	(with respect to the norm induced from the above inner product).
	We denote this space
	\begin{equation}\label{eq:Sobolev0-def}
		H_0^s(\Omega) := \bigl\{ f \in H^s(\R^n) \ : \ 
			\exists (f_n) \subset C_c^\infty(\Omega), \ 
			\norm{f_n - f}_{H^s(\R^n)} \to 0 \bigr\},
	\end{equation}
	and it has the interpretation as the space of functions which
	vanish `sufficiently fast' at the boundary $\partial\Omega$
	(and are identically zero on $\Omega^c$).
	One has $H_0^s(\R^n) = H^s(\R^n)$.
	
	It is important to know that it is possible to define Sobolev spaces
	locally, i.e. on domains $\Omega \subseteq \R^n$,
	without the use of the Fourier transform, 
	however we will not discuss this properly here since it is most naturally done
	using the theory of distributions which goes beyond the course prerequisites.
	We only mention that, given 
	$u \in L^1_\loc(\Omega) \supseteq L^{p \ge 1}_\loc(\Omega)$ 
	it is always possible to define its gradient $\nabla u$ as a generalized function, 
	and in the case that this turns out to be a
	locally integrable function, i.e. $\nabla u \in L^1_\loc(\Omega;\C^n)$,
	we say that $u$ is \keyword{weakly differentiable}
	(it has weak partial derivatives $\partial_j u \in L^1_\loc(\Omega)$, $j=1,\ldots,n$).
	It may even turn out that $\nabla u \in L^q(\Omega;\C^n)$ for some $q \ge 1$, 
	and we define the Sobolev spaces
	$$
		H^1(\Omega) := \bigl\{ u \in L^2(\Omega) :
			\text{$u$ is weakly differentiable and $\nabla u \in L^2(\Omega;\C^n)$}
			\bigr\} 
	$$
	and, by iteration of this procedure to higher orders of derivatives,
	for $k \in \N$
	\begin{multline*}
		H^k(\Omega) := \bigl\{ u \in L^2(\Omega) :
			\text{$u$ has weak partial derivatives all the way up to order $k$, and} \\
			\text{all of which are in $L^2(\Omega)$}
			\bigr\}.
	\end{multline*}
	The inner product in this space is given by
	$$
		\inp{f,g}_{H^k(\Omega)} = \int_\Omega \left( \overline{f(x)}g(x) 
			+ \sum_\alpha \overline{\partial_\alpha f(x)} \partial_\alpha g(x) \right) dx,
	$$
	where the sum runs over all partial derivatives ($\alpha$ a multi-index)
	up to order $k$.
	
	In the case $\Omega = \R^n$ these spaces turn out to coincide with the above
	definitions \eqref{eq:Sobolev-def} and \eqref{eq:Sobolev0-def}.
	When $\Omega \subsetneq \R^n$ they 
	contain but may (or may not) differ
	from $H_0^k(\Omega)$, depending on the geometry of $\Omega$
	(we will return to this question in Section~\ref{sec:uncert-Hardy}),
	and will be associated to Neumann respectively Dirichlet boundary 
	conditions for the Laplace operator.
	
\begin{exc}
	Show that the norm in $H^k(\R^n)$ 
	is equivalent to that
	defined by replacing the weight factor $w(\xi) = (1+|\xi|^2)^k$ 
	in \eqref{eq:Sobolev-innerprod}
	by $w_1(\xi) = 1 + |\xi|^{2k}$ or
	$w_2(\xi) = 1 + \sum_{j=1}^n |\xi_j|^{2k}$.
\end{exc}

\subsection{Forms and operators}\label{sec:prelims-ops}

	An \keyword{operator} on a Hilbert space $\cH$ is a linear map
	$A\colon \cD(A) \to \cH$,
	with 
	the subspace $\cD(A) \subseteq \cH$ 
	the \keyword{domain} of $A$.
	Let $\norm{A} := \sup_{u \in \cD(A) : \norm{u}=1} \norm{Au}$, 
	then $A$ is called \keyword[bounded operator]{bounded} if $\norm{A}<\infty$ 
	and \keyword[unbounded operator]{unbounded} otherwise,
	and 
	\keyword[closed operator]{closed} if its graph 
	\begin{equation}\label{eq:graph-def}
		\Gamma(A) := \{ (u,v) \in \cD(A) \times \cH : v = Au \}
	\end{equation}
	is a closed subspace of $\cH \times \cH$. 
	We will always work with densely defined operators, 
	i.e. with $\cD(A)$ dense in $\cH$, 
	and shall denote by $\cL(\cH)$ 
	the space of such operators on $\cH$.
	We have a subspace $\cB(\cH) \subseteq \cL(\cH)$ 
	of bounded operators $A$, for which one may assume $\cD(A) = \cH$.
	
	An operator $A \in \cL(\cH)$ is \keyword{hermitian} (or \keyword{symmetric})
	if its corresponding sesquilinear form
	\begin{equation}\label{eq:operator-form}
		\cD(A) \times \cD(A) \to \C, \qquad (u,v) \mapsto \inp{u, Av} 
	\end{equation}
	satisfies $\inp{u,Av} = \inp{Au,v}$ for all $u,v \in \cD(A)$
	(compare Definition~\ref{def:forms-norms}.\ref{itm:symmetry}).
	Hermitian operators are always \keyword{closable}, 
	i.e. there is an \keyword{extension}
	$\hat{A}\colon \cD(\hat{A}) \to \cH$, with $\cD(\hat{A}) \supseteq \cD(A)$
	and $\hat{A}|_{\cD(A)} = A$, which is closed.
	Every closable operator $A$ has a smallest closed extension,
	its \keyword{closure} $\bar{A}$, 
	and every hermitian operator has a largest closed
	extension, namely its \keyword{adjoint} $A^*$, which is in general defined 
	with the domain
	\begin{equation}\label{eq:adjoint-domain-def}
		\cD(A^*) := \left\{ u \in \cH \ : 
			\sup_{v \in \cD(A) : \norm{v}=1} |\langle u,Av \rangle|
			< \infty \right\}
	\end{equation}
	and (via the Riesz lemma) the formula
	$\inp{u,Av} = \inp{A^*u,v}$ for all $u \in \cD(A^*)$, $v \in \cD(A)$.
	The operator $A$ is called \keyword{self-adjoint} if $A^* = A$, i.e. if
	$\cD(A^*) = \cD(A)$ and it is hermitian, and 
	\keyword{essentially self-adjoint}
	if it is hermitian and has a unique self-adjoint extension
	$\bar{A} = A^*$.
	Any bounded hermitian operator is self-adjoint.
	
\begin{example}\label{exmp:Laplace}\index{Laplacian}
	The Laplace operator $L\colon u \mapsto -u''$ on the interval $[0,1]$ 
	is unbounded and not defined as an operator on the full Hilbert space 
	$L^2([0,1])$, 
	but its restriction $L|_{C_c^\infty([0,1])}$ to the smooth functions
	with compact support 
	(usually referred to as the \keyword{minimal domain})
	is perfectly well-defined and hermitian.
	Its closure $\overline{L|_{C_c^\infty([0,1])}} = L|_{H^2_0([0,1])}$
	is not self-adjoint, $L|_{H^2_0([0,1])}^* = L|_{H^2([0,1])}$, however
	it has several self-adjoint extensions, such as the
	\keyword{Dirichlet Laplacian} 
	$-\Delta^\eD := L|_{\cD(-\Delta^\eD)}$, with
	$$
		\cD(-\Delta^\eD) := \{ u \in H^{1}_0([0,1]) : u' \in H^{1}([0,1]) \} = H^{1}_0([0,1]) \cap H^{2}([0,1]),
	$$
	and the
	\keyword{Neumann Laplacian} 
	$-\Delta^\eN := L|_{\cD(-\Delta^\eN)}$, with
	$$
		\cD(-\Delta^\eN) := \{ u \in H^{1}([0,1]) : u' \in H^{1}_0([0,1]) \} \subsetneq H^{2}([0,1]).\phantom{ \cap H^{2}([0,1]), }
	$$
	In contrast, when considering $L$ on the full real line,
	$L|_{C_c^\infty(\mathbb{R})}$ is
	essentially self-adjoint in $L^2(\mathbb{R})$, with
	$\overline{L|_{C_c^\infty(\mathbb{R})}} = L|_{H^{2}(\mathbb{R})}$.
	This follows from the properties of Sobolev spaces,
	namely $H^2_0(\R^2) = H^2(\R^2)$.
\end{example}
	
	The form \eqref{eq:operator-form} of an operator $A$
	usually extends to a larger dense subspace 
	$\cQ(A) \subseteq \cH$ called a \keyword{form domain} of $A$.
	We extend the notion of quadratic form to the unbounded case,
	just as in Definition~\ref{def:forms-norms} but with $V$ replaced
	by a dense subspace $\cD(q) \subseteq \cH$,
	$q\colon \cD(q) \to \C$.
	In particular, it is \keyword{hermitian} if $q\colon \cD(q) \to \R$, 
	\keyword{non-negative} if $q\colon \cD(q) \to \Rplus$, 
	and \keyword{strictly positive} if $q(u) > 0$ for all $u \in \cD(q) \setminus \{0\}$.
	We can also compare two hermitian quadratic forms $q$ and $q'$ 
	if their domains intersect, 
	and say that $q \ge q'$ iff $\cD(q) \subseteq \cD(q')$
	and $q(u) \ge q'(u)$ for all $u \in \cD(q)$.
	Thus the form $q$ is \keyword{semi-bounded from below} if $q \ge c$ for some
	constant $c \in \R$, i.e. $q(u) \ge c\norm{u}^2$ for all $u \in \cD(q)$.
	Note that such a form can always be converted to a non-negative one
	by adding to it the constant $-c$.
	Furthermore, for semi-bounded forms we usually distinguish the domain by
	writing $q(u) = +\infty$ iff $u \notin \cD(q)$.
	
	A non-negative (or semi-bounded with the above trick) quadratic form 
	is said to be \keyword[closed form]{closed} iff $\cD(q)$ is complete in the form norm
	$\norm{u}_{q+1}^2 := q(u) + \norm{u}^2$.
	The following then gives a very useful correspondence between forms and
	operators.
	
\begin{theorem}[see e.g. {\cite[Theorem~2.5.18]{Thirring2}}]
	\label{thm:form-operator}
	If the quadratic form $q\colon \cD(q) \to \Rplus$ is non-negative and closed,
	then it is the form $q=q_A$, with $q_A(u) := \inp{u,Au}$,
	$\cQ(A) := \cD(q)$,
	of a unique self-adjoint, non-negative operator $A$.
\end{theorem}

	It follows that for semi-bounded hermitian operators there is always a unique
	self-adjoint extension associated to the form~\eqref{eq:operator-form},
	called the \keyword{Friedrichs extension}.

\begin{theorem}[Friedrichs extension; see e.g. {\cite[Theorem~X.23]{ReeSim2}}]
	\label{thm:Friedrichs}
	For any semi-boun\-ded from below hermitian operator $A$, the quadratic
	form $q_A(u) := \inp{u,Au}$ with $\cD(q_A) = \cD(A)$ is closable
	and its closure $\overline{q_A} = q_{\hat{A}_\sF}$ is the quadratic form
	of a unique semi-bounded self-adjoint operator $\hat{A}_\sF$.
	Furthermore, $\hat{A}_\sF$ is the only self-adjoint extension $\hat{A}$
	of $A$ s.t. $\cD(\hat{A}) \subseteq \cD(\overline{q_A})$ and the largest 
	($\hat{A}_\sF \ge \hat{A}$, i.e. $\cD(\hat{A}_F) \subseteq \cD(\hat{A})$)
	among all semi-bounded self-adjoint extensions of $A$.
\end{theorem}
	
	Note that we compare semi-bounded hermitian operators in terms of their
	quadratic forms:
	$A \ge B$ iff $\cD(A) \subseteq \cD(B)$ and $q_A \ge q_B$.
	
\begin{example}\label{exmp:Laplace-form}\index{Laplacian}
	Associated to the Laplace operator $L$ on $[0,1]$ in Example~\ref{exmp:Laplace} 
	is the quadratic form $q_L(u) = \int_0^1 |u'|^2 \thinspace dx$. 
	The Dirichlet Laplacian $-\Delta^\eD$ is 
	is the Friedrichs extension w.r.t. the non-negative form 
	$q_L|_{C_c^\infty([0,1])}$,
	with the resulting form domain $\cQ(-\Delta^\eD) = H_0^1([0,1])$,
	while the Neumann Laplacian $-\Delta^\eN$ 
	is the Friedrichs extension w.r.t. $q_L|_{C^\infty([0,1])}$,
	with resulting form domain $\cQ(-\Delta^\eN) = H^1([0,1])$.
	We have $-\Delta^\eD \ge -\Delta^\eN$ since 
	$C_c^\infty \subseteq C^\infty$.
\end{example}
	
	Although necessary for some parts of the course, we will for
	simplicity try to avoid operator theory as much as possible 
	and will typically be working with forms and form domains 
	instead of operators and operator domains.
	So for example if $u \in H^1(\Omega)$ and we write
	$$
		\inp{u, (-\Delta)u}_{L^2(\Omega)} 
	$$
	then since $-\Delta = \nabla^* \nabla$, 
	here with $\cQ(-\Delta) = \cD(\nabla) = H^1(\Omega)$ understood, 
	we actually mean
	$$
		\inp{\nabla u, \nabla u}_{L^2(\Omega)} 
		= \int_\Omega |\nabla u|^2 < \infty.
	$$

\begin{exc}
	Verify the statements in Example~\ref{exmp:Laplace}
	(as completely as you can with the preliminaries at hand;
	you may e.g.\ use that any $u \in H^1([0,1])$ is a continuous function
	on $(0,1)$).
\end{exc}

\subsubsection{The spectral theorem}\label{sec:prelims-ops-specthm}

	The \keyword{spectrum} $\sigma(A)$ of an operator $A \in \cL(\cH)$ 
	is the set of points $\lambda \in \C$ 
	for which there does not exist a bounded inverse
	to the operator $A - \lambda = A - \lambda\1$.
	For $\lambda \notin \sigma(A)$ we call such $(A - \lambda)^{-1}$
	the \keyword{resolvent} of $A$ at $\lambda$.
	An operator $A$ is self-adjoint if and only if it is hermitian and 
	$\sigma(A) \subseteq \R$.
	
	To fully describe a self-adjoint operator in terms of its spectrum,
	one needs to have also a notion of spectral measure.
	A \keyword{projection-valued measure} $P$ is a 
	function $\Omega \mapsto P_\Omega$ on the Borel\footnote{The
	$\sigma$-algebra of \keyword{Borel sets} on a topological space 
	is the smallest $\sigma$-algebra containing all open sets.
	A real-valued function $f$ is a \keyword{Borel function} if $f^{-1}((a,b))$
	is a Borel set for any interval $(a,b)$.}
	sets $\Omega \subseteq \mathbb{R}$
	such that each $P_\Omega$ 
	is a projection on a Hilbert space $\cH$, 
	$P_\emptyset = 0$, $P_\R = \1$,
	$P_{\Omega_1}P_{\Omega_2} = P_{\Omega_1 \cap \Omega_2}$,
	and if $\Omega = \cup_{n=1}^{\infty} \Omega_n$,
	with $\Omega_n$ and $\Omega_m$ disjoint for $n \neq m$,
	then $P_\Omega = \textup{s-lim}_{N \to \infty} \sum_{n=1}^N P_{\Omega_n}$.
	Given any $u,v \in \cH$, we then have a complex measure
	$\Omega \mapsto \inp{u,P_\Omega v}$ on the real line.
	
\begin{theorem}[\keyword{Spectral theorem}; 
		see e.g. {\cite[Theorem VIII.6]{ReeSim1}}
		or {\cite[Section~3.1]{Teschl-14}}]
	\label{thm:spectral-theorem} 
	There is a one-to-one correspondence between 
	projection-valued measures $P^A$ 
	and self-adjoint operators 
	$A = \int_{-\infty}^{\infty} \lambda \thinspace dP^A(\lambda)$, with
	$$
		\cD(A) := \left\{ u \in \cH : 
		\int_{-\infty}^{\infty} |\lambda|^2 \thinspace d(\inp{ u, P^A u })(\lambda) 
		< \infty \right\}
	$$
	and
	$$
		\inp{ u, Av } 
		= \int_{-\infty}^{\infty} \lambda \inp{ u, dP^A(\lambda) v }
		:= \int_{-\infty}^{\infty} \lambda \, d(\inp{ u, P^A v })(\lambda).
	$$
	Furthermore, 
	$$
		\cQ(A) = \left\{ u \in \cH : 
		\int_{-\infty}^{\infty} |\lambda| \thinspace d(\inp{ u, P^A u })(\lambda) 
		< \infty \right\} = \cD(\sqrt{|A|})
	$$
	and if $f\colon \R \to \R$ is a Borel function then
	we can define the operator
	$$
		f(A) 
		= \int_{-\infty}^{\infty} \lambda \, dP^{f(A)}(\lambda)
		:= \int_{-\infty}^{\infty} f(\lambda) \, dP^A(\lambda).
	$$
\end{theorem}
	
	One may decompose the spectrum of a self-adjoint operator $A$
	into either
	$$
		\sigma(A) = \sigma_\textup{disc}(A) \sqcup \sigma_\textup{ess}(A)
		\qquad \textrm{or} \qquad
		\sigma(A) = \overline{\sigma_\textup{pp}(A)} \cup \sigma_\textup{ac}(A) \cup \sigma_\textup{sc}(A)
	$$
	(the latter sets may overlap),
	where the \keyword{discrete spectrum}
	$$
		\sigma_\textup{disc}(A) := \{ \lambda \in \sigma(A) : 
		\dim P_{(\lambda - \epsilon, \lambda + \epsilon)}^A \mathcal{H} < \infty \ 
		\textrm{for some $\epsilon > 0$} \}
	$$
	is the set of isolated eigenvalues of finite multiplicity
	and the rest is the \keyword{essential spectrum}
	$$
		\sigma_\textup{ess}(A) := \{ \lambda \in \sigma(A) : 
		\dim P_{(\lambda - \epsilon, \lambda + \epsilon)}^A \mathcal{H} = \infty \ 
		\textrm{for every $\epsilon > 0$} \}
	$$
	(usually referred to as the continuous part of the spectrum),
	while the \keyword{pure point spectrum} $\sigma_\textup{pp}$ 
	is the set of all eigenvalues of $A$,
	and $\sigma_\textup{ac}$ resp. $\sigma_\textup{sc}$ are the
	supports for the \keyword[absolutely continuous spectrum]{absolutely} 
	resp. \keyword[singular continuous spectrum]{singular continuous}
	parts of the spectral measure of $A$.
	
\begin{example}\label{exmp:spectral-projection}
	An operator $A$ with eigenvalues $(\lambda_j)_{j=0}^\infty \subset \R$
	and a corresponding orthonormal basis of eigenfunctions
	$(u_j)_{j=0}^\infty \subset \cH$ can be written
	$A = \int_{-\infty}^\infty \lambda \,dP^A(\lambda)$ with
	projection-valued measure
	$P^A = \sum_{j=0}^\infty \delta_{\lambda_j} u_j\langle u_j, \slot \rangle$.
	Note that if $\lambda_1 = \lambda_2 = \ldots = \lambda_N$
	is a repeated eigenvalue then 
	$\sum_{j=1}^N \delta_{\lambda_j} u_j\langle u_j,\slot \rangle
	= \delta_{\lambda_1} P_W$,
	where $P_W$ is the orthogonal projection on the eigenspace
	$W = \Span \{u_j\}_{j=1}^N$.
\end{example}

\begin{example}\label{exmp:multiplication-operator}
	A \keyword{multiplication operator} on $L^2(\R)$ by $f\colon \R \to \R$,
	$(Au)(x) := f(x)u(x)$ has formally
	$$
		A = \int_{-\infty}^\infty f(\lambda) \,\delta_\lambda 
			\langle\delta_\lambda, \slot\rangle \,d\lambda,
		\qquad \text{i.e.} \qquad
		\inp{u,Av} = \int_{-\infty}^\infty f(\lambda) \,\overline{u(\lambda)} v(\lambda) \,d\lambda,
	$$
	and $\sigma(A) = \overline{f(\R)}$.
\end{example}

	The lowest eigenvalues 
	$\lambda_1(A) \le \lambda_2(A) \le \ldots < \inf \sigma_\textup{ess}(A)$
	(also ordered according to their multiplicity)
	of a semi-bounded from below self-adjoint\footnote{Or just hermitian,
	for which the spectrum analyzed is that of its Friedrichs extension.}
	operator $A$ can be obtained using the so-called 
	\keyword{min-max principle}:
	\begin{equation}\label{eq:min-max}
		\lambda_k(A) = 
		\inf_{ W_k } \ 
		\sup_{u \in W_k \setminus \{0\}} 
		\frac{\inp{u,Au}}{\norm{u}^2},
	\end{equation}
	where the $W_k \subseteq \cH$ are linear subspaces of $\cD(A)$
	such that $\dim W_k = k$.
	If $A \ge B$ then $\lambda_k(A) \ge \lambda_k(B)$ for each $k$,
	and $\inf \sigma_\textup{ess}(A) \ge \inf \sigma_\textup{ess}(B)$.
	The domain $\cD(A)$ can in the above 
	be replaced by the form domain $\cQ(A)$, 
	or even a dense subspace, such as typically $C_c^\infty$.
	
	Furthermore, 
	if one continues to evaluate \eqref{eq:min-max} for $k=1,2,\ldots$
	and eventually only repeated values
	$\lambda_n = \lambda_{n+1} = \lambda_{n+2} = \ldots$
	are obtained, 
	then one has reached the bottom of the essential spectrum,
	$\lambda_n = \inf \sigma_\textup{ess}(A)$.

\begin{example}\label{exmp:Laplace-eigenvalues}
	Computing \eqref{eq:min-max} with $A=L$ from Examples~\ref{exmp:Laplace}
	and \ref{exmp:Laplace-form}, i.e.
	$\inp{u,Lu} = q_L(u) = \int_0^1 |u'|^2$
	on the form domain $H_0^1([0,1])$ or $C_c^\infty([0,1])$ gives the
	eigenvalues $\lambda$ of the Dirichlet Laplacian,\index{Laplacian}
	$$
		-u''(x) = \lambda u(x), \qquad u(0) = u(1) = 0,
	$$
	i.e. $\lambda_n = \pi^2 n^2$, $n=1,2,3,\ldots$.
	On the other hand, taking the form domain $H^1([0,1])$ or $C^\infty([0,1])$
	gives those of the Neumann Laplacian,
	$$
		-u''(x) = \lambda u(x), \qquad u'(0) = u'(1) = 0,
	$$
	i.e. $\lambda_n = \pi^2 n^2$, $n=0,1,2,\ldots$.
\end{example}

\subsubsection{Stone's theorem}\label{sec:prelims-ops-Stone}

	Crucially for quantum mechanics, \keyword{Stone's theorem} tells us that
	self-adjoint operators are the generators of groups of unitary transformations.
	
\begin{theorem}[Stone's theorem; 
		see e.g. {\cite[Theorem VIII.7-8]{ReeSim1}}
		or {\cite[Section~5.1]{Teschl-14}}]
	\label{thm:Stones-theorem} 
	Let $A \in \cL(\cH)$ be a self-adjoint operator 
	and define $U(t) = e^{itA}$ (using the 
	spectral theorem).
	Then $U(t)$ is a \keyword{strongly continuous one-parameter unitary group},
	i.e.
	\begin{enumerate}[label=\textup{(\roman*)}]
	\item $U(t)$ is unitary for all $t \in \R$,
	\item $U(t+s) = U(t)U(s)$ for all $s,t \in \R$, and 
	\item if $\psi \in \cH$ and $t \to t_0$ then $(U(t)-U(t_0))\psi \to 0$.
	\end{enumerate}
	Furthermore,
	\begin{enumerate}[label=\textup{(\roman*)},resume]
	\item if $\lim_{t \to 0} (U(t)-\1)\psi/t$ exists, then $\psi \in \cD(A)$, and 
	\item for such $\psi$, $(U(t)-\1)\psi/t \to iA\psi$ as $t \to 0$.
	\end{enumerate}

	Conversely, if $U(t)$ is a strongly continuous one-parameter
	unitary group acting on $\cH$, 
	then there is a self-adjoint operator $A \in \cL(\cH)$ s.t. 
	$U(t) = e^{itA}$.
\end{theorem}
	
In fact, 
strong continuity may be replaced by just \emph{weak} continuity,
since weak convergence implies strong convergence in this case;
cf. \cite[Theorem~5.3]{Teschl-14} and \cite[Theorem~VIII.9]{ReeSim1}.

\section{A very brief mathematical formulation of classical and quantum mechanics\lect{ [3-5]}}\label{sec:mech}

	We do not assume familiarity with classical and quantum physics in this
	course, and therefore give a very brief account of the essentials here.
	However, a deeper understanding of these concepts is of course helpful
	in the broader perspective and we refer to \cite{Thirring1,Thirring2} 
	and \cite[Part~IV]{Thiemann-07} 
	for introductory material suitable for mathematicians.
	In Section~\ref{sec:mech-CM-instability} we discuss the question of stability
	of matter in classical mechanics,
	and in Section~\ref{sec:mech-QM-matter}
	we define what is meant with stability in quantum mechanics.
	Anyone who is already familiar with many-body quantum mechanics may
	safely skip the chapter except possibly for these parts.

\subsection{Some classical mechanics}\label{sec:mech-CM}
	\index{classical mechanics}

	Although it is important to know that there are several different 
	equivalent formulations of classical mechanics, with their own 
	advantages and disadvantages, we here choose to take the shortest
	mathematical path to quantum mechanics, via Poisson algebras 
	and Hamiltonian mechanics.
	
\subsubsection{Phase space and Poisson brackets}\label{sec:mech-CM-Poisson}

\begin{definition}[Poisson algebra]
	A \keyword{Poisson algebra} $\cA$ is 
	a vector space over $\F$ ($\R$ or $\C$) 
	equipped with an $\F$-bilinear and associative product
	$$
		\cA \times \cA \to \cA, \quad (f,g) \mapsto fg, 
	$$
	and an additional product (typically non-associative, 
	called a \keyword{Poisson bracket}) 
	$$
		\cA \times \cA \to \cA, \quad (f,g) \mapsto \{f,g\},
	$$
	satisfying, for all $\alpha,\beta \in \F$, and $f,g,h \in \cA$:
	\begin{enumerate}[label=\textup{(\roman*)}]
	\item\label{itm:linearity} 
		linearity: $\{f,\alpha g + \beta h\} = \alpha \{f,g\} + \beta \{f,h\}$,
	\item\label{itm:antisymmetry} 
		antisymmetry: $\{f,g\} = -\{g,f\}$,
	\item\label{itm:Jacobi} \index{Jacobi identity}
		Jacobi identity: $\{f,\{g,h\}\} + \{g,\{h,f\}\} + \{h,\{f,g\}\} = 0$,
	\item\label{itm:Leibniz}  \index{Leibniz rule}
		Leibniz rule: $\{f,gh\} = \{f,g\}h + g\{f,h\}$.
	\end{enumerate}
\end{definition}

\begin{remark}
	$\cA$ may also be equipped with a \keyword{unit} $1 \in \cA$ 
	s.t. $f = 1f = f1$. 
	Note that \ref{itm:linearity}+\ref{itm:antisymmetry} implies bilinearity,
	\ref{itm:linearity}+\ref{itm:antisymmetry}+\ref{itm:Jacobi} 
	means that the product $\{\slot,\slot\}$ is a \keyword{Lie bracket} 
	and $(\cA,\{\slot,\slot\})$ a \keyword{Lie algebra},
	and \ref{itm:Leibniz} that it acts as a derivation of the associative product.
\end{remark}
		
	The archetypical example of a Poisson algebra is the algebra 
	$\cA = C^\infty(\phP^n)$ of smooth 
	functions on the \keyword{classical $2n$-dimensional phase space}
	\begin{equation}\label{eq:standard-phase-space}
		\phP^n := \R^{2n} \ni (\sx,\ssp) = (x_1,\ldots,x_n,p_1,\ldots,p_n),
	\end{equation}
	endowed with the Poisson bracket
	\begin{equation}\label{eq:standard-Poisson}
		\{f, g\} := \sum_{j=1}^n \left( 
			\frac{\partial f}{\partial x_j} \frac{\partial g}{\partial p_j} -
			\frac{\partial g}{\partial x_j} \frac{\partial f}{\partial p_j}
			\right).
	\end{equation}
	The first half of the phase space,
	$$
		\phX^n := \R^n \ni \sx = (x_1,\ldots,x_n),
	$$
	is called the \keyword{classical configuration space} and is parameterized 
	by the \keyword{coordinates} or \keyword{position variables} $x_j$, 
	while the second half, parameterized by \keyword{momentum variables} $p_j$,
	is considered dual or \keyword{conjugate} to $\phX^n$ via the Poisson brackets.
	Namely, note that by \eqref{eq:standard-Poisson}, 
	the coordinates and momenta satisfy the following simple relations 
	called the \keyword{canonical Poisson brackets}:
	\begin{equation}\label{eq:classical-CCR}
		\{x_j, x_k\} = 0, \quad
		\{p_j, p_k\} = 0, \quad
		\{x_j, p_k\} = \delta_{jk} 1, \quad
		\forall j,k,
	\end{equation}
	where the constant function $1$ on $\phP^n$ is the unit in $\cA$.
	Hence the $x_j$'s and $p_k$'s Poisson-commute individually, 
	while $p_j$ is Poisson-conjugate to $x_j$ and vice versa
	(though note that there is a certain choice of orientation in the bracket,
	so the coordinates $x_j$'s should come first).
	
	In the case that we allow for complex-valued functions on phase space,
	we note that $\cA$ is also (non-trivially) endowed with the structure of a 
	\keyword{$*$-algebra} 
	in the sense that there is an operation $f \mapsto f^*$, 
	here given by complex-conjugation $f^*(\sx,\ssp) := \overline{f(\sx,\ssp)}$,
	satisfying:
	\begin{enumerate}[label=\textup{(\roman*)}]
	\item $(\alpha f + \beta g)^* = \bar\alpha f^* + \bar\beta g^*$ (conjugate linear)
	\item $(f^*)^* = f$ (involution\index{involution}),
	\item $(fg)^* = g^* f^*$ (antiautomorphism\index{antiautomorphism}).
	\end{enumerate}
	
\begin{remark*}
	The proper mathematical setting for classical mechanics in general 
	is to model the 
	Poisson algebra $\cA = C^\infty(\phP^n)$
	on a geometric object called a \keyword{symplectic manifold} 
	$(\phP^n,\omega)$, where $\omega$ is a symplectic form.
	Then the Poisson bracket is
	$\{f,g\} := \chi_f(g)$, where the vector field $\chi_f \in T(\phP^n)$ 
	is defined via the relation $\omega(\chi_f,\slot) = -df$.
	Typically, $\phP^n$ is defined as the cotangent bundle of a 
	configuration space manifold 
	$\phX^n = \cM$, $\dim \cM = n$,
	i.e. $\phP^n := T^*(\cM)$, with its canonical symplectic structure
	$\omega := d\theta$, $\theta[X](p) := p(\pi X)$,
	given locally by $\theta = \sum_{j=1}^n p_j dx_j$ and \eqref{eq:standard-Poisson}.
	See e.g. \cite{Thiemann-07,Nakahara-03}.
\end{remark*}

\begin{exc}
	Check that \eqref{eq:standard-Poisson} defines a Poisson bracket
	and makes $\cA = C^\infty(\phP^n)$ a Poisson algebra.
	Discuss whether \eqref{eq:classical-CCR} also defines this Poisson algebra
	completely.
\end{exc}

\subsubsection{Hamiltonian mechanics}\label{sec:mech-CM-Hamilton}
	\index{Hamiltonian mechanics}

	Classical mechanics is about time evolution on the configuration space
	$\phX^n$, i.e.\ one considers maps
	$$
		\R \supseteq I \to \phX^n, \quad
		t \mapsto \sx(t) = (x_1(t),\ldots,x_n(t)). 
	$$
	The evolution is typically of second order in time but can instead be 
	formulated in a more advantageous first-order form on the phase space $\phP^n$ 
	if one takes as the momenta $p_j := \dot x_j$ (or similar),
	with the dot denoting the derivative with respect to time $t$,
	$\dot f := df/dt$.

	The desired evolution equation is then determined by a choice of a function $H$ 
	on the phase space called the \keyword{Hamiltonian}, 
	i.e.~an element $H \in \cA$ which depends on the particular physical
	system under consideration.
	The value of this function $H(\sx,\ssp)$
	can usually be interpreted as the energy of the system at the corresponding 
	point $(\sx,\ssp)$ in the phase space.
	The time evolution for general $f \in \cA$ is then defined to be
	governed by the equation
	\begin{equation}\label{eq:Hamilton-Poisson}
		\boxed{\dot f = \{f,H\},}
	\end{equation}
	reducing in particular, by~\eqref{eq:standard-Poisson}, 
	for the case of the coordinates and momenta to
	\keyword{Hamilton's equations of motion}:
	\begin{equation}\label{eq:Hamilton}
		\dot x_j = \frac{\partial H}{\partial p_j}, \qquad
		\dot p_j =-\frac{\partial H}{\partial x_j}.
	\end{equation}
	Note conversely that these equations and the definition~\eqref{eq:standard-Poisson}
	\emph{imply}~\eqref{eq:Hamilton-Poisson},
	\begin{equation}\label{eq:Hamilton-Poisson-recovered}
		\dot f = \frac{d}{dt} f(\sx,\ssp) 
		= \sum_{j=1}^n \left( \frac{\partial f}{\partial x_j}\frac{d x_j}{dt}
			+ \frac{\partial f}{\partial p_j}\frac{d p_j}{dt} \right)
		= \sum_{j=1}^n \left( \frac{\partial f}{\partial x_j}\frac{\partial H}{\partial p_j}
			- \frac{\partial f}{\partial p_j}\frac{\partial H}{\partial x_j} \right)
		= \{f,H\}.
	\end{equation}
	
\begin{example}[Free particle]\index{free particle}
	A particle that is free to move in three-dimensional space has the 
	configuration space $\phX^3 = \R^3$ of positions $\bx = (x_1,x_2,x_3)$
	and the phase space $\phP^3 = \R^6 \ni (\bx,\bp)$, where $\bp = (p_1,p_2,p_3)$ 
	is the canonical momentum of the particle. 
	These variables satisfy the canonical Poisson brackets \eqref{eq:classical-CCR}.
	We take the Hamiltonian to be the 
	(\keyword{non-relativistic}; see also the below remark for an explanation)
	\keyword{free kinetic energy}
	\begin{equation}\label{eq:free-Hamiltonian}
		H(\bx,\bp) = T(\bp), \qquad T(\bp) := \frac{\bp^2}{2m}, 
	\end{equation}
	with $m>0$ known as the \keyword{mass} of the particle 
	which is considered as a fixed (non-dynamical) parameter.
	Indeed, Hamilton's equations of motion \eqref{eq:Hamilton} are then
	\begin{equation}\label{eq:free-Hamiltoneq}
		\dot\bx = \bp/m, \qquad
		\dot\bp = 0,
	\end{equation}
	i.e. $\ddot\bx = 0$, giving straight trajectories $\bx(t) = \ba + t\bb$ 
	in $\phX^3$. Also, we obtain the relationship $\bp = m\bv$ between the 
	momentum and the velocity $\bv := \dot\bx$ of the particle, and therefore
	the well-known formula for its kinetic energy
	\begin{equation}\label{eq:free-kin-energy}
		T(\bp) = \frac{1}{2}m\bv^2.
	\end{equation}
\end{example}

\begin{remark*}\label{rmk:kinetic-energy}
	The special theory of relativity tells us that $(mc^2)^2 = E^2 - \bp^2c^2$,
	where $m c^2$ is the rest energy, or $m$ the rest mass, of a free particle.
	Therefore, for small $\bp/(mc)$,
	\begin{equation}\label{eq:rel-energy}
		E = \sqrt{m^2c^4 + c^2\bp^2} = mc^2 \sqrt{1 + \frac{\bp^2}{m^2c^2}}
		= mc^2 + \frac{\bp^2}{2m} + O(\bp^4),
	\end{equation}
	which after subtracting the constant $mc^2$ yields the 
	non-relativistic (first-order) approximation~\eqref{eq:free-Hamiltonian}
	to the kinetic energy.
	However, one may also study the full relativistic expression
	$T_{\textup{rel},m}(\bp) := mc^2\sqrt{1+\bp^2/(mc)^2} - mc^2$
	or the simpler massless case
	$T_{\textup{rel},0}(\bp) := c|\bp|$.
\end{remark*}
	
\begin{example}[Particle in an external potential]\label{exmp:classical-ext-pot}
	Mechanics would be rather boring if there were only free particles moving 
	in straight lines, but what we may do is to add a 
	\keyword{scalar potential} to the 
	Hamiltonian~\eqref{eq:free-Hamiltonian},
	\begin{equation}\label{eq:pot-Hamiltonian}
		H(\bx,\bp) = T(\bp) + V(\bx),
	\end{equation}
	where $V\colon \R^3 \to \R$ is a function of the coordinates only.
	Hamilton's equations are then modified to
	\begin{equation}\label{eq:pot-Hamiltoneq}
		\dot\bx = \bp/m, \qquad
		\dot\bp = -\nabla V,
	\end{equation}
	where $\bF := -\nabla V$ is the \keyword{force} acting on the particle,
	with its sign chosen to act to \emph{minimize} the potential energy.
	Hence, $m\ddot\bx = \bF$, which is \keyword{Newton's equation of motion}.
\end{example}

\begin{remark*} 
	In fact, the potential may be understood to have a geometric origin
	and is again most naturally formulated in the framework of relativity.
	Namely, one couples the spacetime momentum $\ssp=(E/c,\bp)$ of the particle
	to the \keyword{gauge potential} $\sA(\sx)=(A_0(t,\bx),\bA(t,\bx))$ 
	of the electromagnetic field $F=d\sA$
	using the replacement $\ssp \mapsto \ssp-q\sA$, 
	where $q$ is the charge of the particle.
	Then $(mc)^2 = (\ssp - q\sA)^2 = (E/c-qA_0)^2 - (\bp-q\bA)^2$ 
	implies the electromagnetically coupled version of 
	\eqref{eq:rel-energy}:
	$$
		E = mc^2 + \frac{(\bp-q\bA)^2}{2m} + V + O((\bp-q\bA)^4),
	$$
	with scalar potential $V = qcA_0$.
	See \cite{Nakahara-03} for more on the geometry of electrodynamics.
\end{remark*}

\begin{example}[Harmonic oscillator]\label{exmp:harm-osc}
	The standard example of an external potential, due to its simplicity 
	and also its widespread appearance in real physical systems, 
	is the \keyword{harmonic oscillator potential},
	\begin{equation}\label{eq:harm-osc}
		V_{\textup{osc}}(\bx) := \frac{1}{2} m\omega^2 |\bx|^2,
	\end{equation}
	where $\omega \ge 0$ is a parameter known as the angular frequency
	of the oscillator. The mass $m$ appears here scaled out of $\omega$
	in order to make the dynamics
	independent of $m$,
	namely, Newton's equations become simply
	$\ddot{\bx} = -\omega^2 \bx$,
	with well-known $2\pi\omega$-periodic solutions.
	
	What one usually does in preparation for the quantum version of the
	harmonic oscillator is to introduce the \emph{complex} 
	phase-space variables
	$$
		a_j := \sqrt{\frac{m\omega}{2}}\left(x_j + \frac{i}{m\omega}p_j\right),
		\qquad
		a_j^{*} = \sqrt{\frac{m\omega}{2}}\left(x_j - \frac{i}{m\omega}p_j\right).
	$$
	One may then observe that these satisfy the Poisson algebra
	\begin{equation}\label{eq:oscillator-CCR}
		\{a_j,a_k\} = 0, \qquad
		\{a_j^*,a_k^*\} = 0, \qquad
		\{a_j,a_k^*\} = -i\delta_{jk}1,
	\end{equation}
	and, if $\cN := \sum_j a_j^*a_j = \sum_j |a_j|^2$,
	\begin{equation}\label{eq:oscillator-CCR-N}
		\cN = \cN^* = H/\omega, \qquad
		\{\cN,a_j\} = ia_j, \qquad
		\{\cN,a_j^*\} = -ia_j^*.
	\end{equation}
	Also,
	$x_j = (2m\omega)^{-1/2}(a_j^* + a_j)$ and
	$p_j = i(m\omega/2)^{1/2}(a_j^* - a_j)$.
\end{example}

	Note that by construction of the dynamical equations~\eqref{eq:Hamilton-Poisson}
	and the antisymmetry of the Poisson bracket,
	we have that the Hamiltonian is a conserved quantity under the motion,
	\begin{equation}\label{eq:Hamiltonian-conserved}
		\frac{d}{dt}H(\bx,\bp) = \dot H = \{H,H\} = 0.
	\end{equation}
	In fact, any function $f$ on $\phP^n$ is by \eqref{eq:Hamilton-Poisson} 
	conserved in time iff it Poisson-commutes with $H$
	(the general important relationship between 
	symmetries of the Hamiltonian and conserved quantities 
	admits a more thorough formulation 
	and is known as Noether's theorem).
	
\begin{remark*}
	There are subtleties even in classical mechanics when one considers
	systems with singular behavior which need to be treated using \keyword{constraints}.
	Typical examples are field theories such as electromagnetism,
	where passing from a Lagrangian to a Hamiltonian formulation involves
	redundant degrees of freedom and results in non-invertible transformations.
	However, Dirac has invented a procedure to treat such constrained
	Hamiltonian systems and to perform reductions in the Poisson algebra.
	See \cite{Dirac-64}, \cite[Chapter~24]{Thiemann-07}, 
	and e.g. \cite{deWHopLun-11a} for a recent example in membrane theory.
	Also fully general-relativistic systems --- where there is no 
	canonical time coordinate ---
	may be considered, though in an even more general framework for mechanics, 
	as outlined e.g. in \cite{Rovelli-04,Thiemann-07}.
\end{remark*}

\begin{exc}
	Verify the brackets \eqref{eq:oscillator-CCR} and \eqref{eq:oscillator-CCR-N}.
\end{exc}

\subsection{The instability of classical matter}\label{sec:mech-CM-instability}
	\index{instability}

	Let us briefly discuss why ordinary matter formulated in terms of
	the above-outlined rules for mechanics turns out to be unstable.
	We do not need to construct a very complicated model of 
	matter in order to see the instability. In fact it arises already
	upon considering the simplest model of an atom consisting of a
	single electron moving in three-dimensional space around a fixed nucleus, 
	which we for simplicity place 
	at the origin of the electron's coordinate system $\phX^3 = \R^3$. 
	To justify this assumption, either consider the nucleus
	to be much heavier than the electron (which indeed it is by experiment) 
	so that it experiences only very slow acceleration 
	(according to Newton's equation)
	and thus can be safely 
	considered fixed during a short time frame, 
	or better consider the problem in \keyword{relative coordinates} as will be
	described in Section~\ref{sec:mech-QM-2p}.
	
	The electron and the nucleus have opposite electric charge and therefore
	experience an attractive electric force given by the 
	\keyword{Coulomb potential},
	\begin{equation}\label{eq:Coulomb-potential}
		V_\sC(\bx) := \frac{q_1q_2}{|\bx|},
	\end{equation}
	where $q_1$ and $q_2$ are the particles' respective charges,
	resulting in the Coulomb force
	$$
		\bF_\sC(\bx) = -\nabla V_C(\bx) = q_1q_2\frac{\bx}{|\bx|^3}.
	$$
	In accordance with commonly used conventions and for future simplicity, 
	we will normalize the electron charge to
	$q_1 = -1$ and call the charge of the nucleus $q_2 = Z > 0$.
	For a \emph{neutral} one-electron atom 
	the nucleus consists of a
	single proton with charge $+1$ and we thus have $Z=1$ 
	(this charge $Z$ is known in chemistry as the \keyword{atomic number},
	with $Z=1$ representing 
	the hydrogen atom),
	however we will for generality keep $Z>0$ free as a mathematical parameter.
	We therefore take as our model for the dynamics of the electron 
	in this \keyword{hydrogenic atom}
	the model considered in Example~\ref{exmp:classical-ext-pot}, 
	with the external potential%
	\footnote{%
	Though $V_\sC \notin \cA$,
	we may consider it as a limit of smooth functions (see remark)
	or extend our $\cA$ a bit.}
	$$
		V_\sC(\bx) = -\frac{Z}{|\bx|}.
	$$
	Hence the Hamiltonian defined on the electron's phase space 
	$\phP^3 = \R^3 \times \R^3 \ni (\bx,\bp)$ is
	\begin{equation}\label{eq:hydrogen-Hamiltonian}
		H(\bx,\bp) = T(\bp) + V_\sC(\bx) = \frac{\bp^2}{2m} - \frac{Z}{|\bx|}.
	\end{equation}
	We already observe an obvious problem here: 
	that $H$ is \emph{unbounded from below},
	namely fixing $\bp$ while taking $\bx \to 0$ results in 
	$H(\bx,\bp) \to -\infty$.
	However, one may object that this limit is quite artificial 
	and perhaps cannot be realized in practice, in particular
	because the energy must be conserved throughout the dynamics
	as we already observed in~\eqref{eq:Hamiltonian-conserved}.
	Let us therefore instead consider a possible trajectory: 
	say for simplicity that 
	the electron starts from rest at the point $(1,0,0) \in \R^3$, i.e. 
	$\bx(0) = (1,0,0)$ and $\bp(0) = \0$.
	Then the non-trivial equation of motion to be solved is
	\begin{equation}\label{eq:electron-motion-ex}
		m\ddot{x}_1 = -Zx_1/|x_1|^3, \qquad x_1(0) = 1, \quad \dot{x}_1(0) = 0,
	\end{equation}
	whose solution 
	(see Exercise~\ref{exc:electron-motion})
	can be seen to satisfy $x_1(t) \to 0$ in finite time.
	Therefore the electron described in the framework of classical mechanics 
	\emph{admits dynamics whereby it collapses into the nucleus}.
	
\begin{remark*}
	With a little more physics background,
	one may still object to this conclusion of instability in two ways.
	The first is that the nucleus is actually a composite particle which has some
	spatial extent and therefore it is not clear that a collapse happens
	--- maybe the electron would just bounce around in a continuous charge
	distribution. However, it is known that the size of a nucleus is about 
	$10^{-15}$~m while the typical size of a hydrogen atom is about 
	$10^{-10}$~m (the Bohr radius), 
	so from the perspective of the typical electron orbit the nucleus certainly 
	looks pointlike, and one rather needs to explain why the electron insists
	on staying so far away from the nucleus.
	This leads to the second objection, namely that in 
	analogy to the picture of a planetary system
	(which is completely justified from the model~\eqref{eq:hydrogen-Hamiltonian} 
	since the Newtonian gravitational potential looks exactly the same), 
	the electron could just move in a circular or elliptical 
	orbit with its centripetal acceleration exactly matching
	the Coulomb force. In order to object to this picture of apparent stability
	one needs to know a little more about electromagnetic interactions,
	namely that an accelerating charge necessarily emits electromagnetic radiation
	to its surroundings
	(in order to properly incorporate this --- still purely classical --- 
	effect, called \index{bremsstrahlung}bremsstrahlung, one needs 
	to modify both the above simple Hamiltonian and the phase space severely, 
	and the resulting Hamiltonian describing the electron is then \emph{not} 
	conserved in time).
	The consequence of this radiative effect is that the electron loses 
	energy and therefore transcends into lower and lower orbits, 
	in effect spiraling in towards the nucleus and
	leading to the collapse of classical matter.
\end{remark*}

\begin{exc}\label{exc:electron-motion}
	Find an implicit solution of \eqref{eq:electron-motion-ex}
	for $x_1(t)>0$ and determine the time $T$ for which $x_1(T) = 0$.
	(Hint: start by multiplying the equation by $\dot{x}_1$.)
\end{exc}

\subsection{Some quantum mechanics}\label{sec:mech-QM}
	\index{quantum mechanics}

	The above-discussed problem of instability, 
	together with other unexpected 
	discoveries 
	in the beginning of the 20th century, led to the realization that Hamiltonian
	mechanics on phase space 
	(as well as the other equivalent formulations of classical mechanics) 
	is not sufficient to describe the physical world.
	This was in the 1920's subsequently remedied by the invention
	of a \emph{quantum} representation for mechanics.
	Nowadays this is quite well understood as a 
	kind of mathematical recipe,
	referred to as \keyword{canonical quantization},
	although depending on which systems are considered
	there are still many subtleties to be dealt with, both on a formal level
	and also when it comes to the physical interpretation.
	
\subsubsection{Axioms of canonical quantization}\label{sec:mech-QM-axioms}

	In mathematical terms, the procedure of
	`canonical quantization' amounts to selecting 
	a (sufficiently interesting) Lie-subalgebra $\cO \subseteq \cA$
	and a representation of this $\cO$
	as an algebra of linear operators on a Hilbert space $\cH$,
	$\cO \to \hat{\cO} \subseteq \cL(\cH)$.
	This translates to the following set of 
	\keyword{axioms of quantum mechanics}:

	\begin{enumerate}[label=\textup{A\arabic*.},ref=\textup{A\arabic*}]
	\item\label{axm:states}(States)
		There exists a 
		complex\footnote{One may also consider purely real Hilbert spaces;
		see e.g. \cite{Lundholm-08} and references for a discussion.}
		separable\footnote{This assumption could, and should, sometimes be relaxed; 
		see e.g. \cite{Thiemann-07,AshSin-11}.} 
		Hilbert space $\cH$, 
		which we call the \keyword{quantum configuration space}.
		The non-zero elements $\psi \in \cH \setminus \{0\}$ 
		will describe the \keyword{states} of the quantum system, 
		and furthermore two vectors $\psi$ and $\phi$ in $\cH$ describe 
		the same state if and only if 
		$\phi = c\psi$,
		$c \in \C \setminus \{0\}$. 
		In other words the set of quantum states constitutes 
		a ray representation of $\cH$.
	\item\label{axm:observables}(Observables)
		One has selected a set of \keyword{observables} $a \in \cO$
		which form a closed Lie-subalgebra of $\cA$, 
		and which are real, $a^* = a$.
		To each such observable $a \in \cO$ there is associated a densely defined
		self-adjoint operator $\ha$ acting on $\cH$, 
		i.e. $\ha^* = \ha \in \cL(\cH)$.
		The spectrum of $\ha$ are the possible
		results of a measurement of the observable $a$.
	\item\label{axm:commutators}(Commutators)
		The Poisson bracket in classical mechanics is replaced by the 
		\keyword{commutator}
		\begin{equation}\label{eq:def-commutator}
			[\ha,\hb] := \ha\hb - \hb\ha
		\end{equation}
		of operators, 
		according to:
		\begin{equation}\label{eq:bracket-quantization}
			\{a,b\} \mapsto \frac{1}{i\hbar} [\ha,\hb].
		\end{equation}
		Here we have multiplied the commutator with $-i$ 
		(and the sign is just a convention) 
		in order to make the expression self-adjoint 
		(by the closedness of $\cO$, the bracket of two observables is also an observable 
		and hence should be represented by a self-adjoint operator, 
		but the commutator of
		two self-adjoint operators is \emph{anti}-self-adjoint\footnote{On
		the basis of this one may argue that the more natural thing to do is
		to replace everything by \emph{anti}-self-adjoint operators.}),
		and we furthermore introduced a new parameter $\hbar > 0$ known as 
		\keyword{Planck's constant}.
	\item\label{axm:expectations}(Expectations)
		Given a state $\psi \in \cH$, the \keyword{expectation value} 
		of an observable $a \in \cO$ in this state is given by
		\begin{equation}\label{eq:def-expectation}
			\langle \ha \rangle_\psi := \frac{\inp{\psi,\ha\,\psi}}{\inp{\psi,\psi}}.
		\end{equation}
		The interpretation is that if one prepares a large ensemble of 
		identical systems, each of which is prepared to be in the state $\psi$, 
		and then makes a measurement of the observable $a$ 
		then the result of the measurement will in general be random 
		but the expectation value of the results will be given by 
		the quantity $\inp{\ha}_\psi \in \R$.
		If $a$ has physical meaning but $\psi \notin \cQ(\ha)$ then $\psi$
		may be interpreted as an unphysical state.
	\item\label{axm:time-evolution}(Time evolution)
		The choice of dynamics depends on the choice of a Hamiltonian
		$H \in \cO$, which is represented as a self-adjoint 
		\keyword{Hamiltonian operator} $\hH$ on $\cH$.
		Operators corresponding to other observables may then evolve 
		with time according to \keyword{Heisenberg's equation of motion}
		(compare~\eqref{eq:Hamilton-Poisson}),
		\begin{equation}\label{eq:Heisenberg-eqs}
			\frac{d}{dt}\ha(t) = \frac{1}{i\hbar} [\ha(t),\hH],
		\end{equation}
		and their corresponding expectation value at time $t$ is
		\begin{equation}\label{eq:expectation-timedep}
			\langle \ha(t) \rangle_\psi = \frac{\inp{\psi,\ha(t)\psi}}{\inp{\psi,\psi}}.
		\end{equation}
	\end{enumerate}
	
\begin{remark}\label{rem:self-adjointness}
	The reasons for insisting that observables be represented by
	\emph{self-adjoint} operators are threefold:
	\begin{enumerate}[label=\arabic*.,ref=\arabic*]
	\item\label{itm:s-a-hermitian}
		An observable $a \in \cO$ should represent a real 
		measurable physical quantity, 
		and its expectation value satisfies $\langle \ha \rangle_\psi \in \R$ 
		for all states $\psi \in \cQ(\ha)$ iff $\ha$ is 
		hermitian. 
	\item\label{itm:s-a-spectral-rep}
		Self-adjoint operators have a \keyword{spectral representation}
		given by the spectral theorem, Theorem~\ref{thm:spectral-theorem},
		and the points of the spectrum $\sigma(\ha) \subseteq \R$
		represent the possible values of a measurement of the observable 
		$a \in \cO$.
		Also, if the system is in a state $\psi \in \cH$, 
		then the \keyword{probability} of measuring values of $a$ in the 
		interval $[\lambda,\lambda'] \subseteq \R$ is given by
		the expectation value
		(of the observable\footnote{Note that if we can measure $a$
		then we may also determine if $a \in [\lambda,\lambda']$ for any interval
		$[\lambda,\lambda'] \subseteq \R$,
		and similarly if we have a self-adjoint operator $\ha$ 
		then we also have access to its projection-valued measure 
		$P^{\ha}_{[\lambda,\lambda']}$.
		The commutativity of 
		$P^{\ha}_\Omega$ and $P^{\hb}_{\Omega'}$ 
		expresses what may be known simultaneously about $a$ and $b$.}
		``$a \in [\lambda,\lambda']$'')
		\begin{equation}\label{eq:projection-probability}
			\inp{ P^{\ha}_{[\lambda,\lambda']} }_\psi
			= \frac{\norm{P^{\ha}_{[\lambda,\lambda']}\psi}^2}{\norm{\psi}^2},
		\end{equation}
		where $\Omega \mapsto P^{\ha}_\Omega$ 
		is the corresponding spectral projection
		(see e.g. Example~\ref{exmp:spectral-projection}).
	\item\label{itm:s-a-Stone}
		By Stone's theorem, Theorem~\ref{thm:Stones-theorem},
		self-adjoint operators are the generators of one-parameter 
		\emph{unitary} groups, $t \mapsto U_{\ha}(t) := e^{i\ha t}$,
		which is in particular important for the time evolution by the
		Hamiltonian $\hH$ to \emph{conserve} probabilities,
		$\norm{U_{\hH}(t)\psi} = \norm{\psi}$.
		A \emph{non}-self-adjoint Hamiltonian operator would describe 
		\emph{non}-unitary time evolution, 
		which would however be appropriate when there is 
		energy or \keyword{information loss} from the system to an external
		environment.
	\end{enumerate}
	We also note that:
	\begin{enumerate}[label=\arabic*.,resume]
	\item\label{itm:s-a-measurement}
		A \keyword{measurement} necessarily exchanges information between 
		the system being measured and the observer, 
		and therefore results in non-unitary evolution. 
		In effect, after measurement the state has become 
		projected into the subspace corresponding to the information 
		obtained, $\psi \mapsto P^{\ha}_{[\lambda,\lambda']}\psi$, 
		by means of the spectral projection 
		in~\eqref{eq:projection-probability}.
		For example, if $\ha$ has an isolated simple eigenvalue 
		$\lambda_j \in \sigma(\ha)$ 
		with corresponding normalized eigenstate $u_j$, 
		then the probability~\eqref{eq:projection-probability} 
		of measuring precisely this value in the normalized state 
		$\psi$ is $\inp{P^{\ha}_{[\lambda_j-\eps,\lambda_j+\eps]}}_\psi =
		|\langle u_j,\psi \rangle|^2$, and the state of the system
		$\psi \mapsto u_j$
		after such a measurement. 
		Repeated measurement of $a$ will then produce the same value $\lambda_j$
		with certainty (unless the observable has evolved with time).
		The quantity $\langle u_j,\psi \rangle \in \C$ is called a
		\keyword{probability amplitude}.
	\end{enumerate}
\end{remark}

\begin{remark}
	Instead of evolving the operators in time according to the solution 
	of~\eqref{eq:Heisenberg-eqs},
	\begin{equation}\label{eq:Heisenberg-eqs-soln}
		\ha(t) = e^{it\hH/\hbar} \ha(0) e^{-it\hH/\hbar} \in \cL(\cH),
	\end{equation}
	one may evolve the states,
		$\psi(t) := e^{-it\hH/\hbar} \psi(0) \in \cH$,
	so that by unitarity
	\begin{equation}\label{eq:Heisenberg-Schroedinger}
		\langle \ha(t) \rangle_{\psi(0)} = \langle \ha(0) \rangle_{\psi(t)}.
	\end{equation}
	These states then satisfy the \keyword{Schr\"odinger equation}
	\begin{equation}\label{eq:Schroedinger-eqn}
		\boxed{i\hbar \frac{d}{dt}\psi(t) = \hH \psi(t).}
	\end{equation}
\end{remark}

\subsubsection{The Schr\"odinger representation}\label{sec:mech-QM-Schroedinger}
		
	Let us now implement the above quantization rules on 
	the arche\-typical Poisson algebra $\cA$, namely the
	phase space~\eqref{eq:standard-phase-space}-\eqref{eq:standard-Poisson},
	with the canonical Poisson brackets~\eqref{eq:classical-CCR}.
	We take the canonical coordinate and momentum functions $x_j$ and $p_k$
	as our fundamental observables, with for example
	\begin{equation}\label{eq:standard-observables}
		\cO = \Span_\R \{ 1, x_1, \ldots, x_n, p_1, \ldots, p_n \},
	\end{equation}
	(note that we added $1$ to make this a closed Lie-subalgebra of $\cA$, 
	also known as the Lie algebra of the \keyword{Heisenberg group}).
	However we will typically want to work with
	something slightly larger than~\eqref{eq:standard-observables}
	since by \ref{axm:time-evolution} we also need 
	a Hamiltonian observable $H \in \cO$
	which is some function of $\sx$ and $\ssp$.
	These observables should according to \ref{axm:states}-\ref{axm:observables}
	be promoted to 
	operators $\hat{1}, \hx_j, \hp_k \in \cL(\cH)$
	on some Hilbert space $\cH$, which we leave undetermined 
	for a brief moment.

	The canonical Poisson brackets~\eqref{eq:classical-CCR} 
	should then according to \ref{axm:commutators} be represented by the
	\keyword{canonical commutation relations (CCR)}:
	\begin{equation}\label{eq:quantum-CCR}
		[\hx_j, \hx_k] = 0, \qquad
		[\hp_j, \hp_k] = 0, \qquad
		[\hx_j, \hp_k] = i\hbar \,\delta_{jk} \hat{1}.
	\end{equation}
	We also note that since $\{1,a\} = 0$ for all $a \in \cA$,
	we should have $[\hat{1},\ha] = 0$, at least for all $a \in \cO$,
	so that upon considering irreducible\footnote{That is, if seen as matrices,
	not block-diagonalizable but restricted to just one full block which cannot
	be reduced further. Irreducibility comes in by axiom~\ref{axm:states}
	and the desire to be able to distinguish all states.} 
	representations of this algebra
	we may write $\hat{1} = c\1$, where $c \in \C$, 
	or actually $\bar{c} = c \in \R$ since $\hat{1}^* = \hat{1}$.
	However, as $\hat{1}$ appears only in combination with $\hbar$
	in the r.h.s. of~\eqref{eq:quantum-CCR}, and since we have not yet fixed
	the value of $\hbar$, we may absorb this freedom into $\hbar$ and take 
	$c=1$, i.e. $\hat{1} = \1$, the unit in $\cL(\cH)$.
	
	It is now time to find a Hilbert space on which to represent the 
	operator observables. 
	However, taking the simplest non-trivial choice that comes to mind,
	i.e. $\cH = \C^N$ for some finite dimension $N \ge 1$ and
	with the operators acting as hermitian $N \times N$-matrices,
	is seen not to work, simply by taking the trace on both sides of for example 
	the operator equation corresponding to the first non-trivial commutator in \eqref{eq:quantum-CCR},
	\begin{equation}\label{eq:quantum-CCR-1}
		\hx_1 \hp_1 - \hp_1 \hx_1 = i\hbar \1.
	\end{equation}
	The trace is zero on the l.h.s.~but $i\hbar N$ on the r.h.s., and 
	hence yields a contradiction unless $\hbar = 0$.
	As a result, 
	we cannot represent these relations non-trivially unless
	the space is \emph{infinite}-dimensional, which leads us to the next-most
	natural choice of $\cH = L^2(\phX^n) = L^2(\R^n)$,
	where we may for example take the standard 
	\keyword{Schr\"odinger representation}: for $\psi \in \cH$,
	\begin{equation}\label{eq:def-position-op}
		(\hx_j\psi)(\sx) := x_j \psi(\sx), 
	\end{equation}
	i.e. simply multiplication by the coordinate, and
	\begin{equation}\label{eq:def-momentum-op}
		(\hp_j\psi)(\sx) := -i\hbar \frac{\partial\psi}{\partial x_j}(\sx).
	\end{equation}
	The sign convention on $\hp_j$ here matches that of~\eqref{eq:bracket-quantization}.

	Note that both of these are unbounded operators and therefore
	care has to be taken that they are self-adjoint. 
	They are obviously hermitian when considered as forms on 
	$C_c^\infty(\R^n)$,
	\begin{equation}\label{eq:x-p-hermiticity}
		\inp{\varphi, \hx_j \psi} = \inp{\hx_j \varphi, \psi}, \qquad
		\inp{\varphi, \hp_k \psi} = \inp{\hp_k \varphi, \psi}, \qquad
		\forall \varphi, \psi \in C_c^\infty(\R^n),
	\end{equation}
	but one needs to specify corresponding domains so that 
	$\cD(\hx_j^*) = \cD(\hx_j)$ and $\cD(\hp_k^*) = \cD(\hp_k)$.
	In the case $\hx_j$, the natural (maximal) domain is
	$$
		\cD(\hx_j) := \left\{ \psi \in L^2(\R) : 
			\int_{\R^n} |x_j|^2 |\psi(\sx)|^2 \,d\sx < \infty \right\},
	$$
	while the common joint domain for all $\hx_j$ is
	$$
		\cD(\hsx) := \{ \psi \in L^2(\R^n) : 
			\hsx \psi = (\hx_1\psi,\ldots,\hx_n\psi) \in L^2(\R^n;\C^n) \},
	$$
	and for $\hp_k$,
	\begin{equation}\label{eq:hp-domain}
		\cD(\hp_k) := \left\{ \psi \in L^2(\R) : 
			\int_{\R^n} |\xi_k|^2 |\hat\psi(\xi)|^2 \,d\xi < \infty \right\},
		\qquad
		\cD(\hsp) := H^1(\R^n).
	\end{equation}
	This also corresponds to taking the closure of the minimal operators,
	i.e.\ with $\hx_j$ and $\hp_k$ initially defined on the minimal domain 
	$C_c^\infty(\R^n)$.

	Writing 	$\ssp = \hbar \xi$ (we will later set $\hbar=1$), we have
	in terms of the Fourier transform\footnote{The conventions here are unfortunate
	but standard; hats on states denote their Fourier transform, 
	and otherwise it denotes
	operator representations of phase-space functions.}
	\eqref{eq:Fourier-derivatives}
	\begin{equation}\label{eq:momentum-op-Fourier}
		(\hp_j \psi)^\wedge(\ssp) = p_j \hat\psi(\ssp), \qquad
		(\hx_k \psi)^\wedge(\ssp) = i\hbar\frac{\partial \hat\psi}{\partial p_k}(\ssp),
	\end{equation}
	so that an alternative but equivalent representation 
	(called the \keyword{momentum representation})
	is given by $\cF\cH = \cF L^2(\phX^n) = L^2(\R^n) \ni \hat\psi$ with
	operators $\check{1} = \1$, $\check{x}_k$, $\check{p}_j \in \cL(\cF\cH)$ 
	defined by
	\begin{equation}\label{eq:momentum-rep}
		(\check{p}_j \hat\psi)(\ssp) := p_j \hat\psi(\ssp), \qquad
		(\check{x}_k \hat\psi)(\ssp) := i\hbar\frac{\partial\hat\psi}{\partial p_k}(\ssp),
	\end{equation}
	i.e. $\check{x}_k = \cF\hx_k\cF^{-1}$ and $\check{p}_j = \cF\hp_j\cF^{-1}$.
	
\begin{remark}[Stone--von Neumann uniqueness theorem]\label{rem:Stone-von-Neumann}
	In fact, one may consider the abstract unitary group generated by, 
	say $\hx_1$, $\hp_1$, with the commutation relations~\eqref{eq:quantum-CCR-1},
	via
	\begin{equation}\label{eq:Weyl-generators}
		U(s) := e^{is\hx_1}, \qquad
		V(t) := e^{is\hp_1/\hbar},
	\end{equation}
	which satisfy the \keyword{Weyl algebra}
	\begin{equation}\label{eq:Weyl-algebra}
		U(s)U(s') = U(s+s'), \quad
		V(t)V(t') = V(t+t'), \quad
		V(t)U(s) = e^{ist}U(s)V(t). 
	\end{equation}
	It turns out that, with the only assumptions that the representation
	of this group is unitary, irreducible and \keyword{weakly continuous}, i.e.
	\begin{equation}\label{eq:weak-continuity}
		\lim_{t \to 0} \inp{u, U(t)v} = \inp{u,v} \qquad \forall u,v \in \cH,
	\end{equation}
	and similarly for $V(t)$, it \emph{must} be equivalent to the 
	Schr\"odinger representation: that is, up to conjugation with a unitary
	(such as $\cF$),
	we have $\cH = L^2(\R)$ with
	\begin{equation}\label{eq:Weyl-Schroedinger}
		(U(s)\psi)(x) = e^{isx}\psi(x), \qquad
		(V(t)\psi)(x) = \psi(x+t).
	\end{equation}
	This is known as the \keyword{Stone--von Neumann uniqueness theorem}.
	However, upon relaxing the assumption~\eqref{eq:weak-continuity} 
	on weak continuity other representations may be found, 
	on \emph{non}-separable Hilbert spaces; see Section~\ref{sec:prelims-ops-Stone}
	and \cite[p.~213]{Thiemann-07}, \cite{AshSin-11}.
\end{remark}

	Since the $\hx_j$ are commuting operators
	we may diagonalize them simultaneously,
	with well-defined projections on the joint spectrum of the operators
	$\hsx = (\hx_1,\ldots,\hx_n)$,
	$$
		P^{\hsx}_{I_1 \times I_2 \times \ldots \times I_n} 
		= P^{\hx_1}_{I_1} P^{\hx_2}_{I_2} \ldots P^{\hx_n}_{I_n},
		\qquad I_j \subseteq \R \ \text{intervals},
	$$
	which may be generalized to
	$P^{\hsx}_\Omega$ for any (Borel) $\Omega \subseteq \R^n$.
	This is again the same as taking the Schr\"odinger representation with 
	$P^{\hsx}_\Omega = \1_\Omega$.
	For normalized $\psi \in L^2(\R^n)$, $\|\psi\|_{L^2}=1$, we have by
	Remark~\ref{rem:self-adjointness}.\ref{itm:s-a-spectral-rep}
	also a natural interpretation for 
	$$
		\int_\Omega |\psi|^2 = \inp{ \1_\Omega }_\psi 
		= \inp{ P^{\hsx}_\Omega }_\psi,
	$$
	which is thus the probability of measuring the event $\sx \in \Omega$ 
	given the state $\psi \in \cH$. Furthermore, 
	$$
		|\psi(\sx)|^2 = \inp{\psi, \delta_\sx \psi} 
		= \lim_{\eps \to 0} |B_\eps(\sx)|^{-1} \inp{ P^{\hsx}_{B_\eps(\sx)} }_\psi
	$$
	(valid for a.e. $\sx \in \R^n$ by the Lebesgue differentiation theorem)
	may be interpreted as the \keyword{probability density of measuring 
	the coordinates} $\sx \in \R^n$.
	
	In the case of diagonalizing $\hsp$ instead, i.e. switching to
	the momentum representation where 
	$\cF P^{\hsp}_\Omega \cF^{-1} = P^{\check{\ssp}}_\Omega = \1_\Omega$, 
	one has
	$$
		|\hat\psi(\ssp)|^2 = \inp{\hat\psi, \delta_\ssp \hat\psi} 
		= \lim_{\eps \to 0} |B_\eps(\ssp)|^{-1} \inp{ P^{\check{\ssp}}_{B_\eps(\ssp)} }_{\hat\psi}
		= \lim_{\eps \to 0} |B_\eps(\ssp)|^{-1} \inp{ P^{\hsp}_{B_\eps(\ssp)} }_\psi,
	$$
	the \keyword{probability density of measuring the momenta} $\ssp \in \R^n$.
	
\begin{remark}
	The reason why we cannot just take the full Poisson algebra $\cO = \cA$
	and quantize that is that there will be problems when it comes to the 
	choice of ordering of operators.
	Namely by~\eqref{eq:quantum-CCR-1},
	$\hx_1 \hp_1$ is not the same as $\hp_1 \hx_1$,
	so it matters if we by $x_1p_1 = p_1x_1 \in \cA$ mean
	$\hx_1\hp_1$, or $\hp_1\hx_1$, or perhaps $\frac{1}{2}(\hx_1\hp_1+\hp_1\hx_1)$.
	This is known as the \keyword{factor ordering ambiguity} in quantum mechanics
	and causes many headaches when trying to quantize classical mechanical 
	systems. As a result there are often different routes to quantization
	with obstacles to be overcome and choices to be made of both ordering
	rules and representations, so that, in practice, the procedure of 
	`canonical quantization' may not seem so canonical after all.
	For a very general treatment of the quantization procedure, see e.g.
	\cite{Thiemann-07,Rovelli-04}.
\end{remark}

\begin{exc}
	Verify \eqref{eq:x-p-hermiticity} and the CCR \eqref{eq:quantum-CCR}
	for both the choice \eqref{eq:def-position-op}-\eqref{eq:def-momentum-op} 
	and the alternative \eqref{eq:momentum-rep}, 
	and note the agreement of all sign conventions.
\end{exc}

\begin{exc}
	Verify that the Schr\"odinger representation \eqref{eq:def-position-op}-\eqref{eq:def-momentum-op}
	exponentiates to \eqref{eq:Weyl-Schroedinger}
	and satisfies the Weyl algebra \eqref{eq:Weyl-algebra}.
	Conjugate with $\cF$ and compare with \eqref{eq:momentum-rep}.
	Derive the abstract Weyl algebra starting from the definitions 
	\eqref{eq:Weyl-generators} and the CCR \eqref{eq:quantum-CCR}.
\end{exc}

\subsection{The one-body problem}\label{sec:mech-QM-1p}

	Let us now consider the case of the \keyword{one-body problem},
	i.e. a single particle on a $d$-dimensional classical configuration space 
	$\phX^d = \R^d$ and phase space $\phP^d = \R^{2d}$,
	on which we may take the Hamiltonian $H(\bx,\bp) = T(\bp) + V(\bx)$ 
	from Example~\ref{exmp:classical-ext-pot}
	(in that case we had $d=3$, but let us be more general here and take 
	$d \in \N$).
	Since it generates our dynamics,
	we should promote it to an observable, i.e. add it to 
	\eqref{eq:standard-observables}
	and then construct its quantum representation $\hH \in \cL(\cH)$.
	However, before we do so, let us make sure that we are done with
	our choice of observables $\cO$. Namely, $\cO$ needs to be closed under
	Poisson brackets, and indeed
	$$
		\{x_j,T(\bp)\} = \frac{p_j}{m} \in \cO, \qquad
		\{x_j,V(\bx)\} = 0, \qquad
		\{p_j,T(\bp)\} = 0, 
	$$
	but we find that we might also need to add
	$$
		\{p_j,V(\bx)\} = -\frac{\partial V(\bx)}{\partial x_j} = F_j(\bx),
	$$
	i.e. the components of the corresponding force $\bF$, as well as
	$$
		\{p_j,F_k(\bx)\} = -\frac{\partial F_k}{\partial x_j}(\bx),
	$$
	in case this expression is non-zero, and so on.
	Hence all of these functions on $\phP^d$ 
	need to be represented as operators as well.
	Moreover, there may be other relevant observables such as
	\begin{equation}\label{eq:angular-momentum}
		L_{jk} := x_jp_k - x_kp_j, \qquad 1 \le j<k \le d,
	\end{equation}
	for which $\{L_{jk},T(\bp)\} = 0$ and $\{L_{jk},V(\bx)\} = 0$ 
	if $V(\bx)=f(|\bx|)$
	(i.e. radial potentials), and which describe \keyword{angular momentum}.
	On the other hand, it is not always the case that any of 
	the canonical variables $x_j$ and $p_k$
	ought to be considered observables, 
	and one may in such an extreme circumstance 
	therefore just take the trivial choice 
	$\cO = \Span \{1,H\}$ or $\cO = \R H$
	(but e.g.\ some non-trivial representation based on the concrete expression for $H$)
	and hence only have to worry about quantizing the Hamiltonian $H$ 
	in that case.
	
\begin{example}[Free particle]\label{exmp:free-QM} \index{free particle}
	The simplest example is again the free particle with $V=0$,
	for which we have $H = T(\bp) = \bp^2/(2m)$.
	In the usual Schr\"odinger representation 
	\eqref{eq:def-position-op}-\eqref{eq:hp-domain}
	the natural thing to do is to take
	$$
		\hH = \frac{\hbp^2}{2m} = \frac{\hbar^2}{2m}(-\Delta_{\R^d}),
	$$ 
	with domain $\cD(\hH) = \cD(\hbp^2) = H^2(\R^d)$.
	This is then a self-adjoint operator, and 
	if instead considered on the minimal domain $C_c^\infty(\R^d)$
	it is essentially self-adjoint with the above extension as its closure.
	Hence there is no other choice for the quantum dynamics in this case
	(there could however be other options if one changes $\phX^d$ a bit
	by for example removing or identifying points,
	as will be seen in Section~\ref{sec:mech-QM-statistics}).
	Moreover, taking the Fourier transform and thus the momentum 
	representation as a pure multiplication operator \eqref{eq:momentum-rep} 
	we may even determine its spectrum explicitly:
	$\sigma(\hH) = \sigma(\check{\bp}^2/(2m)) = [0,\infty)$.
\end{example}

	Even if the potential $V$ is non-zero,
	since $T(\bp)$ and $V(\bx)$ depend only on $\bp$ and $\bx$ separately,
	there is luckily no factor ordering problem in $H = T + V$. 
	But we do need to ensure self-adjointness of $\hH$, 
	which could actually be quite difficult depending on $V$.
	The typical procedure would again be to use our earlier Schr\"odinger
	representation for $\hbx$ and $\hbp = -i\hbar\nabla$
	and thus write for the Hamiltonian operator
	\begin{equation}\label{eq:one-body-Hhat}
		\hH := T(\hbp) + V(\hbx) = T(-i\nabla) + V(\bx) 
		= \frac{\hbar^2}{2m}(-\Delta) + V(\bx),
	\end{equation}
	at least on the minimal domain $C_c^\infty(\R^d)$.
	This is a hermitian expression on this domain
	as long as the potential is real-valued, $V\colon \R^d \to \R$, 
	and not too singular (as will be illustrated in Example~\ref{exmp:hard-core}),
	and one may consider it as the sum of two quadratic forms.
	In the case that $V \ge -C$ with a constant $C \ge 0$ 
	the resulting form is bounded from below, $\hH \ge -C$, 
	and therefore by Friedrichs extension, Theorem~\ref{thm:Friedrichs}, 
	there is a \emph{unique} semi-bounded self-adjoint operator
	corresponding to the closure of this form expression.
	An operator on the form \eqref{eq:one-body-Hhat} for some potential $V$
	is conventionally called a \keyword{Schr\"odinger operator}.

\begin{example}[Harmonic oscillator]\label{exmp:harm-osc-QM}
	Our main example for a system with
	non-zero potential 
	is again the harmonic
	oscillator, $V = V_{\textup{osc}} \ge 0$ from Example~\ref{exmp:harm-osc}.
	In this case the Schr\"odinger operator $\hH \ge 0$ 
	may be defined by Friedrichs extension or form closure,
	Theorem~\ref{thm:form-operator}, on $C_c^\infty(\R^d)$ and the form domain is
	$$
		\cQ(\hH) = \cQ(-\Delta) \cap \cQ(V)
		= \bigl\{ \psi \in H^1(\R^d) : {\textstyle\int_{\R^d}} |\bx|^2|\psi(\bx)|^2 \,d\bx < \infty \bigr\}.
	$$
	We note that $\{p_j,V(\bx)\} = -m\omega^2 x_j$, and hence one may take
	$\cO = \Span_\R \{1,x_j,p_k,H\}$ as a closed algebra of observables,
	with quantum representatives 
	$\hat{\cO} = \Span_\R \{\1,\hx_j,\hp_k,\hH\} \subseteq \cL(\cH)$,
	all acting on a common dense domain 
	$\cD(\hH) \subseteq \cQ(\hH) \subseteq \cH = L^2(\R^d)$.
	
	Alternatively, let us consider the closed algebra 
	\eqref{eq:oscillator-CCR}-\eqref{eq:oscillator-CCR-N} 
	spanned by $\{1,a_j,a_k^*,\cN\}$ and try to quantize that. 
	For simplicity we take $\hbar=1$, $d=1$ and drop the index $j=1$.
	Note that $a \neq a^*$ (and also $a \neq -a^*$)
	since $\{a,a^*\} \neq 0$,
	and hence $a$ or $a^*$ are not observables.
	However, any of the expressions 
	$aa^* = a^*a = \cN = (a^*a + aa^*)/2 \in \cA$ 
	may be used.
	Also, forming the real combinations 
	$a+a^*$ and $-i(a-a^*)$ of the original observables subject to 
	\eqref{eq:oscillator-CCR},
	we consider promoting $a$ and $a^*$ to non-self-adjoint operators 
	$\ha$ resp. $\ha^*$ satisfying the 
	commutation relations
	\begin{equation}\label{eq:oscillator-QCCR}
		[\ha,\ha^*] = \1, 		\qquad
		[\hat{\cN},\ha] = -\ha, 	\qquad
		[\hat{\cN},\ha^*] = \ha^*,
	\end{equation}
	where we \emph{defined} $\hat{\cN} := \ha^*\ha$.
	Also note that $\ha\ha^* = \hat{\cN} + \1$ by the first commutator above.
	Hence, the expressions that were the same on $\cA$ are now given by 
	\emph{different} operators. 
	Furthermore, if we demand that $\hat{\cN} = \hat{\cN}^*$
	is self-adjoint and has some non-trivial eigenstate $\psi \in \cH$
	with eigenvalue $\lambda \in \R$, then one may observe that
	$\ha^k \psi$ resp. $(\ha^*)^k \psi$ 
	are also eigenstates with eigenvalues $\lambda - k$ resp. $\lambda + k$.
	Therefore, if also demanding $\hat{\cN}$ to be bounded from below,
	there must exist a state $\psi_0 \in \cH$ such that $\ha \psi_0 = 0$, 
	and the remaining states of an irreducible representation of \eqref{eq:oscillator-QCCR}
	are then given by $\psi_k := (\ha^*)^k \psi_0$ 
	with $\hat{\cN}\psi_k = k\psi_k$,
	$k=0,1,2,\ldots$.
	Taking finally the \emph{symmetrized} expression 
	$\hH/\omega := (\ha^*\ha + \ha\ha^*)/2 = \hat{\cN}+1/2$,
	and $\cH := \overline{\Span\{\psi_k\}_{k=0}^\infty}$,
	this then provides the \keyword{algebraic solution} 
	to the spectrum of the quantum Harmonic oscillator
	(one may finally show that these two representations coincide).
\end{example}

	If $V$ is unbounded from below then it is not certain that Friedrichs 
	extension applies, but there are other tricks and concepts,
	such as relative form boundedness and relatively bounded perturbations,
	which may be used to define the sums of such forms and operators.
	However we will in this course rely solely on proving that the full 
	Schr\"odinger expression \eqref{eq:one-body-Hhat}
	is bounded from below as a quadratic form, 
	so that there is then an unambiguous choice of 
	an associated bounded from below quantum Hamiltonian $\hH$.
		
\begin{example}[Coulomb potential]\label{exmp:Coulomb-QM}
	As already noted in Section~\ref{sec:mech-CM-instability}, 
	the Coulomb potential $V_\sC(\bx) = -Z/|\bx|$ is unbounded from below.
	In the next section we will use the uncertainty principle to
	prove that the form
	$$
		q(\psi) := \inp{\psi, 
			\left[\frac{\hbar^2}{2m}(-\Delta_{\R^3}) - \frac{Z}{|\bx|}\right]
			\psi}
	$$
	is nevertheless bounded from below on the minimal domain $C_c^\infty(\R^3)$
	(also note here that, even though $V_\sC$ is singular at $\bx=\0$, 
	we still have $V_\sC \in L^1_\loc(\R^3)$ which makes the expression
	well defined on this domain),
	and hence it \emph{defines} a semi-bounded from below self-adjoint
	operator $\hH$ by Theorem~\ref{thm:form-operator} or \ref{thm:Friedrichs}.
\end{example}

	We remark that, as soon as we have defined a self-adjoint Hamiltonian
	operator $\hH$ which is bounded from below, $\hH \ge -C$,
	then there cannot be any problems with the physical system 
	for any future time,
	since all states $\psi \in \cH$ then evolve unitarily by Stone's theorem,
	and furthermore arbitrarily negative values of the energy cannot
	be measured at any time since the measurable energy spectrum has a 
	finite lower bound, $\sigma(\hH) \subseteq [-C,+\infty)$.
	Hence, in this precise sense there is then \keyword{stability} 
	for the corresponding quantum system.

\begin{exc}
	Verify the algebraic relations in Example~\ref{exmp:harm-osc-QM}
	and extend the solution to $d > 1$.
\end{exc}

\begin{exc}\label{exc:spin}
	Consider the angular momenta \eqref{eq:angular-momentum} in $\R^3$, with
	$L_1$, $L_2$ and $L_3 := x_1p_2 - x_2p_1$ cyclically defined.
	Verify the Poisson brackets $\{L_1,L_2\} = L_3$ (cyclic)
	and $\{\bL^2,L_k\} = 0$ for all $k$,
	where $\bL^2 := L_1^2 + L_2^2 + L_3^2$.
	By considering $L_{\pm} := L_1 \pm iL_2$ and the corresponding commutation
	relations, show that all possible finite-dimensional
	irreducible quantizations of this algebra 
	(with $\hat{L}_k$ self-adjoint)
	may be labelled by a number
	$\ell \in \Z_{\ge 0}/2$ (called \keyword{spin}), and that the corresponding
	spectrum is $\sigma(\hat{L}_3) = \hbar\{-\ell,-\ell+1,\ldots,\ell\}$ 
	and $\hat{\bL}^2 = \hbar^2\ell(\ell+1)\1$.
\end{exc}

\subsection{The two-body problem and the hydrogenic atom}\label{sec:mech-QM-2p}
	
	In the case that one considers \emph{two} particles on $\R^d$,
	the classical configuration space would be $\phX^{2d} = \R^d \times \R^d$
	and the phase space $\phP^{2d} = \R^{4d}$.
	We write the corresponding coordinates $\sx = (\bx_1,\bx_2)$
	and momenta $\ssp = (\bp_1,\bp_2)$, with $\bx_j, \bp_k \in \R^d$.
	For the Hamiltonian, one could here think of adding 
	the kinetic energies $T_j(\bp_j) = \bp_j^2/(2m_j)$ 
	for each of the two particles $j=1,2$,
	and also allow for some potential on configuration space 
	$V\colon \phX^{2d} \to \R$ 
	which depends on both particles:
	\begin{equation}\label{eq:Hamiltonian-2p}
		H(\sx,\ssp) 
		= T_1(\bp_1) + T_2(\bp_2) + 
			V(\bx_1,\bx_2)
		= \frac{\bp_1^2}{2m_1} + \frac{\bp_2^2}{2m_2} + V(\bx_1,\bx_2)
	\end{equation}
	Note that the masses $m_1$ resp. $m_2$ of the particles could be
	different, and that we may also w.l.o.g. rewrite the potential $V$
	into a sum of independent one-particle parts $V_j$ 
	and a final part $W$ describing any 
	\keyword{correlation} or \keyword{interaction}
	between the two particles,
	$$
		V(\bx_1,\bx_2) = V_1(\bx_1) + V_2(\bx_2) + W(\bx_1,\bx_2),
	$$
	thus
	$$
		H(\sx,\ssp) = \sum_{j=1,2} H_j(\bx_j,\bp_j) 
			+ W(\bx_1,\bx_2),
		\qquad
		H_j(\bx_j,\bp_j) = \frac{\bp_j^2}{2m_j} + V_j(\bx_j).
	$$
	In case there is no correlation between the particles, 
	$W=0$, 
	this hence just describes a sum of two independent 
	one-body Hamiltonians
	for which we may proceed with quantization as in Section~\ref{sec:mech-QM-1p}.
	If we can find self-adjoint representations of the corresponding
	operators $\hH_j \in \cL(\gH)$ on one-particle Hilbert spaces 
	$\gH = L^2(\phX^d) = L^2(\R^d)$ (the same $\phX^d$ for the two particles),
	then we can take the \keyword{two-body Hilbert space} 
	to be the tensor product
	$\cH = \gH \otimes \gH \cong L^2(\R^{2d})$ (see Exercise~\ref{exc:tensor-L2})
	and form a self-adjoint Hamiltonian operator
	$$
		\hH = \hH_1 \otimes \1 + \1 \otimes \hH_2 \in \cL(\cH),
		\qquad \cD(\hH) = \cD(\hH_1) \otimes \cD(\hH_2),
	$$
	with spectrum
	$\sigma(\hH) = \overline{\sigma(\hH_1) + \sigma(\hH_2)}$
	(see e.g.~\cite[Section~4.6]{Teschl-14}).
	For brevity we will usually leave out the trivial factors in the tensor 
	products if it is understood on which part of the space the operator acts.
	Also note that the expression for $\hH$ exponentiates to a unitary time 
	evolution on \keyword{two-body states}
	$\Psi = \psi \otimes \phi \in \cH$ (and linear combinations):
	$$
		e^{it\hH/\hbar}\Psi 
		= (e^{it\hH_1/\hbar} \otimes e^{it\hH_2/\hbar})(\psi \otimes \phi)
		= e^{it\hH_1/\hbar}\psi \otimes e^{it\hH_2/\hbar}\phi.
	$$

\begin{example}\label{exmp:harm-osc-split}
	The $d$-dimensional harmonic oscillator from Example~\ref{exmp:harm-osc-QM},
	$$
		H(\bx,\bp) = \frac{\bp^2}{2m} + \frac{1}{2}m\omega^2 \bx^2
		= \sum_{j=1}^d \left( \frac{p_j^2}{2m} + \frac{1}{2}m\omega^2 x_j^2 \right),
	$$
	separates into $d$ copies of a one-dimensional oscillator, with no 
	correlation between these different degrees of freedom,
	and hence it suffices to solve the one-dimensional problem
	to determine the full spectrum:
	$\sigma(\hH) = \sum_{j=1}^d \omega(\Zplus + 1/2) = \omega(\Zplus + d/2)$.
\end{example}
	
	In the case that there are correlations between the particles,
	$W \neq 0$, 
	it may be helpful to change variables.
	For the Hamiltonian \eqref{eq:Hamiltonian-2p} we define the
	\keyword{center of mass} (COM) and its conjugate momentum
	\begin{equation}\label{eq:COM-coords}
		\bX := \frac{m_1\bx_1 + m_2\bx_2}{m_1+m_2},
		\qquad
		\bP := \bp_1 + \bp_2,
	\end{equation}
	as well as the \keyword{relative coordinate}
	with its conjugate momentum
	\begin{equation}\label{eq:relative-coords}
		\br := \bx_1 - \bx_2,
		\qquad
		\bp_{\br} := \mu(\bp_1/m_1 - \bp_2/m_2).
	\end{equation}
	Here
	\begin{equation}\label{eq:reduced-mass}
		\mu := \frac{m_1m_2}{m_1+m_2}
	\end{equation}
	is the \keyword{reduced mass} of the pair of particles.
	We may then rewrite the two-body Hamiltonian \eqref{eq:Hamiltonian-2p}
	in these coordinates as (see Exercise~\ref{exc:COM-brackets})
	$$
		H(\bx_1,\bx_2;\bp_1,\bp_2) 
		= \frac{\bP^2}{2(m_1+m_2)} + \frac{\bp_{\br}^2}{2\mu}
			+ V\Bigl( \bX + \frac{\mu}{m_1}\br, \bX - \frac{\mu}{m_2}\br \Bigr)
		=: \tilde{H}(\bX,\br;\bP,\bp_\br).
	$$

\begin{example}[The hydrogenic atom]\label{exmp:hydrogenic-2p}
	Recall from Section~\ref{sec:mech-CM-instability} that the 
	hydrogenic atom consists of a nucleus with charge $Z>0$ 
	and an electron with charge $-1$,
	which interact via the Coulomb potential
	$V_\sC(\br) = -Z/|\br|$, 
	where $|\br|$ is the distance between the particles in $\R^3$.
	The proper model is therefore a two-body classical configuration space 
	$\phX^{3 \times 2} = \R^3 \times \R^3$, phase space 
	$\phP^{3 \times 2} = \R^{12}$, and Hamiltonian
	$$
		H(\bx_1,\bx_2;\bp_1,\bp_2) 
		= \frac{\bp_1^2}{2m_1} + \frac{\bp_2^2}{2m_2} - \frac{Z}{|\bx_1-\bx_2|},
	$$
	with $m_1$ and $m_2$ the masses of the nucleus and electron, respectively.
	It is here appropriate to switch to COM and relative coordinates,
	$\phX^{3 \times 2} \cong \phX_\COM \times \phX_\rel
	\ni (\bX,\br)$, 
	in which the Hamiltonian separates,
	$$
		\tilde{H}(\bX,\br;\bP,\bp_\br) 
		= \frac{\bP^2}{2(m_1+m_2)} + \frac{\bp_{\br}^2}{2\mu}
			- \frac{Z}{|\br|}.
	$$
	The first term involving the center-of-mass momentum $\bP$ is just the 
	one-body Hamiltonian $H_\COM(\bX;\bP)$
	of a free particle on $\phX_\COM = \R^3$, 
	which we may quantize uniquely along the lines of Example~\ref{exmp:free-QM},
	while the second two terms constitute the Hamiltonian 
	$H_\rel(\br;\bp_\br)$
	of the one-body Coulomb problem on $\phX_\rel = \R^3$ 
	which was discussed classically in 
	Section~\ref{sec:mech-CM-instability}
	and quantum-mechanically in Example~\ref{exmp:Coulomb-QM}
	of Section~\ref{sec:mech-QM-1p}, and which we shall return to many more times.
	Given a self-adjoint quantization 
	$\hH_\rel \in \cL(L^2(\phX_\rel))$ 
	of this Hamiltonian, we therefore have the two-particle operator
	$$
		\hH = \frac{\hbar^2(-\Delta_{\bX})}{2(m_1+m_2)} \otimes \1 
			+ \1 \otimes \hH_\rel, \qquad
		\cD(\hH) = H^2(\R^3) \otimes \cD(\hH_\rel),
	$$
	on the two-particle quantum configuration space
	$\cH = L^2(\R^3,d\bX) \otimes L^2(\R^3,d\br)$.
	Again then,
	$$
		\sigma(\hH) = \overline{ \sigma(\hH_\COM) + \sigma(\hH_\rel) }
		= \overline{[0,\infty) + \sigma(\hH_\rel)} 
		= [\inf \sigma(\hH_\rel),\infty),
	$$
	and we see that the spectrum of the free center-of-mass motion in the end
	obscures most of the information about the relative one. 
	Therefore, in practice one often removes
	the COM part from the problem altogether and then studies only the more 
	interesting relative part. 
\end{example}
		
\begin{exc}\label{exc:tensor-L2}
	Use Fubini's theorem to show that 
	$L^2(X,\mu) \otimes L^2(Y,\nu) \cong L^2(X \times Y, \mu \times \nu)$, 
	where $\mu,\nu$ are $\sigma$-finite measures 
	and, for two Hilbert spaces $\cH, \cK$, $\cH \otimes \cK$ is defined with
	$$
		\inp{u_1 \otimes u_2,v_1 \otimes v_2}_{\cH \otimes \cK} 
		:= \inp{u_1,v_1}_{\cH} \inp{u_2,v_2}_{\cK}.
	$$
\end{exc}

\begin{exc}\label{exc:COM-brackets}
	Verify that the COM and relative coordinates and momenta 
	\eqref{eq:COM-coords}-\eqref{eq:relative-coords}
	are canonically conjugate, i.e
	$$
		\{(\bX)_j,(\bP)_k\} = \delta_{jk}, \qquad
		\{(\br)_j,(\bp_\br)_k\} = \delta_{jk},
	$$
	and all other Poisson brackets in $\bX,\br,\bP$  and $\bp_\br$ are zero,
	and furthermore that
	$$
		m_1\bx_1^2 + m_2\bx_2^2 = (m_1+m_2)\bX^2 + \mu\br^2, \qquad
		\frac{\bp_1^2}{m_1} + \frac{\bp_2^2}{m_2} 
			= \frac{\bP^2}{m_1+m_2} + \frac{\bp_\br^2}{\mu}.
	$$
\end{exc}

\subsection{The $N$-body problem}\label{sec:mech-QM-Np}

	We may extend much of the above analysis to the \keyword{$N$-body problem},
	that is, we may consider $N$ particles on $\R^d$ with masses $m_j$
	and with independent one-body potentials $V_j(\bx_j)$, 
	two-body correlation potentials $W_{jk}(\bx_j,\bx_k)$ 
	describing pairwise interactions between particles $\bx_j$ and $\bx_k$
	with $j \neq k$,
	as well as three-body interactions $W_{jkl}(\bx_j,\bx_k,\bx_l)$ 
	with $j,k,l$ all distinct,
	and so on.
	Although interaction potentials involving more than two particles 
	are not uncommon in physics,
	they will not be relevant for our stability of matter problem and shall
	hence for simplicity not be considered further in this course
	(except perhaps occasionally).
	Furthermore, it is common that the one-body potentials have already incorporated
	all the dependence on absolute positions such as pairwise centers of mass,
	and hence that the pair-interactions
	$W_{jk}$ are translation-invariant, i.e., 
	depending only on the relative
	coordinate $\br_{jk} := \bx_j - \bx_k$ of each pair
	(w.l.o.g. $j<k$).
	
	With these restrictions or simplifications, 
	the \keyword{$N$-body Hamiltonian} on the 
	classical configuration space $\phX^{d \times N} = (\R^d)^N$ and
	phase space $\phP^{d \times N} = (\R^d)^N \times (\R^d)^N$
	may thus be defined
	\begin{equation}\label{eq:many-body-Hamiltonian}
		H(\sx,\ssp) := \sum_{j=1}^N \left( T_j(\bp_j) + V_j(\bx_j) \right)
			+ \sum_{1 \le j<k \le N} W_{jk}(\bx_j - \bx_k)
			= T(\ssp) + V(\sx) + W(\sx).
	\end{equation}
	Again we see that there is no ordering ambiguity here since the
	terms involve coordinates and momenta separately.
	The natural quantum version of the expression is therefore
	\begin{equation}\label{eq:many-body-Hamiltonian-op}
		\hH := \sum_{j=1}^N \left( T_j(\hbp_j) + V_j(\hbx_j) \right)
			+ \sum_{1 \le j<k \le N} W_{jk}(\hbx_j - \hbx_k)
			= \hT + \hV + \hW,
	\end{equation}
	which is to be
	acting as an operator on $\cH = L^2(\phX^{d \times N}) = L^2(\R^{dN})$,
	that is, we should try to implement these expressions as operators or
	forms on some space $\scF$ of sufficiently well-behaved functions
	$\Psi \in L^2(\phX^{d \times N})$,
	called \keyword{$N$-body quantum states} or \keyword{$N$-body wave functions},
	according to
	\begin{align}
	\label{eq:many-body-kinetic-op}
		\hT\Psi(\sx) &:= \sum_{j=1}^N \frac{\hbp_j^2}{2m_j} \Psi \,(\sx)
		= -\sum_{j=1}^N \frac{\hbar^2}{2m_j} \Delta_{\bx_j}\!\Psi(\sx), \\
	\label{eq:many-body-1p-potential-op}
		\hV\Psi(\sx) &:= \sum_{j=1}^N V_j(\bx_j) \Psi(\sx), \\
	\label{eq:many-body-2p-potential-op}
		\hW\Psi(\sx) &:= \sum_{1 \le j<k \le N} W_{jk}(\bx_j-\bx_k) \Psi(\sx).
	\end{align}
	Which space $\scF$ we may choose depends on details of the potentials $V$ and $W$,
	namely, if for example the interaction $W$ is too singular then this
	may force us to consider only those functions which vanish 
	(sufficiently fast) at the singularities.
	Hence, the usual minimal domain $C_c^\infty(\R^{dN}) \subseteq \cH$ 
	might not always be appropriate, as the following example illustrates.
	
\begin{example}[Hard-core interaction]\label{exmp:hard-core}
	Consider an interaction potential $W_R(\br)$ formally defined by
	$$
		W_R(\br) = \left\{ \begin{array}{ll}
			+\infty,	& \text{if $|\br| < R$,} \\
			0, 		& \text{if $|\br| \ge R$.}
			\end{array}\right.
	$$
	This describes \keyword{hard spheres} (or \keyword{hard cores}) 
	of radius $R/2$, 
	because as soon as the particles are within a distance $|\br| = |\bx_j-\bx_k| < R$
	the energy is infinite --- a very hard collision ---
	and otherwise they do not see each other. 
	Since the corresponding form on the relative Hilbert space $L^2(\R^d,d\br)$ 
	is formally
	$$
		\inp{\psi, W_R \psi} = \int_{\R^d} W_R |\psi|^2 
		= \left\{ \begin{array}{ll}
			+\infty,	& \text{if $|B_R(0) \cap \supp\psi| > 0$,} \\
			0, 		& \text{otherwise,}
			\end{array}\right.
	$$
	the mathematically precise way to incorporate this potential is to
	consider a \emph{new} minimal domain 
	$\cD(W_R) = C_c^\infty(B_R(0)^c) \subseteq \cH$
	and a corresponding restriction in the Hilbert space 
	$\cH = L^2(B_R(0)^c) = \overline{\cD(W_R)}$ 
	(the closure may be taken in the old Hilbert space $L^2(\R^d)$).
\end{example}
		
	Note that (again by Exercise~\ref{exc:tensor-L2})
	we may equivalently think of $\Psi \in L^2(\R^{dN}) \cong \otimes^N \gH$
	as 
	tensor products (including any finite linear combinations and limits thereof)
	of one-particle states $\psi_n \in \gH = L^2(\R^d)$,
	$$
		\Psi = \sum_{n_1=1}^{\infty} \ldots \sum_{n_N=1}^\infty 
			c_{n_1 \ldots n_N}
			\psi_{n_1} \otimes \ldots \otimes \psi_{n_N},
		\qquad c_{n_1 \ldots n_N} \in \C.
	$$
	In the case that all $W_{jk} = 0$ we again have a separation
	of the problem into independent one-body problems,
	\begin{equation}\label{eq:many-body-Hamiltonian-independent}
		\hH = \sum_{j=1}^N \hat{h}_j, \qquad
		\hat{h}_j = -\frac{\hbar^2}{2m_j}\Delta_{\bx_j} + V_j(\bx_j) 
			\ \in \cL(\gH),
	\end{equation}
	and, if these are subsequently realized as self-adjoint operators,
	then
	$
		\sigma(\hH) = \overline{\sum_{j=1}^N \sigma(\hat{h}_j)}.
	$
	
	Also, if $W_{jk} \neq 0$, 
	then one may for each pair of particles instead
	consider the problem in the corresponding relative coordinates $\br_{jk}$.
	However, because for $N>2$ 
	there are more pairs than relative degrees of freedom,
	$\binom{N}{2} > N-1$,
	this forms a redundant set of variables and
	the problem typically does not separate.
	The total center of mass, 
	$$
		\bX = \sum_{j=1}^N m_j\bx_j \bigg/ \sum_{j=1}^N m_j, \qquad
		\bP = \sum_{j=1}^N \bp_j,
	$$
	may still be separated away though if the one-body potential
	$V$ admits such a separation.
	We will not consider the appropriate change of variables in the general case,
	involving Jacobi coordinates, but only in the below special 
	case of identical masses.
	
\subsubsection{Models of matter and notions of stability}\label{sec:mech-QM-matter}

	An important special case is that all the particles are of exactly the same kind
	so that we have the same mass $m_j=m$ and one-particle interaction $V_j=V$ for all $j$,
	and also that the two-particle interaction $W_{jk}=W$ is independent of the pair
	considered and furthermore symmetric w.r.t. particle exchange
	$\br_{jk} \mapsto -\br_{jk} = \br_{kj}$, i.e.
	$W(\br_{jk}) = W(-\br_{jk})$.
	The resulting Hamiltonian operator
	\begin{equation}\label{eq:many-body-Hamiltonian-ident}
		\hH^N := \sum_{j=1}^N \left( \frac{\hbar^2}{2m} (-\Delta_{\bx_j}) + V(\bx_j) \right)
			+ \sum_{1 \le j<k \le N} W(\bx_j - \bx_k)
	\end{equation}
	then defines the \keyword{typical $N$-body quantum system} 
	involving a single type of particle.
	
	In this case it may be useful to write the momenta as a sum of pairs
	using the \keyword{generalized parallelogram identity}\index{parallelogram identity}
	(valid on $\C^d$ or general Hilbert spaces; 
	cf. Example~\ref{exmp:parallelogram-space}),
	\begin{equation}\label{eq:many-body-parallelogram}
		\sum_{j=1}^N |\bp_j|^2 = \frac{1}{N} \Biggl| \sum_{j=1}^N \bp_j \Biggr|^2
			+ \frac{1}{N} \sum_{1 \le j<k \le N} \bigl| \bp_j - \bp_k \bigr|^2,
	\end{equation}
	and thus for the Hamiltonian
	$$
		\hH^N = \frac{\hbar^2}{2mN} (-\Delta_{\bX}) 
			+ \sum_{1 \le j < k \le N} \left( 
				\frac{1}{2mN} \bigl| \hbp_j - \hbp_k \bigr|^2
				+ \frac{1}{N-1}\bigl(V(\bx_j) + V(\bx_k)\bigr)
				+ W(\br_{jk}) \right).
	$$
	Again, although it looks as if we may have separated the problem if $V=0$,
	this is indeed true for the COM variable but
	the particle pairs are actually \emph{not} independent.
	
	Another important case will in fact constitute our \keyword{model of matter} 
	in the sequel.
	Here we have \emph{two} species of particles: $N$ electrons and $M$ nuclei, 
	with positions $\bx_j \in \R^d$ respectively $\bR_k \in \R^d$.
	The quantum Hamiltonian on $L^2(\phX^{d \times N} \times \phX^{d \times M})$
	is
	\begin{equation}\label{eq:many-body-Hamiltonian-matter}
		\hH^{N,M} := \sum_{j=1}^N 
			\frac{\hbar^2}{2m_e} (-\Delta_{\bx_j})
			+ \sum_{k=1}^M 
			\frac{\hbar^2}{2m_n} (-\Delta_{\bR_k})
			+ W_\sC(\sx,\sR),
	\end{equation}
	with masses $m_e > 0$ respectively $m_n > 0$, and
	where we have taken as the interaction the
	\keyword[Coulomb potential]{$(N,M)$-body Coulomb potential}:
	\begin{equation}\label{eq:many-body-Coulomb}
		W_\sC(\sx,\sR) := \sum_{1 \le i<j \le N} \frac{1}{|\bx_i-\bx_j|}
			- \sum_{j=1}^N \sum_{k=1}^M \frac{Z}{|\bx_j-\bR_k|}
			+ \sum_{1 \le k<l \le M}  \frac{Z^2}{|\bR_k-\bR_l|}
	\end{equation}
	This implements the appropriate Coulomb interaction \eqref{eq:Coulomb-potential}
	between each pair of particles, where the charges are again $q_e = -1$
	for the electrons and $q_n = Z > 0$ for the nuclei.
	One may also add external one-body potentials $V_e$ respectively $V_n$,
	although we will not do so here but rather consider the whole system
	\eqref{eq:many-body-Hamiltonian-matter}
	of $N+M$ particles to be completely \emph{free} apart from the internal 
	interactions in $W_\sC$.
	Sometimes we may however consider the kinetic energies of the nuclei
	to be irrelevant for the problem since in reality $m_n \gg m_e$
	and thus we could consider this as a limit $m_n \to \infty$.
	We may in any case drop the non-negative terms 
	$(-\Delta_{\bR_k})/(2m_n) \ge 0$ for a lower bound to $\hH^{N,M}$.
	Upon doing so the positions $\sR = (\bR_1,\ldots,\bR_M)$
	of the nuclei remain as parameters of the 
	resulting $N$-body Hamiltonian $\hH^N(\sR)$
	and some of the terms of the interaction $W_\sC$ are then treated as
	external potentials.
	
\begin{definition}[\keyword{Ground-state energy, and stability of the first and second kind}]
	\label{def:stability}
	Given a quantum system modeled on a Hilbert space $\cH$ with a self-adjoint 
	Hamiltonian operator $\hH \in \cL(\cH)$, we define its 
	\keyword{ground-state energy}
	to be the infimum of the spectrum,
	$$
		E_0 := \inf \sigma(\hH) = \inf_{\psi \in \cQ(\hH) \setminus \{0\}} \inp{\hH}_\psi.
	$$
	We say that the system
	is \keyword{stable of the first kind} iff $\hH$ is bounded from below,
	$$
		E_0 > -\infty.
	$$
	Moreover, in the case that
	the system depends on a total of $N$ particles, with
	$\cH = \cH^N$, $\hH = \hH^N$, and
	$$
		E_0(N) := \inf \sigma(\hH^N) = \inf_{\Psi \in \cQ(\hH^N) \setminus \{0\}} \inp{\hH^N}_\Psi,
	$$
	then it is called \keyword{stable of the second kind} 
	iff $\hH^N$ admits a lower bound which is at most \emph{linearly} divergent in $N$,
	i.e. iff there exists a constant $C\ge0$ such that for all $N \in \N$
	$$
		E_0(N) \ge -CN.
	$$
\end{definition}

\begin{remark}
	A stable system does not necessarily have a \keyword{ground state}, 
	i.e. some eigenstate $\psi \in \cH \setminus \{0\}$ with energy equal to
	the ground-state energy $E_0$, 
	$\hH\psi = E_0\psi$.
	For example the free particle on $\R^d$ (Example~\ref{exmp:free-QM}) \index{free particle}
	is certainly stable with $E_0=0$ but has no ground state,
	since the only sensible candidate would be either the constant function 
	(in the usual form sense, $\inp{\psi,-\Delta\psi}=\norm{\nabla\psi}^2 = 0$)
	or perhaps a harmonic function (in the operator sense, $-\Delta\psi=0$)
	which in either case is not in $L^2(\R^d)$.
\end{remark}

\subsubsection{Density and particle probabilities}\label{sec:mech-QM-density}

	The problem with the models of matter \eqref{eq:many-body-Hamiltonian-ident} and 
	\eqref{eq:many-body-Hamiltonian-matter} is that it is in practice \emph{extremely} 
	difficult to compute their spectra $\sigma(\hH)$,
	even numerically on any foreseeable supercomputer,
	since in reality $N$ is typically extremely large. 
	In just 1 gram of matter there are $N \sim 10^{23}$ particles
	(and, in fact, even the classical many-body problem 
	is then almost impossible to understand on the individual particle level).
	The approach one takes instead is to try to reduce this problem,
	which takes place 
	on the enormous classical configuration space $\R^{dN}$ with $N \gg 1$, 
	to an approximate problem on just $\R^d$ or similar fixed small dimension.

	Recall that if $\Psi$ is normalized in $L^2(\R^{dN})$ 
	--- which we shall assume from now on for our quantum states --- then
	$|\Psi(\sx)|^2$ may be interpreted as the probability density 
	of finding the particles at positions 
	$\sx = (\bx_1,\ldots,\bx_N) \in (\R^d)^N$.
	We may however instead define a corresponding particle density on 
	the one-body configuration space $\R^d$:
	
\begin{definition}[One-body density]
	The \keyword{one-body density} associated to a \emph{normalized}
	$N$-body wave function $\Psi \in L^2((\R^d)^N)$ 
	is the function $\varrho_\Psi \in L^1(\R^d)$ given by
	\begin{equation}\label{eq:def-density}
		\varrho_\Psi(\bx) := 
		\sum_{j=1}^N \int_{\R^{d(N-1)}} |\Psi(\bx_1, \ldots, 
		\bx_{j-1}, \bx, \bx_{j+1}, \ldots, \bx_N)|^2 \prod_{k \neq j}d\bx_k.
	\end{equation}
\end{definition}

	The interpretation of this expression
	is that it is a sum of contributions to a particle
	density, where 
	each term gives the probability of finding particle $j$ at $\bx \in \R^d$
	while all the other particles are allowed to be anywhere in $\R^d$
	and hence have been integrated out.
	Because of the sum, we 
	have no information in $\varrho_\Psi(\bx)$
	which one of the particles
	was at $\bx$ but only how many were there on average.
	Indeed, $\int_\Omega \varrho_\Psi$ will be the expected number of particles 
	to be found on the set $\Omega \subseteq \R^d$, and
	we can write 
	$$
		\int_\Omega \varrho_\Psi = \inp{ \sum_{j=1}^N \1_{\{\bx_j \in \Omega\}} }_\Psi
		\qquad \text{and} \qquad
		\varrho_\Psi(\bx) = \inp{\sum_{j=1}^N \delta_{\bx_j}}_\Psi
	$$
	in accordance with the interpretations of 
	Remark~\ref{rem:self-adjointness}.\ref{itm:s-a-spectral-rep}
	and Section~\ref{sec:mech-QM-Schroedinger}.
	Note that $\int_{\R^d} \varrho_\Psi = N$ 
	since every particle has to be somewhere in $\R^d$.

	Further note that using $\varrho_\Psi$ we can now write for the expectation value 
	of the one-body potential
	$$
		\inp{\hat{V}}_\Psi = \int_{\R^d} V(\bx) \varrho_\Psi(\bx) \,d\bx,
	$$
	which is indeed a tremendous simplification of the full $N$-body form
	to only depend on the density. Unfortunately a similar straightforward
	simplification does not occur for the kinetic and interaction
	energies $\binp{\hT}_\Psi$ and $\binp{\hat{W}}_\Psi$, 
	and the task of physicists and mathematicians working in many-body quantum 
	theory is to try to find such simplifications.
	We shall come across some important instances of this later in the course.
	
	We will also find it useful to extract the 
	\keyword{local particle probability distribution}
	encoded in the full wave function $\Psi$.
	Namely, given a subset $\Omega \subseteq \R^d$ of the one-body 
	configuration space
	and a subset $A \subseteq \{1,2,\ldots,N\}$ of the particles 
	(particle labels),
	we may form the probability to find \emph{exactly} those particles on 
	$\Omega$ and the rest outside $\Omega$ 
	(i.e. on its complement $\Omega^c = \R^d \setminus \Omega$):
	$$
		p_{A,\Omega}[\Psi] := \inp{ \prod_{k \in A} \1_{\{\bx_k \in \Omega\}} 
			\prod_{k \notin A} \1_{\{\bx_k \in \Omega^c\}} }_\Psi
		= \int_{(\Omega^c)^{N-|A|}} \int_{\Omega^{|A|}} |\Psi|^2 
			\prod_{k \in A} d\bx_k \prod_{k \notin A} d\bx_k.
	$$
	The probability of finding exactly $n$ particles on $\Omega$ 
	\emph{irrespective of their labels} is then the sum of all such 
	possibilities
	$$
		p_{n,\Omega}[\Psi] := \sum_{\text{$A \subseteq \{1,\ldots,N\}$ s.t. $|A|=n$}}
			p_{A,\Omega}[\Psi].
	$$
	We then note that (exercise)
	\begin{equation}\label{eq:particle-prob-normalized}
		\sum_{n=0}^N p_{n,\Omega}[\Psi] 
		= \sum_{A \subseteq \{1,\ldots,N\}} p_{A,\Omega}[\Psi]
		= \int_{\R^{dN}} |\Psi|^2 = 1,
	\end{equation}
	in other words, some number of particles (possibly zero)
	or some subset (possibly the empty one) of the particles 
	must always be found on $\Omega$.
	Also, the \keyword{expected number of particles} to be found on 
	$\Omega$ is (exercise)
	\begin{equation}\label{eq:particle-prob-expectation}
		\sum_{n=0}^N n \,p_{n,\Omega}[\Psi] 
		= \sum_{j=1}^N \int_{\R^{dN}} \1_{\Omega}(\bx_j) |\Psi(\sx)|^2 \,d\sx
		= \int_\Omega \varrho_\Psi,
	\end{equation}
	which agrees with our earlier interpretations.

\begin{remark}\label{rem:density-matrix}
	There is also a more general concept of
	\keyword{density matrices}, which certainly is very useful 
	but will not be treated here.
	We refer instead to e.g. \cite[Chapter~3.1.4]{LieSei-09}.
\end{remark}

\begin{remark}[Fock spaces]
	For settings where the number $N$ of particles can vary with time
	it is necessary to introduce an appropriate space containing all different
	particle numbers, known as a \keyword{Fock space}:
	$$
		\cF = \bigoplus_{N=0}^\infty \cH^N = \C \oplus \gH \oplus \ldots
	$$
	where $\cH^N = \otimes^N \gH$ is the $N$-body space.
	The concept will not be applied in this course however.
\end{remark}

\begin{exc}\label{exc:particle-probabilities}
	Prove \eqref{eq:particle-prob-normalized} and \eqref{eq:particle-prob-expectation}
	by inserting and expanding the identity
	\begin{equation}\label{eq:particle-partition-of-unity}
		\1 = \prod_{k=1}^N \bigl(\1_{\Omega}(\bx_k) + \1_{\Omega^c}(\bx_k)\bigr).
	\end{equation}
\end{exc}

\subsection{Identical particles and quantum statistics}\label{sec:mech-QM-statistics}

	It will turn out that our model of matter \eqref{eq:many-body-Hamiltonian-matter}
	as presently formulated is actually \emph{unstable} with $N+M \to \infty$.
	Although it is very hard to see it in \eqref{eq:many-body-Hamiltonian-matter},
	and indeed we have not yet even settled stability for the hydrogen atom $N=M=1$,
	a picture one could keep in mind for now is that of a single atom with
	a large nucleus (or charge $Z \gg 1$) and many electrons 
	(say $N = Z$ to make the system neutral),
	which, if we may ignore their mutual Coulomb repulsion,
	would all prefer to sit in the tightest orbit with the lowest energy,
	and this turns out to diverge too fast with $N$ for stability.
	(You could at this stage think of the atom's energy levels as similar to the 
	harmonic oscillator energy levels,
	though in the attractive Coulomb potential of the nucleus they will be negative and 
	accumulating to zero, with the lowest one proportional to $-Z^2 \sim -N^2$.)
	However, this picture turns out not to be the correct one,
	not only because of the neglected Coulomb repulsion terms, but
	because of an additional \emph{fundamental} property of electrons which 
	was not visible classically.
	Namely, apart from the uncertainty principle arising as 
	a concequence of the
	non-commutativity relations of operator observables,
	an additional 
	pair of intimately related and
	fundamentally new concepts brought in by quantum mechanics 
	is that of \keyword{identical particles} and \keyword{quantum statistics}.

	We have already assumed in \eqref{eq:many-body-Hamiltonian-ident} and 
	\eqref{eq:many-body-Hamiltonian-matter} that all $N$ (or $M$) particles
	are of the same kind, for example electrons, which means that they all have the
	exact same physical properties such as mass and charge, 
	and therefore behave in the exact same way.
	In fact no measurement can ever distinguish 
	one such particle from another.
	In quantum mechanics, where the uncertainty principle sets fundamental
	(logical) limits to distinguishability,
	this becomes a very important logical distinction
	since the particles must therefore be treated as \keyword{logically identical}.
	In particular, 
	considering the probability density of an $N$-body state 
	$\Psi \in L^2((\R^d)^N)$
	of such \keyword{indistinguishable particles}, 
	we must have a symmetry upon exchanging two particles $j$ and $k$,
	\begin{equation}\label{eq:particle-exchange-amplitude}
		|\Psi(\bx_1, \ldots, \bx_j, \ldots, \bx_k, \ldots, \bx_N)|^2 
		= |\Psi(\bx_1, \ldots, \bx_k, \ldots, \bx_j, \ldots, \bx_N)|^2, 
		\quad j \neq k,
	\end{equation}
	since we cannot tell which one is which.
	However, since $\Psi$ takes values in $\C$, 
	this relation involving only the amplitude
	would still allow for a phase difference, 
	\begin{equation}\label{eq:particle-exchange-phase}
		\Psi(\bx_1, \ldots, \bx_j, \ldots, \bx_k, \ldots, \bx_N)
		= e^{i\theta_{jk}} \Psi(\bx_1, \ldots, \bx_k, \ldots, \bx_j, \ldots, \bx_N),
		\quad j \neq k.
	\end{equation}
	One may then argue 
	that a double exchange does nothing to the state
	(the square of a transposition is the identity)
	so $(e^{i\theta_{jk}})^2 = 1$, that is $\theta_{jk} = 0$ or $\pi$.
	Further, by considering expectation values of symmetric $N$-particle operator observables
	one may also realize that these phases cannot depend on $j$ and $k$,
	and the only possibility is then that $\Psi$ satisfies either
	\begin{equation}\label{eq:particle-exchange-bosons}
		\Psi(\bx_{\sigma(1)}, \bx_{\sigma(2)}, \ldots, \bx_{\sigma(N)})
		= \Psi(\bx_1, \bx_2, \ldots, \ldots, \bx_N),
	\end{equation}
	for any permutation $\sigma \in S_N$,
	in which case we refer to these $N$ identical particles as \keyword{bosons},
	or
	\begin{equation}\label{eq:particle-exchange-fermions}
		\Psi(\bx_{\sigma(1)}, \bx_{\sigma(2)}, \ldots, \bx_{\sigma(N)})
		= \sign(\sigma) \Psi(\bx_1, \bx_2, \ldots, \ldots, \bx_N),
	\end{equation}
	for which the particles are instead called \keyword{fermions}.
	Hence, this amounts to a reduction of the full Hilbert space
	$\cH = L^2(\R^{dN})$
	of (distinguishable) $N$-body states into the \emph{symmetric} subspace
	$$
		\cH_\sym = L_\sym^2((\R^d)^N) := \bigl\{ \Psi \in L^2(\R^{dN}) : 
			\text{$\Psi$ satisfies \eqref{eq:particle-exchange-bosons} $\forall \sigma \in S_N$} \bigr\}
			\cong \bigotimes\nolimits_\sym^N \gH, 
	$$
	or the \emph{antisymmetric} subspace
	$$
		\cH_\asym = L_\asym^2((\R^d)^N) := \bigl\{ \Psi \in L^2(\R^{dN}) : 
			\text{$\Psi$ satisfies \eqref{eq:particle-exchange-fermions} $\forall \sigma \in S_N$} \bigr\}
			\cong \bigwedge\nolimits^N \gH.
	$$
	
	The above argument to settle the phase ambiguity of $\Psi$ under
	particle exchange
	was the standard one in the first half of a century after the
	invention of quantum mechanics,
	and is in fact still today commonly applied in physics textbooks without 
	further discussion.
	However, in the 1970's it was clarified 
	(see \cite{LeiMyr-77}, or e.g. \cite{Myrheim-99} for review\footnote{%
	Note however 
	that there were plenty of earlier hints, and the story 
	actually goes back all the way to Gibbs' classical statistical mechanics; 
	see e.g. \cite{Froehlich-90} and references therein, as well as \cite[p.~386]{Souriau-70}.})
	that this is actually not the appropriate way to think about the
	problem for identical particles, but 
	rather that the \emph{classical} configuration space 
	$\phX^{d \times N} = \R^{dN}$
	should be replaced with 
	the \emph{symmetrized} one
	$$
		\phX_\sym^{d \times N} := \left( (\R^d)^N \setminus \bDelta \right) \Big/ S_N,
	$$
	simply because there is no way to distinguish the particles even classically.
	Here 
	we have first removed the set of particle coincidences\footnote{%
	This can be motivated by the fact that if some positions exactly coincide 
	then we cannot tell if there really are $N$ particles, which is what we want to 
	consider here. It may also be justified a posteriori \cite{BorSor-92}.},
	i.e. the \keyword{fat diagonal}\index{diagonal}\footnote{The 
	\emph{thin} diagonal would be the points $\sx \in \R^{dN}$ 
	such that $\bx_1 = \bx_2 = \ldots = \bx_N$.}
	of the configuration space
	\begin{equation}\label{eq:fat-diagonal}
		\bDelta 
		:= \{ (\bx_1,\ldots,\bx_N) \in (\R^d)^N : \text{$\exists\ j \neq k$ s.t. $\bx_j = \bx_k$} \},
	\end{equation}
	and then taken the quotient under the action of the group $S_N$ 
	of particle permutations,
	$$
		\sigma\colon (\bx_1,\ldots,\bx_N) \mapsto (\bx_{\sigma(1)}, \bx_{\sigma(2)}, \ldots, \bx_{\sigma(N)}),
	$$
	to obtain the set of \keyword{proper $N$-point subsets} of $\R^d$:
	$$
		\phX_\sym^{d \times N} = \bigl\{ A = \{\bx_1,\ldots,\bx_N\} \subseteq \R^d : |A|=N \bigr\}.
	$$
	This is the natural space of configurations for $N$ truly indistinguishable 
	\index{indistinguishable particles}
	(logically identical) particles,
	and while it changes very little on the classical side,
	e.g.\ only marginally the space where one may define potentials $V(\sx)$
	since they anyway have to be symmetric under permutations, it does affect the 
	possible quantizations of the free kinetic energy 
	$T(\ssp) = (2m)^{-1}|\ssp|^2$,
	which turn out to depend on the non-trivial \emph{topology} 
	of this configuration space.
	In particlar, particle exchange no longer makes sense as a permutation of indices
	(note that this is 
	the identity operation on the quotient 
	$\phX_\sym^{d \times N}$)
	but should instead be considered as a \emph{continuous} 
	operation which relates points in the configuration space.
	The consequences of this approach have by now been studied in detail 
	and are quite well understood even on a strict mathematical level 
	(although some important questions still remain open). 
	To give the full story would be a course in itself, however,
	and we will only state the main points here, 
	guided by the following simple example:
	
\begin{example}[Two identical particles]
	Consider just one pair of particles on $\R^d$ 
	whose configuration space in the distinguishable case is 
	$\phX^{d \times 2} = \R^d \times \R^d \ni (\bx_1,\bx_2)$,
	while the indistinguishable one is
	$$
		\phX_\sym^{d \times 2} = \Bigl( 
			\R^d \times \R^d \setminus \bDelta 
			\Bigr)\Big/_{\!\sim} \ ,
		\qquad
		\bDelta = \{(\bx,\bx) : \bx \in \R^d \},
	$$
	with the identification $(\bx_1,\bx_2) \sim (\bx_2,\bx_1)$ of the particles.
	The geometry of this space becomes more transparent upon 
	changing to COM and relative coordinates, 
	$(\bX,\br) \in \phX_\COM \times \phX_\rel$.
	The space then separates into
	$\phX_{(\sym)}^{d \times 2} \cong \phX_\COM \times \phX_\rel^{(\sym)}$
	where $\phX_\COM = \R^d$,
	and the distinguishable relative space is
	$\phX_\rel = \R^d$ while the indistinguishable one is
	$\phX_\rel^\sym = (\R^d \setminus \{0\})/_\sim$
	(note that $\bDelta \cong \phX_\COM \times \{0\}$),
	with the antipodal identification $\br \sim -\br$.
	The relative space may finally be parameterized in terms 
	of the pairwise distance $r > 0$
	and a relative angle $\omega \in \S^{d-1}/_\sim$.
	Note that the topology of $\phX_\rel^\sym$ varies 
	markedly with dimension, namely consider the
	fundamental group $\pi_1$ of this topological space:
	\begin{align*}
		\pi_1(\phX_\rel^\sym) \cong \pi_1(\S^{d-1}/_\sim)
		= \left\{ \begin{array}{ccll}
			\pi_1(\R P^{d-1}) &\cong& \Z_2, \qquad & d \ge 3, \\
			\pi_1(\S^1) &\cong& \Z, \qquad & d = 2, \\
			\pi_1(\{1\}) &\cong& 1, \qquad & d = 1.
			\end{array}\right.
	\end{align*}
	This group is by definition the group of continuous loops in the space
	modulo continuous deformations (homotopy equivalences),
	and it is exactly this group which describes the non-trivial
	\emph{continuous} particle exchanges, amounting to continuous loops 
	or particle trajectories
	$\gamma \subset \phX_\rel^\sym$ modulo any loops that are topologically trivial. 
	If we assign a complex phase $e^{i\theta}$ to a \emph{simple} 
	such non-trivial exchange loop,
	i.e.\ to the generator $\tau$ of the group $\pi_1(\phX_\rel^\sym)$, 
	then in the case $d \ge 3$, where $\tau^2 = 1$ and the group is $\Z_2$,
	we must for consistency have that $e^{i2\theta} = 1$
	and hence either $\theta = 0$ or $\theta = \pi$.
\end{example}
	
	Now, very briefly, 
	in the general $N$-particle case
	the possibilities for the free kinetic energy operator depend 
	in a similar way critically
	on the dimension $d$ of the one-particle space, and in particular on the 
	fundamental group $\pi_1$ of the $N$-particle configuration space 
	(again the non-trivial continuous particle exchanges are described by
	the group of loops in the configuration space modulo homotopy equivalences).
	This is 
	$$
		\pi_1(\phX_\sym^{d \times N}) = \left\{ \begin{array}{ll}
			S_N, \qquad & d \ge 3, \\
			B_N, \qquad & d = 2, \\
			1, \qquad & d = 1,
			\end{array}\right.
	$$
	where $B_N$ is called the \keyword{braid group} on $N$ strands, 
	and $1$ is the trivial group.
	The possible quantizations $\hT$ such that they reduce locally
	to the usual one for free distinguishable particles are then labeled by 
	(irreducible, unitary) representations of these groups as complex phases,
	i.e.\ homomorphisms
	\begin{equation}\label{eq:particle-exchange-rep}
		\rho\colon \pi_1(\phX_\sym^{d \times N}) \to U(1),
	\end{equation}
	which in the case $d \ge 3$ of the permutation group
	reduces to only two possibilities:
	\begin{equation}\label{eq:particle-exchange-rep-3D}
		\rho = 1 \ \text{(the trivial representation)}, 
		\qquad \text{or} \qquad \rho = \sign.
	\end{equation}
	These can be shown to correspond to the above-defined bosons 
	respectively fermions, namely, after choosing one of these
	representations one may in fact extend the configuration space again to obtain 
	$\R^{dN} \setminus \bDelta$
	with the corresponding $N$-body wave functions satisfying either
	\eqref{eq:particle-exchange-bosons} or \eqref{eq:particle-exchange-fermions},
	and after closing up the space $\overline{\R^{dN} \setminus \bDelta} = \R^{dN}$
	one is finally left with $\cH_\sym$ or $\cH_\asym$.
	
	On the other hand,
	in the case $d=2$ it turns out
	one has a full unit circle of possibilities:
	$$
		\rho(\tau_j) = e^{i\alpha\pi}, \qquad \alpha \in [0,2) \ \ \text{(periodic)},
	$$
	with $\alpha$ the same for each of the generators 
	$\tau_j$, $j=1,2,\ldots,N-1$, of the group $B_N$ 
	(see Exercise~\ref{exc:braid-relations}).
	The corresponding particles are called \keyword{anyons} 
	(as in `\emph{any} phase' \cite{Wilczek-82b})
	with \keyword{statistics parameter} $\alpha$,
	and in the case $\alpha=0$ one again has bosons in the above common sense 
	and for $\alpha=1$ fermions.
	We shall denote by $\hT_{\alpha}$
	the free kinetic energy operator for anyons, 
	and it turns out that there are two equivalent ways to model them rigorously:
	either by means of topological boundary conditions,
	known in the literature as the \keyword{anyon gauge picture}
	(see \cite{MunSch-95,DelFigTet-97} and below for a proper definition),
	or using ordinary bosons $\Psi \in L^2_\sym(\R^{2N})$ or fermions $\Psi \in L^2_\asym(\R^{2N})$
	but with a peculiar (topological) magnetic interaction,
	which is known as the \keyword{magnetic gauge picture}
	(see \cite[Section~2.2]{LunSol-14} and \cite[Section~1.1]{LarLun-16}
	for a proper definition).
	
	Finally,
	in the one-dimensional case it looks as if there are no non-trivial choices
	since the fundamental group is trivial 
	(the space $\phX_\sym^{1 \times N}$ is simply connected,
	geometrically having the form of a wedge-shaped subset of $\R^N$),
	however in this case there are other ambiguities leading to different
	quantizations $\hT$
	(see e.g.\ \cite{Polychronakos-99,Myrheim-99} for physical reviews,
	and \cite[Section~2.1]{LunSol-14} for mathematical details).
	Intuitively, this is because a continuous exchange of two particles on the real line
	$\R$ necessarily leads to a collision and therefore one needs to prescribe 
	what happens at the collision points, 
	while more formally it is because the removal of the diagonals
	$\bDelta$ introduces boundaries in the configuration space and thus
	demands the specification of boundary conditions.
	However, most if not all of the known quantizations 
	can in fact be modeled using bosons or fermions
	together with some choice of pair interactions $V$, 
	and thus we will in one dimension only consider the usual bosons \eqref{eq:particle-exchange-bosons}
	or fermions \eqref{eq:particle-exchange-fermions}.
	
\begin{remark}
	The observable incorporation 
	\eqref{eq:particle-exchange-amplitude}
	of the indistinguishability of particles, 
	and its lifting to the phase ambiguity
	\eqref{eq:particle-exchange-phase},
	or in general \eqref{eq:particle-exchange-rep},
	may be seen as a
	consequence of the definition of quantum states in axiom~\ref{axm:states}, 
	namely that 
	any state is only defined up to an equivalent ray in the Hilbert space $\cH$.
	Taking into account the overall normalization of the state $\Psi \in \cH$,
	the position observable $\sx \in \phX_\sym^{d \times N}$,
	and its projection operator $P^{\hsx}_\Omega$, 
	which together determine the amplitude of $\Psi$ at every point,
	the only remaining ambiguity is the pointwise phase of $\Psi$ 
	(see below for further details).
	
	Also note that, because of the symmetry \eqref{eq:particle-exchange-amplitude},
	the one-body density \eqref{eq:def-density} simplifies to
	$$
		\varrho_\Psi(\bx) = N \int_{\R^{d(N-1)}} 
		|\Psi(\bx,\bx_1, \ldots, \bx_{N-1})|^2 \,d\bx_1 \ldots d\bx_{N-1}
	$$
	for indistinguishable particles.
\end{remark}
	
\begin{remark}\label{rem:spin}
	A further complication, which we will not find room to discuss in detail here, 
	is the concept of \keyword{spin}.
	Namely, relativistic quantum mechanics
	predicts that there must be additional geometric degrees of freedom associated
	to every particle in the form of
	a representation of angular momentum, 
	labelled by its spin quantum number (cf. Exercise~\ref{exc:spin}),
	and furthermore that there is a direct connection between 
	spin and statistics; see e.g. \cite{Froehlich-90} for review.
	We will return to 
	some consequences of this theory
	in Section~\ref{sec:exclusion-weaker}.
\end{remark}

\begin{remark*}\label{rem:anyons-fiber-bundles}
	The proper geometric setting to think about the above quantization problem 
	for identical particles is in
	the language of \keyword{fiber bundles} and \keyword{connections}
	(see e.g. \cite{Nakahara-03}).
	Namely, locally 
	$\Psi\colon \Omega \subseteq \R^{dN} \to \C$ is a function,
	but globally $\Psi$ is a section of a complex line bundle over the configuration 
	space $\phX_\sym^{d \times N}$.
	If assumed to be \keyword{locally flat} 
	(which physically means that if the particles are not
	moving too much then they are certainly distinguishable 
	and should thus have the usual free kinetic energy $\hT$) 
	such bundles/connections are fully classified by the maps \eqref{eq:particle-exchange-rep}.
	Furthermore, just as we may have reason to consider 
	a larger Hilbert space
	$\cH = L^2(\R^{dN};\C^n)$
	for distinguishable particles,
	where $\C^n$ is called an \keyword{internal space},
	containing additional degrees of freedom on top of the spatial ones,
	the one-dimensional fiber $\C$ may also be changed to $\C^n$ 
	(or possibly even some infinite-dimensional Hilbert space). 
	This leads then
	for $d=2$ to the notion of \keyword{non-abelian anyons}, 
	which are classified by irreducible unitary representations
	\begin{equation}\label{eq:particle-exchange-rep-nonabelian}
		\rho\colon B_N \to U(n)
	\end{equation}
	(the ones considered in \eqref{eq:particle-exchange-rep} 
	with $n=1$ are in fact \keyword{abelian anyons}
	since phases commute;
	a similar generalization also exists in the case $d \ge 3$ 
	but one may then argue that 
	it can be incorporated into the frameworks of ordinary bosons and fermions 
	\cite{Froehlich-90,DopRob-90}).
	
	Let us consider how the above notions arise
	starting strictly from the axioms of quantum mechanics.
	We thus attempt to sketch the formal procedure here, although we are not aware 
	of it having been done in complete detail elsewhere
	(see however \cite{MueDoe-93,DoeGroHen-99,DoeStoTol-01}).
	Assume generally that we have been given a configuration space manifold $\phX$
	which contains the observable positions $\sx \in \phX$ of the system
	and whose topology describes how such positions are logically related.
	Consider the Borel subsets $\Omega \subseteq \phX$,
	and the corresponding observables $``\sx \in \Omega" \in \cO$.
	These should be represented by self-adjoint operators
	$\widehat{``\sx \in \Omega"} \in \cL(\cH)$ on some Hilbert space $\cH$,
	and must have eigenvalues $0$ (false) or $1$ (true) to represent the outcome 
	of such a measurement. In other words, these 
	$\widehat{``\sx \in \Omega"} = P^{\hsx}_\Omega \in \cB(\cH)$ 
	are in fact projection operators on $\cH$.
	Considering the ranges of such projection operators,
	$$
		\gh_\Omega := P^{\hsx}_\Omega \cH \ \subseteq \cH,
	$$
	we thus have 
	a correspondence
	$$
		\begin{array}{ccl}
			\gh_\Omega & \subseteq & \cH \\
			\updownarrow \\
			 P^{\hsx}_\Omega & \in & \cB(\cH) \\
			\updownarrow \\
			\Omega & \subseteq & \phX
		\end{array}
	$$
	between Borel subsets of $\phX$ and closed subspaces of $\cH$.
	
	Now take a smaller subset $\Omega' \subseteq \Omega$ and observe that clearly
	$$
		``\sx \in \Omega'" \Rightarrow ``\sx \in \Omega"
	$$
	so that the information brought by the former observable is finer than
	that brought by the latter,
	implying $P^{\hsx}_{\Omega'} P^{\hsx}_\Omega = P^{\hsx}_{\Omega'}$,
	and therefore $\gh_{\Omega'} \hookrightarrow \gh_\Omega$.
	Also, one may consider two subsets $\Omega_i, \Omega_j \subseteq \phX$,
	and in the case that they overlap $\Omega_i \cap \Omega_j \neq \emptyset$
	we have two diagrams
	\begin{equation}\label{eq:fiber-subsets}
		\begin{array}{c}
			\Omega_i \quad \Omega_j \\
			\rotatebox{90}{$\subseteq$} \quad \ \rotatebox{90}{$\subseteq$} \\
			\Omega_i \cap \Omega_j
		\end{array}
		\qquad \text{and} \qquad
		\begin{array}{c}
			\gh_{\Omega_i} \quad \gh_{\Omega_j} \\
			\rotatebox{90}{$\hookrightarrow$} \quad \ \rotatebox{90}{$\hookrightarrow$} \\
			\gh_{\Omega_i \cap \Omega_j}
		\end{array}.
	\end{equation}
	Thus, one may relate the space $\gh_{\Omega_i}$ to $\gh_{\Omega_j}$,
	and vice versa, via the intersecting space $\gh_{\Omega_i \cap \Omega_j}$.
	
	Consider now the particular system at hand, that is the configuration space
	$\phX = \phX_\sym^{d \times N}$ of $N$ indistinguishable particles on $\R^d$.
	We start from the classical expression for the kinetic energy
	$$
		T(\ssp) = \frac{1}{2m} \sum_{j=1}^N |\bp|^2 \quad \in \cO,
	$$
	where $\ssp \in T^*_\sx(\phX) \cong \R^{dN}$ is the cotangent vector at 
	$\sx \in \phX$, and wish to look for quantizations $\hT$ represented on a
	corresponding Hilbert space $\cH$.
	We require that any such quantization must reduce to the usual one 
	\eqref{eq:many-body-kinetic-op} for
	distinguishable particles as soon as the particles indeed are distinguishable.
	In other words, upon restricting to a small enough subset $\Omega \subseteq \phX$
	--- thereby imposing the knowledge with certainty 
	that the particles are distinguishable ---
	we may consider the corresponding subspace $\gh_\Omega = P^{\hsx}_\Omega \cH$ 
	as sitting in some Hilbert space $\cH_\textup{dist}$
	of distinguishable particles.
	By the Stone--von Neumann uniqueness theorem 
	(see Remark~\ref{rem:Stone-von-Neumann}),
	with $\sx \in \R^{dN}$ and $\ssp \in \R^{dN}$ satisfying the CCR 
	\eqref{eq:classical-CCR}/\eqref{eq:quantum-CCR}
	and represented as self-adjoint operators on $\cH_\textup{dist}$,
	we must have $\cH_\textup{dist} \cong L^2(\R^{dN}) \otimes \cF$
	for some Hilbert space $\cF$ on which any remaining observables $a \in \cO$
	of the system may be represented
	(we are here relaxing the irreducibility requirement in the Schr\"odinger
	representation in order to be as general as possible
	and allow for other observables).
	
	More precisely, and in order to also take operator domain issues into account,
	we may initially consider states $\Psi \in \cH$ that are completely localized 
	on topologically trivial subsets
	$\phX \supseteq \Omega_j \hookrightarrow \R^{dN} \setminus \bDelta$
	of configurations of distinguishable particles.
	We thus take such $\Psi \in \gh_{\Omega_j}$ and furthermore demand that 
	$0 \le \inp{ T(\hsp) }_\Psi < \infty$ in order for such states to be physical.
	The Hilbert space on which we represent $\hsx$ as a multiplication operator
	and $\hsp = -i\hbar\nabla$ as a differentiation operator, 
	in order to implement the 
	CCR and the Weyl algebra, is then 
	$L^2(\Omega_j;\cF_j) \cong L^2(\Omega_j) \otimes \cF_j \subseteq \cH_\textup{dist}$
	for some undetermined space $\cF_j \subseteq \cF$.
	Taking the minimal domain, $\Psi \in C_c^\infty(\Omega_j;\cF_j)$,
	certainly ensures that states are physical and amounts, 
	upon taking the closure of such states, to a form domain
	$\cQ(\hT) = H^1_0(\Omega_j;\cF_j)$,
	i.e. the Dirichlet realization.
	On the other hand, one may also consider a \emph{maximal} domain, namely
	any states $\Psi \in L^2(\Omega;\cF_j)$ for which one can make finite sense of
	the form $\langle\Psi,\hT\Psi\rangle = (2m)^{-1}\norm{\nabla \Psi}^2$,
	which is $\cQ(\hT) = H^1(\Omega_j;\cF_j)$,
	the Neumann realization for distinguishable particles.
	Since the former domain sits in the latter, and we should try to be as
	unrestrictive as possible in our choices, let us then 
	only require that
	$\gh_{\Omega_j} \cong L^2(\Omega_j;\cF_j) \ni \Psi$
	with $(\hsx\Psi)(\sx) = \sx\Psi(\sx)$ and 
	$(\hsp\Psi)(\sx) = -i\hbar\nabla\Psi(\sx)$
	s.t. $\int_{\Omega_j} \|\nabla\Psi\|_{\cF}^2 < \infty$
	for any physical states $\Psi$.
	
	Hence we have obtained for each topologically trivial $\Omega_j \subseteq \phX$
	an isomorphism
	$\gh_{\Omega_j} \cong L^2(\Omega_j) \otimes \cF_j$.
	Let us then cover $\phX$ by a finite collection $\{\Omega_j\}_{j \in J}$ 
	of suitable such subsets.
	In the case of $N=2$ we may for example choose as $\Omega_j$,
	in terms of relative coordinates $(\bX,\br)$,
	$$
		\Omega_j = \Bigl\{ [(\bx_1,\bx_2)] \in \phX_\sym^{d \times 2} : 
			\bX \in \R^d, \ \br = r\omega, \ r>0, \ \omega \in B_\eps(\bn_j) \cap \S^{d-1} \Bigr\}
	$$
	where $\{\bn_j\} \subseteq \S^{d-1}$ is a finite collection of unit vectors 
	and $\eps>0$ is small enough.
	Furthermore, 
	we may impose the symmetry in the system on the collection $\{\Omega_j\}_{j \in J}$
	by requiring that the subsets are related $\Omega_i = R_{ij} \Omega_j$
	via symmetry transformations $R_{ij}\colon \R^{dN} \to \R^{dN}$ 
	which extend to unitary operators
	$U_{ij}\colon \gh_{\Omega_j} \to \gh_{\Omega_i}$.
	We must therefore for all $i,j \in J$ have that
	$$
		L^2(\Omega_i) \otimes \cF_i \cong \gh_{\Omega_i}
		\cong \gh_{\Omega_j} \cong L^2(\Omega_j) \otimes \cF_j
	$$
	and hence, since $L^2(\Omega_i) \cong L^2(\Omega_j)$, 
	we find that $\cF_i \cong \cF_j$ $\forall i,j \in J$.
	Let us therefore denote this prototype \keyword{fiber Hilbert space} by $\cF$.
	
	Coming back to the general observation \eqref{eq:fiber-subsets},
	now with $\gh_{\Omega_j} \cong L^2(\Omega_j) \otimes \cF$ for each subset
	of the covering $\{\Omega_j\}$, we must on the subspace
	$$
		L^2(\Omega_i;\cF)
		\supseteq
		L^2(\Omega_i \cap \Omega_j;\cF)
		\hookleftarrow
		\gh_{\Omega_i \cap \Omega_j} 
		\hookrightarrow
		L^2(\Omega_i \cap \Omega_j;\cF)
		\subseteq
		L^2(\Omega_j;\cF)
	$$
	have an isomorphism acting locally in the fiber $\cF$,
	$$
		t_{ij}\colon L^2(\Omega_i \cap \Omega_j;\cF) \subseteq 
			L^2(\Omega_i;\cF) \to L^2(\Omega_j;\cF),
		\qquad
		t_{ij}(\sx) \in U(\cF).
	$$
	The data $(\{\Omega_j\},\{t_{ij}\},\cF)$ defines a fiber bundle
	$E \to \phX$ 
	with structure group $G = U(\cF)$ 
	whose geometry is encoded in the transition functions $\{t_{ij}\}$
	and the connection.
	Given that the connection is flat on each local piece $\Omega_j$
	(the operator $i\hsp$ is
	the usual gradient $\nabla$ on a piece of flat $\R^{dN}$) 
	and that transitions ought to be trivial whenever particles remain 
	distinguishable, we thus have a locally flat bundle, 
	whose only non-trivial geometry is classified using the non-trivial 
	topology of $\phX$, and more precisely 
	the fundamental group $\pi_1(\phX)$
	(a well-known correspondence; 
	see e.g. \cite{MueDoe-93} and \cite[Chapter~5]{Michiels-13} for details),
	by homomorphisms (representations)
	$$
		\rho\colon \pi_1(\phX) \to U(\cF).
	$$
	In the case 
	that $\cF \cong \C$ resp. $\cF \cong \C^n$ this then reduces to 
	\eqref{eq:particle-exchange-rep} resp. \eqref{eq:particle-exchange-rep-nonabelian}.
	In general the dimension $n$ of the fiber $\cF$ would depend on 
	whether there are additional observables in $\cO$ which could distinguish $n$.
	If there are no such observables, then we should, by demanding that 
	all states be distinguishable (irreducibility), simply take $n=1$.
	
	We considered above local sections 
	$\Psi \in L^2(\Omega_j;\cF) \subseteq \Gamma(\Omega_j,E)$
	which should be extended to \emph{continuous} global sections on $E$, 
	for example by taking the closure of smooth sections.
	The full Hilbert space is then the space of 
	square-integrable sections
	$$
		\cH = \left\{ \Psi \in \Gamma(\phX;E) :  \int_{\phX} \|\Psi\|_{\cF}^2 < \infty \right\},
	$$
	with the local requirement $\int_{\Omega_j} \|\nabla\Psi\|_{\cF}^2 < \infty$
	for physical states $\Psi \in \Gamma(\Omega_j,E)$ then lifting globally to
	yield the form domain
	$$
		\cQ(\hT) = \left\{ \Psi \in \cH : \int_{\phX} \|\nabla\Psi\|_{\cF}^2 < \infty \right\}.
	$$
	By its non-negativity, this form $q_{\hT}\colon \cQ(\hT) \to \R_+$
	finally defines for us a \emph{unique} 
	(given the bundle $E$, i.e.\ the representation $\rho$)
	non-negative self-adjoint 
	operator $\hT \in \cL(\cH)$, which reduces on each local domain $\Omega_j$
	to the usual free kinetic energy \eqref{eq:many-body-kinetic-op} 
	for distinguishable particles with domains 
	$\cQ(\hT|_{\gh_{\Omega_j}}) = H^1(\Omega_j;\cF)$ and
	$\cD(\hT|_{\gh_{\Omega_j}}) \subseteq H^2(\Omega_j;\cF)$.
	
	Note that the procedure of defining $\hT$ by means of the form 
	domain really does matter, namely 
	if we start for simplicity from the flat but punctured bundle 
	of distinguishable particles
	$E = (\R^{dN} \setminus \bDelta) \times \cF$
	with trivial transition functions
	and consider all possible operator extensions from the minimal domain 
	$C_c^\infty(\R^{dN} \setminus \bDelta;\cF)$,
	then those may (if $d \le 3$) include \keyword{point interactions} 
	\cite{Albeverio-etal-05,BorSor-92}.
	However, by considering instead the form domain as above one obtains
	for $d \ge 2$ uniquely the extension corresponding to free particles,
	with $\cQ(\hT) = H^1(\R^{dN};\cF)$ and $\cD(\hT) = H^2(\R^{dN};\cF)$
	\cite{LunSol-14}.
	
	We should finally remark that the choice of observables employed above 
	corresponds to the usual Schr\"odinger representation,
	however one could alternatively start from a different choice of exchange-symmetric 
	observables and arrive at a \emph{different} quantization
	(this is sometimes referred to as the \keyword{Heisenberg representation};
	see e.g. \cite{Myrheim-99}).
\end{remark*}

\begin{exc}\label{exc:braid-relations}
	The braid group $B_N$ can be defined as the abstract group generated by elements 
	$\tau_j$, $j=1,2,\ldots,N-1$, satisfying the \keyword{braid relations}
	$$
		\tau_j \tau_k = \tau_k \tau_j, \ \ \text{for} \ \ |j-k| \ge 2,
		\qquad \text{and} \qquad
		\tau_j \tau_{j+1} \tau_j = \tau_{j+1} \tau_j \tau_{j+1}, \ \ 1 \le j \le N-2,
	$$
	(the latter are also called \keyword{Yang-Baxter relations}).
	Show that these relations imply that if $\rho\colon B_N \to U(1)$
	is a representation, i.e. $\rho(xy) = \rho(x)\rho(y)$,
	then it is 
	uniquely defined by its values on the generators 
	$\rho(\tau_j) = e^{i\alpha_j\pi}$, $j=1,2,\ldots,N-1$,
	and furthermore 
	$e^{i\alpha_1\pi} = \ldots = e^{i\alpha_{N-1}\pi} =: e^{i\alpha\pi}$.
	Also, show that if the additional relations $\tau_j^2 = 1$, $j=1,\ldots,N-1$, 
	are added then the resulting group is $S_N$ and that the only options for
	$\rho$ are then \eqref{eq:particle-exchange-rep-3D}.
\end{exc}

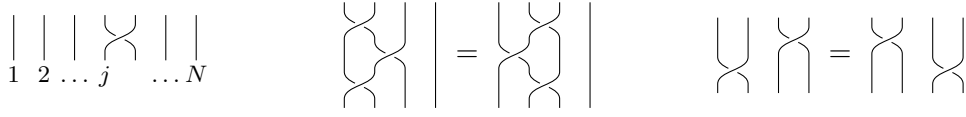
\begin{figure}[t]\label{fig:braids}
	\centering
	\begin{tikzpicture}[scale=0.4,font=\footnotesize,anchor=mid,baseline={([yshift=-.5ex]current bounding box.center)}]
		\braid[number of strands=7] s_4^{-1};
		\node at (1, -2.0) {$1$};
		\node at (2, -2.0) {$2$};
		\node at (3, -2.0) {$\ldots$};
		\node at (4, -2.0) {$j$};
		\node at (6, -2.0) {$\ldots$};
		\node at (7, -2.0) {$N$};
	\end{tikzpicture}
	\qquad
	\qquad
	\begin{tikzpicture}[scale=0.4,font=\footnotesize,anchor=mid,baseline={([yshift=-.5ex]current bounding box.center)}]
		\braid[number of strands=4] s_1^{-1} s_2^{-1} s_1^{-1};
	\end{tikzpicture}
	\ = \ 
	\begin{tikzpicture}[scale=0.4,font=\footnotesize,anchor=mid,baseline={([yshift=-.5ex]current bounding box.center)}]
		\braid[number of strands=4] s_2^{-1} s_1^{-1} s_2^{-1};
	\end{tikzpicture}
	\qquad
	\qquad
	\begin{tikzpicture}[scale=0.4,font=\footnotesize,anchor=mid,baseline={([yshift=-.5ex]current bounding box.center)}]
		\braid[number of strands=4] s_3^{-1} s_1^{-1};
	\end{tikzpicture}
	\ = \ 
	\begin{tikzpicture}[scale=0.4,font=\footnotesize,anchor=mid,baseline={([yshift=-.5ex]current bounding box.center)}]
		\braid[number of strands=4] s_1^{-1} s_3^{-1};
	\end{tikzpicture}
	\caption{Braid diagrams corresponding to the generator $\tau_j$ of $B_N$, 
		i.e. a counterclockwise exchange of particles/strands $j$ and $j+1$
		with time running upwards,
		and the relations $\tau_1 \tau_2 \tau_1 = \tau_2 \tau_1 \tau_2$
		respectively $\tau_1 \tau_3 = \tau_3 \tau_1$ of $B_4$.}
\end{figure}

\begin{exc}\label{exc:braid-enclosed-phases}
	We may represent a general particle exchange in two dimensions, 
	or an element of $B_N$, 
	using a \keyword{braid diagram}, i.e.\ an arbitrary composition of 
	\keyword{simple braids}
	of $N$ strands, with each simple braid $\tau_j$ formed by 
	taking the $j$:th strand over the $j+1$:st strand;
	see Figure~\ref{fig:braids}.
	Show using such braid diagrams
	that if the phase associated to a simple two-particle exchange 
	is $\rho(\tau_j) = e^{i\alpha\pi}$,
	then the phase that will arise as one particle encircles $p$ other
	particles in a simple loop is $e^{i2p\alpha\pi}$, 
	while if a pair of particles is exchanged once and in the exchange loop
	they enclose $p$ other particles then the phase must be $e^{i(2p+1)\alpha\pi}$.
\end{exc}

\begin{exc}\label{exc:braid-S2}
	Consider the corresponding braid group defined with the particles on the 
	surface of the sphere $S^2$ instead of $\R^2$.
	Show that in this case there is an extra topological condition on the generators,
	$$
		\tau_1 \tau_2 \ldots \tau_{N-1} \tau_{N-1} \ldots \tau_2 \tau_1 = 1,
	$$
	and determine the possible values of $\alpha$.
\end{exc}

\section{Uncertainty principles\lect{ [6-10]}}\label{sec:uncert}

	We will now investigate the most important feature of quantum mechanics
	as compared to classical mechanics, 
	namely the consequence of the non-commutativity
	of position and momentum observables
	referred to as the \keyword{uncertainty principle}.
	Namely, as made famous by Heisenberg, 
	for two non-commuting
	observables $a$ and $b$ and a state $\psi \in \cH$,
	we have that
	$$
		\inp{\bigl(\ha - \langle\ha\rangle_\psi\bigr)^2}_\psi
		\inp{\bigl(\hb - \langle\hb\rangle_\psi\bigr)^2}_\psi
		\ge \frac{1}{4}\left| \inp{[\ha,\hb]}_\psi \right|^2
		= \frac{\hbar^2}{4}\left| \inp{\widehat{\{a,b\}}}_\psi \right|^2.
	$$
	On the l.h.s.\ stands the product of the variances of a measurement of 
	$a$ and $b$,
	while the r.h.s.\ depends on the quantization of the observable
	$\{a,b\}$ and its expectation in $\psi$,
	which in the case that $\{a,b\} = \text{const.} \neq 0$ is
	strictly positive independently of $\psi$.
	Therefore this inequality has the physical interpretation 
	that it is \emph{impossible} to determine
	the value of both $a$ and $b$ simultaneously to arbitrary precision.
	The most important case is the canonical position and momentum operators,
	$\hx_j$ and $\hp_k$,
	and we will in this section formulate various versions of the 
	uncertainty principle involving these operators.
	We set $\hbar=1$ for simplicity.

\begin{exc}
	Prove the more general relation, known as the 
	\keyword{Robertson--Schr\"odinger uncertainty relation},
	(with $\psi$ in a common dense domain of all the operators)
	$$
		\inp{\bigl(\ha - \langle\ha\rangle_\psi\bigr)^2}_\psi
		\inp{\bigl(\hb - \langle\hb\rangle_\psi\bigr)^2}_\psi
		\ge \left| \frac{1}{2}\inp{\ha\hb + \hb\ha}_\psi 
				- \langle\ha\rangle_\psi \langle\hb\rangle_\psi \right|^2
			+ \left| \frac{1}{2}\inp{[\ha,\hb]}_\psi \right|^2,
	$$
	by considering $\inp{f,g}$ with 
	$f = (\ha-\langle \ha \rangle_\psi)\psi$ and
	$g = (\hb-\langle \hb \rangle_\psi)\psi$.
\end{exc}

\subsection{Heisenberg}\label{sec:uncert-Heisenberg}

	Recall the canonical commutation relations~\eqref{eq:quantum-CCR},
	$$
		[\hx_j,\hp_k] = i \delta_{jk} \1.
	$$
	In particular, 
	$$
		i (\hat{\bp} \cdot \hat{\bx} - \hat{\bx} \cdot \hat{\bp}) 
		= i \sum_{j=1}^d(\hat{p}_j \hat{x}_j - \hat{x}_j \hat{p}_j)
		= \sum_{j=1}^d \delta_{jj} \1 = d\1,
	$$
	or in the usual Schr\"odinger representation 
	($\nabla$ here acts as an operator on everything to the right)
	\begin{equation}\label{eq:nabla-x-commutator}
		\nabla \cdot \bx - \bx \cdot \nabla = d\1.
	\end{equation}
	This identity can be used together with the Cauchy-Schwarz inequality to prove
	the most famous version of the uncertainty principle of quantum mechanics:

\begin{theorem}[\keyword{Heisenberg's uncertainty principle}]
	For any $\psi \in L^2(\R^d)$ 
	with $\|\psi\|_{L^2} = 1$, we have
	\begin{equation}\label{eq:Heisenberg}
		\langle \psi, \hat{\bp}^2 \psi \rangle 
		\langle \psi, \hat{\bx}^2 \psi \rangle \ge \frac{d^2}{4}.
	\end{equation}
\end{theorem}
\begin{remark}
	By means of the Fourier transform, the l.h.s. should be understood to be
	$$
		\int_{\R^d} |\nabla\psi(\bx)|^2 d\bx \int_{\R^d} |\bx \psi(\bx)|^2 d\bx 
		= \int_{\R^d} |\bp \hat\psi(\bp)|^2 d\bp \int_{\R^d} |\bx \psi(\bx)|^2 d\bx,
	$$
	which is finite if and only if both these integrals converge,
	i.e. $\psi \in H^1(\R^d) \cap \cD(\hbx)$.
	Moreover, by replacing $\psi(\bx)$ by $e^{i\bx \cdot \bp_0}\psi(\bx + \bx_0)$,
	the inequality may also be written
	\begin{equation}\label{eq:Heisenberg-shift}
		\inp{(\hbp - \bp_0)^2}_\psi \inp{(\hbx - \bx_0)^2}_\psi \ge \frac{d^2}{4}.
	\end{equation}
\end{remark}

	This formulation of the uncertainty principle is the
	original and most well-known one \cite{Heisenberg-27}, and indeed 
	it tells us that if the state 
	$\psi$ localizes around the origin so that the r.h.s. of the inequality
	\begin{equation}\label{eq:Heisenberg-inverse}
		\langle \psi, \hat{\bp}^2 \psi \rangle 
		\ge \frac{d^2}{4} \langle \psi, \hat{\bx}^2 \psi \rangle^{-1}
	\end{equation}
	tends to infinity, then this also implies a large momentum. However,
	as stressed e.g.\ in \cite{Lieb-76} (see also \cite{LieSei-09}),
	this is unfortunately not sufficient for proving stability 
	of the hydrogenic atom.
	Namely, the expectation value $\inp{\psi,\hbx^2\psi}$
	is a poor measure of how localized the state is,
	since it is possible to make this value very large
	without changing the kinetic energy much.
	Consider for example the state
	$$
		\psi(\bx) = \sqrt{1-\eps^2} u(\bx) + \eps v(\bx - \by),
	$$
	with $\eps \in (0,1)$, 
	$u,v \in C_c^\infty(B_1(0))$ normalized in $L^2$, 
	and $|\by|>2$.
	Then $\int_{\R^d} |\psi|^2 = 1$ and
	$$
		\int_{\R^d} |\bx|^2|\psi(\bx)|^2 \,d\bx
		= (1-\eps^2) \int_{B_1(0)} \!|\bx|^2|u(\bx)|^2 \,d\bx
			+ \eps^2 \int_{B_1(0)} \!|\by+\bz|^2|v(\bz)|^2 \,d\bz
		\ge \eps^2(|\by|-1)^2,
	$$
	while
	$$
		\int_{\R^d} |\nabla\psi(\bx)|^2 \,d\bx
		= (1-\eps^2) \int_{B_1(0)} |\nabla u(\bx)|^2 \,d\bx
			+ \eps^2 \int_{B_1(0)} |\nabla v(\bx)|^2 \,d\bx.
	$$
	Hence, we may take simultaneously $\eps \ll 1$ and $|\by| \gg \eps^{-1}$
	to make the r.h.s. of \eqref{eq:Heisenberg-inverse} small while the l.h.s. 
	stays essentially the same.
	
	Before considering formulations that are more useful for our stability problem, 
	we note that
	there is also the following version of the Heisenberg uncertainty principle
	which explains 
	that a state $\psi \in \cH \setminus \{0\}$ 
	and its Fourier transform
	$\hat\psi$ cannot both have compact support
	(see e.g. \cite[Theorem~7.12]{Teschl-14} for a proof,
	and e.g. 
	\cite{Benedicks-85,Hedenmalm-12}
	for various generalizations):

\begin{theorem} 
	Suppose $f \in L^2(\R^n)$. If both $f$ and $\hat{f}$ have compact
	support, then $f = 0$.
\end{theorem}

	\noindent
	In other words, recalling our interpretations of 
	$|\psi(\sx)|^2$ respectively $|\hat\psi(\ssp)|^2$
	for a normalized state $\psi$ as probability densities for 
	the observables $\sx$ respectively $\ssp$,
	this theorem tells us that it is 
	impossible to know with certainty that
	both $\sx \in B_R(0)$ and $\ssp \in B_{R'}(0)$,
	regardless of the size of the radii $R,R' > 0$.

\begin{exc}
	Prove 
	\eqref{eq:Heisenberg} and \eqref{eq:Heisenberg-shift}
	for $\psi \in C_c^\infty(\R^d)$ by taking 
	expectation values of the relation~\eqref{eq:nabla-x-commutator}.
	How can this be extended to $L^2(\R^d)$?
\end{exc}

\begin{exc}
	Check that the \keyword{Gaussian wave packet}
	$$
		\psi(x) = (\lambda/\pi)^{1/4} e^{-\frac{\lambda}{2}|x-x_0|^2 + ip_0x},
	$$
	for any $\lambda > 0$, is normalized and
	realizes the minimum of \eqref{eq:Heisenberg-shift} in dimension $d=1$.
\end{exc}

\subsection{Hardy}\label{sec:uncert-Hardy}

	A more powerful version of the uncertainty principle is Hardy's inequality:

\begin{theorem}[{\keyword[Hardy inequality]{The Hardy inequality}}]\label{thm:Hardy}
	For any $u \in H^1(\R^d)$ in dimension $d \ge 2$, and for any 
	$u \in H_0^1(\R \setminus \{0\})$ in dimension $d=1$, we have that
	\begin{equation} \label{eq:Hardy}
		\int_{\R^d} |\nabla u(\bx)|^2 \,d\bx 
		\ge \frac{(d-2)^2}{4} \int_{\R^d} \frac{|u(\bx)|^2}{|\bx|^2} \,d\bx.
	\end{equation}
\end{theorem}
\begin{remark}
	The constant $(d-2)^2/4$ is sharp (and the inequality trivial for $d=2$),
	in the sense that for any larger constant
	there is a function $u \in C_c^\infty(\R^d \setminus \{0\})$ 
	which violates it
	(see, e.g., \cite[Appendix~B]{Lundholm-15} for a discussion continuing 
	the context outlined below).
\end{remark}
\begin{remark}
	The inequality \eqref{eq:Hardy} of quadratic forms translates to the following
	operator inequality:
	\begin{equation} \label{eq:Hardy-operator}
		-\Delta 
		\ge \frac{(d-2)^2}{4} \frac{1}{|\bx|^2},
	\end{equation}
	with both sides
	interpreted as non-negative operators on $L^2(\R^d)$ having a common form domain
	$\cQ(-\Delta_{\R^d}) = H^1(\R^d)$, $d \ge 2$, respectively
	$\cQ(-\Delta_{\R \setminus \{0\}}) = H_0^1(\R \setminus \{0\})$.
\end{remark}

	Many other types of Hardy inequalities exist, and their general defining
	characteristic is that they provide a bound for the Laplacian 
	(and hence for the kinetic energy $\hT$)
	from below in terms of a positive potential $V$ which scales in the same way, 
	i.e.\ as an inverse-square length or distance, and which is singular at some 
	point of the configuration space or its boundary.
	In other words, in case such a non-trivial inequality holds, 
	we clearly have that
	the kinetic energy is not only non-negative but it even tends to infinity
	if the state is sufficiently localized close to a singularity of $V$.
	We refer to the recent book \cite{BalEvLew-15} 
	and the more classic references given in \cite{Tidblom-05}
	for general treatments of Hardy inequalities.
	Although the basic inequality \eqref{eq:Hardy} is fairly straightforward to prove 
	directly, we will instead take a very general approach to proving Hardy inequalities,
	involving a formulation
	referred to as the `ground state representation'.
	This approach is not covered by the above standard references but has been discussed 
	in various forms in for example
	\cite{Birman-61, 
	AlbHoeStr-77,FraSei-08,FraSimWei-08,Seiringer-10,Lundholm-15}.

\begin{exc}
	Show that $\inp{\psi, |\bx|^{-2} \psi} \ge \inp{\psi, |\bx|^2 \psi}^{-1}$
	if $\norm{\psi} = 1$,
	and hence that Hardy \eqref{eq:Hardy} directly implies Heisenberg \eqref{eq:Heisenberg}
	but with a slightly weaker constant.
\end{exc}

\subsubsection{The ground state representation}

	We consider the following identity involving
	the quadratic form of the Dirichlet Laplacian on a domain in $\R^n$,
	which we refer to as the \keyword{ground state representation (GSR)}:

\begin{proposition}[GSR] \label{prop:GSR}
	Let $\Omega$ be an open set in $\mathbb{R}^n$ and
	let $f\colon \Omega \to \R^+ := (0,\infty)$ be twice differentiable. 
	Then, for any $u \in C_c^\infty(\Omega)$ and $\alpha \in \R$,
	\begin{equation} \label{eq:GSR}
		\int_\Omega |\nabla u|^2 
		= \int_\Omega \left( 
				\alpha(1 - \alpha) \frac{|\nabla f|^2}{f^2} + \alpha\frac{-\Delta f}{f}
			\right) |u|^2 
		+ \int_\Omega |\nabla v|^2 f^{2\alpha}, 
	\end{equation}
	where $v = f^{-\alpha}u$. 
\end{proposition}
\begin{proof}
	We have for $u = f^{\alpha}v$ that
	$
		\nabla u = \alpha f^{\alpha-1} (\nabla f) v
		+ f^{\alpha} \nabla v,
	$
	and hence
	$$
		|\nabla u|^2 = \alpha^2 f^{2(\alpha-1)} |\nabla f|^2 |v|^2
		+ \alpha f^{2\alpha-1} (\nabla f) \cdot \nabla |v|^2 
		+ f^{2\alpha} |\nabla v|^2.
	$$
	Now let us integrate this expression over $\Omega$, and note that the middle term 
	on the r.h.s. produces after a partial integration
	(note that $v$ vanishes on $\partial\Omega$)
	$$
		- \alpha \int_{\Omega} |v|^2 \nabla \cdot (f^{2\alpha-1} \nabla f). 
	$$
	Finally, using that
	$$
		\nabla \cdot (f^{2\alpha-1} \nabla f) 
		= (2\alpha-1) f^{2\alpha-2} |\nabla f|^2 + f^{2\alpha-1} \Delta f,
	$$
	and collecting the terms we arrive at \eqref{eq:GSR}.
\end{proof}

	The idea of the ground state representation is to factor out a positive function 
	$f^\alpha$ from the kinetic energy form, to the cost of a new potential
	\begin{equation} \label{eq:GSR-potential}
		\alpha(1 - \alpha) \frac{|\nabla f|^2}{f^2} + \alpha\frac{-\Delta f}{f},
	\end{equation}
	which we will call the \keyword{GSR potential}.
	Note that if $f$ is chosen to be an exact zero-eigenfunction of the Laplacian,
	$\Delta f = 0$, i.e. a harmonic function, 
	or a \keyword{generalized ground state} of the kinetic energy operator
	(not necessarily normalizable on $L^2(\Omega)$ or in the form domain),
	then the last term of \eqref{eq:GSR-potential} vanishes 
	while the first is positive and maximal for the choice $\alpha = 1/2$.
	Since the last integral in \eqref{eq:GSR} is also non-negative this then 
	yields a potentially useful estimate of the kinetic energy form in terms
	of this positive potential.
	Also in the case that $f$ is not an exact zero-eigenfunction but for example
	an approximation to the true ground state, the first term in the potential \eqref{eq:GSR-potential}
	may still be able to control the second one and produce a useful positive bound.
	The parameter $\alpha$ may then be used in a variational sense to find
	the best possible bound given the ansatz $f$, which 
	if the exact ground state is unknown instead
	may be taken 
	of a form which is convenient for computations.

\begin{exc}
	Show that a modification of Proposition~\ref{prop:GSR}
	to involve a product ground state ansatz $g^\alpha h^\beta$
	(i.e. $u = g^\alpha h^\beta v$)
	produces the corresponding GSR potential 
	\begin{equation} \label{eq:product-GSR}
		\alpha(1-\alpha) \frac{|\nabla g|^2}{g^2} + \alpha \frac{-\Delta g}{g}
		+ \beta(1-\beta) \frac{|\nabla h|^2}{h^2} + \beta \frac{-\Delta h}{h}
		- 2\alpha\beta \frac{\nabla g \cdot \nabla h}{gh}.
	\end{equation}
\end{exc}

\begin{exc}
	Apply the ground state approach to prove a Hardy inequality on $[0,1]$,
	i.e. find functions $g$ and $h$ on $[0,1]$ s.t. $g(0) = 0$ and $g''=0$,
	and $h(1)=0$ and $h''=0$,
	compute the GSR potential \eqref{eq:product-GSR} and try to optimize it 
	w.r.t.\ $\alpha$ and $\beta$.
\end{exc}

\subsubsection{The standard Hardy inequalities in $\R^d$}

	Our approach to prove the standard Hardy inequality \eqref{eq:Hardy} is to
	first prove that it holds for all $u \in C_c^\infty(\R^d \setminus \{0\})$
	(this is a minimal domain respecting the singularity of the potential)
	using the above ground state representation, and then conclude that it
	also holds on the Sobolev space $H^1_0(\R^d \setminus \{0\})$ by density.
	Finally, we prove that actually $H^1_0(\R^d \setminus \{0\}) = H^1(\R^d)$
	if the dimension $d$ is large enough, 
	a result which is also very useful in itself.

	A natural choice of 
	ground states $f$ for the Laplacian in $\R^d$
	are the \keyword{fundamental solutions}:
	\begin{align}\label{eq:fundamental-solutions}
		d \neq 2:	&& f_d(\bx) &:= |\bx|^{-(d-2)}, & \Delta_{\R^d} f_d &= c_d \delta_0, \\
		d = 2:		&& f_2(\bx) &:= \ln |\bx|, & \Delta_{\R^2} f_2 &= c_2 \delta_0, 
	\end{align}
	where $\delta_0$ are Dirac delta distributions 
	supported at the origin 
	and $c_d$ some 
	constants which only depend on $d$
	(see Exercise~\ref{exc:fundamental-solutions}).
	These functions are smooth and strictly positive for all $\bx \neq \0$ 
	if $d \neq 2$, while for $d=2$ we may cure the sign problem
	by taking absolute values,
	but still need to avoid both the origin $\bx=\0$ and the circle $|\bx|=1$
	(a different nonzero radius may be chosen by rescaling).
	Indeed we cannot expect to have a non-trivial Hardy inequality 
	on all of $\R^2$, as indicated by the vanishing constant in \eqref{eq:Hardy}.
	
	Hence, for $d \neq 2$ we consider the domain 
	$\Omega := \R^d \setminus \{0\}$, on which $f := f_d > 0$ and $\Delta f = 0$,
	while $|\nabla f|^2/f^2 = (d-2)^2/|\bx|^2$.
	The GSR potential \eqref{eq:GSR-potential} is therefore optimal for 
	$\alpha = \frac{1}{2}$, and \eqref{eq:GSR} yields the ground state 
	representation associated to the standard Hardy inequality 
	\eqref{eq:Hardy} in $\mathbb{R}^d$:
	\begin{equation} \label{eq:standard-Hardy}
		\int_\Omega |\nabla u(\bx)|^2 \,d\bx 
			- \frac{(d-2)^2}{4} \int_\Omega \frac{|u(\bx)|^2}{|\bx|^2} \,d\bx
		= \int_\Omega |\nabla v(\bx)|^2 |\bx|^{-(d-2)} \,d\bx \ \ge \ 0.
	\end{equation}
	The inequality \eqref{eq:standard-Hardy} 
	thus holds for all $u \in C_c^\infty(\Omega)$.
	Now, if we take an arbitrary function $u$ in the Sobolev space $H^1_0(\Omega)$,
	then we have by its definition a sequence $(u_n) \subset C_c^\infty(\Omega)$
	s.t.
	$$
		\int_\Omega |u-u_n|^2 \to 0, \qquad \text{and} \qquad
		\int_\Omega |\nabla(u-u_n)|^2 \to 0.
	$$
	We then find by \eqref{eq:standard-Hardy} that this sequence is also
	Cauchy in the space $L^2(\Omega,|\bx|^{-2} d\bx)$ 
	weighted by the singular GSR potential:
	$$
		\int_\Omega |u_n - u_m|^2 \,|\bx|^{-2} d\bx 
		\le \frac{4}{(d-2)^2} \int_\Omega |\nabla(u_n - u_m)|^2 \to 0,
	$$
	which implies for the limit $u \in L^2(\Omega,|\bx|^{-2} d\bx)$.
	Furthermore, taking $n \to \infty$,
	$$
		\int_\Omega |u - u_m|^2 \,|\bx|^{-2} d\bx 
		\le \frac{4}{(d-2)^2} \int_\Omega |\nabla(u - u_m)|^2 \to 0,
	$$
	and therefore, by approximating both sides of the desired
	inequality in terms of $u_m \in C^\infty_c$,
	$$
		\int_\Omega |u|^2 \,|\bx|^{-2} d\bx 
		\le \frac{4}{(d-2)^2} \int_\Omega |\nabla u|^2,
	$$
	which proves the Hardy inequality \eqref{eq:Hardy} on the space 
	$H^1_0(\R^d \setminus \{0\}) \subseteq H^1(\R^d)$.

	It remains then to prove that actually 
	$H^1_0(\R^d \setminus \{0\}) = H^1(\R^d)$,
	so that the Hardy inequality indeed holds 
	on the maximal form domain of the Laplacian. 
	
\begin{lemma}\label{lem:H10}
	We have that $H^1_0(\R^d \setminus \{0\}) = H^1(\R^d)$ for $d \ge 2$.
\end{lemma}
\begin{remark}
	It is not true that $H_0^1(\R \setminus \{0\}) = H^1(\R)$,
	but one rather has a decomposition
	$$
		H_0^1(\R \setminus \{0\}) \cong H_0^1(\R_\limminus) \oplus H_0^1(\R_\limplus)
	$$
	in one dimension.
\end{remark}
\begin{proof}[Proof of Lemma~\ref{lem:H10}]
	We aim to prove that $C_c^\infty(\R^d \setminus \{0\})$ is dense in $H^1(\R^d)$.
	Since $C_c^\infty(\R^d)$ is dense in $H^1(\R^d)$
	(recall from Section~\ref{sec:prelims-Hk} that
	we have $H^1_0(\R^d) = H^1(\R^d)$) it suffices to prove that an
	arbitrary $u \in C_c^\infty(\R^d)$ can be approximated arbitrarily well
	by a sequence $(u_n) \subset C_c^\infty(\R^d \setminus \{0\})$ 
	in the $H^1(\R^d)$-norm.
	For $d \ge 3$ we take a cut-off function
	$\vphi \in C^\infty(\R^d;[0,1])$ such that
	$\vphi(\bx) = 0$ for $|\bx| \le 1$ and $\vphi(\bx) = 1$ for $|\bx| \ge 2$,
	and let
	$$
		u_\eps(\bx) := \vphi_\eps(\bx) u(\bx),
		\qquad
		\vphi_\eps(\bx) := \vphi(\bx/\eps),
	$$
	for $\eps > 0$.
	Then $u_\eps \in C_c^\infty(\R^d \setminus B_\eps(0))$ and,
	as $\eps \to 0$,
	$$
		\norm{u - u_\eps}_{L^2}^2 
		= \int_{\R^d} |u-u_\eps|^2 
		= \int_{\R^d} |u|^2 |1-\vphi_\eps|^2 
		\le \int_{B_{2\eps}(0) \setminus B_\eps(0)} C
		\to 0.
	$$
	Furthermore, by the product rule and the triangle inequality,
	$$
		\norm{\nabla(u - u_\eps)}_{L^2}
		\le \norm{(1-\vphi_\eps) \nabla u}_{L^2}
			+ \norm{(\nabla\vphi_\eps)u}_{L^2},
	$$
	with $\norm{(1-\vphi_\eps) \nabla u}_{L^2}^2 \to 0$ as above and
	$$
		\norm{(\nabla\vphi_\eps)u}_{L^2}^2
		= \int_{\R^d} |u|^2 |\nabla\vphi_\eps|^2 
		\le C \int_{B_{2\eps}(0) \setminus B_\eps(0)} |\nabla\vphi_\eps|^2 
		\le C \eps^{-2} |B_{2\eps}(0)|
		\le C \eps^{d-2}
		\to 0,
	$$
	as $\eps \to 0$.
	
	In the case $d=2$ the above choice fails but we may instead take 
	$\vphi_\eps(\bx) = \vphi(\eps\ln |\bx|)$
	with $\vphi \in C^\infty(\R;[0,1])$
	such that $\vphi = 0$ on $(-\infty,-2)$ and $\vphi = 1$ on $(-1,\infty)$.
	Then $u_\eps := \vphi_\eps u \in C_c^\infty(\R^2 \setminus B_{e^{-2/\eps}}(0))$,
	and again $\norm{u - u_\eps}_{L^2} \to 0$ and 
	$\norm{(1-\vphi_\eps) \nabla u}_{L^2} \to 0$ as above.
	Moreover, by first switching to polar coordinates and then making 
	the change of variable $r = e^s$,
	$$
		\frac{1}{2\pi} \int_{\R^2} |\nabla\vphi_\eps|^2 \,d\bx
		= \eps^2 \int_0^\infty |\vphi'(\eps \ln r)|^2 \frac{dr}{r}
		= \eps^2 \int_{-2/\eps}^{-1/\eps} |\vphi'(\eps s)|^2 \,ds
		\le C \eps \to 0,
	$$
	as $\eps \to 0$, which completes the proof.
\end{proof}

	The above proves Theorem~\ref{thm:Hardy}, and may also be straightforwardly
	generalized in numerous directions.
	For $d=2$ we can instead take the two-component domain
	$\Omega := \R^2 \setminus (\{0\} \cup \mathbb{S}^{1})$ 
	and the ground state $f:=|f_2|$, so that $f>0$ and $\Delta f = 0$ on $\Omega$.
	Also, $|\nabla f(\bx)| = 1/|\bx|$.
	Proposition~\ref{prop:GSR}
	then produces the corresponding two-dimensional GSR
	for $u \in C_c^\infty(\Omega)$
	\begin{equation}\label{eq:standard-Hardy-2d-GSR}
		\int_\Omega |\nabla u|^2 \,d\bx 
			- \frac{1}{4} \int_\Omega \frac{|u(\bx)|^2}{|\bx|^2 (\ln |\bx|)^2} \,d\bx
		= \int_\Omega |\nabla v|^2 \big| \ln |\bx| \big| \,d\bx \ \ge \ 0.
	\end{equation}
	By taking the closure of $u \in C_c^\infty(\Omega)$ as above
	and using Lemma~\ref{lem:H10} locally around the point $\0$, 
	the l.h.s.~of \eqref{eq:standard-Hardy-2d-GSR} is non-negative for all
	$u \in H^1_0(\Omega) = H^1_0(\R^2 \setminus \mathbb{S}^1)$.
	Hence we proved the following two-dimensional Hardy inequality:
	
\begin{theorem}
	For any $u \in H^1_0(\R^2 \setminus \mathbb{S}^1)$
	we have that
	\begin{equation}\label{eq:standard-Hardy-2d}
		\int_{\R^2} |\nabla u|^2 \,d\bx 
		\ge \frac{1}{4} \int_{\R^2} 
			\frac{|u(\bx)|^2}{|\bx|^2 (\ln |\bx|)^2} \,d\bx.
	\end{equation}
\end{theorem}

	Also note that we are free to choose the location of the singularity
	in the Hardy inequality, namely by translation invariance of the
	kinetic energy we also have for example
	\begin{equation}\label{eq:Hardy-translated}
		\int_{\R^d} |\nabla u(\bx)|^2 \,d\bx 
		\ge \frac{(d-2)^2}{4} \int_{\R^d} \frac{|u(\bx)|^2}{|\bx - \bx_0|^2} \,d\bx,
	\end{equation}
	for any $\bx_0 \in \R^d$ and $u \in H^1(\R^{d \ge 3})$ 
	respectively $u \in H^1_0(\R^1 \setminus \{\bx_0\})$.
	
\begin{exc}\label{exc:fundamental-solutions}
	Verify that \eqref{eq:fundamental-solutions} are the fundamental solutions, 
	i.e. zero-eigenfunctions of the Laplacian outside $\0$, 
	and, if you know distribution theory, compute the constants $c_d$
	(i.e. consider $\int_{\R^d} f \Delta \varphi$ for $\varphi \in C_c^\infty(\R^d)$
	and make partial integrations).
\end{exc}

\begin{exc}
	Prove that the inequality \eqref{eq:standard-Hardy-2d} 
	does not extend to $H^1(\R^2)$.
\end{exc}

\begin{exc}
	Prove a Hardy inequality outside the hard-sphere potential 
	of Example~\ref{exmp:hard-core} in $d=3$,
	by taking the ground state $f(\bx) = 1 - R/|\bx|$.
	Can it be improved by using an additional ground state $g(\bx) = 1/|\bx|$?
\end{exc}

\subsubsection{Many-body Hardy inequalities}\label{sec:uncert-Hardy-many}
	
	In \cite{Tidblom-05,HofLapTid-08} a number of interesting and useful many-body
	versions of the Hardy inequality were derived.
	Some of these were extended geometrically in \cite{Lundholm-15}.
	We only mention the simplest, one-dimensional, case here,
	and return to some other, fermionic, versions in conjunction with
	exclusion principles in Section~\ref{sec:exclusion-fermions-Hardy}.
	
	Recall the definition \eqref{eq:fat-diagonal} of the many-body configuration
	space diagonal $\bDelta$.

\begin{theorem}\label{thm:many-body-Hardy-1d}
	For any $u \in H^1_0(\R^N \setminus \bDelta)$
	we have that
	\begin{equation}\label{eq:many-body-Hardy-1d}
		\int_{\R^N} |\nabla u|^2 \,d\sx 
		\ge \frac{1}{2}\sum_{1 \le j < k \le N} \int_{\R^N} 
			\frac{|u(\sx)|^2}{|x_j-x_k|^2} \,d\sx.
	\end{equation}
\end{theorem}

	This inequality is useful for the analysis of 
	\keyword[Calogero--Sutherland model]{Calogero--Sutherland}
	\cite{Calogero-71,Sutherland-71}
	and similar models in many-body quantum mechanics 
	involving inverse-square interactions.
	It proves immediately that an interacting many-body Hamiltonian of
	the form 
	$$
		\hH^N = \hT + \beta \hW 
		= -\sum_{j=1}^N \frac{\partial^2}{\partial x_j^2} 
			+ \beta \sum_{1 \le j<k \le N} \frac{1}{|x_j-x_k|^2},
	$$
	with form domain $\cQ(\hH^N) = H^1_0(\R^N \setminus \bDelta)$,
	is trivially \keyword[stability]{stable} (of both first and second kind) 
	if the interaction coupling strength parameter satisfies $\beta \ge -1/2$,
	i.e.\ if it is not too attractive.

\begin{exc}
	Prove that for any three distinct points $x_1,x_2,x_3 \in \R$,
	$$
		(x_1-x_2)^{-1}(x_1-x_3)^{-1} + 
		(x_2-x_3)^{-1}(x_2-x_1)^{-1} +
		(x_3-x_1)^{-1}(x_3-x_2)^{-1} = 0.
	$$
\end{exc}
\begin{exc}
	Use this identity and the ansatz
	$f(\sx) := \prod_{j<k} |x_j-x_k|$
	on $\Omega := \R^N \setminus \bDelta$
	to prove the one-dimensional many-body GSR
	$$
		\int_{\Omega} |\nabla u|^2 \,d\sx 
		- \frac{1}{2}\sum_{j < k} \int_{\Omega} 
			\frac{|u(\sx)|^2}{|x_j-x_k|^2} \,d\sx
		= \int_\Omega |\nabla(f^{-1/2}u)|^2 f
		\ge 0,
	$$
	for $u \in C^\infty_c(\Omega)$,
	and hence Theorem~\ref{thm:many-body-Hardy-1d}.
\end{exc}

\subsection{Sobolev}\label{sec:uncert-Sobolev}

Another very powerful formulation of the uncertainty principle is given
by Sobolev's inequality:

\begin{theorem}[{\keyword[Sobolev inequality]{The Sobolev inequality}}]\label{thm:Sobolev}
	For $d \ge 3$ and all $u \in H^1(\R^d)$, it holds that
	\begin{equation} \label{eq:Sobolev}
		\int_{\R^d} |\nabla u|^2 
		\ge S_d \norm{u}_{2d/(d-2)}^2,
	\end{equation}
	with $S_d = d(d-2)|\S^d|^{2/d}/4$.
	For $d=2$ and every $2 < p < \infty$
	there exists a constant $S_{2,p}>0$ such that for any $u \in H^1(\R^2)$,
	\begin{equation} \label{eq:Sobolev-2d}
		\int_{\R^2} |\nabla u|^2 
		\ge S_{2,p} \norm{u}_2^{-4/(p-2)} \norm{u}_p^{2p/(p-2)}.
	\end{equation}
	For $d=1$ one has for $u \in H^1(\R)$,
	\begin{equation} \label{eq:Sobolev-1d}
		\int_{\R} |u'|^2 
		\ge \norm{u}_2^{-2} \norm{u}_{\infty}^4.
	\end{equation}
\end{theorem}
\begin{remark}
	The assumption $u \in H^1(\R^d)$ may be weakened slightly in the case $d \ge 3$;
	see \cite[Section~8.2-8.3]{LieLos-01}.
\end{remark}
	
	The constants $S_d$ respectively $S_1=1$ in the case $d \ge 3$ and $d=1$ are
	sharp (see \cite{Aubin-75,Talenti-76}, and e.g. \cite{LieLos-01}), 
	with $S_3 = 3(\pi/2)^{4/3} \approx 5.478$,
	while the value of the optimal constant $S_{2,p}$ for $d=2$ is presently unknown
	(a very rough but explicit estimate 
	may be obtained from \cite[Theorem~8.5]{LieLos-01}).
	Also note as usual the necessary match of dimensions in these inequalities 
	(cf. the remark after Proposition~\ref{prop:Hoelder})
	which not only helps to remember them 
	but also clarifies why an inequality of the simpler form \eqref{eq:Sobolev}
	cannot extend to $d \le 2$.
	
\begin{proof}[Proof for $d = 1$]
	By the fundamental theorem of calculus applied to an approximating sequence
	$u_n \in C_c^\infty(\R)$
	(see \cite[Theorem~8.5]{LieLos-01} for details),
	one has for any $u \in H^1(\R)$ and a.e. $x \in \R$
	$$
		u(x)^2 = \int_{-\infty}^x u(y)u'(y) \,dy - \int_x^\infty u(y)u'(y) \,dy.
	$$
	Therefore, by the triangle inequality,
	$$
		|u(x)|^2 \le \int_{-\infty}^x |u||u'| + \int_x^\infty |u||u'|
		= \int_{-\infty}^\infty |u||u'|,
	$$
	and thus by the Cauchy--Schwarz inequality,
	$$
		\|u\|_\infty^2 \le \|u\|_2 \|u'\|_2,
	$$
	which is \eqref{eq:Sobolev-1d}.
\end{proof}

	We shall not give a proof of Theorem~\ref{thm:Sobolev} 
	for $d=2$ for general $p>2$ here, 
	but refer instead to e.g. \cite[Theorem~8.5]{LieLos-01}. 
	A proof for the important special case $p=4$ will be given below.
	For $d \ge 3$ we follow a proof which is closer in spirit to those of
	upcoming specializations of the uncertainty principle,
	and which was given in \cite{Lenzmann-13} based on \cite{CheXu-97}.
	It also generalizes straightforwardly to fractional Sobolev spaces but
	does not yield the optimal constant $S_d$ however.

\begin{proof}[Proof for $d \ge 3$]
	Let $u \in H^1(\R^d)$ and set $q = 2d/(d-2)$.
	Using the unitary Fourier transform $\hu = \cF u$ we may decompose 
	$u$ into low- and high-frequency parts,
	$u = u^\limminus_P + u^\limplus_P$, with
	$$
		u^\limminus_P := \cF^{-1}\left[ \1_{B_P(0)} \,\hu \right]
		\qquad \text{and} \qquad
		u^\limplus_P  := \cF^{-1}\left[ \1_{B_P(0)^c} \,\hu \right],
	$$
	for an arbitrary momentum/frequency $P > 0$ to be chosen below.
	We then use that
	\begin{equation} \label{eq:Sobolev-layercake}
		\|u\|_q^q = \int_{t=0}^\infty \bigl|\{ |u| > t \}\bigr| \,d(t^q),
	\end{equation}
	by the layer-cake representation \eqref{eq:layer-cake-p},
	and that by the triangle inequality $|u| \le |u^\limminus_P| + |u^\limplus_P|$,
	\begin{equation} \label{eq:Sobolev-triangle}
		\{ |u| > t \} \subseteq \{ |u^\limminus_P| > t/2 \} \cup 
			\{ |u^\limplus_P| > t/2 \}.
	\end{equation}
	Now, note that by the Fourier inversion formula,
	$\|f\|_\infty \le (2\pi)^{-d/2} \|\hf\|_1$, 
	so that 
	\begin{align*}
		(2\pi)^{d/2} \|u^\limminus_P\|_\infty 
		&\le \norm{\cF u^\limminus_P}_1
		= \int_{B_P(0)} \frac{1}{|\bp|} |\bp \hu(\bp)| \,d\bp
		\le \left( \int_{B_P(0)} \frac{d\bp}{|\bp|^2} \right)^{1/2}
			\norm{\cF(\nabla u)}_2 \\
		&= \left( \frac{|\S^{d-1}|}{d-2} \right)^{1/2} P^{\frac{d-2}{2}} \norm{\nabla u}_2,
	\end{align*}
	by Cauchy--Schwarz.
	Hence, if we choose 
	$$
		P = P(t) := \left( \frac{(d-2)(2\pi)^d}{|\S^{d-1}| \|\nabla u\|_2^2} 
			\frac{t^2}{4}
			\right)^{\frac{1}{d-2}}
		=: C_d (t/\|\nabla u\|_2)^{\frac{2}{d-2}}
	$$
	then $|\{ |u^\limminus_P| > t/2 \}| = 0$, and we obtain in 
	\eqref{eq:Sobolev-layercake}-\eqref{eq:Sobolev-triangle}
	$$
		\|u\|_q^q \le \int_{t=0}^\infty \bigl|\{ |u^\limplus_{P(t)}| > t/2 \}\bigr| \,d(t^q)
		\le \int_{t=0}^\infty 4\|u^\limplus_{P(t)}\|_2^2/t^2 \,d(t^q),
	$$
	by Chebyshev's inequality \eqref{eq:Chebyshev}. Thus,
	$$
		\|u\|_q^q \le 4q \int_0^\infty \|\cF u^\limplus_{P(t)}\|_2^2 \,t^{q-3} dt
		= 4q \int_0^\infty \int_{B_{P(t)}(0)^c} |\hu(\bp)|^2 d\bp \,t^{q-3} dt,
	$$
	and by Fubini's theorem and inverting the relation
	$$
		|\bp| \ge P(t) \ \Leftrightarrow \ 
		t \le \|\nabla u\|_2 (|\bp|/C_d)^{\frac{d-2}{2}} =: \Lambda(\bp),
	$$
	we have
	\begin{align*}
		\|u\|_q^q &\le 4q \int_{\R^d} |\hu(\bp)|^2 
			\int_0^{\Lambda(\bp)} t^{q-3} dt \,d\bp 
		= \frac{4q}{q-2} C_d^{-2} \|\nabla u\|^{q-2} \int_{\R^d} |\hu(\bp)|^2 
			|\bp|^2 \,d\bp \\
		&= 2d C_d^{-2} \|\nabla u\|^q.
	\end{align*}
	This proves the Sobolev inequality \eqref{eq:Sobolev} with the 
	constant
	$$
		S_d' = (2dC_d^{-2})^{-2/q} = \frac{(2\pi)^2}{(2d)^{\frac{d-2}{d}}} 
			\left( \frac{d-2}{4} \right)^{\frac{2}{d}} |\S^{d-1}|^{-\frac{2}{d}},
	$$
	which for $d=3$ is $S_3' = \pi^{4/3}/(2 \cdot 3^{1/3}) \approx 1.595$.
\end{proof}
		
\subsubsection{Sobolev from Hardy}
	
	Alternatively, the Sobolev inequality for $d \ge 3$ actually also 
	follows from the Hardy inequality,
	by the method of rearrangements; see \cite{FraSei-08,Seiringer-10}.
	Namely, for any radial, non-negative decreasing function 
	$u\colon \R^d \to \R_+$ one has the inequality
	\begin{equation}\label{eq:radial-symm-bound}
		\norm{u}_p^p = \int_{\R^d} u(\by)^p \,d\by
		\ge u(\bx)^p |\bx|^d |\B^d|
	\end{equation}
	for any $\bx \in \R^d$ and $p > 2$,
	where $\B^d = B_1(0)$ denotes the unit ball in $\R^d$.
	Taking both sides to the power $1-2/p$, 
	multiplying by $u(\bx)^2 |\bx|^{-d(1-2/p)}$
	and integrating over $\bx$, one obtains
	\begin{equation}\label{eq:radial-symm-Hardy-Sobolev}
		\int_{\R^d} \frac{u(\bx)^2}{|\bx|^{d(1-2/p)}} \,d\bx
		\ge |\B^d|^{1-2/p} \norm{u}_p^2.
	\end{equation}
	Taking $p=2d/(d-2)$, the l.h.s.~reduces to the r.h.s.~of the Hardy inequality
	\eqref{eq:Hardy}
	and thus Hardy implies Sobolev for such $u \in H^1(\R^d)$.
	Finally, one may use a symmetric-decreasing rearrangement 
	(we refer to e.g. \cite[Chapter~3]{LieLos-01} for details)
	to reduce an arbitrary $u \in H^1(\R^d)$ to such a 
	non-negative decreasing radial function $u^*$,
	with the properties
	$\|u\|_p = \|u^*\|_p$ and 
	$\|\nabla u\|_2 \ge \|\nabla u^*\|_2$.
	The first property follows by the layer-cake representation
	while last property is the non-trivial one, and
	it would be too much of a detour to try to
	cover this approach here.

\subsection{Gagliardo--Nirenberg--Sobolev}

	Note that by an application of H\"older's inequality, 
	for $d \ge 3$,
	\begin{equation}\label{eq:Hoelder-GNS}
		\int_{\R^d} |u|^{2(1+2/d)} 
		\le \|u\|_2^{4/d} \|u\|_{2d/(d-2)}^2,
	\end{equation}
	and thus by the Sobolev inequality \eqref{eq:Sobolev}, 
	\begin{equation}\label{eq:Sobolev-GNS}
		\left( \int_{\R^d} |\nabla u|^2 \right)
			\left( \int_{\R^d} |u|^2 \right)^{2/d}
		\ge S_d \int_{\R^d} |u|^{2(1+2/d)}.
	\end{equation}
	This is another formulation of the uncertainty principle
	known as a \keyword{Gagliardo--Nirenberg--Sobolev (GNS) inequality}.
	Note that such an inequality also follows in $d=2$ directly from 
	\eqref{eq:Sobolev-2d} with $p=4$, 
	and that in $d=1$ one has from \eqref{eq:Sobolev-1d} that
	$$
		\int_{\R} |u|^6 \le \|u\|_\infty^4 \|u\|_2^2 \le \|u\|_2^4 \|u'\|_2^2.
	$$
	Hence \eqref{eq:Sobolev-GNS} takes the same form in all dimensions $d \ge 1$,
	and we shall here present an independent and simple proof for it 
	which also allows for many useful generalizations
	(this is the one-body version of a proof 
	due to Rumin of a more general kinetic energy inequality,
	given later in Theorem~\ref{thm:LT-kinetic-1p}).
	
\begin{theorem}[Gagliardo--Nirenberg--Sobolev inequality --- one-body version]
	\label{thm:GNS}
	For any $d \ge 1$ there exists a constant $G_d > 0$ such that for all
	$u \in H^1(\R^d)$
	\begin{equation}\label{eq:GNS}
		\int_{\R^d} |\nabla u|^2
		\ge G_d \left( \int_{\R^d} |u|^2 \right)^{-2/d} \int_{\R^d} |u|^{2(1+2/d)}.
	\end{equation}
\end{theorem}
\begin{remark}\label{rem:GNS-constant}
	The optimal constant satisfies 
	$G_1 = \pi^2/4 > 1$, $G_2 = S_{2,4}$, respectively $G_d \ge S_d$ for $d \ge 3$,
	and for all $d \ge 1$ we also have that
	\begin{equation}\label{eq:GNS-constant}
		G_d \ge G_d' := \frac{(2\pi)^2 d^{2+2/d} |\S^{d-1}|^{-2/d}}{(d+2)(d+4)}.
	\end{equation}
	The exact value of the optimal constant $G_{d \ge 2}$ is presently unknown 
	but numerical work suggests
	$G_3 \approx 9.578$, to be contrasted with $S_3 \approx 5.478$
	and $G_3' \approx 3.907$ 
	(see \cite{LieThi-76,Lieb-76}, and also 
	\cite{Levitt-14} for more recent related numerical work).
\end{remark}
\begin{proof}
	We decompose an arbitrary $u \in H^1(\R^d)$ into parts corresponding
	to low respectively high kinetic energy according to
	$u = u_{E,-} + u_{E,+}$, 
	with an energy cut-off $E > 0$, and
	$$
		u_{E,-} := \cF^{-1}\left[ \1_{\{|\bp|^2 \le E\}} \,\hu \right]
		\qquad \text{and} \qquad
		u_{E,+}  := \cF^{-1}\left[ \1_{\{|\bp|^2 > E\}} \,\hu \right].
	$$
	Then by the unitarity of the Fourier transform, and Fubini,
	\begin{align}\label{eq:GNS-high-energy}
		\int_0^\infty &\int_{\R^d} |u_{E,+}(\bx)|^2 \,d\bx \,dE
		= \int_0^\infty \int_{\R^d} |\widehat{u_{E,+}}(\bp)|^2 \,d\bp \,dE
		= \int_{\R^d} \int_0^{|\bp|^2} |\hu(\bp)|^2 \,dE \,d\bp \ \\
		&= \int_{\R^d} |\bp|^2|\hu(\bp)|^2 \,d\bp
		= \int_{\R^d} |\nabla u(\bx)|^2 \,d\bx.
	\end{align}
	For the low-energy part we use that by Fourier inversion and Cauchy--Schwarz
	\begin{align}\label{eq:GNS-low-energy}
		|u_{E,-}(\bx)| &= \left| (2\pi)^{-d/2} \int_{\R^d} 
			\1_{\{|\bp|^2 \le E\}} \hu(\bp) e^{i\bp\cdot\bx} \,d\bp \right| \\
		&\le (2\pi)^{-d/2} |B_{E^{1/2}}(0)|^{1/2} \|\hu\|_2
		= (2\pi)^{-d/2} d^{-1/2} |\S^{d-1}|^{1/2} E^{d/4} \|u\|_2.
	\end{align}
	Now, combining \eqref{eq:GNS-high-energy} and \eqref{eq:GNS-low-energy}
	with the pointwise triangle inequality
	\begin{equation}\label{eq:GNS-triangle-ineq}
		|u_{E,+}(\bx)| \ge \Bigl[ |u(\bx)| - |u_{E,-}(\bx)| \Bigr]_+,
	\end{equation}
	yields the bound
	\begin{align*}
		\int_{\R^d} |\nabla u(\bx)|^2 \,d\bx
		\ge \int_0^\infty &\int_{\R^d} \left[
			|u(\bx)| - (2\pi)^{-d/2} d^{-1/2} |\S^{d-1}|^{1/2} \|u\|_2 E^{d/4}
			\right]_+^2 \,d\bx \,dE.
	\end{align*}
	Again changing the order of integration and then carrying 
	out the integral over $E$, with
	\begin{equation}\label{eq:GNS-integral-identity}
		\int_0^\infty \left[ A - B t^{d/4} \right]_+^2 dt
		= \frac{d^2 A^{2+4/d} B^{-4/d}}{(d+2)(d+4)},
	\end{equation}
	one finally obtains
	\begin{align*}
		\int_{\R^d} |\nabla u(\bx)|^2 \,d\bx
		\ge \frac{(2\pi)^2 d^{2+2/d} |\S^{d-1}|^{-2/d}}{(d+2)(d+4)} 
			\|u\|_2^{-4/d}
			\int_{\R^d} |u|^{2(1+2/d)}.
	\end{align*}
	This also produces the bound \eqref{eq:GNS-constant}
	for the optimal constant $G_d$ while
	for $d \ge 3$ this may be improved by \eqref{eq:Sobolev-GNS}.
\end{proof}

	We have also the following many-body version of the GNS inequality:
	
\begin{theorem}[Gagliardo--Nirenberg--Sobolev inequality --- many-body version]
	\label{thm:GNS-many-body}
	For any $d \ge 1$, $N \ge 1$, and every $L^2$-normalized $N$-body state 
	$\Psi \in H^1(\R^{dN})$,
	$$
		\sum_{j=1}^N \int_{\R^{dN}} |\nabla_j \Psi|^2
		\ge G_d \,N^{-2/d} \int_{\R^d} \varrho_\Psi^{1+2/d}.
	$$
\end{theorem}
	
	This can be proved either by directly generalizing the above proof (exercise) 
	or by using the following inequality followed by an application of
	Theorem~\ref{thm:GNS} to $u = \sqrt{\varrho_\Psi}$.
	
\begin{lemma}[\keyword{Hoffmann-Ostenhof inequality}]
	For any $d \ge 1$, $N \ge 1$, and every $L^2$-normalized $N$-body state 
	$\Psi \in H^1(\R^{dN})$,
	\begin{equation}\label{eq:HO}
		\sum_{j=1}^N \int_{\R^{dN}} |\nabla_j \Psi|^2
		\ge \int_{\R^d} |\nabla \sqrt{\varrho_\Psi}|^2.
	\end{equation}
\end{lemma}

	The inequality \eqref{eq:HO} is actually equivalent to its 
	one-body version, i.e.\ the simple inequality
	\begin{equation}\label{eq:diamagnetic}
		\int_{\R^d} |\nabla u|^2 \ge \int_{\R^d} \bigl|\nabla |u|\bigr|^2,
	\end{equation}
	which is known as a \keyword{diamagnetic inequality}
	(because it holds in greater generality also involving magnetic fields;
	see e.g. \cite[Theorem~7.21]{LieLos-01}).
	The many-body version \eqref{eq:HO} was first proved in \cite{Hof-77} 
	(see also e.g. \cite[Lemma~3.2]{Lewin-15} for a simple generalization and proof).

\begin{exc}
	Prove the inequality \eqref{eq:Hoelder-GNS}.
\end{exc}

\begin{exc}\label{exc:GNS-many-body}
	Prove Theorem~\ref{thm:GNS-many-body} by defining for each
	$\sx = (\bx_1,\ldots,\bx_{j-1},\bx_{j+1},\ldots,\bx_N) \in \R^{d(N-1)}$
	a collection of functions
	$$
		u_j(\bx,\sx') := \Psi(\bx_1,\ldots,\bx_{j-1},\bx,\bx_{j+1},\ldots,\bx_N)
	$$
	and using instead of \eqref{eq:GNS-triangle-ineq}
	the triangle inequality on $L^2(\R^{d(N-1)};\C^N)$,
	\begin{align*}
		&\left( \int_{\R^{d(N-1)}} \sum_{j=1}^N |u_j^{E,+}(\bx,\sx')|^2 \,d\sx' \right)^{1/2} \\
		&\ge \left[
		 \left( \int_{\R^{d(N-1)}} \sum_{j=1}^N |u_j(\bx,\sx')|^2 \,d\sx' \right)^{1/2}
		-\left( \int_{\R^{d(N-1)}} \sum_{j=1}^N |u_j^{E,-}(\bx,\sx')|^2 \,d\sx' \right)^{1/2}
		\right]_+.
	\end{align*}
\end{exc}

\subsection{Applications to stability}\label{sec:uncert-stability}

	In this subsection we follow mainly \cite{Lieb-76,LieSei-09,Seiringer-10}.
	\index{stability}

\subsubsection{The stability of the hydrogenic atom}

	We now return to the hydrogenic atom of Example~\ref{exmp:hydrogenic-2p},
	which after factoring out the free center-of-mass kinetic energy
	leaves the more relevant relative Hamiltonian operator
	on $\cH = L^2(\phX_\rel) = L^2(\R^d)$:
	\begin{equation}\label{eq:hydrogenic-Hamiltonian}
		\hH_\rel = -\Delta - \frac{Z}{|\bx|}.
	\end{equation}
	Here we have put for simplicity $2\mu=1$ for the reduced mass or, 
	equivalently, rescaled the operator and the value of $Z$,
	which is no loss in generality.
	This operator should be understood to be defined via the energy form
	$$
		\cE[\psi] := q_{\hH_\rel}(\psi)
		= \inp{\psi, \left[ -\Delta - \frac{Z}{|\bx|} \right] \psi}_{L^2(\R^3)}
		= \int_{\R^3} |\nabla \psi|^2 - Z\int_{\R^3} \frac{|\psi(\bx)|^2}{|\bx|} d\bx,
	$$
	where we may take $\psi$ in the minimal form domain $C_c^\infty(\R^d)$ 
	or a larger closed domain $\cQ(\hH_\rel) \subseteq H^1(\R^3)$. 
	Any self-adjoint realization 
	of $\hH_\rel$ associated to this form
	then has the ground-state energy
	$$
		E_0 = \inf \bigl\{ \cE[\psi] : \psi \in H^1(\R^3), \ \|\psi\|_2 = 1 \bigr\}.
	$$

	Using Heisenberg's uncertainty principle \eqref{eq:Heisenberg-inverse}, 
	one has the bound
	\begin{equation}\label{eq:Heisenberg-fail}
		\cE[\psi] 
		\ge \frac{9}{4}\left( \int_{\R^3} |\bx|^2|\psi(\bx)|^2 \,d\bx \right)^{-1} 
		- Z\int_{\R^3} \frac{|\psi(\bx)|^2}{|\bx|} d\bx,
	\end{equation}
	whose r.h.s.~can be made arbitrarily negative (Exercise~\ref{exc:Heisenberg-fail}),
	therefore not being able to settle the stability question.
	However, using instead the Hardy inequality \eqref{eq:Hardy}, 
	we obtain the lower bound
	$$
		\cE[\psi] 
		\ge \int_{\R^3} \left[ \frac{1}{4|\bx|^2} - \frac{Z}{|\bx|} \right] 
			|\psi(\bx)|^2 \,d\bx.
	$$
	We may then proceed
	by minimizing the expression in brackets pointwise:
	$$
		\min_{\bx \in \R^3} \left[ \frac{1}{4|\bx|^2} - \frac{Z}{|\bx|} \right]
		= -Z^2,
		\qquad \text{for} \ |\bx| = \frac{1}{2Z},
	$$
	and by the normalization of $\psi$ we therefore obtain the finite lower bound
	\begin{equation}\label{eq:hydrogenic-bound-Hardy}
		E_0 \ge -Z^2,
	\end{equation}
	and thus stability for the hydrogenic atom for any finite charge $Z > 0$
	(and it is trivially stable for $Z \le 0$ according to our definitions).
	In summary, we have thus found that the uncertainty principle 
	introduces an effective repulsion
	around the origin which overcomes the attraction from the nucleus
	by its stronger scaling property,
	scaling quadratically in the inverse distance as opposed to linearly,
	here resulting in an equilibrium around $|\bx| = 1/(2Z)$.

	Another approach is to use the Gagliardo--Nirenberg--Sobolev inequality 
	\eqref{eq:GNS}, that is
	\begin{equation}\label{eq:GNS-hydrogenic}
		\cE[\psi] 
		\ge G_3 \int_{\R^3} |\psi|^{10/3}
		- Z\int_{\R^3} \frac{|\psi(\bx)|^2}{|\bx|} d\bx,
	\end{equation}
	with the constant $G_3 \ge S_3 \approx 5.478$.
	In this case we are led to a constrained optimization problem for the density
	$\varrho := |\psi|^2$,
	\begin{equation}\label{eq:GNS-hydrogenic-minimization}
		E_0 \ge
		\inf \left\{ \int_{\R^3} \left( G_3 \varrho(\bx)^{5/3}
			- Z \frac{\varrho(\bx)}{|\bx|} \right) d\bx
			\ : \ \varrho\colon \R^3 \to \R_+, \int_{\R^3} \varrho = 1 \right\},
	\end{equation}
	whose minimum can be shown (Exercise~\ref{exc:GNS-stability}) to be 
	\begin{equation}\label{eq:GNS-hydrogenic-minimum}
		-9(\pi/2)^{4/3} Z^2/(5G_3) \ge -3Z^2/5
	\end{equation}
	for 
	$$
		\varrho(\bx) = \left( \frac{3}{5} \frac{Z}{G_3} (|\bx|^{-1} - R^{-1})_+ \right)^{3/2},
	$$ 
	with $R = \frac{3}{5}(2/\pi)^{4/3} G_3/Z$.
	Therefore the 
	GNS inequality, which arose 
	as a weaker implication of the Sobolev inequality
	(and thus a yet weaker implication of the Hardy inequality),
	is still strong enough to enforce stability.
	One may even note that formally replacing the exponent $10/3$ 
	in the kinetic term 
	in \eqref{eq:GNS-hydrogenic} by anything strictly greater than $3$, 
	i.e. the $5/3$ in \eqref{eq:GNS-hydrogenic-minimization}
	by any exponent $p > 3/2$
	(which would require the constant $G_3$ to be dimensionful however),
	would have been sufficient for stability (exercise).

	It is actually possible to solve for the complete spectrum 
	$\sigma(\hH_\rel)$ of the hydrogenic atom, 
	which was indeed worked out shortly after the birth of quantum mechanics.
	In particular, the exact ground state may be seen to be
	(with a normalization constant $C>0$)
	$$
		\psi_0(\bx) = C e^{-Z|\bx|/2},
	$$
	since this function is positive, square-integrable, and solves the Schr\"odinger 
	eigenvalue equation
	\begin{equation}\label{eq:hydrogenic-groundstate}
		\hH_\rel \psi_0 = \left( -\Delta - \frac{Z}{|\bx|} \right) \psi_0 
		= -\frac{Z^2}{4} \psi_0
	\end{equation}
	(see e.g. \cite[Section~2.2.2]{LieSei-09} 
	and \cite[Section~11.10]{LieLos-01} for details).
	Thus, by the min-max principle we have
	$$
		\cE[\psi] =
		\int_{\R^3} |\nabla \psi|^2 - Z\int_{\R^3} \frac{|\psi(\bx)|^2}{|\bx|} d\bx
		\ge -\frac{Z^2}{4} \int_{\R^3} |\psi(\bx)|^2 d\bx
	$$
	for all $\psi \in \cQ(\hH_\rel)$.
	We thus see that the above-obtained bounds using the Hardy and Sobolev/GNS
	uncertainty principles are quite close to the
	actual ground-state energy 
	$E_0 = -Z^2/4$.
	Moreover, the entire spectrum of the operator \eqref{eq:hydrogenic-Hamiltonian} 
	turns out to be
	(see e.g. \cite[Chapter~10]{Teschl-14})
	$$
		\sigma(\hH_\rel) 
		= \left\{ -\frac{Z^2}{4(n+1)^2} \right\}_{n=0,1,2,\ldots} 
			\cup [0,\infty).
	$$
	The infinite sequence of negative eigenvalues of finite multiplicity
	are the energy levels of the bound electron\index{bound states},
	with eigenstates corresponding to the ground state and the excited orbitals of the atom,
	while the non-negative essential spectrum describes states where the
	electron is not bound to the nucleus but rather scatters off of it,
	i.e.\ \keyword{scattering states}.

\begin{exc}\label{exc:Heisenberg-fail}
	Prove that the r.h.s.\ of \eqref{eq:Heisenberg-fail} 
	tends to $-\infty$ for some sequence of $L^2$-normalized states 
	$\psi \in H^1(\R^d)$.
\end{exc}
\begin{exc}\label{exc:GNS-stability}
	Compute the minimizer for the variational problem 
	\eqref{eq:GNS-hydrogenic-minimization} in the generalized case with exponent
	$p > 3/2$ (why is this bound necessary?),
	and the minimum \eqref{eq:GNS-hydrogenic-minimum} in the case $p=5/3$.
\end{exc}
\begin{exc}
	Verify the Schr\"odinger equation \eqref{eq:hydrogenic-groundstate}.
	How can we be sure that $\psi_0$ is the ground state?
\end{exc}

\subsubsection{General criteria for stability of the first kind}

	In the case of a one-body Schr\"odinger Hamiltonian operator
	$\hH = -\Delta + V$ on $\cH = L^2(\R^d)$ 
	with a general potential $V\colon \R^d \to \R$, $d \ge 3$,
	we have using Sobolev that for all $\psi \in H^1(\R^d)$
	$$
		\cE[\psi] := q_{\hH}(\psi)
		= \int_{\R^d} |\nabla \psi|^2 + \int_{\R^d} V|\psi|^2
		\ge S_d \|\psi\|_{2d/(d-2)}^2 + \int_{\R^d} (V_+ - |V_-|)|\psi|^2,
	$$
	with $V_\pm := (V \pm |V|)/2$.
	Furthermore, if $V_- \in L^{d/2}(\R^d)$
	then we have using H\"older that
	$$
		\int_{\R^d} |V_-||\psi|^2 \le \|V_-\|_{d/2} \|\psi\|_{2d/(d-2)}^2,
	$$
	and therefore
	$$
		\cE[\psi] \ge \left( S_d - \|V_-\|_{d/2} \right) \|\psi\|_{2d/(d-2)}^2.
	$$
	Assuming $\|V_-\|_{d/2} \le S_d$
	then implies $\cE[\psi] \ge 0$, 
	and hence clearly stability for such potentials.
	However, one may also extract an arbitrary negative constant from the 
	potential without changing this conclusion.
	In general, if
	$$
		V(\bx) = U(\bx) + v(\bx),
	$$
	where $U \ge -C$, i.e. $U_- \in L^\infty(\R^d)$, and $v \in L^{d/2}(\R^d)$,
	then by the layer-cake principle there exists for any $\eps \in (0,1)$
	some constant $A_\eps \ge 0$
	such that $\|(A_\eps+v)_-\|_{d/2} \le \eps S_d$,
	and thus
	\begin{align*}
		\cE[\psi] &= T[\psi] + V[\psi] 
		= (1-\eps)T[\psi] + \eps T[\psi] 
			+ \int \bigl(U - A_\eps + (A_\eps+v) \bigr)|\psi|^2 \\
		&\ge (1-\eps)T[\psi] + \eps T[\psi] - \|U_-\|_\infty - A_\eps 
			- \int \bigl|(A_\eps+v)_-\bigr||\psi|^2 \\
		&\ge (1-\eps)T[\psi] - \|U_-\|_\infty - A_\eps 
			+ \bigl( \eps S_d - \|(A_\eps + v)_-\|_{d/2} \bigr) \|\psi\|_{2d/(d-2)}^2 \\
		&\ge (1-\eps)T[\psi] - \|U_-\|_\infty - A_\eps.
	\end{align*}
	This is summarized in the following theorem, where the case $d \le 2$
	is left as an exercise:

\begin{theorem}\label{thm:stability-by-Sobolev}
	Given a Schr\"odinger Hamiltonian \eqref{eq:one-body-Hhat}
	on $\R^d$ with quadratic form
	$$
		\cE[\psi]
		= \int_{\R^3} \left( |\nabla \psi|^2 + V|\psi|^2 \right),
	$$
	for some potential $V\colon \R^d \to \R$
	and finite kinetic energy, $\psi \in H^1(\R^d)$,
	there is \keyword{stability} for the corresponding quantum system, i.e.
	$$
		E_0 = \inf \bigl\{ \cE[\psi] : \psi \in H^1(\R^3), \ \|\psi\|_2 = 1 \bigr\}
		\quad > \ -\infty,
	$$
	if
	$$
		V_- \in \left\{\begin{array}{ll}
			L^{d/2}(\R^d) + L^\infty(\R^d), & d \ge 3, \\
			L^{1+\eps}(\R^2) + L^\infty(\R^2), & d = 2, \\
			L^1(\R^1) + L^\infty(\R^1), & d = 1.
		\end{array}\right.
	$$
\end{theorem}

	The Hardy inequality can in fact be even stronger than the above theorem, 
	namely we have immediately by \eqref{eq:Hardy} 
	that $E_0 > -\infty$ if
	$$
		V(\bx) \ge -\frac{(d-2)^2}{4|\bx|^2} -C,
	$$
	(note that the r.h.s.\ is not in $L^{d/2}_\loc(\R^d)$ for $d \ge 3$
	so Theorem~\ref{thm:stability-by-Sobolev} does not apply),
	or even if (exercise)
	\begin{equation}\label{eq:many-points-Hardy-pot}
		V(\bx) \ge -\frac{(d-2)^2}{4} \sum_{k=1}^M |\bx - \bR_k|^{-2} - C,
	\end{equation}
	for finitely many distinct points $\bR_j \neq \bR_k$ in $\R^d$.

	If $V_-$ is not too singular then
	one may also obtain an explicit bound for $E_0$ directly 
	from the GNS inequality \eqref{eq:GNS},
	which even turns out to be equivalent to such a bound:

\begin{theorem}[GNS---Schr\"odinger equivalence]\label{thm:stability-by-GNS}
	The ground-state energy $E_0$ of the Schr\"o\-dinger form $\cE[\psi]$
	in Theorem~\ref{thm:stability-by-Sobolev} 
	is bounded from below by
	\begin{equation}\label{eq:one-body-LT}
		E_0 \ge -L^1_d \int_{\R^d} |V_-|^{1+d/2},
	\end{equation}
	with the positive constant 
	$L^1_d = \frac{2}{d+2} \bigl(\frac{d}{d+2}\bigr)^{d/2} G_d^{-d/2}$.
	
	Conversely, if a bound of the form \eqref{eq:one-body-LT} holds for
	arbitrary potentials $V$ and some constant $L^1_d > 0$, then 
	the GNS inequality \eqref{eq:GNS} holds for all $u \in H^1(\R^d)$ with
	the positive constant
	$G_d = \frac{d}{d+2} \bigl(\frac{2}{d+2}\bigr)^{2/d} (L^1_d)^{-2/d}$.
\end{theorem}
\begin{proof}
	To obtain \eqref{eq:one-body-LT}, note that
	by GNS \eqref{eq:GNS} and H\"older we have that
	for any $L^2$-normalized $\psi \in H^1(\R^d)$
	\begin{align*}
		\cE[\psi] &\ge G_d \int_{\R^d} |\psi|^{2(1+2/d)} 
			- \left( \int_{\R^d} |V_-|^{(d+2)/2} \right)^{2/(d+2)}
				\left( \int_{\R^d} |\psi|^{2(d+2)/d} \right)^{d/(d+2)} \\
		&\ge - L^1_d \int_{\R^d} |V_-|^{(d+2)/2},
	\end{align*}
	where we used the fact that the function
	$\Rplus \ni t \mapsto At - Bt^{d/(d+2)}$ for $A,B > 0$ has the minimal value
	$-\frac{2}{d+2} \bigl(\frac{d}{d+2}\bigr)^{d/2} A^{-d/2} B^{(d+2)/2}$.
	
	On the other hand, if \eqref{eq:one-body-LT} holds, then 
	first assume that $\psi \in H^1(\R^d)$ with $\|\psi\|_2 = 1$
	and let us write
	for an arbitrary potential $V$:
	\begin{align*}
		T[\psi] &= T[\psi] + V[\psi] - V[\psi]
		\ge E_0 - \int_{\R^d} V|\psi|^2
		\ge -L^1_d \int_{\R^d} |V_-|^{1+d/2} - \int_{\R^d} V|\psi|^2.
	\end{align*}
	Now, take $V(\bx) := -c|\psi(\bx)|^\alpha$ 
	and demand that the above integrals involving $|\psi|$ match modulo constants,
	i.e. $\alpha(1+d/2) = \alpha + 2$, or equivalently, $\alpha=4/d$. Thus,
	$$
		T[\psi] \ge \left( c - c^{1+d/2} L^1_d \right) \int_{\R^d} |\psi|^{2(1+2/d)},
	$$
	and we may finally optimize in $c > 0$ to obtain \eqref{eq:GNS}
	with the claimed relationship between $G_d$ and $L^1_d$.
	In the case that $\lambda := \|\psi\|_2 \neq 1$ 
	the homogeneous GNS inequality \eqref{eq:GNS} is obtained by simple rescaling 
	$\psi = \lambda \tilde\psi$.
\end{proof}

\begin{exc}
	Prove Theorem~\ref{thm:stability-by-Sobolev} in the case $d=1$ and $d=2$.
\end{exc}

\begin{exc}
	Prove that $E_0 > -\infty$ for \eqref{eq:many-points-Hardy-pot}.
\end{exc}

\subsection{Poincar\'e}\label{sec:uncert-Poincare}

	The Heisenberg, Hardy and Sobolev inequalities were all \emph{global} 
	in the sense that they involved the full configuration space $\R^n$
	(or Dirichlet restrictions of it, by simple restriction of the domain
	to $H^1_0(\Omega) \subseteq H^1(\R^n)$).
	We shall now consider some \emph{local} formulations of uncertainty principles
	(amounting to Neumann restrictions which potentially \emph{increase} the domain),
	the prime example being the Poincar\'e inequality.
	
\begin{definition}[Poincar\'e inequality]
	A \keyword{Poincar\'e inequality} on a domain $\Omega \subseteq \R^n$ 
	with finite measure $|\Omega|$
	is a lower bound of the form
	\begin{equation}\label{eq:Poincare}
		\int_\Omega |\nabla u|^2
		\ge C_\sP \int_\Omega |u - u_\Omega|^2
	\end{equation}
	for some constant $C_\sP = C_\sP(\Omega) > 0$
	and for all $u \in H^1(\Omega)$, 
	where in the r.h.s.~we have subtracted the average of $u$ on $\Omega$,
	\begin{equation}\label{eq:u_Omega}
		u_\Omega := |\Omega|^{-1} \int_\Omega u.
	\end{equation}
\end{definition}

\begin{remark}
	Note that if $C_\sP > 0$ then $\Omega \subseteq \R^d$ has to be a connected set,
	for otherwise we may choose $u$ to be a non-zero constant on each connected 
	component and such that $u_\Omega = 0$, for example
	$u = |\Omega_1|^{-1}\1_{\Omega_1} - |\Omega_2|^{-1}\1_{\Omega_2}$,
	so that the l.h.s.~of \eqref{eq:Poincare} is zero but the r.h.s.~non-zero.
	Also note that by dimensional scaling, 
	$C_\sP(\Omega) = |\Omega|^{-2/d} C_\sP(\Omega/|\Omega|)$,
	where $C_\sP(\Omega/|\Omega|)$ only depends on the shape of $\Omega$.
\end{remark}

	It is useful to reformulate the inequality \eqref{eq:Poincare} 
	as an operator relation for the Laplacian on $\Omega$.
	We note that, since $\nabla u_\Omega = 0$, the l.h.s. of \eqref{eq:Poincare}
	can also be written
	\begin{equation}\label{eq:Poincare-LHS}
		\int_\Omega |\nabla u|^2 = \int_\Omega |\nabla(u - u_\Omega)|^2
		= \langle (u - u_\Omega), -\Delta^\eN (u - u_\Omega) \rangle
	\end{equation}
	and, with $u_0 := |\Omega|^{-1/2}$ the $L^2$-normalized zero-eigenfunction
	of the Neumann Laplacian on $\Omega$, we can write
	$u_\Omega = u_0 \langle u_0, u \rangle$, and thus
	\begin{equation}\label{eq:Poincare-RHS}
		\langle u_0, u - u_\Omega \rangle 
		= \langle u_0,u \rangle - \langle u_0,u_0 \rangle \langle u_0,u \rangle = 0.
	\end{equation}
	In other words, 
	if we introduce $P_0 := u_0 \langle u_0, \cdot \rangle$ 
	the orthogonal projection operator on the ground-state eigenspace $W_0 = \C u_0$
	and $P_0^\perp = \1-P_0$ the projection on the orthogonal subspace $W_0^\perp$,
	we have $u_\Omega = P_0u$ and $u - u_\Omega = P_0^\perp u$.
	Hence, the Poincar\'e inequality \eqref{eq:Poincare} equivalently says
	\begin{equation}\label{eq:Poincare-Hilbert}
		\langle P_0^\perp u, (-\Delta_\Omega^\eN) P_0^\perp u \rangle 
		\ge C_\sP \langle P_0^\perp u, P_0^\perp u\rangle
	\end{equation}
	or, as an operator inequality,
	\begin{equation}\label{eq:Poincare-Hilbert-ops}
		(-\Delta_\Omega^\eN) P_0^\perp \ge C_\sP P_0^\perp.
	\end{equation}
	Hence, we see that finding the best possible constant $C_\sP$ 
	for a given domain $\Omega$
	is the same as finding the second lowest eigenvalue 
	$\lambda_1 \ge \lambda_0=0$ for the Laplace operator
	$-\Delta_\Omega^\eN$ (with Neumann boundary conditions) on $\Omega$,
	$$
		-\Delta_\Omega^\eN = \sum_{k=0}^\infty \lambda_k P_k,
		\qquad P_0 = u_0 \inp{u_0, \slot},
		\qquad P_0^\perp = \sum_{k \ge 1} P_k = \sum_{k \ge 1} u_k \inp{u_k, \slot}.
	$$
	This is actually just the content of the min-max theorem
	of Section~\ref{sec:prelims-ops-specthm}, 
	applied to the form \eqref{eq:Poincare}.
	Also, we see that $C_\sP = \lambda_1-\lambda_0>0$ 
	if and only if there is a gap in the spectrum
	between the lowest eigenvalue $\lambda_0=0$ (the ground-state energy)
	and the second-lowest one 
	(the first excited energy level) $\lambda_1$.
	
\begin{example}\index{Laplacian}
	As a prototype case one may consider
	the Laplacian on the unit interval $[0,1]$.
	The eigenfunctions of the Neumann problem are
	$u_k(x) = C\cos (\pi k x)$ with eigenvalues $\lambda = \pi^2 k^2$, 
	$k=0,1,2,\ldots$. Hence we have a Poincar\'e inequality
	\begin{equation}\label{eq:Poincare-interval}
		\int_0^1 |u'|^2
		\ge \pi^2 \int_0^1 \left| u - {\textstyle\int_0^1} u \right|^2,
	\end{equation}
	for $u \in H^1([0,1])$, with the optimal Poincar\'e constant $C_\sP = \pi^2$.
	In the case of the Dirichlet problem, with $u_k(x) = C\sin (\pi k x)$,
	$\lambda = \pi^2k^2$, $k=1,2,\ldots$, one has an inequality
	\begin{equation}\label{eq:Poincare-interval-Dir}
		\int_0^1 |u'|^2
		\ge \pi^2 \int_0^1 |u|^2,
	\end{equation}
	for any $u \in H_0^1([0,1])$, without any projection in this case.
\end{example}
		
\begin{example}
	A Poincar\'e inequality of the form \eqref{eq:Poincare} cannot hold on 
	the unbounded interval $\R$ or $\Rplus$, not only because of the lack of an
	integrable ground state $u_0$ to project out as in \eqref{eq:u_Omega},
	but more crucially because of the lack of a spectral gap in this case.
	We have already seen and used that 
	$\sigma(-\Delta_\R) = \sigma(\hp_1^2) = [0,\infty)$,
	but also on the half-line $\Rplus$ one may consider a sequence of trial states
	such as
	$$
		u_L(x) = \sqrt{2/L} \sin (\pi x/L) \1_{[0,L]} 
			\ \in H_0^1(\Rplus) \subseteq H^1(\Rplus).
	$$
	By taking $L \to \infty$, such states have an arbitrarily low energy, 
	and one may furthermore pick an orthogonal sequence such as 
	$\{u_L(\slot + nL)\}_{n=0,1,2,\ldots}$ and use the min-max principle to find
	that the essential spectrum $\sigma_\textup{ess}(-\Delta_{\Rplus}^{\eN/\eD})$ 
	starts at zero
	(such a sequence is known as a \keyword{Weyl sequence}).
	Furthermore, by e.g. multiplying $u_L$ with a phase $e^{i\kappa x}$,
	any $\lambda = \kappa^2 \ge 0$ may be seen to be a point of the essential 
	spectrum as well, so
	$\sigma(-\Delta_{\Rplus}^{\eN/\eD}) = [0,\infty)$.
\end{example}

	Note that the existence of a gap in the spectrum
	\index{spectral gap}
	of the Laplacian on a domain $\Omega$
	can also be interpreted as a form of the uncertainty principle
	of $\hx_j$ and $\hp_j$, 
	since the corresponding Hamiltonian describing 
	the free kinetic energy of a particle on $\Omega$
	of mass $m=1/2$ is actually 
	$\hH = \hbp_\Omega^2 = \hbar^2(-\Delta_\Omega)$, 
	which then has a gap of size proportional to $\hbar^2$.
	If $\hx_j$ and $\hp_j$ would be made to commute, as they do classically, by
	formally taking $\hbar \to 0$ in \eqref{eq:quantum-CCR}, then 
	$$
		\sigma(\hH) = \hbar^2\{\lambda_0,\lambda_1,\ldots\} \to [0,\infty)
	$$
	and the gap would therefore close\footnote{In fact any gap in the spectrum will 
	close in this limit because the eigenvalues $\lambda_k$ 
	are distributed rather uniformly;
	they can on bounded regular domains $\Omega$ be shown to satisfy
	$(\lambda_{k+1}-\lambda_{k})/\lambda_{k+1} \to 0$ as $k \to \infty$ 
	(cf. Exercise~\ref{exc:Weyl-cube}).}.
	Another way to think about this limit is that $\hbar$ may be compensated
	for by rescaling the domain, 
	$\Omega \mapsto \Omega/\hbar$, and $\Omega/\hbar \to \R^d$ as $\hbar \to 0$,
	so that
	$$
		\sigma(\hH) = \sigma(-\Delta_{\Omega/\hbar})
		\xrightarrow{\hbar \to 0} \sigma(-\Delta_{\R^d}) 
		= \sigma(\hbp_{\R^d}^2)
		= \sigma(\check{\bp}_{\R^d}^2)
		= [0,\infty),
	$$
	by Fourier transform.
	
\begin{example}
	Poincar\'e inequalities also extend to other contexts where there is
	a gap in the spectrum, such as on compact, connected manifolds.
	One has for example the following Poincar\'e inequality 
	on the unit sphere $\S^{d-1}$ in $\R^d$:
	\begin{equation}\label{eq:Poincare-sphere}
		\int_{\S^{d-1}} |\nabla u|^2 \ge (d-1) \int_{\S^{d-1}} 
			\left|u - |\S^{d-1}|^{-1} {\textstyle\int_{\S^{d-1}}} u \right|^2,
	\end{equation}
	for $u \in H^1(\S^{d-1})$.
	This 
	follows from the following theorem concerning the spectrum of
	the \keyword{Laplace-Beltrami operator} on $\S^{d-1}$,
	i.e. 
	the operator $-\Delta_{\S^{d-1}}$
	associated to the non-negative quadratic form
	of the l.h.s.~of \eqref{eq:Poincare-sphere}
	(see e.g. \cite[Chapter~22]{Shubin-01} for further details and proofs).
\end{example}
	
	\index{Laplacian}
\begin{theorem}[Spectrum of the Laplacian on the sphere; 
	see e.g. {\cite[Theorem~22.1 and Corollary~22.2]{Shubin-01}}]
	The spectrum of the operator $-\Delta_{\S^{d-1}}$ is discrete
	and its eigenvalues are given by $\lambda = k(k+d-2)$, $k=0,1,2,\ldots$,
	with multiplicity given by the dimension of the space of homogeneous,
	harmonic polynomials on $\R^d$ of degree $k$,
	which is $\binom{k+d-1}{d-1} - \binom{k+d-3}{d-1}$
	for $k \ge 2$.
\end{theorem}
	
\begin{exc}
	Prove the Poincar\'e inequality \eqref{eq:Poincare-sphere}
	on the unit circle $\S^1$, by treating it as an
	interval $[0,2\pi]$ with periodic boundary conditions,
	$u(0) = u(2\pi)$, $u'(0) = u'(2\pi)$.
\end{exc}
\begin{exc}
	Prove a Poincar\'e inequality on an annulus 
	$\Omega_{R_1,R_2} = B_{R_2}(0) \setminus \bar{B}_{R_1}(0)$
	and give some explicit non-zero lower bound for the constant $C_\sP$
	depending on $R_2 > R_1 > 0$.
\end{exc}

\subsection{Local Sobolev-type inequalities}\label{sec:uncert-local-Sobolev}

	There is a family of important inequalities which combine the properties 
	of the global Sobolev inequality of Section~\ref{sec:uncert-Sobolev} 
	with the local properties
	of the Poincar\'e inequality of Section~\ref{sec:uncert-Poincare},
	and which are thus called \keyword{Poincar\'e-Sobolev inequalities}.
	However, these are typically a little more involved to prove and we shall 
	therefore instead take a more direct route to obtain the inequalities that we
	will need, of the form of the Gagliardo-Nirenberg-Sobolev inequalities
	of Theorems~\ref{thm:GNS} and \ref{thm:GNS-many-body}, 
	and which only relies on 
	knowledge of the eigenvalues for the Laplacian on a cube $Q$.
	They could be considered variants of the above-mentioned 
	Poincar\'e-Sobolev inequalities though.
	See \cite{BenValVan-18} for very recent generalizations 
	and improvements of the bounds given below.
	
\subsubsection{Laplacian eigenvalues on the cube}\label{sec:uncert-local-cube}

	\index{Laplacian}
	Consider the Neumann Laplacian on a cube $Q = [0,L]^d \subseteq \R^d$
	and the number $N(E)$ of its eigenvalues $\lambda_k$ below an energy $E>0$
	(note that since $\lambda_0 = 0$ we always have $N(E) \ge 1$).
	In the case $d=1$,
	$$
		\lambda_k = \frac{\pi^2}{|Q|^2} k^2,
		\qquad k = 0,1,2,\ldots,
	$$
	and
	$$
		N(E)-1 = \#\{ k : 0 < \lambda_k < E \} 
		= \#\{ k \in \Z : 0 < k < E^{1/2}|Q|/\pi \}
		\le E^{1/2}|Q|/\pi.
	$$
	In the case $d \ge 2$ we have
	$$
		\lambda_\bk = \frac{\pi^2}{|Q|^{2/d}} |\bk|^2,
		\qquad \bk \in \Z_{\ge 0}^d,
	$$
	and
	\begin{align*}
		N(E)-1 &= \#\{ \bk : 0 < \lambda_\bk < E \} 
		= \#\bigl\{ \bk \in \Z_{\ge 0}^d : 0 < |\bk| < E^{1/2}|Q|^{1/d}/\pi \bigr\} \\
		&\le 2^d(E^{1/2}|Q|^{1/d}/\pi)^d,
	\end{align*}
	where we for $E^{1/2}|Q|^{1/d}/\pi \ge 1$
	roughly bounded the number of integer points of the first quadrant
	inside a sphere of radius $R$ by the volume of an enclosing cube of 
	side length $R+1$.
	Hence,
	\begin{equation}\label{eq:cube-bound-eigenvalues}
		N(E) \le 1 + 2^d|Q|/\pi^d \cdot E^{d/2}
	\end{equation}
	for all $d \ge 1$.
	Also note that the orthonormal eigenfunctions are given explicitly by
	$$
		u_\bk(\bx) = |Q|^{-1/2} \prod_{j=1}^d c_{k_j} \cos \frac{\pi k_j x_j}{|Q|^{1/d}},
	$$
	with $c_0 = 1$ and $c_{k \ge 1} = \sqrt{2}$, so that
	\begin{equation}\label{eq:cube-bound-eigenfunctions}
		\|u_\bk\|_\infty \le |Q|^{-1/2} \prod_{j=1}^d c_{k_j} \le 2^{d/2}|Q|^{-1/2}.
	\end{equation}

\subsubsection{A Poincar\'e-Sobolev-type bound}

	The following is a local version of the Gagliardo--Nirenberg--Sobolev
	inequality of Theorem~\ref{thm:GNS}
	(it was given in this form as Theorem~13 in \cite{LunSol-13a}):

\begin{theorem} 
	\label{thm:GNS-local}
	For any $d \ge 1$ there exists a constant $C_d > 0$ such that for 
	any $d$-cube $Q \subseteq \R^d$ and every $u \in H^1(Q)$
	\begin{equation}\label{eq:GNS-local}
		\int_Q |\nabla u|^2
		\ge C_d \left( \int_Q |u|^2 \right)^{-2/d} 
			\int_Q \left[ |u| - \left(\frac{{\textstyle\int_Q} |u|^2}{|Q|}\right)^{1/2} \right]_+^{2(1+2/d)}.
	\end{equation}
\end{theorem}
\begin{proof}
	We make a decomposition of $u$ similar to the one in the proof of 
	Theorem~\ref{thm:GNS},
	but this time w.r.t.~the Neumann kinetic energy on $Q$.
	Namely, we define
	$u = u_{E,-} + u_{E,+}$
	with an energy cut-off $E > 0$
	and the spectral projections
	$$
		u_{E,-} := P^{-\Delta_Q^\eN}_{[0,E)} u
		\qquad \text{and} \qquad
		u_{E,+} := P^{-\Delta_Q^\eN}_{[E,\infty)} u.
	$$
	In this case we have by the spectral theorem and the
	properties of the projection-valued measures that
	$$
		\int_0^\infty \|u_{E,+}\|_2^2 \,dE
		= \int_0^\infty \inp{u, P^{-\Delta_Q^\eN}_{[E,\infty)} u} \,dE
		= \inp{u, -\Delta_Q^\eN u},
	$$
	since $P^A_{[E,\infty)} = \int_\R \1_{\{\lambda \ge E\}} dP^A(\lambda)$
	and thus $\int_0^\infty P^A_{[E,\infty)} dE = \int_\R \lambda \,dP^A(\lambda) = A$
	for any self-adjoint operator $A$
	(a reader worried about such formal manipulation may note that it is applied
	here in the form sense with non-negative integrands).
	
	Let us denote the eigenvalues and orthonormal eigenfunctions of $-\Delta_Q^\eN$,
	ordered according to their multiplicity,
	as usual by $\{\lambda_k\}_{k=0}^\infty$ and $\{u_k\}_{k=0}^\infty$
	For the low-energy part we have then for each $\bx \in Q$
	\begin{align*}
		|u_{E,-}(\bx)|^2 &= \left| P^{-\Delta_Q^\eN}_{[0,E)}u (\bx) \right|^2
		= \left| \sum_{\lambda_k < E} \langle u_k,u \rangle u_k (\bx) \right|^2
		= \left|\inp{ \sum_{\lambda_k < E} \overline{u_k(\bx)} u_k, u }\right|^2 \\
		&\le \left( \sum_{\lambda_k < E} |u_k(\bx)|^2 \right) \|u\|_2^2,
	\end{align*}
	by Cauchy--Schwarz and the orthonormality of $\{u_k\}$.
	Furthermore, by \eqref{eq:cube-bound-eigenvalues} and \eqref{eq:cube-bound-eigenfunctions}, 
	we have that
	$$
		\sum_{\lambda_k < E} |u_k(\bx)|^2 
		\le \frac{1}{|Q|} + \sum_{0 < \lambda_k < E} \frac{2^d}{|Q|}
		\le \frac{1}{|Q|} + \frac{2^{2d}}{\pi^d} E^{d/2}.
	$$
	
	After these preparations we may finally use the triangle inequality 
	\eqref{eq:GNS-triangle-ineq}
	and the integral identity \eqref{eq:GNS-integral-identity} to obtain
	\begin{align*}
		\int_Q |\nabla u|^2 &= \int_Q \int_0^\infty |u_{E,+}(\bx)|^2 \,dE \,d\bx \\
		&\ge \int_Q \int_0^\infty \left[
			|u(\bx)| - \left( |Q|^{-1} + 2^{2d} \pi^{-d} E^{d/2} \right)^{1/2} \|u\|_2
			\right]_+^{2} dE \,d\bx \\
		&\ge \int_Q \int_0^\infty \left[
			|u(\bx)| - |Q|^{-1/2}\|u\|_2 - 2^{d} \pi^{-d/2} \|u\|_2 E^{d/4}
			\right]_+^{2} dE \,d\bx \\
		&= C_d \|u\|_2^{-4/d} \int_Q \left[
			|u(\bx)| - |Q|^{-1/2}\|u\|_2
			\right]_+^{2(1+2/d)} d\bx,
	\end{align*}
	with $C_d = d^2(2^{-4}\pi^2)/((d+2)(d+4))$.
\end{proof}

\subsection{Local uncertainty and density formulations}\label{sec:uncert-local}

	We finish this chapter on uncertainty principles with some local many-body
	formulations involving the one-body density $\varrho_\Psi$.
	These will in later chapters be supplemented with local formulations of the 
	exclusion principle to prove powerful global kinetic energy inequalities 
	of wide applicability.

	The following is a local version of the many-body GNS inequality of
	Theorem~\ref{thm:GNS-many-body}
	(it was given in this form as Theorem~14 in \cite{LunSol-13a}):

\begin{theorem} 
	\label{thm:GNS-local-many-body}
	For any $d \ge 1$ there exists a constant $C_d > 0$ 
	(same as in Theorem~\ref{thm:GNS-local}) such that for 
	any $d$-cube $Q \subseteq \R^d$, all $N \ge 1$, 
	and $L^2$-normalized $\Psi \in H^1(\R^{dN})$
	\begin{equation}\label{eq:GNS-local-many-body}
		\sum_{j=1}^N \int_{\R^{dN}} |\nabla_j \Psi|^2 \1_Q(\bx_j) \,d\sx
		\ge C_d \left( \int_Q \varrho_\Psi \right)^{-2/d} 
			\int_Q \left[ \varrho_\Psi^{1/2} - \left(\frac{{\textstyle\int_Q} \varrho_\Psi}{|Q|}\right)^{1/2} \right]_+^{2(1+2/d)}.
	\end{equation}
\end{theorem}
	The proof is a straightforward modification of the one-body 
	Theorem~\ref{thm:GNS-local}, using
	$$
		\sum_{j=1}^N \int_{\R^{dN}} |\nabla_j \Psi|^2 \1_Q(\bx_j) \,d\sx
		= \sum_{j=1}^N \int_{\R^{d(N-1)}} \int_Q |\nabla_j \Psi|^2 \,d\bx_j \,d\sx'
		= \sum_{j=1}^N \int_{\R^{d(N-1)}} \int_0^\infty \norm{u^{E,+}_j}^2 dE
	$$
	as in Exercise~\ref{exc:GNS-many-body}, with 
	$\int_{\R^{d(N-1)}} \sum_{j=1}^N |u_j(\bx,\sx')|^2 \,d\sx' = \varrho_\Psi(\bx)$
	and
	$$
		\int_{\R^{d(N-1)}} \sum_{j=1}^N |u_j^{E,-}(\bx,\sx')|^2 \,d\sx'
		\le \left( \frac{1}{|Q|} + \frac{2^{2d}}{\pi^d} E^{d/2} \right)
			\int_{\R^{d(N-1)}} \sum_{j=1}^N \|u_j(\slot,\sx')\|_{L^2(Q)}^2 \,d\sx'.
	$$

	Now, let us write for the total expected kinetic energy 
	of an $N$-body wave function $\Psi \in H^1(\R^{dN})$
	$$
		T[\Psi] = \inp{\hT}_\Psi = \int_{\R^{dN}} |\nabla \Psi|^2.
	$$
	We also introduce the \keyword{local expected kinetic energy} on the cube 
	$Q \subseteq \R^d$ 
	(which may of course also be replaced by a general subdomain $\Omega$)
	$$
		T^Q[\Psi] := \sum_{j=1}^N \int_{\R^{dN}} |\nabla_j \Psi|^2 \,\1_Q(\bx_j) \,d\sx
		= \sum_{j=1}^N \norm{ \1_{\bx_j \in Q} \,\hbp_j \Psi}^2.
	$$
	Using the above inequality we may obtain a bound for this quantity 
	of the particularly convenient form
	$$
		T^Q[\Psi] \ge 
			C_1 \frac{\int_Q \varrho_\Psi^{1 + 2/d}}{(\int_Q \varrho_\Psi)^{2/d}}
			- C_2 \frac{\int_Q \varrho_\Psi}{|Q|^{2/d}},
	$$
	which we shall refer to as a \keyword{local uncertainty principle}.

\begin{lemma}[Local uncertainty principle]
	\label{lem:local-uncertainty}
	For any $d$-cube $Q$ and any $\eps \in (0,1)$ we have
	\begin{equation}\label{eq:local-uncertainty}
		T^Q[\Psi] \ \ge \ 
		C_d \eps^{1+4/d} \ \frac{\int_Q \varrho_\Psi^{1 + 2/d}}{(\int_Q \varrho_\Psi)^{2/d}}
			- C_d \left(1 + \left(\frac{\eps}{1-\eps}\right)^{1+4/d} \right)
				\frac{\int_Q \varrho_\Psi}{|Q|^{2/d}},
	\end{equation}
	with $C_d$ as in Theorems~\ref{thm:GNS-local} and \ref{thm:GNS-local-many-body}.
\end{lemma}
\begin{proof}
	By Theorem~\ref{thm:GNS-local-many-body} we have that
	$$
		T^Q[\Psi]
		\ \ge \ \frac{C_d}{(\int_Q \varrho_\Psi)^{2/d}} \int_Q \left[
			\varrho_\Psi(\bx)^{\frac{1}{2}} - \left( \frac{\int_Q \varrho_\Psi}{|Q|} \right)^{\frac{1}{2}}
			\right]_+^{2+4/d} d\bx,
	$$
	with 
	\begin{align*}
		\int_Q &\left[
			\varrho_\Psi(\bx)^{\frac{1}{2}} - \left( \frac{\int_Q \varrho_\Psi}{|Q|} \right)^{\frac{1}{2}}
			\right]_+^{2+4/d} d\bx \\
		&\ge \int_Q \left|
			\varrho_\Psi(\bx)^{\frac{1}{2}} - \left( \frac{\int_Q \varrho_\Psi}{|Q|} \right)^{\frac{1}{2}}
			\right|^{2+4/d} d\bx \ 
			- \int_Q \left( \frac{\int_Q \varrho_\Psi}{|Q|} \right)^{1+2/d} \\
		&= \left\| \varrho_\Psi^{\frac{1}{2}} - \left( \frac{\int_Q \varrho_\Psi}{|Q|} \right)^{\frac{1}{2}} \right\|^{2+4/d}_{2+4/d} 
			- \frac{(\int_Q \varrho_\Psi)^{1+2/d}}{|Q|^{2/d}}.
	\end{align*}
	The first term is bounded below by 
	$$
		\left( \| \varrho_\Psi^{\frac{1}{2}} \|_{2+4/d} 
		- \| ({\textstyle\int_Q} \varrho_\Psi / |Q|)^{\frac{1}{2}} \|_{2+4/d} \right)^{2+4/d}
	$$
	using the triangle inequality on $L^p(Q)$. 
	Furthermore, by convexity of the function $x \mapsto x^p$ for $p \ge 1$ 
	we have for any $a,b \in \R$ and
	$\eps \in (0,1)$ that
	$$
		\left( \eps a + (1-\eps) b \right)^p \le \eps a^p + (1-\eps) b^p,
	$$
	and hence with $a = A-B$ and $b = \frac{\eps}{1-\eps}B$,
	$$
		(A-B)^p \ge \eps^{p-1} A^p - \left( \frac{\eps}{1-\eps} \right)^{p-1} B^p.
	$$
	Applying this inequality to the norms
	above with $p=2+4/d$, we finally arrive at 
	\eqref{eq:local-uncertainty}.
\end{proof}

	Finally, given a partition $\eP$ of the one-particle configuration space $\R^d$
	into disjoint cubes,
	$$
		\R^d = \bigcup_{Q \in \eP} \bar{Q},
		\qquad Q \cap Q' = \varnothing 
		\ \ \forall Q,Q' \in \eP \ \text{s.t.} \ Q \neq Q',
	$$
	we have, with $\1 = \sum_{Q \in \eP} \1_Q$, that
	\begin{equation}\label{eq:local-uncertainty-partition}
		T[\Psi] = \sum_{Q \in \eP} T^Q[\Psi]
		\ge \sum_{Q \in \eP} \left( 
			C_1 \frac{\int_Q \varrho_\Psi^{1 + 2/d}}{(\int_Q \varrho_\Psi)^{2/d}}
			- C_2 \frac{\int_Q \varrho_\Psi}{|Q|^{2/d}}
			\right).
	\end{equation}
	This global 
	bound for the expected kinetic energy of $\Psi$ 
	can only be useful 
	if the positive terms 
	are stronger than the negative ones,
	i.e.\ if the expected number of particles
	$\sum_{j=1}^N \inp{\1_{\bx_j \in Q}}_\Psi = \int_Q \varrho_\Psi$
	on each cube $Q$ is not too large, 
	\emph{and} if the density is sufficiently localized on $Q$,
	$$
		\frac{1}{|Q|} \int_Q \varrho_\Psi^{1+2/d} 
		\gg \left( \frac{1}{|Q|} \int_Q \varrho_\Psi \right)^{1+2/d}.
	$$
	For the case of rather homogeneous density distributions
	the above inequality fails,
	and the local uncertainty principle \eqref{eq:local-uncertainty-partition}
	will therefore have to be supplemented with
	for example an exclusion principle in order to yield a non-trivial global bound,
	and this will be the topic of the next section.

\section{Exclusion principles\lect{ [11,12]}}\label{sec:exclusion}

	In this section we consider consequences of the theory for identical particles
	and exchange phases that was outlined in Section~\ref{sec:mech-QM-statistics}, 
	as well as related concepts, both locally and globally on the configuration space.
	
	Recall that in general we have a division of the full $N$-particle Hilbert space
	$\cH \cong \bigotimes^N \gH$ of distinguishable particles
	into subspaces of symmetric (\keyword[bosonic state]{bosonic}) 
	respectively antisymmetric (\keyword[fermionic state]{fermionic}) states
	of \emph{in}distinguishable particles,
	$$
		\cH_\sym \cong \bigotimes\nolimits_\sym^N \gH, \qquad
		\cH_\asym \cong \bigwedge\nolimits^N \gH,
	$$
	where $\gH$ denotes the one-particle Hilbert space.
	Let us consider here for illustration 
	the usual space $\gH = L^2(\R^d)$ of a particle in $\R^d$, 
	but other spaces will be of interest as well.
	When acting with the non-interacting 
	$N$-body Hamiltonian operator
	(as given e.g. in \eqref{eq:many-body-Hamiltonian-independent})
	$$
		\hH = \sum_{j=1}^N \hh_j, \qquad 
		\hh_j = 
		\hh(\hbx_j,\hbp_j) = -\Delta_{\bx_j} + V(\bx_j) \ \in \cL(\gH),
	$$ 
	on $\cH_\sym$ respectively $\cH_\asym$,
	we may observe a crucial difference in the resulting spectra.
	Namely,
	let us assume for simplicity that we can diagonalize the one-body operator
	$\hh$ into a complete discrete set of 
	eigenvalues 
	$\sigma(\hh) = \{\lambda_n\}_{n=0}^\infty \subset \R$ 
	and a corresponding basis of orthonormal one-body eigenstates
	$\{u_n\}_{n=0}^\infty \subset \gH$,
	with an ordering $\lambda_0 \le \lambda_1 \le \ldots$ 
	according to multiplicity.
	Then the corresponding $N$-body eigenstates 
	$\Psi \in \cH$ of $\hH$ are simply
	\begin{equation}\label{eq:many-body-basis}
		\Psi = \Psi_{\{n_j\}} 
		:= u_{n_1} \otimes u_{n_2} \otimes \ldots \otimes u_{n_N},
		\qquad n_j \in \N_0,
	\end{equation}
	with
	$$
		\hH \Psi_{\{n_j\}} 
		= \sum_{j=1}^N u_{n_1} \otimes \ldots \otimes \hat{h} u_{n_j} \otimes \ldots \otimes u_{n_N}
		= \sum_{j=1}^N \lambda_{n_j} \Psi_{\{n_j\}}.
	$$
	In other words, 
	the $N$-body energy eigenvalues in the case of distinguishable particles are
	$$
		E_{\{n_j\}} = \sum_{j=1}^N \lambda_{n_j},
		\qquad \{n_j\} \in \N_0^N.
	$$
	However, with the symmetry restriction in $\cH_\asym$,
	the basis \eqref{eq:many-body-basis} reduces
	to the antisymmetric product (also known as a \keyword{Slater determinant})
	$$
		\Psi_\asym = u_{n_1} \wedge u_{n_2} \wedge \ldots \wedge u_{n_N} 
		:= \frac{1}{\sqrt{N!}}\sum_{\sigma \in S_N} \sign(\sigma) \,
			u_{\sigma(n_1)} \otimes u_{\sigma(n_2)}
			 \otimes \ldots \otimes u_{\sigma(n_N)},
	$$
	with $n_1 < n_2 < \ldots < n_N$.
	Because of the antisymmetry we cannot use the same one-body state $u_n$ 
	more than once in the expression, namely
	$u_n \wedge u_n = 0$, and this symmetry restriction is known in physics as the 
	\keyword{Pauli principle}\index{exclusion principle} 
	and is thus obeyed by fermions, 
	such as the electrons of an atom \cite{Pauli-47}.
	
	On the other hand, in a basis of the bosonic space
	$\cH_\sym$ we must take symmetric tensor products, 
	and for example the state
	$$
		\Psi_\sym = \otimes^N u_0 := u_0 \otimes u_0 \otimes \ldots \otimes u_0 
		\quad \text{($N$ factors)}
	$$
	with $u_0$ corresponding to the lowest eigenvalue $\lambda_0$ of $\hat{h}$,
	is allowed and will in fact be the ground state of $\hH$ 
	(both when considered as an operator on $\cH$ and on $\cH_\sym$, 
	but \emph{not} on $\cH_\asym \reflectbox{$\notin$} \Psi_\sym$).
	Namely, note that the 
	energy eigenvalue of this state $\Psi_\sym$ is 
	$$
		E_{0,\sym}(N) = N \lambda_0 \ \le \ \sum_{j=1}^N \lambda_{n_j} = E_{\{n_j\}},
	$$
	for any multi-index $\{n_j\}$, and hence this is the ground-state energy.
	In contrast,
	on states $\Psi_\asym$ we necessarily obtain a sum of higher and higher energies,
	and 
	the smallest possible value is
	\begin{equation}\label{eq:fermionic-energy-bound}
		E_{0,\asym}(N) = \sum_{k=0}^{N-1} \lambda_k \ \le \ \sum_{j=1}^N \lambda_{n_j} = E_{\{n_j\}},
	\end{equation}
	for any admissible $\{n_j\}$, i.e. $n_1 < n_2 < \ldots < n_N$.
	Because of this symmetry restriction the fermionic g.s.\ energy must
	(unless the lowest eigenvalue $\lambda_0$ is infinitely degenerate)
	be strictly larger than the bosonic (or the distinguishable) one for 
	large enough $N$,
	and in realistic systems with finite degeneracies it will actually be
	\emph{significantly} larger as $N \to \infty$.
	
	These differences in the rules for distributing bosons and fermions
	into one-body states, enforced by the Pauli principle, has remarkable
	macroscopic consequences when one considers large ensembles of particles,
	which is the aim 
	of \keyword[quantum statistics]{quantum statistical mechanics}.
	Bosons are then said to obey \keyword{Bose--Einstein statistics}
	while fermions are subject to \keyword{Fermi--Dirac statistics}.
	Distinguishable particles on the other hand obey \keyword{Maxwell--Boltzmann statistics} 
	and are sometimes called \keyword{boltzons}.
	
	We shall in this chapter also consider 
	exclusion principles in a more general context.
	In particular, we allow for a weakening of the above Pauli principle
	--- which is actually relevant for real fermions appearing in nature 
	such as electrons with spin ---
	and we also extend the notion of exclusion to encompass other important 
	cases with similar features such as bosons with repulsive pair interactions, 
	as well as anyons in two dimensions.
	The generalization of quantum statistics to allow for several particles in each
	one-body state (beyond spin) has a 
	long history, going back at least to
	Gentile \cite{Gentile-40,Gentile-42} and is thus known as 
	\keyword{Gentile statistics} 
	or \keyword{intermediate (exclusion) statistics}.
	More recent, further generalizations of such concepts have been reviewed in 
	\cite{Haldane-91,Isakov-94,Wu-94,Myrheim-99,Polychronakos-99,Khare-05}.

\begin{exc}\label{exc:fermions-harm-osc}
	Compute the ground-state energy $E_{0,\asym}$ for $N$ fermions in a harmonic trap 
	$V_{\textup{osc}}(\bx) = \frac{1}{2}m\omega^2|\bx|^2$ 
	in $d=1$ and in $d=2$ (see Example~\ref{exmp:harm-osc-QM}).
	Note that in the latter case there are certain ``magic numbers''\index{magic number}
	$N = n(n+1)/2$, $n \in \N$, such that
	\begin{equation}\label{eq:fermion-gs-harm-osc}
		E_{0,\asym}(N) 
		= (1 + 2^2 + \ldots + n^2) \hbar\omega
		= \frac{1}{3}N\sqrt{8N+1} \hbar\omega.
	\end{equation}
\end{exc}

\begin{exc}\label{exc:Weyl-cube}
	Consider $N$ fermions in a cube $Q \subseteq \R^d$.
	Use the values in Section~\ref{sec:uncert-local-cube}
	and an integral approximation of a Riemann sum 
	to show that, as a leading-order approximation,
	\begin{equation}\label{eq:Weyl-asymptotics}
		E_{0,\asym}(N) = 
		\sum_{k=0}^{N-1} \lambda_k(-\Delta_Q^{\eN/\eD}) 
		\approx K_d^\cl \frac{N^{1+2/d}}{|Q|^{2/d}}, 
		\qquad K_d^\cl := 4\pi \frac{d}{d+2}\left( \frac{2}{d+2} \right)^{\frac{2}{d}} 
			\Gamma\left(2+\frac{d}{2}\right)^{\frac{2}{d}}.
	\end{equation}
	This is known as \keyword{Weyl's asymptotic 
	formula} for the sum of eigenvalues,
	and it thus determines the g.s.\ energy of a \keyword{free Fermi gas} 
	confined to a box, with the constants
	\begin{equation}\label{eq:Kcl-low-dims}
		K_1^\cl = \frac{\pi^2}{3}, \qquad
		K_2^\cl = 2\pi, \qquad
		K_3^\cl = \frac{3}{5}(6\pi^2)^{\frac{2}{3}}.
	\end{equation}
\end{exc}

\begin{exc}\label{exc:fermions-localization}
	Compute an upper bound to $E_{0,\asym}(N)$ of \eqref{eq:Weyl-asymptotics} in $d=1,2,3$
	by constructing a trial state $\Psi_\asym$ by localizing each particle
	to a separate cube,
	e.g.\ taking $u_n \in \gH$ to be Dirichlet ground states on the respective cubes,
	and then antisymmetrizing the full expression.
	Does one gain anything by localizing on balls instead? \\
	Hint: the optimal packing density of circles is $\pi/(2\sqrt{3})$ and of spheres is $\pi/(3\sqrt{2})$.
\end{exc}

\subsection{Fermions}\label{sec:exclusion-fermions}
	\index{fermions}

	Let us first consider some local consequences of the Pauli principle 
	for fermions that will turn out to be particularly usful 
	in our context of stability.
	We denote $H^1_\asym = H^1 \cap \cH_\asym$, and so on.

\subsubsection{The Pauli principle}\label{sec:exclusion-fermions-Pauli}
	
	One has the following simple consequence of the Pauli principle for 
	fermions on a cube $Q$:

\begin{proposition}\label{prop:DL-bound}
	Let $Q \subseteq \R^d$ be a $d$-cube, $d \ge 1$.
	For any $N \ge 1$ and $\Psi \in H^1_\asym(Q^N)$, we have
	\begin{equation}\label{eq:DL-bound-cube}
		\int_{Q^N} |\nabla \Psi|^2 \ge (N-1) \frac{\pi^2}{|Q|^{2/d}} 
			\int_{Q^N} |\Psi|^2.
	\end{equation}
\end{proposition}
	
\begin{proof}
	Using the same argument of simple eigenvalue estimation as in~\eqref{eq:fermionic-energy-bound}, 
	but now locally with the one-particle space $\gH = L^2(Q)$ and
	$0 \neq \Psi \in \bigwedge^N \gH$,
	we find in this case
	\begin{equation}\label{eq:DL-bound-proof}
		\frac{\int_{Q^N} |\nabla \Psi|^2}{\int_{Q^N} |\Psi|^2} 
		\ge E_0\left( -\Delta^\eN_{Q^N} \Big|_{\bigwedge^N \gH} \right)
		= \sum_{k=0}^{N-1} \lambda_k(-\Delta^\eN_Q).
	\end{equation}
	Now recall the energy levels $\lambda_k$ of the Neumann Laplacian $-\Delta^\eN_Q$
	given in Section~\ref{sec:uncert-local-cube}:
	$$
		\lambda_0 = 0  \ \ \text{(here $\bk = \0$)}, \qquad 
		\lambda_{k \ge 1} \ge \lambda_1 = \frac{\pi^2}{|Q|^{2/d}} \ \ 
			\text{(here $|\bk| \ge 1$)},
	$$
	which proves the proposition.
\end{proof}

	The above bound gives a concrete local measure of the exclusion principle, 
	as it tells us that, 
	while a single particle may have zero energy on a finite domain 
	(as is possible here due to the Neumann b.c.), 
	as soon as there are two or more particles on the same domain, 
	the energy must be strictly positive.
	However, if the particles were bosons or distinguishable 
	then one could have chosen 
	all particles to be in the ground state,
	i.e. the constant function, 
	and thus obtained zero energy.
	The domain considered above was a cube but similar bounds are naturally valid 
	on other domains as well, and in fact the
	corresponding bound on a ball (see \eqref{eq:DL-bound-ball} below)
	was used by Dyson and Lenard in their original 
	proof of stability of fermionic matter \cite[Lemma~5]{DysLen-67}.
	Indeed the only property of fermionic systems
	(compared to 
	bosonic or boltzonic) 
	used in their proof was
	this remarkably weak implication of the exclusion principle, with an energy
	which only grows \emph{linearly} with $N$, as opposed to the true energy which 
	grows much faster with $N$ according to the Weyl asymptotics as we saw above
	(Exercise~\ref{exc:Weyl-cube}). 
	This curious fact, 
	that matter turns out to be stabilized sufficiently by the exclusion principle
	acting effectively between pairs and triplets of neighboring electrons,
	was discussed briefly in \cite{Dyson-68} and \cite{Lenard-73},
	and will be further clarified in the coming sections.
	
	One may alternatively prove \eqref{eq:DL-bound-cube} resp. \eqref{eq:DL-bound-ball} 
	directly via the Poincar\'e inequality on the domain; 
	see \cite[Theorem~8]{Lenard-73}.
	See also \cite[Lemma~11]{LunNamPor-16} for an extension of this type of exclusion 
	bound to kinetic energy operators involving arbitrary powers of the momentum
	(even fractional, which is relevant for relativistic stability).

\begin{exc}
	Prove a corresponding Pauli bound on a disk or a ball $B = B_R(0)$ of radius $R$,
	$\Psi \in H^1_\asym(B^N)$,
	\begin{equation}\label{eq:DL-bound-ball}
		\int_{B^N} |\nabla \Psi|^2 \ge (N-1) \frac{\xi^2}{R^2} 
			\int_{B^N} |\Psi|^2,
	\end{equation}
	where for the two-dimensional disk, $\xi \approx 1.841$ denotes the first
	non-trivial zero of the derivative of the Bessel function $J_1$, 
	and for the three-dimensional ball, $\xi \approx 2.082$ denotes the
	smallest positive root $x$ of the equation
	$\frac{d^2}{dx^2} \frac{\sin x}{x} = 0$.
\end{exc}

\subsubsection{Fermionic uncertainty and statistical repulsion}\label{sec:exclusion-fermions-Hardy}
	\index{fermionic uncertainty}\index{statistical repulsion}

	Another useful and concrete 
	measure of the exclusion principle betweeen
	fermions comes 
	in the form of a strengthened uncertainty principle.
	Namely, as we shall see below,
	fermions turn out to satisfy an effective pairwise repulsion,
	and the corresponding 
	mathematical statement may be referred to as a 
	\keyword{fermionic many-body Hardy inequality}.
	Inequalities of this form were introduced in \cite{HofLapTid-08} 
	in the global case,
	and were subsequently generalized to anyons in two dimensions and to 
	local formulations in 
	\cite{LunSol-13a,LarLun-16}.
	The optimal constant in the inequality for fermions in the global case for 
	$d \ge 3$ was also discussed in \cite{FraHofLapSol-06}.
	We will here only discuss the simpler global case, 
	although it is important to stress
	that the repulsion persists also locally, which is not the case with the usual
	Hardy inequality without antisymmetry.
	
	Let us start with the following simple one-body version of the inequality
	(which has been pointed out already in \cite{Birman-61}):
	
\begin{lemma}[\keyword{One-body Hardy with antisymmetry}]\label{lem:fermionic-Hardy-1p}
	If $u \in H^1(\R^d)$ is antipodal-antisym\-metric, i.e. $u(-\bx) = -u(\bx)$, then
	\begin{equation} \label{eq:fermionic-Hardy-1p}
		\int_{\R^d} |\nabla u(\bx)|^2 \,d\bx
		\ge \frac{d^2}{4} \int_{\R^d} \frac{|u(\bx)|^2}{|\bx|^2} \,d\bx.
	\end{equation}
\end{lemma}
\begin{remark}
	Note that this improves upon the constant of the standard Hardy inequality
	of Theorem~\ref{thm:Hardy} for $d \ge 2$, 
	and that for $d=1$ antisymmetry and continuity implies $u(0) = 0$ and hence
	$u \in H^1_0(\R \setminus \{0\})$, reducing therefore to the standard inequality.
\end{remark}
\begin{proof}
	We may write in terms of spherical coordinates $\bx = \bx(r,\omega)$,
	$r \ge 0$, $\omega \in \S^{d-1}$,
	\begin{equation} \label{eq:kinetic-u-spherical}
		\int_{\R^d} |\nabla u(\bx)|^2 \,d\bx
		= \int_{r=0}^\infty \int_{\S^{d-1}} \left(
			|\partial_r u|^2 + \frac{|\nabla_\omega u|^2}{r^2}
			\right) r^{d-1} dr d\omega
	\end{equation}
	and change the order of integration by Fubini.
	The usual Hardy inequality \eqref{eq:Hardy}
	actually concerns the radial part here (exercise),
	namely for any $\omega \in \S^{d-1}$,
	\begin{equation} \label{eq:Hardy-radial}
		\int_0^\infty |\partial_r u(r,\omega)|^2 \,r^{d-1} dr 
		\ge \frac{(d-2)^2}{4} \int_0^\infty \frac{|u(r,\omega)|^2}{r^2} \,r^{d-1} dr.
	\end{equation}
	For the angular part of the derivative we may use the Poincar\'e inequality 
	on $\S^{d-1}$ given in \eqref{eq:Poincare-sphere},
	\begin{equation} \label{eq:Poincare-antisymmetry}
		\int_{\S^{d-1}} |\nabla_\omega u(r,\omega)|^2 \,d\omega 
		\ge (d-1) \int_{\S^{d-1}} |u(r,\omega)|^2 \,d\omega,
	\end{equation}
	since, by antipodal antisymmetry, 
	$\int_{\S^{d-1}} u(r,\omega) \,d\omega = 0$ for all $r \ge 0$.
	The lemma then follows by combining these two inequalities.
\end{proof}

	By applying this lemma in pairwise relative coordinates, one obtains the
	following manybody Hardy inequality
	\cite[Theorem~2.8]{HofLapTid-08}: 

\begin{theorem}[\keyword{Many-body Hardy with antisymmetry}]\label{thm:fermionic-Hardy}
	If $\Psi \in H^1_\asym(\R^{dN})$ then
	\begin{equation} \label{eq:fermionic-Hardy}
		\int_{\R^{dN}} \sum_{j=1}^N |\nabla_j \Psi|^2 \,d\sx 
		\ \ge \ \frac{d^2}{N} \int_{\R^{dN}} \sum_{1 \le j<k \le N} 
			\frac{|\Psi(\sx)|^2}{|\bx_j - \bx_k|^2} \,d\sx
		+ \frac{1}{N} \int_{\R^{dN}} \left| \sum_{j=1}^N \nabla_j \Psi \right|^2 d\sx.
	\end{equation}
\end{theorem}
\begin{proof}
	We first use the many-body parallelogram identity \eqref{eq:many-body-parallelogram}
	$$
		\sum_{j=1}^N |\nabla_j\Psi|^2 
		= \frac{1}{N} \sum_{1 \le j<k \le N} \bigl| \nabla_j\Psi - \nabla_k\Psi \bigr|^2
			+ \frac{1}{N} \Biggl| \sum_{j=1}^N \nabla_j\Psi \Biggr|^2,
	$$
	and then for each pair $(j,k)$ of particles we introduce relative coordinates,
	$$
		\br_{jk} := (\bx_j - \bx_k)/2, \qquad
		\bX_{jk} := (\bx_j + \bx_k)/2, \qquad
		\nabla_{\br_{jk}} = \nabla_j - \nabla_k, \qquad
		\nabla_{\bX_{jk}} = \nabla_j + \nabla_k.
	$$
	Thus, splitting the coordinates according to $\sx = (\bx_j,\bx_k;\sx')$, 
	with $\sx' \in \R^{d(N-2)}$ and $j<k$ fixed,
	we may define the relative function
	$$
		u(\br,\bX) := \Psi(\bx_1,\bx_2,\ldots, \bx_j = \bX+\br, \ldots, \bx_k = \bX-\br, \ldots, \bx_N),
	$$
	for which we have that $u(-\br,\bX) = -u(\br,\bX)$ for any $\bX \in \R^d$
	by the antisymmetry of $\Psi$.
	Then, by the 1-to-1 change of variables, with $d\bx_j d\bx_k = 2^d d\br d\bX$, 
	and Lemma~\ref{lem:fermionic-Hardy-1p},
	\begin{align*}
		\int_{\R^{d(N-2)}} &\int_{\R^d \times \R^d} |(\nabla_j-\nabla_k)\Psi|^2 
			\,d\bx_j d\bx_k \,d\sx'
		= \int_{\R^{d(N-2)}} \int_{\R^d \times \R^d} |\nabla_{\br} u|^2 
			\,2^d d\br d\bX \,d\sx' \\
		&\ge \frac{d^2}{4} \int_{\R^{d(N-2)}} \int_{\R^d \times \R^d} \frac{|u|^2}{|\br|^2}
			\,2^d d\br d\bX \,d\sx'
		= d^2 \int_{\R^{d(N-2)}} \int_{\R^d \times \R^d} \frac{|\Psi|^2}{|\bx_j-\bx_k|^2}
			\,d\bx_j d\bx_k \,d\sx',
	\end{align*}
	which proves the theorem.
\end{proof}

	The above theorem shows that fermionic particles always feel an 
	inverse-square pairwise repulsion,
	which is not just due to the energy cost of localization as encoded in the 
	usual uncertainty principle,
	but which is \emph{strictly stronger} 
	(and in two dimensions therefore non-trivial in contrast 
	to the standard Hardy inequality).
	Its origin is the relative antipodal antisymmetry 
	and thereby the Poincar\'e inequality 
	\eqref{eq:Poincare-antisymmetry}
	which comes weighted by the inverse-square of the distance 
	$r_{jk} = |\br_{jk}|$ between each pair of particles. 
	The repulsion persists also locally, 
	i.e. in tubular domains around the diagonals $\bDelta$ 
	of the configuration space
	and independently of the considered boundary conditions.
	Indeed, an alternative version of \eqref{eq:fermionic-Hardy}, 
	without having applied the $\R^d$-global Hardy inequality in
	\eqref{eq:kinetic-u-spherical}
	but instead the local (diamagnetic) inequality\index{diamagnetic inequality}
	$|\partial_r \Psi| \ge \bigl|\partial_r|\Psi|\bigr|$, 
	is
	\begin{align} \label{eq:fermionic-Hardy-split}
		\int_{\R^{dN}} \sum_{j=1}^N |\nabla_j \Psi|^2 \,d\sx 
		\ \ge &\ \frac{1}{N} \int_{\R^{dN}} \sum_{1 \le j<k \le N} 
			\left( \big|\partial_{r_{jk}}|\Psi|\big|^2
				+ 4(d-1) \frac{|\Psi(\sx)|^2}{|\bx_j - \bx_k|^2} \right)d\sx \\
		&+ \frac{1}{N} \int_{\R^{dN}} \left| \sum_{j=1}^N \nabla_j \Psi \right|^2 d\sx.
	\end{align}
	Although the first integral term of the
	r.h.s. involves the \emph{bosonic} function $|\Psi| \in L^2_\sym(\R^{dN})$,
	the singular repulsive term forces the probability amplitude $|\Psi|^2$ to 
	be smaller (or even to vanish for $d=2$) where particles meet,
	and we will therefore refer to this effect as a \keyword{statistical repulsion}%
	\footnote{Though the term should perhaps be used with some caution \cite{MulBla-03}.}.
	One may note however that the constant in \eqref{eq:fermionic-Hardy} resp. \eqref{eq:fermionic-Hardy-split}
	and the number of terms in the sum combine to yield an overall \emph{linear} 
	growth in $N$, 
	which matches the number of terms of the kinetic energy, 
	but differs from the case of a usual pair interaction of fixed strength,
	such as in \eqref{eq:many-body-Hamiltonian-ident}.
	Also, the last integral term in \eqref{eq:fermionic-Hardy} resp. \eqref{eq:fermionic-Hardy-split}
	involves the total center-of-mass motion and will typically only contribute
	to a lower order and may thus be discarded 
	in many applications.
	
	Local, but necessarily more complicated, 
	versions of these inequalities were given for $d=2$ 
	in \cite{LunSol-13a,LarLun-16}.
	These may be applied to eventually 
	give rise to the same type of exclusion bounds as 
	\eqref{eq:DL-bound-cube} and \eqref{eq:DL-bound-ball}, 
	with a slightly weaker constant, but 
	with the advantage of opening up for generalizations of the statistical repulsion
	such as to anyons as discussed in Section~\ref{sec:exclusion-anyons} below.
	Also note that for dimension $d=1$ 
	where $H^1_\asym(\R^N) \subseteq H^1_0(\R^N \setminus \bDelta)$
	we actually have a much better inequality from 
	before, namely the many-body Hardy inequality of Theorem~\ref{thm:many-body-Hardy-1d},
	for which the overall dependence of the r.h.s.\ is rather \emph{quadratic} in $N$.
	Furthermore, in the one-dimensional case with bosons, a vanishing condition 
	on the diagonals $\bDelta$ is actually sufficient
	to impose the usual Pauli principle, 
	since there is an equivalence between symmetric functions in 
	$H^1_0(\R^N \setminus \bDelta)$
	and antisymmetric functions in $H^1(\R^N)$ \cite{Girardeau-60},
	however this is not the case in $d \ge 2$.

\begin{exc}
	Prove the radial Hardy inequality \eqref{eq:Hardy-radial} in two ways:
	by reducing the usual Hardy inequality to radial functions, 
	and, by modifying the GSR approach of Proposition~\ref{prop:GSR}.
\end{exc}

\subsection{Weaker exclusion}\label{sec:exclusion-weaker}

	In the case that one would need to weaken the Pauli principle a bit to allow 
	for $q$ particles in each one-body state,
	this could be modeled using a modified $N$-particle Hilbert space
	$$
		\cH_{\asym(q)} := \bigotimes\nolimits^q \left( \bigwedge\nolimits^K \gH \right),
		\qquad N = qK,
	$$
	which may be thought of as having $q$ different (distinguishable) 
	flavors (or species) of fermions,
	with each such flavor being subject to the Pauli principle.
	In the context of the non-interacting Hamiltonian 
	$\hH = \sum_j \hh_j$ 
	discussed in the beginning of this chapter, 
	the typical ground state would then be
	$$
		\Psi = \left( u_0 \wedge \ldots \wedge u_K \right) \otimes \ldots \otimes \left( u_0 \wedge \ldots \wedge u_K \right)
	$$
	with energy $E_0 = q\sum_{k=0}^{K-1} \lambda_k$.
	We consider here $N$ to be a multiple of $q$ for simplicity, 
	but this assumption may be relaxed with 
	a slightly more 
	involved 
	framework to specify 
	which flavors are being added.
	An equivalent and more flexible way to characterize the elements of $\cH_{\asym(q)}$
	is as functions $\Psi(\bx_1,\ldots,\bx_N)$ in $L^2(\R^{dN})$
	for which there is a partition of the variables $\bx_j$ into $q$ groups,
	with the variables in each such group being antisymmetric under 
	permutations of the labels.
	
	Because there is also a notion of \emph{spin} in quantum mechanics 
	(recall Exercise~\ref{exc:spin} and Remark~\ref{rem:spin})
	it is in fact a realistic assumption that each particle comes equipped 
	with such an additional flavor degree of freedom,
	modeled using an \keyword{internal space} $\C^q$, where $q \ge 1$ is the 
	\keyword{spin dimension} of the particle
	(while the number $s = (d-1)/2 \in \Z_{\ge 0}/2$ is called its \keyword{spin}),
	and the full Hilbert space may then be modeled correctly using $\cH_{\asym(q)}$.
	Although the introduction of spin appears a bit artificial 
	here in our context of non-relativistic quantum mechanics,
	it turns out to be a non-trivial consequence of 
	\emph{relativistic} quantum theory that in nature,
	i.e.\ for elementary particles moving in $\R^3$,
	bosons always have $q$ odd (integer spin) while fermions have $q$ even
	(half-integer spin).
	In the case of electrons we have $q=2$, 
	with a basis of $\C^2$ modeling spin-up respectively spin-down states.
	The fermions that were considered previously and modeled by 
	the simpler antisymmetric space $\cH_\asym$ 
	with $q=1$ are therefore known as \keyword{spinless fermions}
	and may seem a bit artificial, 
	however such particles may become manifest in certain spin-polarized systems.
	
	A simple modification of the proof of Proposition~\ref{prop:DL-bound} 
	to the case $\cH_{\asym(q)}$ yields:

\begin{proposition}\label{prop:DL-bound-q}
	Let $Q \subseteq \R^d$ be a $d$-cube, $d \ge 1$.
	For any $N = qK$ with $q,K \ge 1$, and $\Psi \in H^1_{\asym(q)}(Q^N)$, we have
	\begin{equation}\label{eq:DL-bound-cube-q}
		\int_{Q^N} |\nabla \Psi|^2 \ge (N-q)_+ \frac{\pi^2}{|Q|^{2/d}} 
			\int_{Q^N} |\Psi|^2.
	\end{equation}
\end{proposition}

	Again, $N$ being a multiple of $q$ is not important for the proof and may thus
	be relaxed, as illustrated by the following version
	(see also \cite[Lemma~3]{FraSei-12} for generalizations):

\begin{proposition}\label{prop:DL-bound-q-subset}
	Let $\Psi \in H^1_{\asym(q)}(\R^{dN})$ be an $N$-body wave function with
	$N = qK$, $q,K \ge 1$,
	and let $Q \subseteq \R^d$ be a $d$-cube, $d \ge 1$.
	For any subset $A \subseteq \{1,\ldots,N\}$ of the particles, 
	$\sx = (\sx_A;\sx_{A^c})$, $n=|A| \ge 0$, we have
	\begin{equation}\label{eq:DL-bound-cube-q-subset}
		\int_{Q^n} \sum_{j \in A} |\nabla_j \Psi(\sx_A;\sx_{A^c})|^2 \,d\sx_A
		\ge (n-q)_+ \frac{\pi^2}{|Q|^{2/d}} 
			\int_{Q^n} |\Psi(\sx_A;\sx_{A^c})|^2 \,d\sx_A.
	\end{equation}
\end{proposition}

\begin{exc}\label{exc:Weyl-cube-weak}
	The expression \eqref{eq:Weyl-asymptotics} approximates 
	the ground-state energy for the \keyword{free Fermi gas} 
	(non-interacting and homogeneous) on a box $Q$, 
	i.e. the infimum 
	$$
		E_{0,\asym}(N) := \inf \left\{ T[\Psi] : \Psi \in H^1_{\asym}(Q^N),\ \|\Psi\|_2 = 1 \right\}
		\approx K_d^\cl \rho^{2/d} N,
	$$
	where $\rho = N/|Q|$ is the density.
	Obtain the approximation for 
	the g.s. energy for the 
	Fermi gas with a fixed number $q$ species of fermions, or $q$-dimensional spin
	\begin{equation}\label{eq:q-fermion-gs-energy}
		E_{0,\asym(q)}(N) := \inf \left\{ T[\Psi] : \Psi \in H^1_{\asym(q)}(Q^N),\ \|\Psi\|_2 = 1 \right\}
		\approx q^{-2/d} K_d^\cl \rho^{2/d} N.
	\end{equation}
	Show also that boundary conditions on $\partial Q$ are unimportant for the leading-order
	approximation of the energy.
\end{exc}

\subsection{Local exclusion and density formulations}\label{sec:exclusion-local}

	We may define the \keyword{local $n$-particle kinetic energy} 
	on cubes $Q$ for fermions,
	i.e. particles obeying the Pauli principle,
	$$
		e_n(|Q|;\asym) 
		:= 
			\inf_{\substack{\psi \in H^1_\asym(Q^n) \\ \int_{Q^n} |\psi|^2 = 1}}
			\int_{Q^n} \sum_{j=1}^n |\nabla_j\psi|^2,
	$$
	and in general, 
	allowing for $q$ particles to occupy the same state
	and taking $n$ an arbitrary subset of $N$ total particles on $\R^d$,
	$$
		e_n(|Q|;\asym(q)) 
		:= 
			\inf_{\substack{\Psi \in H^1_{\asym(q)}(\R^{dN}) \\ \int_{Q^n} |\Psi(\slot;\sx')|^2 = 1 \\ \sx' \in \R^{d(N-n)}}}
			\int_{Q^n} \sum_{j=1}^n |\nabla_j\Psi(\slot;\sx')|^2.
	$$
	We found from Proposition~\ref{prop:DL-bound} respectively \ref{prop:DL-bound-q-subset} 
	that for these energies
	\begin{equation}\label{eq:local-energy-bound}
		e_n(|Q|;\asym(q)) \ge \frac{\pi^2}{|Q|^{2/d}} (n-q)_+.
	\end{equation}
	
	Now, consider an $N$-body wave function of particles on $\R^d$,
	$\Psi \in H^1_{\asym(q)}(\R^{dN})$, 
	and a partition $\eP$ of the configuration space into $d$-cubes $Q$ with a
	corresponding kinetic energy
	$$
		\inp{\Psi, \hT \Psi} = T[\Psi] = \sum_{Q \in \eP} T^Q[\Psi],
	$$
	where we recall the local expected kinetic energy 
	(see Section~\ref{sec:uncert-local})
	\begin{equation}\label{eq:local-exp-kin-en-excl}
		T^Q[\Psi] := \sum_{j=1}^N \int_{\R^{dN}} |\nabla_j \Psi|^2 \,\1_Q(\bx_j) \,d\sx.
	\end{equation}

\begin{lemma}[\keyword{Local exclusion principle}]
	\label{lem:local-exclusion}
	For any $d$-cube $Q$ and $N$-body state $\Psi \in H^1_{\asym(q)}(\R^{dN})$ 
	we have
	\begin{equation}\label{eq:local-exclusion}
		T^Q[\Psi] \ \ge \ 
		\frac{\pi^2}{|Q|^{2/d}} \left( \int_Q \varrho_\Psi(\bx) \,d\bx \ - q \right)_+.
	\end{equation}
\end{lemma}
\begin{proof}
	We insert the partition of unity \eqref{eq:particle-partition-of-unity}
	into the definition \eqref{eq:local-exp-kin-en-excl},
	producing
	\begin{align*}
		T^Q[\Psi] &= \sum_{A \subseteq \{1,\ldots,N\}} 
			\int_{(Q^c)^{N-|A|}} \int_{Q^{|A|}} \sum_{j \in A} 
			|\nabla_j \Psi|^2 \prod_{k \in A} d\bx_k \prod_{k \notin A} d\bx_k \\
		&\ge \sum_{A \subseteq \{1,\ldots,N\}} 
			\int_{(Q^c)^{N-|A|}} e_{|A|}(|Q|;\asym(q)) \int_{Q^{|A|}} 
			|\Psi|^2 \prod_{k \in A} d\bx_k \prod_{k \notin A} d\bx_k \\
		&= \sum_{n=0}^N e_n(|Q|;\asym(q)) \,p_{n,Q}[\Psi].
	\end{align*}
	Now, we apply the bound \eqref{eq:local-energy-bound} and use convexity
	of the function $x \mapsto (x-q)_+$,
	\begin{align*}
		T^Q[\Psi] &\ge \sum_{n=0}^N \frac{\pi^2}{|Q|^{2/d}} (n-q)_+ \,p_{n,Q}[\Psi]
		\ge \frac{\pi^2}{|Q|^{2/d}} \left( \sum_{n=0}^N np_{n,Q}[\Psi] - q \right)_+,
	\end{align*}
	which by \eqref{eq:particle-prob-expectation} 
	is exactly the r.h.s. of \eqref{eq:local-exclusion}.
\end{proof}

\begin{example}[{\keyword[free Fermi gas]{The free Fermi gas}}]\label{exmp:exclusion-gas}
	As a simple application of this local bound we may prove a global lower bound 
	for the ground-state energy of the ideal Fermi gas which matches the approximation
	\eqref{eq:q-fermion-gs-energy} 
	apart from the explicit value of the constant,
	i.e.\ proving that the fermionic energy is extensive in the number of particles.
	Namely, consider an arbitrary wave function of $N$ fermions 
	confined to a large cube $Q_0 = [0,L]^d$,
	$$
		\Psi \in H^1_{\asym(q)}(Q_0^N) 
		\quad \Rightarrow \quad \varrho_\Psi \in L^1(Q_0;\R_+).
	$$
	Taking a partition of $Q_0$ into exactly $M^d$ smaller cubes $Q \in \eP$,
	$M \in \N$, of equal size $|Q| = |Q_0|/M^d$,
	the energy is by Lemma~\ref{lem:local-exclusion} 
	bounded as
	\begin{align*}
		T[\Psi] &= \sum_{Q \in \eP} T^Q[\Psi]
		\ge \sum_{Q \in \eP} \frac{\pi^2}{|Q|^{2/d}} \left( \int_Q \varrho_\Psi(\bx) \,d\bx \ - q \right)
		= \frac{\pi^2}{|Q|^{2/d}} \left( \int_{Q_0} \varrho_\Psi(\bx) \,d\bx \ - q M^d \right) \\
		&= \frac{\pi^2}{|Q_0|^{2/d}} \left( N M^2 - q M^{d+2} \right).
	\end{align*}
	Optimizing this expression in $M$ gives 
	$M \sim \left( \frac{2N}{(d+2)q} \right)^{1/d}$ (rounded to the nearest integer)
	and thus
	$$
		E_{0,\asym(q)}(N) \gtrsim \frac{d}{d+2}\left( \frac{2}{d+2} \right)^{2/d} \frac{\pi^2}{q^{2/d}} \frac{N^{1+2/d}}{|Q_0|^{2/d}}.
	$$
	Hence, with $q$ fixed and taking both $N \to \infty$ and $|Q_0| \to \infty$ 
	at fixed density $\rho := N/|Q_0|$,
	what is commonly referred to as the \keyword{thermodynamic limit},
	the energy per particle is
	$$
		\liminf_{N,|Q_0| \to \infty} \frac{E_{0,\asym(q)}(N)}{N} \geq \frac{d}{d+2}\left( \frac{2}{d+2} \right)^{2/d} \pi^2\frac{\rho^{2/d}}{q^{2/d}},
	$$
	to be compared with \eqref{eq:q-fermion-gs-energy} and \eqref{eq:Weyl-asymptotics}.
	Note that when $q = N$, i.e. for bosons, $M$ cannot be chosen large
	and in fact the local bound \eqref{eq:local-exclusion} is trivial for any $Q$.
	This reflects the fact that the energy for the ideal Bose gas on $Q_0$ is
	$E_{0,\sym}(N) = N\lambda_0(-\Delta_{Q_0})$ which is linear in $N$ and also
	depends crucially on the boundary conditions chosen on $\partial Q_0$.
	In the thermodynamic limit one obtains in any case $E_{0,\sym}(N)/N \to 0$ 
	for noninteracting bosons.
\end{example}

\subsection{Repulsive bosons}\label{sec:exclusion-bosons}
	\index{bosons}

	In the case that there is a given pair-interaction $W$ between particles,
	we define the corresponding $n$-particle energy on the box $Q$,
	$$
		e_n(|Q|;W) := 
			\inf_{\int_{Q^n} |\psi|^2 = 1}
			\int_{Q^n} \left( \sum_{j=1}^n |\nabla_j\psi|^2 + \sum_{1\le j<k \le n} W(\bx_j-\bx_k) |\psi|^2
				\right)
	$$
	Note that $e_n$ with $n \ge 2$ can be reduced to a bound in terms of only $e_2$, 
	namely:

\begin{lemma}
	\label{lem:en-from-e2}
	For any pair-interaction potential $W$ and any $d$-cube $Q$, we have
	\begin{equation}\label{eq:en-from-e2}
		e_n(|Q|;W) \ge \frac{n}{2} e_2(|Q|;(n-1)W).
	\end{equation}
\end{lemma}
\begin{proof}
	Use the simple identity
	$$
		(n-1)\sum_{j=1}^n |\nabla_j \psi|^2 = \sum_{1 \le j < k \le n} 
			\left( |\nabla_j\psi|^2 + |\nabla_k\psi|^2 \right)
	$$
	to bound the $n$-body energy in terms of a sum of two-body energies.
\end{proof}

	This tells us that, as soon as the two-particle energy is 
	strictly positive, such as for repulsive interactions,
	there will be positive energy also for $n \ge 2$ particles
	analogously to the Pauli principle.
	We also note that for any $n$ and non-negative pair potential $W$
	the function $f(\mu) :=  e_n(|Q|;\mu W)$ 
	is monotone increasing and concave in $\mu \ge 0$, and $f(0) = 0$.

\subsubsection{The stupid bound}\label{sec:exclusion-bosons-stupid}

	In the case that $W(\bx) = W_\beta(\bx) := \beta|\bx|^{-2}$,
	we may use the following very crude bound for $e_n$:
	$$
		e_n(|Q|;W_\beta) 
		\ge \sum_{j<k} \inf_{\int_{Q^n} |\psi|^2 = 1} \int_{Q^n} W_\beta(\bx_j-\bx_k) |\psi|^2
		\ge \binom{n}{2} \inf_{\bx_1,\bx_2 \in Q} W(\bx_1-\bx_2)
		\ge \frac{\beta n(n-1)}{2d|Q|^{2/d}} 
	$$
	In particular,
	\begin{equation}\label{eq:stupid-bound}
		e_n(|Q|;W_\beta) \ge \frac{\beta}{d|Q|^{2/d}} (n-1)_+.
	\end{equation}

\subsubsection{Hard-core bosons}\label{sec:exclusion-bosons-hardcore}

	For the case of a hard-sphere interaction $W^\hs_R$ in $d=3$
	(see Example~\ref{exmp:hard-core}), one has the rough bound
	\cite[Proposition~10]{LunPorSol-15} 
	$$
		e_2(|Q|;W^\hs_R) \ge \frac{2}{\sqrt{3}} \frac{R}{|Q|}(2 - R/|Q|^{1/3})_+^{-2}.
	$$

\subsubsection{Local exclusion for bosons}\label{sec:exclusion-bosons-local}
	
	Define the corresponding \keyword{expected interaction energy} on $Q$,
	$$
		W^Q[\Psi] := \frac{1}{2} \sum_{j=1}^N \sum_{(j \neq) k=1}^N
			\int_{\R^{dN}} W(\bx_j - \bx_k) |\Psi|^2 \1_Q(\bx_j) \,d\sx,
	$$
	as well as the combined energies
	$$
		(T+W)^Q[\Psi] := T^Q[\Psi] + W^Q[\Psi].
	$$
	Then, for a partition $\eP$ of $\R^d$,
	$$
		T[\Psi] + W[\Psi] = \sum_{Q \in \eP} (T+W)^Q[\Psi].
	$$
	
\begin{lemma}[\keyword{Local exclusion principle for repulsive bosons}]
	\label{lem:local-exclusion-bosons}
	Let $W \ge 0$ be a repulsive pair interaction.
	For any $d$-cube $Q$ and $N$-body wave function $\Psi \in H^1(\R^{dN})$ we have
	\begin{equation}\label{eq:local-exclusion-bosons}
		(T+W)^Q[\Psi] \ \ge \ 
		\frac{1}{2} e_2(|Q|;W) \left( \int_Q \varrho_\Psi(\bx) \,d\bx \ - 1 \right)_+.
	\end{equation}
\end{lemma}
\begin{proof}
	The proof is a straightforward extension of the proof of Lemma~\ref{lem:local-exclusion}
	where we use Lemma~\ref{lem:en-from-e2}, 
	monotonicity $e_2(|Q|;(n-1)W) \ge e_2(|Q|;W)$,
	and that $|A|/2 \ge (|A|-1)_+/2$ for $|A| \ge 2$.
	Non-negativity of $W$ is used in order to estimate interactions between
	particles inside and outside $Q$ trivially.
\end{proof}

	Application: just as a local application of Lemma~\ref{lem:local-exclusion} 
	gave rise to bounds for the homogeneous Fermi gas as in Example~\ref{exmp:exclusion-gas},
	Lemma~\ref{lem:local-exclusion-bosons} 
	can be used to prove lower bounds for homogeneous interacting Bose gases; 
	see \cite{LunPorSol-15}.

\subsection{Anyons}\label{sec:exclusion-anyons}
	\index{anyons}

	We end this chapter with a short discussion on the exclusion properties
	of anyons in two dimensions (recall their definition in Section~\ref{sec:mech-QM-statistics}).
	Two-particle energies and other pairwise statistics-dependent 
	properties for anyons have been known since the
	original works \cite{LeiMyr-77,Wilczek-82b,AroSchWilZee-85} in the abelian case, 
	and at least since \cite{Verlinde-91,LeeOh-94} for certain non-abelian anyons, 
	however the method outlined below to account for statistical repulsion 
	in the full many-body context is fairly recent and developed in
	\cite{LunSol-13a,LunSol-13b,LunSol-14,LarLun-16,Lundholm-16,Qvarfordt-17,LunSei-17}.

	For (ideal abelian) anyons one has the following many-body Hardy inequality, 
	which was
	generalized from the fermionic one 
	\eqref{eq:fermionic-Hardy-split}
	in \cite[Theorem~4]{LunSol-13a} and \cite[Theorem~1.3]{LarLun-16}:

\begin{theorem}[\keyword{Many-anyon Hardy}]\label{thm:anyonic-Hardy}
	For any $\alpha \in \R$, $N \ge 1$, and $\Psi \in \cQ(\hT_\alpha)$,
	one has the many-body Hardy inequality
	\begin{equation} \label{eq:anyonic-Hardy}
		\inp{\Psi, \hT_\alpha \Psi}
		\ge \frac{1}{N} \int_{\R^{dN}} \sum_{1 \le j<k \le N} 
			\left( \big|\partial_{r_{jk}}|\Psi|\big|^2
				+ \alpha_N^2 \frac{|\Psi(\sx)|^2}{r_{jk}^2} \right) \,d\sx
		+ \frac{1}{N} \int_{\R^{dN}} \left| \sum_{j=1}^N \nabla_j \Psi \right|^2 d\sx,
	\end{equation}
	where $r_{jk} = |\bx_j-\bx_k|/2$, and the strength of the statistical repulsion
	term depends on the anyonic statistics parameter $\alpha$ via
	\begin{equation} \label{eq:alpha_N}
		\alpha_{N} 
		:= \min\limits_{p \in \{0, 1, \ldots, N-2\}} \min\limits_{q \in \Z} |(2p+1)\alpha - 2q|.
	\end{equation}
\end{theorem}
	
	The expression \eqref{eq:alpha_N} is a piecewise linear and 
	$2$-periodic function of $\alpha$ 
	(in accordance with the periodicity of the phase)
	and for $N=2$ it reduces 
	to the simple form of a saw-tooth wave with maxima at $\alpha \in 2\Z+1$
	and minima at $\alpha \in 2\Z$,
	$\alpha_2 = \alpha$ for $\alpha \in [0,1]$.
	However, in the limit as $N \to \infty$ the expression depends non-trivially
	on arithmetic properties of $\alpha$
	(see \cite[Proposition~5]{LunSol-13a}):
	\begin{equation} \label{eq:alpha_*}
		\alpha_* := \lim_{N \to \infty} \alpha_N 
			= \inf_{N \ge 2} \alpha_N
			= \left\{ \begin{array}{ll}
			\frac{1}{\nu}, & \text{if $\alpha = \frac{\mu}{\nu} \in \Q$ reduced, $\mu$ \emph{odd} and $\nu \ge 1$,} \\
			0, & \text{otherwise.}
			\end{array}\right.
	\end{equation}
	In other words, 
	it is supported only on the rationals with odd numerator,
	with a magnitude inversely proportional to the denominator,
	and may thus be considered a variant 
	of a function known as the \keyword[Thomae function]{Thomae}, 
	or \keyword{popcorn function}.

\begin{remark*}
	In order to understand the origin of the above expressions
	and their peculiar dependence on $\alpha$,
	recall that anyons may be modeled correctly using connections on fiber bundles;
	cf. Remark~\ref{rem:anyons-fiber-bundles}.
	The kinetic energy for $N$ anyons may in fact be written as 
	$$
		\hT_\alpha = \sum_{j=1}^N \bigl(-i\nabla_j^{\bA_\alpha}\bigr)^2,
	$$
	where $\bA_\alpha$ denotes a connection one-form on the bundle which implements
	the statistics, i.e. which is such that the holonomies produced under continuous
	exchanges of the particles yield the corresponding representation of the
	braid group, $\rho\colon B_N \to U(\cF)$, of the 
	anyon model. 
	For abelian anyons, with fiber $\cF = \C$ and statistics parameter 
	$\alpha \in [0,2)$, 
	this is thus the phase 
	$\rho(\tau_{n_1} \ldots \tau_{n_k}) = e^{i\alpha\pi k}$ 
	of the exchange as
	discussed in Exercises~\ref{exc:braid-relations}-\ref{exc:braid-enclosed-phases}.
	
	The Poincar\'e inequality \eqref{eq:Poincare-antisymmetry} 
	that was used for the statistical repulsion\index{statistical repulsion} 
	of fermions is here replaced by the inequality
	\begin{equation}\label{eq:anyonic-Poincare}
		\int_{r\S^1} \bigl|\nabla_{\br}^{\bA_\alpha} \Psi\bigr|^2
		\ge \int_{r\S^1} \left( \bigl|\partial_{r}|\Psi|\bigr|^2
			+ \min_{q \in \Z} |\Phi(r) - 2q|^2 \frac{1}{r^2} |\Psi|^2 \right),
	\end{equation}
	where the integration is performed over the circle $|\br| = r$ 
	in the relative coordinates
	$\br = \br_{jk}$ of a fixed pair $\bx_j$, $\bx_k$ of particles,
	and $2\pi\Phi(r)$ is defined as the statistics phase 
	(i.e.\ the holonomy on the bundle)
	obtained under exchange of this pair
	as $\br \to -\br \to \br$ continuously along the full circle $r\S^1$.
	Note that already after half of the circle has been traversed
	one has actually completed a full particle exchange,
	$\br \to -\br \sim \br$, 
	and thus the corresponding phase factor in this case must be 
	(see Exercise~\ref{exc:braid-enclosed-phases})
	$$
		e^{i\pi\Phi(r)}, \qquad \Phi(r) = \bigl(1+2p(r)\bigr)\alpha,
	$$
	where $p(r) \in \{0,1,\ldots,N-2\}$ 
	denotes the number of other particles that happen to become 
	enclosed under such an exchange.
	This depends both on the positions of the other $N-2$ particles 
	$\sx' = (\bx_l)_{l \neq j,k}$,
	on the center of mass $\bX$ of the particle pair, and on the radius $r$ of the circle.
	However, recall that the phase is only determined uniquely up to multiples of 
	$2\pi$,
	i.e. $e^{i\pi\Phi(r)} = e^{i\pi(\Phi(r) - 2q)}$ for any $q \in \Z$.
	Let us for definiteness take the representative closest to the identity,
	$$
		e^{i\pi\beta_0}, \qquad 
		\beta_0 = \beta_0(r) := \pm \min_{q \in \Z} |\Phi(r) - 2q|,
	$$
	where one of the signs apply.
	This phase factor may be considered as a non-trivial boundary condition 
	that the function 
	$u(\varphi) := \Psi\bigl(\br(r;\varphi), \bX; \sx'\bigr)$
	of the relative angle $\varphi$ (with $r$, $\bX$ and $\sx'$ fixed)
	must satisfy:
	\begin{equation}\label{eq:semiperiodic-bc}
		u(\pi) = e^{i\pi\beta_0} u(0).
	\end{equation}
	It is a straightforward exercise (see below)
	to show that the Poincar\'e inequality
	\index{Poincar\'e inequality}
	\begin{equation}\label{eq:semiperiodic-Poincare}
		\int_0^\pi |u'(\varphi)|^2 \,d\varphi 
		\ge \beta_0^2 \int_0^\pi |u(\varphi)|^2 \,d\varphi
	\end{equation}
	holds for such semi-periodic functions on the (half) circle, 
	by expanding in the basis of energy eigenstates 
	$u_q(\varphi) = e^{i(\beta_0 + 2q)\varphi}$, $q \in \Z$.
	Proceeding as in \eqref{eq:kinetic-u-spherical} 
	one then obtains \eqref{eq:anyonic-Poincare}, 
	and finally \eqref{eq:anyonic-Hardy} after
	minimizing over all possibilities for $p(r)$, 
	i.e. $\beta_0(r) \ge \alpha_N$ for all $r$.
	Note that in the case $\alpha=0$ one has $\beta_0 \equiv 0$ and thus bosons
	and no Poincar\'e inequality, 
	while for $\alpha=1$ one has $\beta_0 = \min_{q \in \Z} |1+2p(r)-2q| \equiv 1$
	and thus the fermionic Poincar\'e inequality~\eqref{eq:Poincare-antisymmetry}.
	The above procedure may even be extended to certain families of non-abelian anyons 
	\cite{Qvarfordt-17,LunQva-18}.
\end{remark*}

	\bigskip

	Let us denote by (which needs to be interpreted in the correct form sense \cite{LunSol-14,LarLun-16})
	$$
		e_n(\alpha) 
		:= \inf_{\int_{Q_0^n} |\Psi|^2 = 1}
			\inp{\Psi, \hT_\alpha \Psi}_{L^2(Q_0^n)}
	$$
	the local (Neumann) $n$-particle kinetic energy for anyons on the unit square 
	$Q_0 = [0,1]^2$
	(the corresponding energy on a general square $Q \subseteq \R^2$ 
	is obtained by simple scaling due to homogeneity of the kinetic energy).
	A local version of 
	Theorem~\ref{thm:anyonic-Hardy}
	may be used to prove that $e_n$ satisfies a lower bound 
	for all $n$ of the form \cite{LarLun-16}
	$$
		e_n(\alpha) \ge 
		f\bigl((j_{\alpha_n}')^2\bigr) (n-1)_+,
	$$
	where $j_\nu'$ denotes the 
	first positive zero of the derivative of the Bessel function $J_\nu$
	of the first kind,
	$$
		\sqrt{2\nu} \le j_\nu' \le \sqrt{2\nu(1+\nu)}
		\qquad \text{(and $j_0' := 0$)},
	$$
	and $f\colon [0,(j_1')^2] \to \R_+$ is a function satisfying
	\begin{equation}\label{eq:anyon-f-props}
		t/6 \le f(t) \le 2\pi t 
		\qquad \text{and} \qquad
		f(t) = 2\pi t \bigl(1 - O(t^{1/3})\bigr) \quad \text{as $t\to 0$.}
	\end{equation}
	Proceeding as in Example~\ref{exmp:exclusion-gas},
	one may then prove that the ground-state energy per particle 
	and unit density of the ideal anyon gas in the thermodynamic limit
	at fixed density $\rho = N/|Q|$, 
	$$
		e(\alpha) :=
		\liminf_{N \to \infty, |Q| \to \infty} \frac{|Q|^{-1} e_N(\alpha)}{N\rho}
		= \liminf_{N \to \infty} \frac{e_N(\alpha)}{N^2},
	$$
	is bounded from below by
	$$
		e(\alpha)
		\ge {\textstyle\frac{1}{24}} (j_{\alpha_*}')^2
		\ge {\textstyle\frac{1}{12}} \alpha_*,
	$$
	and moreover, as $\alpha_* \to 0$ 
	the bound improves to
	\begin{equation}\label{eq:anyon-energy-small-alpha*}
		e(\alpha)
		\ge \pi\alpha_* \bigl(1 - O(\alpha_*^{1/3})\bigr).
	\end{equation}
	
	The dependence of the above expressions on $\alpha_*$ 
	comes about by 
	assuming (as a lower bound) that the measure
	of relative radii $r = |\bx_j-\bx_k|/2$ such that the factor $\beta_0(r)^2$ 
	of the potential in \eqref{eq:anyonic-Poincare}
	differs from its absolute minimum $\alpha_*^2$ can be vanishingly small.
	This requires the gas to be dilute and with its particles arranged in tiny clusters
	\cite{LarLun-16,Lundholm-16}, 
	and it is not clear that such configurations will be beneficial with
	respect to the uncertainty principle.
	Indeed, very recently the above bounds have been improved 
	by using the scale invariance of ideal anyons and the uncertainty principle, 
	to yield a dependence only on the two-particle energy
	\cite{LunSei-17}:
	
\begin{lemma}[\keyword{Local exclusion principle for anyons} {\cite{LunSei-17}}]
	\label{lem:local-exclusion-anyons}
	For any $\alpha \in \R$ and $n \ge 2$, it holds
	$$
		e_n(\alpha) \ge c(\alpha)n,
	$$
	where
	$$
		c(\alpha) := \frac{1}{4} \min\{ e_2(\alpha), e_3(\alpha), e_4(\alpha) \} 
		\ge \frac{1}{4} \min\{ e_2(\alpha), 0.147 \}.
	$$
	Furthermore, for any $N$-anyon wave function $\Psi \in \cQ(\hT_\alpha)$ 
	and any square $Q$
	we have the local exclusion principle
	\begin{equation}\label{eq:local-exclusion-anyons}
		T_\alpha^{Q}[\Psi] \ \ge \ 
		\frac{c(\alpha)}{|Q|} \left( \int_Q \varrho_\Psi(\bx) \,d\bx \ - 1 \right)_+.
	\end{equation}
\end{lemma}
	
	In fact, these lower bounds in terms of
	$e_2(\alpha) \ge \frac{1}{3}\alpha_2$
	may be complemented with 
	upper bounds of the same form,
	and one has the following
	leading behavior for the ground state energy of the ideal anyon gas,
	showing that it has a similar extensitivity as the Fermi gas for
	all types of anyons except for bosons:
	
\begin{theorem}[\keyword{Extensivity of the ideal anyon gas energy} {\cite{LunSei-17}}]\label{thm:anyon-gas}
	There exist constants $0 < C_1 \le C_2 < \infty$ such that for any $\alpha \in \R$
	$$
		C_1 \alpha_2 \le e(\alpha) \le C_2 \alpha_2,
	$$
	and moreover, in the limit $\alpha_2 \to 0$,
	$$
		e(\alpha) \ge \frac{\pi}{4} \alpha_2 \bigl(1-O(\alpha_2^{1/3})\bigr).
	$$
\end{theorem}

	However, 
	the exact energy is not known, and in fact a recent conjecture \cite{CorLunRou-16}
	in the context of a common approximation known as average-field theory
	could imply for the full energy that the simple linear interpolation 
	$e(\alpha) = 2\pi\alpha_2$ (i.e. $C_1 = C_2 = 2\pi$)
	cannot hold.
	Furthermore, the picture might change with an additional attraction 
	between the anyons which promotes clustering.
	Whether the true energy $e(\alpha)$ could be lower for even-numerator rational $\alpha$
	than for odd numerators due to the above form of statistical repulsion 
	is an interesting possibility, discussed in more detail in \cite{Lundholm-16}.

\begin{exc}
	Prove the Poincar\'e inequality \eqref{eq:semiperiodic-Poincare}
	for functions $u \in H^1([0,\pi])$ subject to the semi-periodic 
	boundary condition \eqref{eq:semiperiodic-bc}.
	
	Note that the self-adjoint operator corresponding to the form
	\eqref{eq:semiperiodic-Poincare} is $D^2 = -\partial_\varphi^2$,
	with $D = -i\partial_\varphi$ defined as a self-adjoint operator
	on $L^2([0,\pi])$ with the b.c. \eqref{eq:semiperiodic-bc}.
	Then the natural domain of $D^2$ is the space of functions $u \in H^2([0,\pi])$
	satisfying both
	\begin{equation}\label{eq:semiperiodic-bc-both}
		u(\pi) = e^{i\pi\beta_0} u(0)
		\qquad \text{and} \qquad
		u'(\pi) = e^{i\pi\beta_0} u'(0).
	\end{equation}
\end{exc}

\section{The Lieb--Thirring inequality\lect{ [13,14]}}\label{sec:LT}

	In Section~\ref{sec:uncert} we found that the kinetic energy of an
	arbitrary $L^2$-normalized $N$-body state $\Psi \in \cQ(\hT) = H^1(\R^{dN})$ 
	is bounded from below by
	(as usual we assume $\hbar^2/(2m) = 1$)
	$$
		T[\Psi] := \inp{\Psi, \sum_{j=1}^N (-\Delta_j) \Psi}
		\ge G_d N^{-2/d} \int_{\R^d} \varrho_\Psi^{1+2/d},
	$$
	which encodes the uncertainty principle by yielding an increase in the 
	kinetic energy for localized densities, 
	but unfortunately becomes overall very weak with $N \to \infty$.
	It turns out however that if one restricts to fermionic,
	i.e. antisymmetric, states $\Psi \in \cH_\asym$,
	then this inequality can be improved to
	\begin{equation}\label{eq:LT}
		T[\Psi] 
		\ge K_d \int_{\R^d} \varrho_\Psi^{1+2/d},
	\end{equation}
	with a constant $K_d > 0$ that is independent of $N$.
	This \keyword{fermionic kinetic energy inequality}, 
	which encodes \emph{both} the uncertainty principle 
	\emph{and} the exclusion principle,
	is also known as a \keyword{Lieb--Thirring inequality} and it was introduced by 
	Lieb and Thirring in 1975 \cite{LieThi-75,LieThi-76} in order to give a 
	new and drastically simplified proof of stability of matter, 
	as compared to the original tour-de-force 
	proof due to Dyson and Lenard in 1967.

	We will apart from proving the celebrated inequality \eqref{eq:LT}
	for fermions
	also prove that the assumption on antisymmetry, i.e.~the Pauli principle,
	may be replaced by a strong enough repulsive interaction 
	which then effectively imposes an exclusion principle on the states $\Psi$ 
	as discussed in the previous chapter.
	In particular, 
	for an inverse-square pair interaction $W(\bx) = \beta|\bx|^{-2}$,
	and for any $\Psi \in H^1(\R^{dN})$ 
	(hence also for bosons $\Psi \in \cH_\sym$ or distinguishable particles),
	the following Lieb--Thirring-type inequality holds:
	\begin{equation}\label{eq:LT-interaction}
		T[\Psi] + W[\Psi]
		\ge K_d(\beta) \int_{\R^d} \varrho_\Psi^{1+2/d},
	\end{equation}
	with a constant $K_d(\beta) > 0$ for any $\beta > 0$.
	Also other forms of LT inequalities are valid if for instance there is 
	some local statistical repulsion, such as for anyons in two dimensions.

\subsection{One-body and Schr\"odinger formulations}\label{sec:LT-1p}

	Note that for a fermionic basis state
	$\Psi = u_1 \wedge u_2 \wedge \ldots \wedge u_N$,
	i.e. a Slater determinant,
	where $\{u_j\}_{j=1}^N \subset \gH = L^2(\R^d)$ denotes 
	an orthonormal set of one-body states,
	we have that (exercise)
	\begin{equation}\label{eq:Slater-T}
		T[\Psi] 
		= \sum_{j=1}^N \int_{\R^d} |\nabla u_j|^2,
	\end{equation}
	and furthermore (exercise)
	\begin{equation}\label{eq:Slater-rho}
		\varrho_\Psi(\bx) 
		= \sum_{j=1}^N |u_j(\bx)|^2.
	\end{equation}
	The inequality \eqref{eq:LT} on such a state then follows straightforwardly
	from the following simple generalization of the
	Gagliardo--Nirenberg--Sobolev inequality of Theorem~\ref{thm:GNS}.
	This very simple approach to proving Lieb--Thirring inequalities directly
	by means of the kinetic energy inequality is quite recent and
	due to Rumin \cite{Rumin-10,Rumin-11}
	(see also e.g. \cite{Solovej-11} and \cite{Frank-14} for generalizations).

\begin{theorem}[\keyword{Kinetic energy inequality}]\label{thm:LT-kinetic-1p}
	Given an $L^2$-orthonormal set $\{u_j\}_{j=1}^N \subset H^1(\R^d)$,
	we have that
	$$
		\sum_{j=1}^N \int_{\R^d} |\nabla u_j(\bx)|^2 \,d\bx
		\ge K_d \int_{\R^d} \left( \sum_{j=1}^N |u_j(\bx)|^2 \right)^{1+\frac{2}{d}} d\bx,
	$$
	with a constant satisfying $G_d' \le K_d \le \min\{G_d,K_d^\cl\}$.
\end{theorem}
\begin{remark}\label{rem:LT-conjecture}
	See Remark~\ref{rem:GNS-constant} concerning the constants $G_d' \le G_d$.
	The presently known best bound for the optimal constant $K_d$ is 
	$K_d \ge (\pi/\sqrt{3})^{-2/d}K_d^\cl$ 
	\cite{DolLapLos-08},
	and it is conjectured that $K_d = G_d < K_d^\cl$ for $d \le 2$ 
	while $K_d = K_d^\cl < G_d$ for $d \ge 3$ \cite{LieThi-76}.
	See also \cite{Laptev-12}.
\end{remark}
\begin{proof}
	We follow the proof of Theorem~\ref{thm:GNS}, 
	with the crucial difference that we use the
	Bessel inequality in the bound corresponding to~\eqref{eq:GNS-low-energy}.
	Namely, we have for the orthonormal system of functions $\{u_j\}_j$ that
	\begin{equation}\label{eq:LT-high-energy}
		\sum_{j=1}^N \int_{\R^d} |\nabla u_j(\bx)|^2 \,d\bx
		= \sum_{j=1}^N \int_0^\infty \int_{\R^d} |u_j^{E,+}(\bx)|^2 \,d\bx \,dE,
	\end{equation}
	and by the triangle inequality on $\C^N$,
	\begin{equation}\label{eq:LT-triangle-ineq}
		\sum_{j=1}^N \bigl|u_j^{E,+}(\bx)\bigr|^2 
		\ge \left[ \left( \sum_{j=1}^N |u_j(\bx)|^2 \right)^{1/2} 
			- \left( \sum_{j=1}^N \bigl|u_j^{E,-}(\bx)\bigr|^2 \right)^{1/2} \right]_+^2.
	\end{equation}
	The bound on the low-energy part is then done using Fourier transform and
	Bessel's inequality~\eqref{eq:Bessel}
	(note that $\{\hu_j\}_j$ are also orthonormal by the unitarity of $\cF$)
	according to
	\begin{align}\label{eq:LT-low-energy}
		\sum_{j=1}^N \bigl|u_j^{E,-}(\bx)\bigr|^2
		&= \sum_{j=1}^N \left| (2\pi)^{-d/2} \int_{\R^d} 
			\1_{\{|\bp|^2 \le E\}} \hu_j(\bp) e^{i\bp\cdot\bx} \,d\bp \right|^2 \\
		&= (2\pi)^{-d} \sum_{j=1}^N 
			\Bigl|\inp{\hu_j, \1_{\{|\bp| \le E^{1/2}\}} e^{i\bp\cdot\bx}}\Bigr|^2 
		\le (2\pi)^{-d} |B_{E^{1/2}}(0)|,
	\end{align}
	after which the remainder of the previous proof goes through 
	with the replacement $\|u\| = 1$,
	and $K_d \ge G_d'$ of \eqref{eq:GNS-constant}.
	Also, since for $N=1$ the inequality is exactly GNS, 
	we cannot have $K_d > G_d$ for the optimal constants.
	Furthermore, taking as $u_j$ the eigenfunctions of the Dirichlet
	Laplacian on a cube, we may as $N \to \infty$ compare to the
	Weyl asymptotics \eqref{eq:Weyl-asymptotics}, see e.g.\ \cite{Kroeger-94},
	and in fact one may thus prove that also $K_d \le K_d^\cl$.
\end{proof}
	
\begin{theorem}[\keyword{Many-body kinetic energy inequality}]\label{thm:LT-kinetic}
	The inequality 
	\begin{equation}\label{eq:LT-thm}
		\int_{\R^{dN}} |\nabla \Psi|^2
		\ge K_d \int_{\R^d} \varrho_\Psi^{1+2/d}
	\end{equation}
	holds for any 
	$\Psi \in H^1_\asym((\R^d)^N)$.
\end{theorem}
\begin{proof}
	The proof in the many-body case would again be a straightforward modification 
	of the above one-body case, along the lines of Exercise~\ref{exc:GNS-many-body},
	if we only knew that
	the corresponding partial traces of $\Psi$ 
	are orthonormal. 
	This is the case for pure product states, i.e. Slater determinants
	(see Exercise~\ref{exc:LT-Slater}), but not necessarily so for a general $\Psi$.
	What we may do instead is to use a diagonalization trick from the 
	abstract theory of density matrices 
	(which goes slightly outside the course, 
	cf. Remark~\ref{rem:density-matrix}, 
	but we nevertheless give here for the interested reader).
	Alternatively, we will find below that the theorem also follows 
	directly from the above one-body theorem 
	together with the equivalence between the kinetic energy
	and Schr\"odinger forms of the inequality; 
	Theorem~\ref{thm:LT-equivalence} and Corollary~\ref{thm:LT-eigenvalues}.

	Given $\Psi \in \cH_\asym$ we may form the corresponding one-body density matrix
	$\gamma_\Psi\colon \gH \to \gH$, $\gH = L^2(\R^d)$,
	defined via the integral kernel
	\begin{align*}
		\gamma_\Psi(\bx,\by) := \sum_{j=1}^N \int_{\R^{d(N-1)}} 
		&\Psi(\bx_1,\dots,\bx_{j-1},\bx,\bx_{j+1},\dots,\bx_N) \times \\
		&\times \overline{\Psi(\bx_1,\dots,\bx_{j-1},\by,\bx_{j+1},\dots,\bx_N)}
			\, \prod\limits_{k \neq j} d\bx_k,
	\end{align*}
	which turns out to be a bounded self-adjoint trace-class operator, with
	$0 \le \gamma_\Psi \le \1$ and $\Tr \gamma_\Psi = N$.
	It may thus be diagonalized, 
	$$
		\gamma_\Psi = \sum_{j=1}^\infty \lambda_j u_j \inp{u_j,\slot}
	$$
	with $\lambda_j \in [0,1]$ and $\{u_j\}_j \subset \gH$ orthonormal.
	Furthermore,
	$$
		\varrho_\Psi(\bx) = \gamma_\Psi(\bx,\bx) 
		= \sum_{j=1}^\infty \lambda_j |u_j(\bx)|^2.
	$$
	Then
	$$
		T[\Psi] = \sum_{j=1}^\infty \inp{\Psi, (-\Delta_j)\Psi}
		= \sum_{j=1}^\infty \lambda_j \inp{u_j, (-\Delta)u_j}
		= \sum_{j=1}^\infty \lambda_j \int_{\R^d} |\nabla u_j|^2
		\ge \sum_{j=1}^M \lambda_j \int_{\R^d} |\nabla u_j|^2,
	$$
	for any $M \in \N$.
	Now we may modify the proof of Theorem~\ref{thm:LT-kinetic-1p}
	by attaching $\sqrt{\lambda_j}$ to each $u_j$ 
	(or taking the triangle inequality on $\C^M$ weighted with $\lambda$),
	with 
	$$
		\sum_{j=1}^M \left| \inp{ \sqrt{\lambda_j}\hu_j , v} \right|^2 
		= \sum_{j=1}^M \lambda_j \bigl|\inp{\hu_j,v}\bigr|^2
		\le \sum_{j=1}^M \bigl|\inp{\hu_j,v}\bigr|^2
		\le \|v\|^2,
	$$
	again by Bessel's inequality for the orthonormal set $\{\hu_j\}$,
	and $\lambda_j \le 1$. Thus,
	$$
		T[\Psi] \ge \sum_{j=1}^M \int_{\R^d} \lambda_j |\nabla u_j(\bx)|^2 \,d\bx
		\ge K_d \int_{\R^d} \left( \sum_{j=1}^M \lambda_j|u_j(\bx)|^2 \right)^{1+\frac{2}{d}} d\bx,
	$$
	and taking $M \to \infty$ this proves the theorem.
\end{proof}

\begin{exc}\label{exc:LT-Slater}
	Show \eqref{eq:Slater-T}, \eqref{eq:Slater-rho} 
	and thus that \eqref{eq:LT-thm} holds for all such basis states $\Psi$
	immediately by Theorem~\ref{thm:LT-kinetic-1p}.
\end{exc}

\subsubsection{Equivalent Schr\"odinger operator formulation}

	A common and indeed very useful equivalent reformulation of the inequality
	\eqref{eq:LT-thm}
	is in the form of an operator inequality involving the negative eigenvalues
	of the one-body Schr\"odinger operator on $\R^d$
	$$
		\hat{h} = -\Delta + V.
	$$
	Both of these formulations are referred to as Lieb--Thirring inequalities.
	We will be a bit more general here, however, 
	allowing to replace the exclusion
	principle for fermions with a repulsive pair interaction $W$ 
	or some other statistical repulsion.
	We denote as usual
	\begin{align*}
		T[\Psi] &= \int_{\R^{dN}} |\nabla \Psi|^2 
		= \int_{\R^{dN}} \sum_{j=1}^N |\nabla_j \Psi|^2, \\
		V[\Psi] &= \int_{\R^{dN}} V|\Psi|^2 
		= \int_{\R^{dN}} \sum_{j=1}^N V(\bx_j)|\Psi(\sx)|^2 \,d\sx, \quad \text{and}\\
		W[\Psi] &= \int_{\R^{dN}} W|\Psi|^2 
		= \int_{\R^{dN}} \sum_{1 \le j<k \le N} W(\bx_j-\bx_k)|\Psi(\sx)|^2 \,d\sx,
	\end{align*}
	so that
	$$
		T[\Psi] + V[\Psi] =
		\sum_{j=1}^N \int_{\R^{dN}} \left( |\nabla_j \Psi|^2 + V(\bx_j)|\Psi|^2 \right)
		= \sum_{j=1}^N \inp{\Psi, (-\Delta_j + V(\bx_j)) \Psi}
		= \sum_{j=1}^N \inp{\Psi, \hh_j \Psi}.
	$$
		
\begin{theorem}[{\keyword[Lieb--Thirring inequality]{Lieb--Thirring inequalities}}]
	\label{thm:LT-equivalence}
	There is an equivalence between 
	\keyword[exclusion-kinetic energy inequality]{exclusion-kinetic energy inequalities}
	of the form
	\begin{equation}\label{eq:LT-general-K}
		T[\Psi] + W[\Psi]
		\ge K_d(W) \int_{\R^d} \varrho_\Psi^{1+2/d},
	\end{equation}
	and inequalities for Schr\"odinger operators of the form
	\begin{equation}\label{eq:LT-general-V}
		T[\Psi] + V[\Psi] + W[\Psi]
		\ge -L_d(W) \int_{\R^d} |V_-|^{1+d/2},
	\end{equation}
	with 
	the relationship between the constants
	\begin{equation}\label{eq:L-K-relation}
		L_d(W) = \frac{2}{d+2} \left(\frac{d}{d+2}\right)^{d/2} K_d(W)^{-d/2}.
	\end{equation}
	In the above, the interaction $W$ may be replaced by a restriction 
	of the domain of $\hT$ such as to $H^1_{\asym(q)}$,
	or by a different many-body operator such as $\hT_\alpha$ for anyons.
\end{theorem}
\begin{remark}
	For $N=1$ (for which $W=0$ and $\cH = \cH_\sym = \cH_\asym$)
	the equivalence is Theorem~\ref{thm:stability-by-GNS} 
	concerning only the uncertainty principle.
\end{remark}
\begin{proof}
	The proof is completely analogous to that of Theorem~\ref{thm:stability-by-GNS}, 
	namely, assuming that \eqref{eq:LT-general-K} holds, we obtain
	by H\"older and optimization
	\begin{align*}
		T[\Psi] + W[\Psi] + V[\Psi] 
		&= T[\Psi] + W[\Psi] + \int_{\R^d} |V_+|\varrho_\Psi - \int_{\R^d} |V_-|\varrho_\Psi \\
		&\ge K_d \int_{\R^d} \varrho_\Psi^{1+2/d}
			- \left( \int_{\R^d} |V_-|^{1+d/2} \right)^{2/(d+2)}
				\left( \int_{\R^d} \varrho_\Psi^{1+2/d} \right)^{d/(d+2)} \\
		&\ge - L_d \int_{\R^d} |V_-|^{1+d/2}.
	\end{align*}
	On the other hand,
	if \eqref{eq:LT-general-V} holds, then by taking the one-body potential
	$V(\bx) := -c\varrho_\Psi^{2/d}$ we obtain that
	\begin{align*}
		T[\Psi] + W[\Psi] &= T[\Psi] + W[\Psi] + V[\Psi] - V[\Psi]
		\ge -L_d \int_{\R^d} |V_-|^{1+d/2} - \int_{\R^d} V\varrho_\Psi \\
		&= \left( c - c^{1+d/2} L_d \right) \int_{\R^d} \varrho_\Psi^{1+2/d},
	\end{align*}
	which after optimization in $c > 0$ again yields \eqref{eq:LT-general-K}.
\end{proof}

	Since for fermions $\Psi \in H^1_{\asym}(\R^{dN})$ (with $W=0$) one has
	$$
		\inf_{\|\Psi\| = 1} \left( T[\Psi] + V[\Psi] \right) 
		= \sum_{k=0}^{N-1} \lambda_k(\hh),
	$$
	the above equivalence then implies 
	the following inequality for
	the negative eigenvalues $\lambda_k^-$ of the one-body 
	Schr\"odinger operator $\hh = -\Delta + V$ on $\R^d$:

\begin{corollary}[Inequality for the sum of Schr\"odinger eigenvalues]\label{thm:LT-eigenvalues}
	Let $\{\mu_k\}_{k=0}^\infty$ denote 
	the min-max values
	of the Schr\"odinger operator $\hh = -\Delta + V$ on $\R^d$.
	Then
	\begin{equation}\label{eq:LT-eigenvalues}
		\sum_{k=0}^{\infty} |[\mu_k]_-| 
		\le L_d \int_{\R^d} |V_-(\bx)|^{1+d/2} \,d\bx.
	\end{equation}
\end{corollary}
\begin{remark}
	Note that if the r.h.s. of \eqref{eq:LT-eigenvalues} is finite
	then the bottom of the essential spectrum of $\hh$ must satisfy
	$\inf \sigma_\textup{ess}(\hh) \ge 0$, because otherwise all 
	$\mu_k(\hh) \le \inf \sigma_\textup{ess}(\hh) < 0$
	and thus $\sum_{k=0}^\infty |[\mu_k]_-| = \infty$.
	It follows that the negative min-max values are actually eigenvalues and
	hence that \eqref{eq:LT-eigenvalues} is an inequality for the sum of
	negative eigenvalues of $\hh$.
	
	In the spectral theory literature it is usually this inequality 
	that one refers to as the Lieb--Thirring inequality\index{Lieb--Thirring inequality}.
	Also note that it can be generalized to other powers of the eigenvalues
	as well as to other powers of the Laplacian, including fractional
	\cite{LieSei-09,Laptev-12}.
\end{remark}
\begin{proof}
	Let $n \le D := \dim P_{(-\infty,0)}^{\hh} \gH$ be finite, and
	denote by $\{u_k\}_{k=0}^{n-1}$ an orthonormal sequence of $C^2_c(\R^d)$ functions corresponding
	to $n$ lowest negative min-max values 
	$\tilde\mu_k := \langle u_k,\hh u_k \rangle$ (approximating $\mu_k$),
	and $\varrho(\bx) := \sum_{k=0}^{n-1} |u_k(\bx)|^2$.
	The one-body kinetic energy inequality of Theorem~\ref{thm:LT-kinetic-1p}
	then yields
	\begin{align*}
		-\sum_{k=0}^{n-1} |\tilde\mu_k|
		&= \sum_{k=0}^{n-1} \inp{u_k, \hh u_k}
		= \sum_{k=0}^{n-1} \int_{\R^d} \left( |\nabla u_k|^2 + V \varrho \right) \\
		&\ge K_d \int_{\R^d} \varrho^{1+2/d} - \int_{\R^d} |V_-|\varrho
		\ \ge - L_d \int_{\R^d} |V_-|^{1+d/2},
	\end{align*}
	where we again estimated by H\"older and optimized as in Theorem~\ref{thm:LT-equivalence}.
	Taking $n \to D$ and using that the bound is uniform in the approximation
	$\tilde\mu_k \to \mu_k$,
	this proves \eqref{eq:LT-eigenvalues}.
\end{proof}

	Given the inequality \eqref{eq:LT-eigenvalues},
	one obtains the lower bound
	$$
		\inf_{\substack{\Psi \in \cH_{\asym} \\ \|\Psi\| = 1}} 
			\left( T[\Psi] + V[\Psi] \right) 
		= \sum_{k=0}^{N-1} \lambda_k(\hh) 
		\ge -\sum_{k=0}^{D-1} |\lambda_k^-(\hh)|
		\ge -L_d \int_{\R^d} |V_-|^{1+d/2},
	$$
	which by the equivalence of Theorem~\ref{thm:LT-equivalence}
	proves the many-body kinetic energy inequalty
	of Theorem~\ref{thm:LT-kinetic}.
	
\begin{corollary}[LT with weaker exclusion]\label{thm:LT-q}
	For any $\Psi \in H^1_{\asym(q)}((\R^d)^N)$, 
	the exclusion-kinetic energy inequality 
	\begin{equation}\label{eq:LT-kinetic-q}
		\int_{\R^{dN}} |\nabla \Psi|^2
		\ge q^{-2/d} K_d \int_{\R^d} \varrho_\Psi^{1+2/d}
	\end{equation}
	holds, and is equivalent to the uniform bound
	\begin{equation}\label{eq:LT-thm-q}
		T[\Psi] + V[\Psi]
		\ge -q L_d \int_{\R^d} |V_-|^{1+d/2}.
	\end{equation}
\end{corollary}
\begin{proof}
	One may use the relationship $L_d(\asym(q)) = qL_d(\asym(1)) = qL_d$, 
	following from
	$$
		\inf_{\substack{\Psi \in \cH_{\asym(q)} \\ \|\Psi\| = 1}} 
			\left( T[\Psi] + V[\Psi] \right) 
		= q\sum_{k=0}^{N/q-1} \lambda_k(\hh),
	$$
	and hence $K_d(\asym(q)) = q^{-2/d}K_d(\asym(1)) = q^{-2/d}K_d$ 
	by the correspondence \eqref{eq:L-K-relation}.
\end{proof}

\subsection{Local approach to Lieb--Thirring inequalities}\label{sec:LT-local}

The above formulations of LT were global in the sense that they always involved the full
one-body configuration space $\R^d$. We shall now consider a local approach
to proving LT inequalities, which was first developed in \cite{LunSol-13a}
for anyons, and has since been generalized in various directions,
including point-interacting fermions \cite{FraSei-12}, 
other types of generalized statistics \cite{LunSol-13b,LunSol-14},
inhomogeneously scaling repulsive interactions \cite{LunPorSol-15},
operators involving fractional powers and critical Hardy terms \cite{LunNamPor-16},
and most recently gradient corrections to TF \cite{Nam-18}.

\subsubsection{Covering lemma}\label{sec:LT-covering}

For the local approach it is convenient to 
use the following lemma which originates in the construction used
in \cite{LunSol-13a}, and was generalized in \cite[Lemmas~9 and 12]{LunNamPor-16}
and in \cite{Nam-18}. The following version includes all those as special cases,
though not with the optimal constant.

\begin{lemma}[\keyword{Covering lemma}] \label{lem:covering}
	Let $Q_0$ be a $d$-cube in $\R^d$, $d \ge 1$, and let $\Lambda>0$. 
	Let $0\le f\in L^1(Q_0)$ satisfy $\int_{Q_0} f \ge \Lambda>0$. 
	Then $Q_0$ can be partitioned into 
	a collection $\cQ$ of disjoint sub-cubes $Q \in \cQ$, 
	i.e. $Q_0 = \overline{\bigsqcup \cQ}$,
	such that:
	\begin{itemize}
	\item For all $Q \in \cQ$,
	$$
		\int_{Q} f \le \Lambda.
	$$
	\item For all $\alpha,\beta>0$ and $\gamma \ge 0$, 
	\begin{equation}\label{eq:covering-bound}
		\sum_{Q \in \cQ} \frac{1}{|Q|^{\alpha}} \left[ 
			\left(\int_{Q} f \right)^\beta 
			- \frac{\Lambda^{\beta-\gamma}}{C_{d,\alpha,\beta}} 
				\left(\int_{Q} f \right)^\gamma
			\right] \ge 0,
	\end{equation}
	where 
	\begin{equation}\label{eq:covering-const}
		C_{d,\alpha,\beta} := \frac{ 2^{d(\alpha+\beta+1)} }{ 2^{d\alpha} - 1 }.
	\end{equation}
	\end{itemize}
\end{lemma}
\begin{proof}
	Note that if $\int_{Q_0} f = \Lambda$ then there is nothing to prove since
	$C_{d,\alpha,\beta} > 1$, hence we may assume $\int_{Q_0} f > \Lambda$.
	We then start by dividing $Q_0$ into $2^d$ subcubes $Q$ of equal size,
	$|Q| = 2^{-d}|Q_0|$,
	and consider the mass $\int_Q f$ on each such subcube.
	If $\int_Q f \le \Lambda$ we do nothing, while
	if $\int_Q f > \Lambda$ then we may iterate the procedure on $Q$ and
	divide that cube into $2^d$ smaller cubes, and so on.
	This procedure stops after finitely many iterations since $f$ is integrable.
	We may organize the resulting divisions of cubes into a $2^d$-ary tree rooted at 
	$Q_0$ and with the leaves of the tree representing the resulting 
	disjoint cubes $Q$ that form the sought collection $\cQ$, 
	with $\cup_{Q \in \cQ} \bar{Q} = Q_0$ and $\int_Q f \le \Lambda$.
	See Figure~\ref{fig:splitting} for an example.
	
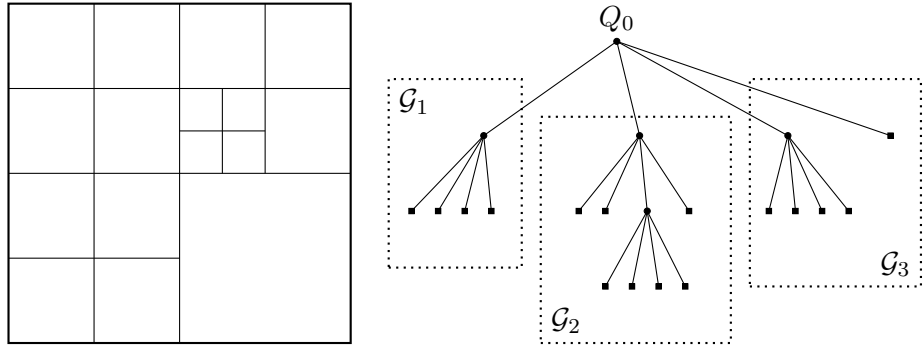
\begin{figure}
	\begin{tikzpicture}
		%SQUARE
		\def\L{4.5}
		\draw[thick] (0,0) rectangle (\L, \L);
		
		\draw (\L/2,0) -- (\L/2,\L);
		\draw (0,\L/2) -- (\L,\L/2);
		
		\draw (0,\L/4) -- (\L/2,\L/4);
		\draw (0,3*\L/4) -- (\L,3*\L/4);
		\draw (\L/4,0) -- (\L/4,\L);
		\draw (3*\L/4,\L/2) -- (3*\L/4,\L);
		
		\draw (5*\L/8,\L/2) -- (5*\L/8,3*\L/4);
		\draw (\L/2,5*\L/8) -- (3*\L/4,5*\L/8);
		
		%SPLITTING
		\draw [fill] (8,4) circle [radius = 0.04];
		\node [above] at (8,4) {$Q_0$};
		\draw (8,4) -- (6.25,2.75);
		\draw (8,4) -- (8.3,2.75);
		
		%F_1
		\draw [dotted, thick] (5,1) rectangle (6.75,3.5);
		\node [below right] at (5,3.5) {$\cG_1$};
		\draw [fill] (6.25,2.75) circle [radius = 0.04];
		\foreach \x in {0,...,3}
			\filldraw ([xshift=-1pt,yshift=-1pt]5.3+0.35*\x,1.75) rectangle ++(2pt,2pt);
		\foreach \x in {0,...,3}
			\draw (6.25,2.75) -- (5.3+0.35*\x,1.75);
		
		%F_2
		\draw [dotted, thick] (7,0) rectangle (9.5,3);
		\node [above right] at (7,0) {$\cG_2$};
		\draw [fill] (8.3,2.75) circle [radius = 0.04];
		
		\filldraw ([xshift=-1pt,yshift=-1pt]7.5,1.75) rectangle ++(2pt,2pt);
		\filldraw ([xshift=-1pt,yshift=-1pt]7.85,1.75) rectangle ++(2pt,2pt);
		\draw [fill] (8.4,1.75) circle [radius = 0.04];
		\filldraw ([xshift=-1pt,yshift=-1pt]8.95,1.75) rectangle ++(2pt,2pt);
		
		\draw (8.3,2.75) -- (7.5,1.75);
		\draw (8.3,2.75) -- (7.85,1.75);
		\draw (8.3,2.75) -- (8.4,1.75);
		\draw (8.3,2.75) -- (8.95,1.75);
		
		\foreach \x in {0,...,3}
			\filldraw ([xshift=-1pt,yshift=-1pt]7.85+0.35*\x,0.75) rectangle ++(2pt,2pt);
		\foreach \x in {0,...,3}
			\draw (8.4,1.75) -- (7.85+0.35*\x,0.75);
		
		%F_3
		\draw [dotted, thick] (9.75,0.75) rectangle (12,3.5);
		\node [above left] at (12,0.75) {$\cG_3$};
		\draw [fill] (10.25,2.75) circle [radius = 0.04];
		\draw (8,4) -- (10.25,2.75);
		\foreach \x in {1,...,1}
			\filldraw ([xshift=-1pt,yshift=-1pt]10.7+0.9*\x,2.75) rectangle ++(2pt,2pt);
		\foreach \x in {1,...,1}
			\draw (8,4) -- (10.7+0.9*\x,2.75);
		\foreach \x in {0,...,3}
			\filldraw ([xshift=-1pt,yshift=-1pt]10+0.35*\x,1.75) rectangle ++(2pt,2pt);
		\foreach \x in {0,...,3}
			\draw (10.25,2.75) -- (10+0.35*\x,1.75);			
	\end{tikzpicture}
	\caption{Example in $d=2$ of a division of a square $Q_0$  
		and a corresponding tree with three disjoint groups $\cG_j$,
		together covering all subsquares $Q \in \cQ$.}
	\label{fig:splitting}
\end{figure}

	Moreover, we note that the cubes $Q \in \cQ$ may be distributed into 
	a finite number of disjoint groups $\cG \subseteq \cQ$, 
	such that in each group:
	\begin{itemize}
	\item There is a smallest size of cubes in $\cG$, 
		denoted $m := \min_{Q \in \cG} |Q|$, 
		and the total mass of such cubes is 
		\begin{equation}\label{eq:covering-smallcube-bound}
			\sum_{Q \in \cG : |Q|=m}\int_Q f \ge \Lambda.
		\end{equation}
	\item There are at most $2^d$ cubes in $\cG$ of every given size.
	\end{itemize}
	Such a grouping may be constructed for example by starting with each
	collection of $2^d$ leaves which stem from a final split of a cube
	with $\int_Q f > \Lambda$, 
	and then add leaves to the group by going back in the tree 
	(arbitrarily many times, possibly all the way to the root $Q_0$)
	and then one step forward.
	Note that all leaves will then be covered by at least one such group,
	and if several, we may choose one arbitrarily.
	
	Now, consider an arbitrary group $\cG$.
	Because of \eqref{eq:covering-smallcube-bound}, 
	there must be at least one cube $Q \in \cG$ with $|Q|=m$ and
	$\int_Q f \ge \Lambda/2^d$. Hence,
	$$
		\max_{\substack{Q \in \cG \\ |Q|=m}} \frac{1}{|Q|^\alpha} 
			\left( \int_Q f \right)^\beta
		\ge \frac{1}{m^\alpha} \left( \frac{\Lambda}{2^d} \right)^\beta 
		= \frac{\Lambda^\beta}{2^{d\beta} m^\alpha}.
	$$
	Furthermore,
	\begin{align*}
		\sum_{Q \in \cG} \frac{1}{|Q|^\alpha} \left( \int_Q f \right)^\gamma
		\le \sum_{k=0}^\infty \sum_{\substack{Q \in \cG \\ |Q|= m2^{dk}}} 
			\frac{1}{|Q|^\alpha} \Lambda^\gamma
		\le \sum_{k=0}^\infty \frac{2^d \Lambda^\gamma}{m^\alpha 2^{dk\alpha}}
		= \frac{2^d \Lambda^\gamma}{m^\alpha} \frac{1}{1-2^{-d\alpha}}
		= \frac{C_{d,\alpha,\beta}}{\Lambda^{\beta-\gamma}} 
			\frac{\Lambda^\beta}{2^{d\beta} m^\alpha},
	\end{align*}
	which proves \eqref{eq:covering-bound} on the group $\cG$, 
	and since every cube $Q \in \cQ$ belongs to some group, 
	therefore also the lemma.
\end{proof}

One obtains from this also the following version, which was proven directly
and with a slightly different constant in \cite[Lemma~12]{LunNamPor-16}:

\begin{corollary}[Weaker exclusion version]\label{thm:covering-q}
	Under the same conditions as in Lemma~\ref{lem:covering},
	the cube $Q_0$ may be partitioned into a finite collection $\cQ$ of 
	disjoint sub-cubes $Q \in \cQ$, such that:
	\begin{itemize}
	\item For all $Q \in \cQ$,
	$$
		\int_{Q} f \le \Lambda.
	$$
	\item For any $q \ge 0$, 
	\begin{equation}\label{eq:covering-bound-q}
		\sum_{Q \in \cQ} \frac{1}{|Q|^{\alpha}} \left( 
			\left[\int_{Q} f \ - q\right]_+ - b \int_{Q} f 
			\right) \ge 0,
	\end{equation}
	where 
	\begin{equation}\label{eq:covering-const-q}
		b := 1 - \frac{q}{\Lambda} \frac{2^{d(\alpha+2)}}{2^{d\alpha}-1}.
	\end{equation}
	\end{itemize}
\end{corollary}
\begin{proof}
	Note that for any collection $\cQ$ of cubes $Q$
	\begin{align}
		\sum_{Q \in \cQ} \frac{1}{|Q|^{\alpha}} \left( 
			\left[\int_{Q} f \ - q\right]_+ - b \int_{Q} f \right)
		&\ge \sum_{Q \in \cQ} \frac{1}{|Q|^{\alpha}} \left( 
			(1-b) \int_{Q} f \ - q \right) \nonumber\\
		&= (1-b)\sum_{Q \in \cQ} \frac{1}{|Q|^{\alpha}} \left( 
			\int_{Q} f \ - \frac{q/\Lambda}{1-b} \Lambda \right).
		\label{eq:covering-bound-q-comp}
	\end{align}
	By choosing $\cQ$ as in Lemma~\ref{lem:covering}, 
	with $\beta=1$ and $\gamma=0$, 
	the r.h.s. of \eqref{eq:covering-bound-q-comp} 
	is non-negative if $b \le 1$ and if 
	$$
		\frac{q/\Lambda}{1-b} = C_{d,\alpha,1}^{-1} 
		= \frac{2^{d\alpha}-1}{2^{d(\alpha+2)}},
	$$
	that is \eqref{eq:covering-const-q}.
	This proves \eqref{eq:covering-bound-q}
	(the inequality is trivially true for $b \le 0$).
\end{proof}

\begin{exc}
	Extend the covering lemma to a split into $k^d$ subcubes, with $k \ge 2$.
	(One practical usefulness of this version for $k=3$ (or any $k$ odd) 
	is that there is
	always a cube in $\cQ$ whose centerpoint coincides with the centerpoint
	of $Q_0$.)
\end{exc}

\subsubsection{Local proof of LT for fermions}\label{sec:LT-fermions}

Let us now use the above covering lemma to prove the fermionic LT inequality of 
Theorems~\ref{thm:LT-kinetic} and \ref{thm:LT-q}
directly by means of the local formulations of the uncertainty and the exclusion
principle, although with a weaker constant.

	Namely,
	take $\Psi \in H^1_{\asym(q)}(\R^{dN})$, and recall that for any partition
	$\eP$ of $\R^d$ into cubes $Q$ we have by the local uncertainty principle
	\eqref{eq:local-uncertainty-partition} that
	\begin{equation}\label{eq:LT-local-uncertainty}
		T[\Psi] = \sum_{Q \in \eP} T^Q[\Psi]
		\ge \sum_{Q \in \eP} \left( 
			C_1 \frac{\int_Q \varrho_\Psi^{1 + 2/d}}{(\int_Q \varrho_\Psi)^{2/d}}
			- C_2 \frac{\int_Q \varrho_\Psi}{|Q|^{2/d}}
			\right),
	\end{equation}
	for some positive constants $C_1,C_2 > 0$,
	and that the local exclusion principle of Lemma~\ref{lem:local-exclusion} 
	yields
	\begin{equation}\label{eq:LT-local-exclusion}
		T^Q[\Psi] \ge 
		\frac{\pi^2}{|Q|^{2/d}} \left[ \int_Q \varrho_\Psi \ - q \right]_+.
	\end{equation}
	By the GNS inequality on the full space,
	$$
		T[\Psi] \ge G_d N^{-2/d} \int_{\R^d} \varrho_\Psi^{1+2/d},
	$$
	we see that $\int_{\R^d \setminus Q_0} \varrho_\Psi^{1+2/d}$
	can be made arbitrarily small by taking some large enough cube $Q_0$,
	and thus in a standard approximation argument we may in fact assume
	$\supp \Psi \subseteq Q_0^N$.
	
	We then take a partition $\eP = \cQ$ of $Q_0$ according to Corollary~\ref{thm:covering-q},
	and combine the energies \eqref{eq:LT-local-uncertainty} and \eqref{eq:LT-local-exclusion}
	with $T = \eps T + (1-\eps)T$ and an arbitrary $\eps \in [0,1]$,
	to obtain
	\begin{align*}
		T[\Psi] 
		&\ge \sum_{Q \in \cQ} \left( 
			\eps C_1 \frac{\int_Q \varrho_\Psi^{1 + 2/d}}{(\int_Q \varrho_\Psi)^{2/d}}
			- \eps C_2 \frac{\int_Q \varrho_\Psi}{|Q|^{2/d}}
			+ (1-\eps) \frac{\pi^2}{|Q|^{2/d}} \left[ \int_Q \varrho_\Psi \ - q \right]_+
			\right) \\
		&\ge \sum_{Q \in \cQ} \left( 
			\eps C_1 \frac{\int_Q \varrho_\Psi^{1 + 2/d}}{\Lambda^{2/d}}
			+ \left( (1-\eps) \pi^2 b - \eps C_2 \right)
				\frac{\int_Q \varrho_\Psi}{|Q|^{2/d}}
			\right),
	\end{align*}
	with $b = 1 - \frac{4}{3}4^d q/\Lambda$ (here $\alpha = 2/d$).
	Taking $\Lambda = \frac{8}{3}4^d q$ so that $b=1/2$,
	and $\eps = \pi^2/(4C_2)$, the last term becomes nonnegative and
	$$
		T[\Psi] \ge C q^{-2/d} \sum_Q \int_Q \varrho_\Psi^{1 + 2/d}
		= C q^{-2/d} \int_{\R^d} \varrho_\Psi^{1 + 2/d},
	$$
	which proves the theorem.
\qed

\subsubsection{Local proof of LT for inverse-square repulsion}\label{sec:LT-bosons}

Using the above local approach we may also prove the Lieb--Thirring inequality~\eqref{eq:LT-interaction}
for the repulsive pair potential $W_\beta(\bx) = \beta|\bx|^{-2}$.
In this case the natural form domain to be considered is
\begin{equation}\label{eq:BLT-domain}
	\cH^1_{d,N} := \left\{ \Psi \in H^1(\R^{dN}) : 
		\sum_{1 \le j<k \le N} \int_{\R^{dN}} \frac{|\Psi(\sx)|^2}{|\bx_j-\bx_k|^2} \,d\sx < \infty
		\right\},
\end{equation}
and $0 \le T[\Psi] + W_\beta[\Psi] < \infty$ for such $\Psi \in \cH^1_{d,N}$.
The following theorem was proved 
in \cite{LunPorSol-15} 
(see also \cite{LunSol-14} concerning $d=1$ 
and the Calogero--Sutherland model\index{Calogero--Sutherland model};
cf.\ Section~\ref{sec:uncert-Hardy-many})
and generalized to fractional kinetic energy operators in \cite{LunNamPor-16}:

\begin{theorem}[BLT with inverse-square repulsion]\label{thm:BLT}
	For any $d\ge1$, $\beta>0$ there exists a constant $K_d(\beta)>0$
	such that for any $N \ge 1$ and $\Psi \in \cH^1_{d,N}$, 
	the exclusion-kinetic energy inequality\index{exclusion-kinetic energy inequality}
	\begin{equation}\label{eq:BLT-kinetic}
		T[\Psi] + W_\beta[\Psi]
		\ge K_d(\beta) \int_{\R^d} \varrho_\Psi^{1+2/d}
	\end{equation}
	holds. Furthermore, for any external one-body potential $V\colon \R^d \to \R$,
	\begin{equation}\label{eq:BLT-Schroedinger}
		T[\Psi] + V[\Psi] + W_\beta[\Psi]
		\ge -L_d(\beta) \int_{\R^d} |V_-|^{1+d/2},
	\end{equation}
	with the correspondence \eqref{eq:L-K-relation} between the constants.
\end{theorem}
\begin{remark}
	The domain $\cH^1_{d,N}$ imposes no symmetry on the wave function $\Psi$,
	but in fact it may be shown \cite[Corollary~3.1]{LieSei-09}
	that the minimizers of the l.h.s.\ of
	\eqref{eq:BLT-kinetic} respectively \eqref{eq:BLT-Schroedinger} 
	must nevertheless be positive and symmetric, i.e.\ \emph{bosonic}.
	We will therefore refer to such an inequality as an (interacting)
	\keyword[bosonic Lieb--Thirring inequality]{bosonic Lieb--Thirring (BLT) inequality}.
\end{remark}
\begin{remark}\label{rem:BLT-constant}
	The behavior of the optimal constant $K_d(\beta)$ in \eqref{eq:BLT-kinetic}
	as a function of $\beta$ was discussed in \cite[Section~3.5]{LunNamPor-16}.
	One has for all $d \ge 1$ that $K_d(\beta)$ is monotone increasing and concave,
	and
	$$
		C_d \min\{1,\beta^{2/d}\} \le K_d(\beta) \le G_d,
	$$
	for some constant $C_d > 0$, while for
	$$
		\begin{array}{ll}
			d=1: & \lim_{\beta \to 0} K_1(\beta) = K_1 > 0, \\
			d=2: & \lim_{\beta \to 0} K_2(\beta) = 0, \\
			d\ge3: & K_d(\beta) \sim \beta^{2/d} \ \text{as} \ \beta \to 0,
		\end{array}
	$$
	where $K_1$ is the usual \emph{fermionic} LT constant for $d=1$.
	It is furthermore conjectured that
	$\lim_{\beta \to \infty} K_d(\beta) = G_d$ 
	for all $d \ge 1$, and hence that
	the conjecture on the optimal constant in the fermionic LT 
	(Remark~\ref{rem:LT-conjecture}) for $d=1$
	is equivalent to proving that $K_1(\beta)$ is constant in $\beta$.
\end{remark}
\begin{proof}
	We proceed similarly to the previous proof,
	starting from the same formulation of the local uncertainty principle
	\eqref{eq:LT-local-uncertainty}.
	However, \eqref{eq:LT-local-exclusion} is now replaced by
	Lemma~\ref{lem:local-exclusion-bosons}:
	\begin{equation}\label{eq:BLT-local-exclusion}
		(T+W_\beta)^Q[\Psi] \ \ge \ 
		\frac{1}{2} e_2(|Q|;W_\beta) \left( \int_Q \varrho_\Psi \ - 1 \right)_+,
	\end{equation}
	with $e_2(|Q|;W_\beta) \ge \beta d^{-1} |Q|^{-2/d}$ by \eqref{eq:stupid-bound}.
	Thus, we may just take $q=1$ and replace $\pi^2$ with $\beta/(2d)$
	in the previous proof, to obtain
	\begin{align*}
		T[\Psi] + W_\beta[\Psi]
		&\ge \sum_{Q \in \cQ} \left( 
			\eps C_1 \frac{\int_Q \varrho_\Psi^{1 + 2/d}}{\Lambda^{2/d}}
			+ \left( (1-\eps) \frac{\beta b}{2d} - \eps C_2 \right)
				\frac{\int_Q \varrho_\Psi}{|Q|^{2/d}}
			\right)
		\ge \eps C \int_{\R^d} \varrho_\Psi^{1 + 2/d},
	\end{align*}
	with $\Lambda$ fixed and $\eps \sim \beta$ 
	as $\beta \to 0$.
	
	In order to get an improved dependence for small $\beta$,
	one may instead use the convexity in $n$ of the bound
	$e_n(|Q|;W_\beta) \ge \frac{\beta n(n-1)_+}{2d|Q|^{2/d}}$
	to replace \eqref{eq:BLT-local-exclusion} by the improved local exclusion
	principle
	\begin{equation}\label{eq:BLT-local-exclusion-improved}
		(T+W_\beta)^Q[\Psi] \ \ge \ 
		\frac{\beta}{2d|Q|^{2/d}} \left( 
			\left( \int_Q \varrho_\Psi \right)^2
			- \left( \int_Q \varrho_\Psi \right) \right)_+.
	\end{equation}
	Then, by Lemma~\ref{lem:covering} with 
	$\alpha=2/d$, $\beta=2$, and $\gamma=1$,
	\begin{align*}
		T[\Psi] + W_\beta[\Psi]
		&\ge \sum_{Q \in \cQ} \left( 
			\eps C_1 \frac{\int_Q \varrho_\Psi^{1 + 2/d}}{\Lambda^{2/d}}
			+ \left( (1-\eps) \frac{\beta\Lambda}{2dC_{d,2/d,2}} 
				- (1-\eps)\frac{\beta}{2d} - \eps C_2 \right)
				\frac{\int_Q \varrho_\Psi}{|Q|^{2/d}}
			\right) \\
		&\ge \eps C_1 \Lambda^{-2/d} \int_{\R^d} \varrho_\Psi^{1 + 2/d},
	\end{align*}
	with suitable fixed $\eps>0$ and $\Lambda \sim \beta^{-1}$ as $\beta \to 0$.
\end{proof}

\subsubsection{Local proof of LT for anyons}\label{sec:LT-anyons}

We may furthermore extend the Lieb--Thirring inequality straightforwardly to anyons 
in two dimensions using their local exclusion principle given in 
Lemma~\ref{lem:local-exclusion-anyons} \cite{LunSei-17}.
Recall the periodization of the statistics parameter 
$\alpha_2 = \min_{q \in \Z}|\alpha-2q| \in [0,1]$.

\begin{theorem}[LT for anyons]\label{thm:LT-anyons}
	There exists a constant $K_2>0$ such that for any $\alpha \in \R$,
	any $N \ge 1$ and $\Psi \in \cQ(\hT_\alpha)$, 
	the exclusion-kinetic energy inequality\index{exclusion-kinetic energy inequality}
	\begin{equation}\label{eq:LT-anyons-kinetic}
		T_\alpha[\Psi]
		\ge K_2 \alpha_2 \int_{\R^2} \varrho_\Psi^{2}
	\end{equation}
	holds. Furthermore, for any external one-body potential $V\colon \R^2 \to \R$,
	\begin{equation}\label{eq:LT-anyons-Schroedinger}
		T_\alpha[\Psi] + V[\Psi]
		\ge -L_2 \alpha_2^{-1} \int_{\R^d} |V_-|^{2},
	\end{equation}
	with the correspondence \eqref{eq:L-K-relation} between the constants
	$K_2$ and $L_2$.
\end{theorem}
\begin{proof}
	The proof goes through exactly as in Section~\ref{sec:LT-fermions},
	replacing $\pi^2$ in \eqref{eq:LT-local-exclusion} with 
	$c(\alpha) \ge \alpha_2/12$ from Lemma~\ref{lem:local-exclusion-anyons}.
\end{proof}

\subsubsection{Local LT with semiclassical constant}\label{sec:LT-semiclassical}

Very recently, Nam has considered improvements w.r.t. the constant in the local 
approach, getting arbitrarily close to the semiclassical one,
to the cost of a gradient error term \cite{Nam-18}.
Namely, using the Weyl asymptotics on cubes one may first prove the following:

\begin{lemma}[\keyword{Local density approximation}]\label{lem:LT-LDA}
	Given any $\Psi \in H^1_{\asym}((\R^d)^N)$
	and $d$-cube $Q_0$ such that $\int_{Q_0} \varrho_\Psi \ge \Lambda > 0$,
	there exists a partition $\cQ$ of $Q_0$ into sub-cubes $Q$
	such that $\int_Q \varrho_\Psi \le \Lambda$
	and
	\begin{equation}\label{eq:LT-LDA}
		\int_{\R^{dN}} |\nabla \Psi|^2
		\ge 
			K_d^\cl \bigl(1-C_d\Lambda^{-1/d}\bigr) 
			\sum_{Q \in \cQ} |Q| \left( \frac{\int_Q \varrho_\Psi}{|Q|} \right)^{1+2/d},
	\end{equation}
	where $C_d$ is a constant depending only on $d$.
\end{lemma}
\begin{proof}
	The Weyl asymptotics \eqref{eq:Weyl-asymptotics} for the sum of Laplacian 
	eigenvalues on a cube $Q$ can be rigorously justified with a uniform 
	lower bound (see e.g. \cite{Kroeger-94}) 
	$$
		e_N(|Q|;\asym) = \sum_{k=0}^{N-1} \lambda_k(-\Delta_Q^\eN)
		\ge \frac{K_d^\cl}{|Q|^{2/d}} \left( N^{1+2/d} - C N^{1+1/d} \right)_+
	$$
	for some constant $C>0$ and all $N \in \N$. 
	Using that the r.h.s. is convex in $N > 0$, one may then also prove the
	following improvement of the local exclusion principle of 
	Lemma~\ref{lem:local-exclusion}:
	\begin{equation}\label{eq:local-exclusion-Weyl}
		T^Q[\Psi] \ \ge \ 
		\frac{K_d^\cl}{|Q|^{2/d}} \left( \left(\int_Q \varrho_\Psi\right)^{1+2/d} 
			- C\left(\int_Q \varrho_\Psi\right)^{1+1/d} \right)_+.
	\end{equation}
	Now, applying Lemma~\ref{lem:covering} in the second term
	with $\alpha=2/d$, $\beta=1+2/d$ and $\gamma=1+1/d$, one obtains
	$$
		T[\Psi] \ge \sum_{Q \in \cQ} T^Q[\Psi] \ge 
		\sum_{Q \in \cQ} \frac{K_d^\cl}{|Q|^{2/d}} \left( \left(\int_Q \varrho_\Psi\right)^{1+2/d} 
			- C \frac{C_{d,\alpha,\beta}}{\Lambda^{1/d}} 
			\left(\int_Q \varrho_\Psi\right)^{1+2/d} \right),
	$$
	for the corresponding partition $\cQ$ of $Q_0$, which proves \eqref{eq:LT-LDA}.
\end{proof}

Lemma~\ref{lem:LT-LDA} 
is a local density approximation for the Thomas-Fermi energy.
In fact, Nam proved the following consequence, which shows that if
the density is sufficiently constant then one indeed obtains the Thomas-Fermi 
energy as a lower bound, with the semiclassical constant:

\begin{theorem}[\keyword{LT with gradient correction}]\label{thm:LT-grad}
	There exists a constant $C_d>0$ s.t. the inequality 
	\begin{equation}\label{eq:LT-grad}
		\int_{\R^{dN}} |\nabla \Psi|^2
		\ge (1-\eps) K_d^\cl \int_{\R^d} \varrho_\Psi^{1+2/d}
			- \frac{C_d}{\eps^{3+4/d}} \int_{\R^d} \left| \nabla \sqrt{\varrho_\Psi} \right|^2
	\end{equation}
	holds for any $\eps > 0$ and
	$\Psi \in H^1_\asym((\R^d)^N)$.
\end{theorem}

\subsection{Some direct applications of LT}\label{sec:LT-apps}

	As a direct application of the above kinetic energy inequalities
	one may consider the ground-state energy of a system 
	of $N$ particles with a given external potential $V$.
	Namely, given a bound of the form \eqref{eq:LT-general-K},
	one may estimate the $N$-body energy in terms of the density,
	\begin{equation}\label{eq:LT-density-functional}
		T[\Psi] + V[\Psi] + W[\Psi] 
		\ge \int_{\R^d} \left( K_d(W) \varrho_\Psi^{1+2/d}(\bx) + V \varrho_\Psi(\bx) \right) d\bx
		=: \tilde\cE[\varrho_\Psi],
	\end{equation}
	and therefore
	$$
		E_0(N) \ge \inf \left\{ \tilde\cE[\varrho] : \varrho \in L^{1+2/d}(\R^d;\R_+),
		\ \int_{\R^d} \varrho = N \right\}.
	$$
	Another bound for $E_0$ is given directly in terms of an integral of $V_-$ 
	by \eqref{eq:LT-general-V},
	but note that the above bound is also valid for arbitrary, 
	even non-negative, one-body potentials $V$.

\begin{example}[{\keyword[homogeneous gas]{The homogeneous gas}}]
	We model the case of a homogeneous gas on a finite domain 
	$\Omega \subseteq \R^d$ by formally taking 
	for the one-body potential a flat external trap
	with infinite walls,
	$$
		V(\bx)
		= \left\{ \begin{array}{ll}
			0, 		& \text{on $\Omega$,} \\
			+\infty,	& \text{on $\Omega^c$.} 
			\end{array}\right.
	$$
	More precisely, we take Dirichlet boundary conditions,
	$\Psi \in H^1_0(\Omega^N)$.
	Then also $\supp \varrho_\Psi \subseteq \Omega$, 
	$\int_\Omega \varrho_\Psi = N$,
	and since we have by H\"older that
	$$
		\tilde\cE[\varrho] 
		= K_d(W) \int_{\Omega} \varrho^{1+2/d}
		\ge K_d(W) |\Omega|^{-2/d} \left( \int_\Omega \varrho \right)^{1+2/d}
	$$
	for any $\varrho \in L^1(\Omega;\R_+)$,
	we then obtain by \eqref{eq:LT-density-functional},
	for the ground state energy per particle of the homogeneous gas
	with density $\rho = N/|\Omega|$,
	$$
		\frac{E_0(N)}{N} 
		= \inf_{\substack{\Psi \in H^1_0(\Omega^N) \\ \|\Psi\| = 1}} 
			\frac{1}{N}\left( T[\Psi] + V[\Psi] + W[\Psi] \right) 
		\ge K_d(W) \rho^{2/d}.
	$$
	\index{free Fermi gas}%
	In the case of fermions (compare Example~\ref{exmp:exclusion-gas})
	this lower bound will exactly match the correct energy
	given by the Weyl asymptotics \eqref{eq:Weyl-asymptotics},
	even including the conjectured optimal constant $K_d(\asym) \stackrel{\text{conj.}}{=} K_d^\cl$ 
	in higher dimensions $d \ge 3$ (see Remark~\ref{rem:LT-conjecture}).
\end{example}

\begin{exc}[Harmonic traps]
	Apply the bound \eqref{eq:LT-density-functional} to a harmonic oscillator 
	potential\index{harmonic oscillator potential} 
	$V_\textup{osc}(\bx) = \frac{1}{4}\omega^2 |\bx|^2$.
	Show that 
	$\varrho_{\mathrm{min}}(\bx) = K_d(W)^{-d/2} \left(\frac{d}{d+2}\right)^{d/2} \left( \lambda - \frac{\omega^2}{4}|\bx|^2 \right)_+^{d/2}$ 
	for some constant $\lambda = \lambda(d,N) \sim N^{1/d}$.
	Compute a lower bound $\tilde\cE[\varrho_{\min}] \sim N^{1+1/d}$
	to the energy,
	and compare to Exercise~\ref{exc:fermions-harm-osc}.
\end{exc}

\section{The stability of matter\lect{ [15,16]}}\label{sec:stability}
	\index{stability}

	Recall from Section~\ref{sec:mech-QM-matter} that we consider the following
	many-body Hamiltonian for matter 
	\index{model of matter}
	consisting of $N$ particles of one type (`electrons') 
	and $M$ particles of another type (`nuclei') moving in $\R^3$:
	\begin{equation}\label{eq:matter-Hamiltonian}
		\hH^{N,M} := \sum_{j=1}^N 
			\frac{\hbar^2}{2m_e} (-\Delta_{\bx_j})
			+ \sum_{k=1}^M 
			\frac{\hbar^2}{2m_n} (-\Delta_{\bR_k})
			+ W_\sC(\sx,\sR),
	\end{equation}
	with the full Coulomb interaction 
	\begin{equation}\label{eq:matter-Coulomb}
		W_\sC(\sx,\sR) := \sum_{1 \le i<j \le N} \frac{1}{|\bx_i-\bx_j|}
			- \sum_{j=1}^N \sum_{k=1}^M \frac{Z}{|\bx_j-\bR_k|}
			+ \sum_{1 \le k<l \le M}  \frac{Z^2}{|\bR_k-\bR_l|}.
	\end{equation}
	As usual we assume that the electrons have charge $-1$ 
	and the nuclei charge $Z > 0$, and put $\hbar = 1$ by a change of mass scale.
	One may also have different masses and charges for the nuclei
	but we will stick to this technically simplifying assumption here.
	Furthermore, we could also consider the nuclei to be fixed 
	at the positions $\sR = (\bR_k)_{k=1}^M$
	by formally taking 
	$m_n = +\infty$, or just remove their kinetic energy terms in 
	\eqref{eq:matter-Hamiltonian}:
	\begin{equation}\label{eq:matter-Hamiltonian-fix}
		\hH^N(\sR) := \sum_{j=1}^N \frac{1}{2m_e} 
			(-\Delta_{\bx_j}) + W_\sC(\sx,\sR).
	\end{equation}
	Since $\hH^{N,M} \ge \hH^N(\sR)$ as quadratic forms on $H^1(\R^{d(N+M)})$, 
	obtaining a lower bound for \eqref{eq:matter-Hamiltonian-fix}
	which is uniform in $\sR$ will then also be valid as a lower bound for
	$\hH^{N,M}$.

\subsection{Stability of the first kind}\label{sec:stability-first}

	That the Hamiltonian operators
	\eqref{eq:matter-Hamiltonian} and \eqref{eq:matter-Hamiltonian-fix} 
	are bounded from below, i.e. stability of the \emph{first} kind in the 
	terminology of Definition~\ref{def:stability},
	is a simple consequence of the uncertainty principle.
	Namely, recall the stability of the single hydrogenic atom,
	\begin{equation}\label{eq:first-stability-hydrogenic}
		-\frac{1}{2m}\Delta_{\R^3} - \frac{Z}{|\bx|} 
		\ \ge \ -\frac{1}{2}mZ^2,
	\end{equation}
	which is the sharp lower bound from the ground-state solution \eqref{eq:hydrogenic-groundstate}
	or, with just a slightly worse constant,
	from Hardy \eqref{eq:hydrogenic-bound-Hardy}
	or Sobolev \eqref{eq:GNS-hydrogenic-minimization}.

\begin{theorem}[{\keyword[stable of the first kind]{Stability of the first kind}}]
	For any number of particles $N,M \ge 1$, 
	for any positions of the nuclei $\sR \in \R^{3M}$, 
	and for any mass $m > 0$ and charge $Z > 0$, we have
	\begin{equation}\label{eq:first-stability-bound-fix}
		\hH^N(\sR) 
		\ \ge \ -\frac{1}{2} m_e Z^2 N M^2,
	\end{equation}
	with respect to the form domain $H^1(\R^{3N})$ (with unrestricted symmetry),
	from which $\hH^{N}(\sR)$ extends to a bounded-from-below 
	(uniformly in $\sR$)
	self-adjoint operator on $\cH^N = L^2(\R^{3N})$.
	Hence,
	\begin{equation}\label{eq:first-stability-bound}
		\hH^{N,M} \ge \inf_{\sR \in \R^{3M}} \hH^N(\sR) 
		\ge \ -\frac{1}{2} m_e Z^2 N M^2,
	\end{equation}
	in the sense of forms on $H^1(\R^{3(N+M)})$,
	from which $\hH^{N,M}$ extends to a bounded-from-below 
	self-adjoint operator on $\cH^{N,M} = L^2(\R^{3(N+M)})$.
\end{theorem}
\begin{proof}
	Let us simply throw away the positive terms in $W_\sC$ and write
	in terms of forms
	\begin{align*}
		\hH^N(\sR) \ge \sum_{j=1}^N \sum_{k=1}^M \left[
			\frac{1}{2m_e M}(-\Delta_{\bx_j}) - \frac{Z}{|\bx_j-\bR_k|} \right]
			\ge -\sum_{j=1}^N \sum_{k=1}^M \frac{1}{2} m_e MZ^2
			= -\frac{1}{2} m_e Z^2 NM^2,
	\end{align*}
	by \eqref{eq:first-stability-hydrogenic}.
	This proves \eqref{eq:first-stability-bound-fix}, and by the inequality
	$$
		\inp{\Psi,\hH^{N,M}\Psi}_{L^2(\R^{3(N+M)})} \ge 
			\int_{\R^{3M}} \inp{\Psi(\slot,\sR),\hH^N(\sR)\Psi(\slot,\sR)}_{L^2(\R^{3N)}} d\sR
	$$
	applied on arbitrary $\Psi \in H^1(\R^{3(N+M)})$,
	also \eqref{eq:first-stability-bound}.
\end{proof}

\subsection{Some electrostatics}\label{sec:stability-electro}

	In order to improve the above bound to a linear one in $N+M$,
	i.e.\ prove stability of the second kind, we also need some important results
	on electrostatics.
	See \cite{Seiringer-10} or \cite{LieSei-09} for now...

\subsubsection{Coulomb interactions and Newton's theorem}\label{sec:stability-electro-Newton}

\subsubsection{Baxter's inequality}\label{sec:stability-electro-Baxter}

	The main result that we need for the proof of stability is the following
	inequality \cite{Baxter-80}, which replaces the $N^2 + NM + M^2$ terms of the full Coulomb 
	interaction \eqref{eq:matter-Coulomb} by only $N+M$ nearest-neighbor terms
	\cite[Theorem~5.4]{LieSei-09}:

\begin{theorem}[\keyword{Baxter's inequality}]\label{thm:Baxter}
	For any $\sx \in \R^{3N}$, $\sR \in \R^{3M}$, and $Z \ge 0$ we have
	\begin{equation}\label{eq:Baxter}
		W_\sC(\sx,\sR) \ge -(2Z+1) \sum_{j=1}^N \frac{1}{\dist(\bx_j,\sR)}
			+ \frac{Z^2}{4} \sum_{k=1}^M \frac{1}{\dist(\bR_k,\sR_{\setminus k})}.
	\end{equation}
\end{theorem}

\subsection{Proof of stability of the second kind}\label{sec:stability-proof}

	We will now apply the Lieb--Thirring inequality 
	in the simplest available proof of stability for fermions,
	due to Solovej \cite{Solovej-06b} (see also \cite[Section~7.2]{LieSei-09}),
	as well as in a straightforward generalization
	for inverse-square repulsive bosons \cite{LunPorSol-15}
	in 3D.

\subsubsection{Fermions}\label{sec:stability-proof-fermions}

	We consider first fermions, possibly with $q$ different species 
	or spin states.

\begin{theorem}[\keyword{Stability for fermionic matter}]
	\label{thm:fermionic-stability}
	For any number of particles $N,M \ge 1$, 
	for any positions of the nuclei $\sR \in \R^{3M}$, 
	and for any mass $m > 0$ and charge $Z > 0$, we have
	\begin{equation}\label{eq:stability-bound-fermions}
		\hH^N(\sR) := \sum_{j=1}^N \frac{1}{2m} 
			(-\Delta_{\bx_j}) + W_\sC(\sx,\sR)
		\ \ge \ -1.073 \,q^{2/3} m(2Z+1)^2 (N+M),
	\end{equation}
	where the domain of the operator $\hH^N(\sR)$ is defined w.r.t. 
	$q$-\emph{antisymmetric}
	functions with the form domain $H^1_{\asym(q)}(\R^{3N})$.
\end{theorem}
\begin{remark}
	Note that the result concerns only the symmetry type of one of the 
	species of particles involved (here the electrons), 
	and that the other particles at $\bR_k$ may thus 
	be of any type: quantum mechanical bosons, distinguishable,
	or even fixed classical particles.
\end{remark}
\begin{proof}
	We only need the simpler lower bound given by Baxter's inequality,
	Theorem~\ref{thm:Baxter},
	$$
		W_\sC(\sx,\sR) \ge -(2Z+1) \sum_{j=1}^N \frac{1}{\dist(\bx_j,\sR)}.
	$$
	The trick is then just to add and subtract a constant $b>0$ to the 
	Hamiltonian, writing for simplicity
	\begin{align*}
		\hH^N(\sR) &\ge \frac{1}{2m}\sum_{j=1}^N \left[
			-\Delta_{\bx_j} - 2m(2Z+1) \left( \frac{1}{\dist(\bx_j,\sR)} - b \right) \right]
			- (2Z+1)Nb \\
		&\ge -\frac{1}{2m} qL_3 \bigl( 2m(2Z+1) \bigr)^{5/2} 
			\int_{\R^3} 	\left[ \frac{1}{\dist(\bx,\sR)} - b \right]_+^{5/2} d\bx
			- (2Z+1)Nb,
	\end{align*}
	where in the second step we crucially used the domain of the operator 
	and the fermionic Lieb--Thirring inequality of Corollary~\ref{thm:LT-q}.
	Furthermore, using that
	$$
		\left[ \frac{1}{\dist(\bx,\sR)} - b \right]_+^{5/2}
		= \max_{1 \le k \le M} \left[ \frac{1}{|\bx - \bR_k|} - b \right]_+^{5/2}
		\le \sum_{k=1}^M \left[ \frac{1}{|\bx - \bR_k|} - b \right]_+^{5/2},
	$$
	one may bound
	$$
		\int_{\R^3} \left[ \frac{1}{\dist(\bx,\sR)} - b \right]_+^{5/2} d\bx
		\le M \int_{|\bx| \le 1/b} \left( \frac{1}{|\bx|} - b \right)^{5/2} d\bx
		= \frac{5\pi^2}{4} Mb^{-1/2}.
	$$
	Hence,
	$$
		\hH^N(\sR) \ge -qL_3 \frac{5\pi^2}{4} (2m)^{3/2} (2Z+1)^{5/2} M b^{-1/2} 
			- (2Z+1) Nb,
	$$
	and after optimizing in $b>0$,
	$$
		\hH^N(\sR) \ge -\frac{3}{2} (5\pi^2 L_3)^{2/3} q^{2/3} m(2Z+1)^2 M^{2/3} N^{1/3}
		\ge -1.073 \,q^{2/3} m(2Z+1)^2 (N+M),
	$$
	where in the final step we used the best presently known value for $L_3$
	(see Remark~\ref{rem:LT-conjecture}) and Young's inequality~\eqref{eq:Young}.
\end{proof}

\subsubsection{Inverse-square repulsive bosons}\label{sec:stability-proof-bosons}

	We now consider a system of $N+M$ charged particles, 
	subject to the usual Coulomb interaction $W_\sC$,
	but $N$ of which also experience an additional inverse-square 
	repulsive interaction with coupling parameter $\beta$.
	If we do not impose any symmetry conditions on the particles,
	the ground state is known to be bosonic. 
	However, thanks to the bosonic Lieb--Thirring inequality of Theorem~\ref{thm:BLT},
	we nevertheless do have stability for $\beta > 0$:

\begin{theorem}[\keyword{Stability for inverse-square repulsive bosons}]
	\label{thm:invsquare-stability}
	For any coupling strength $\beta > 0$ 
	there exists a positive constant $C(\beta) = CL_3(\beta)^{2/3} > 0$,
	such that for any number of particles $N,M \ge 1$, 
	for any positions $\sR \in \R^{3M}$, 
	and for any mass $m > 0$ and charge $Z > 0$, 
	we have
	\begin{align}\label{eq:stability-bound-invsquare}
		\hH^N_\beta(\sR) &:= \frac{1}{2m} \left[ \sum_{j=1}^N (-\Delta_{\bx_j}) 
			+ \sum_{1 \le j<k \le N} \frac{\beta}{|\bx_j - \bx_k|^2} \right]
			+ W_\sC(\sx,\sR) \nonumber\\
		&\ge \ -C(\beta) \,m(2Z+1)^2 (N+M),
	\end{align}
	where the operator $\hH^N_\beta(\sR)$ is defined as the Friedrichs extension 
	w.r.t.\ the form domain 
	$\cH^1_{d,N}$ in \eqref{eq:BLT-domain}
	(with no symmetry assumptions on 
	the wave functions).
\end{theorem}
\begin{remark}\label{rem:BLT-instability}
	We note that since $K_d(\beta) \sim \beta^{2/d}$
	or equivalently $L_d(\beta) \sim \beta^{-1}$ as $\beta \to 0$ for $d \ge 3$
	(see Remark~\ref{rem:BLT-constant}), one has
	$C(\beta) \sim \beta^{-2/3} \to \infty$ in the limit of weak interactions.
	Taking $\beta \sim N^{-1}$ for example corresponds to a 
	\keyword{mean-field scaling} 
	of the interaction, and evidently leads to bound 
	similar to $q \sim N$ for fermions, which is too weak for stability.
	Also note that it may be more natural to 
	incorporate the mass factor $2m$ by rescaling the parameter $\beta$.
\end{remark}
\begin{proof}
	We mimic the proof of Theorem~\ref{thm:fermionic-stability},
	replacing the fermionic LT with the BLT \eqref{eq:BLT-Schroedinger} 
	of Theorem~\ref{thm:BLT}. Thus,
	\begin{align*}
		\hH^N_\beta(\sR) &\ge \frac{1}{2m} \left[ \hT + \hW_\beta + \hV \right]
			- (2Z+1)Nb \\
		&\ge -\frac{1}{2m} L_3(\beta) \int_{\R^3} |V_-|^{5/2}
			- (2Z+1)Nb,
	\end{align*}
	where
	$$
		V(\bx) := - 2m(2Z+1) \left( \frac{1}{\dist(\bx,\sR)} - b \right).
	$$
	The remainder of the proof then goes through exactly as before, 
	with the replacement
	$qL_3 \to L_3(\beta)$.
\end{proof}

\subsubsection{Two dimensions}\label{sec:stability-proof-2D}

	Consider particles that have been confined to a thin layer, 
	for example by means of a strong transverse external potential, such as
	$$
		V(x,y,z) = V_2(x,y) + C_3|z|^2, \qquad C_3 \gg 1.
	$$
	Hence they may effectively only move in the two-dimensional plane $z=0$,
	but their interactions could still be the usual Coulomb interactions 
	due to electromagnetism that propagates freely in three dimensions.
	(If electromagnetism would instead be propagating in only two dimensions, one
	could argue to instead use the logarithmic 2D analog of the Coulomb potential;
	see \cite{ManSir-06} for this case.)
	We may thus keep the interaction term $W_\sC$ as it is 
	and consider stability under the assumption that all particles are situated
	in the plane $\R^2$:
	\begin{equation}\label{eq:matter-Hamiltonian-2D}
		\hH^{N,M} := \sum_{j=1}^N 
			\frac{\hbar^2}{2m_e} (-\Delta_{\bx_j})
			+ \sum_{k=1}^M 
			\frac{\hbar^2}{2m_n} (-\Delta_{\bR_k})
			+ W_\sC(\sx,\sR),
	\end{equation}
	or, in the case of fixed `nuclei' (or impurities in the layer),
	\begin{equation}\label{eq:matter-Hamiltonian-2D-Rfix}
		\hH^N(\sR) := \sum_{j=1}^N \frac{\hbar^2}{2m_e} (-\Delta_{\bx_j})
			+ W_\sC(\sx,\sR),
	\end{equation}
	with $\sx = (\bx_1,\ldots,\bx_N) \in \R^{2N}$
	and $\sR = (\bR_1,\ldots,\bR_M) \in \R^{2M}$.

	In this case there are some additional technical complications due to
	the dimensionality, and
	one needs to use the positive terms in Baxter's inequality
	\eqref{eq:Baxter} as well;
	see e.g. \cite[Section~6.2]{LunSol-14}
	for details, leading to the following theorem:

\begin{theorem}[\keyword{Stability for fermions in 2D}]
	There exists a constant $C>0$ such that
	for any number of particles $N,M \ge 1$, 
	for any positions of the nuclei $\sR \in \R^{2M}$, 
	and for any mass $m > 0$ and charge $Z > 0$, we have
	\begin{equation}\label{eq:stability-bound-fermions-2D}
		\hH^N(\sR) := \sum_{j=1}^N \frac{1}{2m} 
			(-\Delta_{\bx_j}) + W_\sC(\sx,\sR)
		\ \ge \ -C \,q m(2Z+1)^2 (N+M),
	\end{equation}
	where the domain of the operator $\hH^N(\sR)$ is defined w.r.t. 
	$q$-\emph{antisymmetric}
	functions with the form domain $H^1_{\asym(q)}(\R^{2N})$.
\end{theorem}

	Stability has also been extended to anyons in \cite[Theorem~21]{LunSol-14}
	and \cite{LunSei-17}, i.e. replacing the fermionic kenetic energy with $\hT_\alpha$,
	where the resulting difference is just to replace $q$ in the 
	r.h.s.\ of \eqref{eq:stability-bound-fermions-2D} with $1/\alpha_2$
	as of Theorem~\ref{thm:LT-anyons}.
	In other words, one has stability for any type of anyons except for bosons
	for which $\alpha_2 = \alpha = 0$.

\begin{exc}
	Fill in the details above to prove stability for fermions in 2D. 
\end{exc}

\subsection{Instability for bosons}\label{sec:stability-instability}

	We have seen that bulk matter consisting to at least one part 
	(a nonvanishing fraction) of fermions
	is stable, and that stability also holds if the particles are bosons with 
	an additional inverse-square repulsion. However, 
	switching off this repulsion with $\beta \to 0$ according to Remark~\ref{rem:BLT-instability},
	or relaxing the Pauli principle to allow for arbitrarily many particles 
	$q = N \to \infty$ in each one-body state,
	seems to lead to \keyword{instability},
	and indeed this may be shown to be the correct conclusion. 
	
	We define two purely bosonic ground-state energies 
	for $N$ electrons and $M$ nuclei; 
	one with the nuclei being fixed, classical particles 
	but at their worst possible positions,
	$$
		E_0(N,M) := \inf_{\sR \in \R^{3M}} \inf \sigma \bigl(\hH^N(\sR)|_{\cH^{N}}\bigr),
	$$
	and one where they are true quantum particles with a finite mass $m_n > 0$,
	$$
		\tilde{E}_0(N,M) := \inf \sigma\bigl(\hH^{N,M}|_{\cH^{N,M}}\bigr).
	$$
	In the first case, 
	which was settled early on by \cite{DysLen-68,Lieb-79},
	one has the following instability result, matching the lower bound
	of Theorem~\ref{thm:fermionic-stability} with $q=N$:

\begin{theorem}[\keyword{Instability for bosons with fixed nuclei}]
	There exist constants $C_- > C_+ > 0$ (depending on $m_e$, $m_n$ and $Z$) 
	such that for any $N=M \in \N$ particles,
	$$
		-C_- N^{5/3} < E_0(N,N) < -C_+ N^{5/3}.
	$$
\end{theorem}

	However, in the case of the second definition of energy $\tilde{E}_0$, 
	it is potentially increased and thus less divergent due to
	the uncertainty principle for the quantum nuclei,
	and indeed this was shown to be the case,
	although not sufficiently so to make the system thermodynamically stable.
	The result is given
	in the following theorem which was worked out 
	in several steps and over many years 
	\cite{Dyson-67,ConLieYau-88,LieSol-04,Solovej-06}:

\begin{theorem}[\keyword{Instability for bosons with moving nuclei}]
	Let $Z=1$ and $m_e = m_n = 1$.
	There exists a constant $C > 0$ (explicitly defined) such that 
	$$
		\min_{N+M=K} \tilde{E}_0(N,M) = -CK^{7/5} + o(K^{7/5}), \quad \text{as} \ K \to \infty.
	$$
\end{theorem}

	We refer to \cite{LieSei-09,LieSeiSolYng-05} for pedagogically outlined 
	proofs of these results and for further discussions on instability.

\subsection{Extensivity of matter}\label{sec:stability-extensivity}

	Finally, we have the following formulation of the extensivity of matter,
	taken from \cite[Theorem~4.3.3]{Thirring2}, 
	and extended here also to inverse-square repulsive bosons:

\begin{theorem}[{\keyword[extensivity of the volume]{Extensivity of the volume}}]
	Consider the $N$-body operators
	$$
		\hH := \hT|_{\cH_{\asym(q)}} + \hV
		\qquad \text{or} \qquad
		\hH := \hT + \hW_\beta + \hV,
	$$
	where $\hV$ is an operator assumed to scale with the inverse length, 
	such as the Coulomb interaction $\hV = W_\sC$ with some fixed set of nuclei.
	If we have stability of the second kind, $\hH > -cN$, 
	and if $\Psi \in \cQ(\hH)$ 
	is a state with nonpositive expectation value, 
	$\langle \hH \rangle_\Psi \le 0$,
	then in this state no volume $|\Omega| \le \eps N$ contains more than
	$(4c/K_3)^{3/5} \eps^{2/5} N$ particles.
\end{theorem}
\begin{proof}
	Use that
	$$
		\frac{1}{2}\inp{\hT + \hW_\beta}_\Psi
		\le -\inp{ \frac{1}{2}(\hT+\hW_\beta) + \hV}_\Psi \le 2cN,
	$$
	by the scaling property of $\hV$.
	Then, by the (B)LT
	$$
		\inp{\hT + \hW_\beta}_\Psi \ge K_3(\beta) \int_{\R^3} \varrho_\Psi^{5/3},
	$$
	we obtain using H\"older
	$$
		\int_\Omega \varrho_\Psi
		\le \left( \int_\Omega \varrho_\Psi^{5/3} \right)^{3/5} \left( \int_\Omega 1 \right)^{2/5}
		\le \left( K_3(\beta)^{-1} 4cN \right)^{3/5} (\eps N)^{2/5}
	$$
	for the number of particles on $\Omega$.
\end{proof}
	
	This then explains why fermionic matter occupies a volume that
	grows at least linearly with the number of particles.
	See also \cite[Theorem~7.2]{LieSei-09} for more general formulations 
	of extensivity for fermions.

% ------------------  APPENDIX  --------------------

% ------------------  Bibilography  --------------------

\newpage

\def\MR#1{} %kill the MRNUMBER
%\bibliographystyle{amsalphahyperdupernd}
%\bibliography{biblio}

\begin{thebibliography}{HOHOLT08}

\bibitem[AGHKH05]{Albeverio-etal-05}
S. Albeverio, F. Gesztesy, R. H{\o}egh-Krohn, and H. Holden, \emph{Solvable
  models in quantum mechanics}, second ed., AMS Chelsea Publishing, Providence,
  RI, 2005, With an appendix by Pavel Exner. \MR{2105735}

\bibitem[AHKS77]{AlbHoeStr-77}
S. Albeverio, R. H{\o}egh-Krohn, and L. Streit, \emph{Energy forms,
  {H}amiltonians, and distorted {B}rownian paths}, J. Math. Phys. \textbf{18}
  (1977), no.~5, 907--917, \href {http://dx.doi.org/10.1063/1.523359}
  {\path{doi}}. \MR{0446236}

\bibitem[AS11]{AshSin-11}
A. Ashtekar and P. Singh, \emph{Loop quantum cosmology: a status report},
  Class. Quantum Grav. \textbf{28} (2011), no.~21, 213001, \href
  {http://dx.doi.org/10.1088/0264-9381/28/21/213001} {\path{doi}}.

\bibitem[ASWZ85]{AroSchWilZee-85}
D.~P. Arovas, R. Schrieffer, F. Wilczek, and A. Zee, \emph{Statistical
  mechanics of anyons}, Nuclear Physics B \textbf{251} (1985), 117 -- 126,
  \href {http://dx.doi.org/10.1016/0550-3213(85)90252-4} {\path{doi}}.

\bibitem[Aub75]{Aubin-75}
T. Aubin, \emph{Probl{\`e}mes isop{\'e}rim{\'e}triques et espaces de
  {S}obolev}, C. R. Acad. Sci. Paris S{\'e}r. A-B \textbf{280} (1975), no.~5,
  Aii, A279--A281. \MR{MR0407905 (53 \#11675)}

\bibitem[Bax80]{Baxter-80}
J.~R. Baxter, \emph{Inequalities for potentials of particle systems}, Illinois
  J. Math. \textbf{24} (1980), no.~4, 645--652,
  \url{http://projecteuclid.org/euclid.ijm/1256047480}. \MR{MR586803
  (82j:81065)}

\bibitem[BEL15]{BalEvLew-15}
A.~A. Balinsky, W.~D. Evans, and R.~T. Lewis, \emph{{The Analysis and Geometry
  of Hardy's Inequality}}, Universitext, Springer International Publishing,
  2015, \url{https://books.google.de/books?id=NOTHCgAAQBAJ}.

\bibitem[Ben85]{Benedicks-85}
M. Benedicks, \emph{On {F}ourier transforms of functions supported on sets of
  finite {L}ebesgue measure}, J. Math. Anal. Appl. \textbf{106} (1985), no.~1,
  180--183, \href {http://dx.doi.org/10.1016/0022-247X(85)90140-4}
  {\path{doi}}. \MR{780328}

\bibitem[Bir61]{Birman-61}
M.~S. Birman, \emph{On the spectrum of singular boundary-value problems}, Mat.
  Sb. (N.S.) \textbf{55 (97)} (1961), 125--174, English translation in: Eleven
  Papers on Analysis, AMS Transl. 53, 23--80, AMS, Providence, R.I., 1966.

\bibitem[BS92]{BorSor-92}
M. Bourdeau and R.~D. Sorkin, \emph{When can identical particles collide?},
  Phys. Rev. D \textbf{45} (1992), 687--696, \href
  {http://dx.doi.org/10.1103/PhysRevD.45.687} {\path{doi}}.

\bibitem[BVV18]{BenValVan-18}
R.~D. {Benguria}, C. {Vallejos}, and H. {Van Den Bosch},
  \emph{{Gagliardo-Nirenberg-Sobolev inequalities for convex domains in
  $\mathbb{R}^d$}}, arXiv e-prints, 2018, \href
  {http://arxiv.org/abs/1802.01740} {\path{arXiv:1802.01740}}.

\bibitem[Cal71]{Calogero-71}
F. Calogero, \emph{Solution of the one-dimensional {$N$}-body problems with
  quadratic and/or inversely quadratic pair potentials}, J. Math. Phys.
  \textbf{12} (1971), 419--436, \href {http://dx.doi.org/10.1063/1.1665604}
  {\path{doi}}. \MR{0280103 (43 \#5824)}

\bibitem[CLR17]{CorLunRou-16}
M. Correggi, D. Lundholm, and N. Rougerie, \emph{Local density approximation
  for the almost-bosonic anyon gas}, Analysis \& PDE \textbf{10} (2017),
  1169--1200, \href {http://dx.doi.org/10.2140/apde.2017.10.1169} {\path{doi}}.

\bibitem[CLY88]{ConLieYau-88}
J.~G. Conlon, E.~H. Lieb, and H.-T. Yau, \emph{The {$N\sp {7/5}$} law for
  charged bosons}, Commun. Math. Phys. \textbf{116} (1988), no.~3, 417--448.
  \MR{MR937769 (89d:81125)}

\bibitem[CX97]{CheXu-97}
J.-Y. Chemin and C.-J. Xu, \emph{Inclusions de {S}obolev en calcul de
  {W}eyl-{H}{\"o}rmander et champs de vecteurs sous-elliptiques}, Ann. Sci.
  \'Ecole Norm. Sup. (4) \textbf{30} (1997), no.~6, 719--751, \href
  {http://dx.doi.org/10.1016/S0012-9593(97)89937-5} {\path{doi}}. \MR{1476294
  (98i:35026)}

\bibitem[DFT97]{DelFigTet-97}
G. Dell'Antonio, R. Figari, and A. Teta, \emph{Statistics in space dimension
  two}, Lett. Math. Phys. \textbf{40} (1997), no.~3, 235--256, \href
  {http://dx.doi.org/10.1023/A:1007361832622} {\path{doi}}. \MR{1453763}

\bibitem[DGH99]{DoeGroHen-99}
H.-D. Doebner, W. Groth, and J.~D. Hennig, \emph{On quantum mechanics of
  {$n$}-particle systems on {$2$}-manifolds---a case study in topology}, J.
  Geom. Phys. \textbf{31} (1999), no.~1, 35--50, \href
  {http://dx.doi.org/10.1016/S0393-0440(98)00070-9} {\path{doi}}. \MR{1704809}

\bibitem[Dir67]{Dirac-64}
P.~A.~M. Dirac, \emph{Lectures on quantum mechanics}, Belfer Graduate School of
  Science Monographs Series, vol.~2, Belfer Graduate School of Science, New
  York; produced and distributed by Academic Press, Inc., New York, 1967,
  Second printing of the 1964 original. \MR{2220894}

\bibitem[DL67]{DysLen-67}
F.~J. Dyson and A. Lenard, \emph{Stability of matter. {I}}, J. Math. Phys.
  \textbf{8} (1967), no.~3, 423--434, \href
  {http://dx.doi.org/10.1063/1.1705209} {\path{doi}}.

\bibitem[DLL08]{DolLapLos-08}
J. Dolbeault, A. Laptev, and M. Loss, \emph{Lieb-{T}hirring inequalities with
  improved constants}, J. Eur. Math. Soc. (JEMS) \textbf{10} (2008), no.~4,
  1121--1126, \href {http://dx.doi.org/10.4171/JEMS} {\path{doi}}. \MR{2443931
  (2010h:35067)}

\bibitem[DR90]{DopRob-90}
S. Doplicher and J.~E. Roberts, \emph{Why there is a field algebra with a
  compact gauge group describing the superselection structure in particle
  physics}, Comm. Math. Phys. \textbf{131} (1990), no.~1, 51--107,
  \url{http://projecteuclid.org/euclid.cmp/1104200703}. \MR{1062748}

\bibitem[D{\v S}T01]{DoeStoTol-01}
H.-D. Doebner, P. {\v S}{\v t}ov{\'\i}{\v c}ek, and J. Tolar,
  \emph{Quantization of kinematics on configuration manifolds}, Rev. Math.
  Phys. \textbf{13} (2001), no.~7, 799--845, \href
  {http://arxiv.org/abs/math-ph/0104013} {\path{arXiv:math-ph/0104013}}, \href
  {http://dx.doi.org/10.1142/S0129055X0100079X} {\path{doi}}. \MR{1843854}

\bibitem[dWHL11]{deWHopLun-11a}
J. de~Woul, J. Hoppe, and D. Lundholm, \emph{Partial {H}amiltonian reduction of
  relativistic extended objects in light-cone gauge}, J. High Energy Phys.
  \textbf{2011} (2011), no.~1, 31, \href
  {http://dx.doi.org/10.1007/JHEP01(2011)031} {\path{doi}}.

\bibitem[Dys67]{Dyson-67}
F.~J. Dyson, \emph{Ground-state energy of a finite system of charged
  particles}, J. Math. Phys. \textbf{8} (1967), no.~8, 1538--1545.

\bibitem[Dys68]{Dyson-68}
F.~J. Dyson, \emph{Stability of matter}, Statistical Physics, Phase Transitions
  and Superfluidity, Brandeis University Summer Institute in Theoretical
  Physics 1966, Gordon and Breach, New York, 1968, pp.~179--239.

\bibitem[FHOLS06]{FraHofLapSol-06}
R.~L. Frank, T. Hoffmann-Ostenhof, A. Laptev, and J.~P. Solovej, \emph{Hardy
  inequalities for large fermionic systems}, unpublished, 2006.

\bibitem[Fra14]{Frank-14}
R.~L. Frank, \emph{Cwikel's theorem and the {CLR} inequality}, J. Spectr.
  Theory \textbf{4} (2014), no.~1, 1--21, \href
  {http://dx.doi.org/10.4171/JST/59} {\path{doi}}. \MR{3181383}

\bibitem[Fr{\"o}90]{Froehlich-90}
J. Fr{\"o}hlich, \emph{Quantum statistics and locality}, Proceedings of the
  {G}ibbs {S}ymposium ({N}ew {H}aven, {CT}, 1989), Amer. Math. Soc.,
  Providence, RI, 1990, pp.~89--142. \MR{1095329}

\bibitem[FS08]{FraSei-08}
R.~L. Frank and R. Seiringer, \emph{Non-linear ground state representations and
  sharp {H}ardy inequalities}, J. Funct. Anal. \textbf{255} (2008), no.~12,
  3407--3430, \href {http://dx.doi.org/10.1016/j.jfa.2008.05.015} {\path{doi}}.
  \MR{2469027}

\bibitem[FS12]{FraSei-12}
R.~L. Frank and R. Seiringer, \emph{Lieb-{T}hirring inequality for a model of
  particles with point interactions}, J. Math. Phys. \textbf{53} (2012), no.~9,
  095201, 11, \href {http://dx.doi.org/10.1063/1.3697416} {\path{doi}}.
  \MR{2905783}

\bibitem[FSW08]{FraSimWei-08}
R.~L. Frank, B. Simon, and T. Weidl, \emph{Eigenvalue bounds for perturbations
  of {S}chr\"odinger operators and {J}acobi matrices with regular ground
  states}, Comm. Math. Phys. \textbf{282} (2008), no.~1, 199--208, \href
  {http://dx.doi.org/10.1007/s00220-008-0453-1} {\path{doi}}. \MR{2415477}

\bibitem[Gen40]{Gentile-40}
G. Gentile, \emph{Osservazioni sopra le statistiche intermedie}, Il Nuovo
  Cimento \textbf{17} (1940), no.~10, 493--497, \href
  {http://dx.doi.org/10.1007/BF02960187} {\path{doi}}.

\bibitem[Gen42]{Gentile-42}
G. Gentile, \emph{Le statistiche intermedie e le propriet{\`a} dell'elio
  liquido}, Il Nuovo Cimento \textbf{19} (1942), no.~4, 109--125, \href
  {http://dx.doi.org/10.1007/BF02960192} {\path{doi}}.

\bibitem[Gir60]{Girardeau-60}
M. Girardeau, \emph{Relationship between systems of impenetrable bosons and
  fermions in one dimension}, J. Mathematical Phys. \textbf{1} (1960),
  516--523. \MR{0128913 (23 \#B1950)}

\bibitem[Hal91]{Haldane-91}
F.~D.~M. Haldane, \emph{{``Fractional statistics'' in arbitrary dimensions: A
  generalization of the Pauli principle}}, Phys. Rev. Lett. \textbf{67} (1991),
  937--940, \href {http://dx.doi.org/10.1103/PhysRevLett.67.937} {\path{doi}}.

\bibitem[Hed12]{Hedenmalm-12}
H. Hedenmalm, \emph{Heisenberg's uncertainty principle in the sense of
  {B}eurling}, J. Anal. Math. \textbf{118} (2012), no.~2, 691--702, \href
  {http://dx.doi.org/10.1007/s11854-012-0048-9} {\path{doi}}. \MR{3000695}

\bibitem[Hei27]{Heisenberg-27}
W. Heisenberg, \emph{{{\"U}ber den anschaulichen Inhalt der
  quantentheoretischen Kinematik und Mechanik}}, Z. Phys. \textbf{43} (1927),
  no.~3, 172--198, \href {http://dx.doi.org/10.1007/BF01397280} {\path{doi}}.

\bibitem[HH77]{Hof-77}
M. {Hoffmann-Ostenhof} and T. {Hoffmann-Ostenhof}, \emph{Schr{\"o}dinger
  inequalities and asymptotic behavior of the electron density of atoms and
  molecules}, Phys. Rev. A \textbf{16} (1977), no.~5, 1782--1785.

\bibitem[HOHOLT08]{HofLapTid-08}
M. Hoffmann-Ostenhof, T. Hoffmann-Ostenhof, A. Laptev, and J. Tidblom,
  \emph{{Many-particle Hardy Inequalities}}, J. London Math. Soc. \textbf{77}
  (2008), 99--114, \href {http://dx.doi.org/10.1112/jlms/jdm091} {\path{doi}}.

\bibitem[Isa94]{Isakov-94}
S.~B. Isakov, \emph{Statistical mechanics for a class of quantum statistics},
  Phys. Rev. Lett. \textbf{73} (1994), 2150--2153, \href
  {http://dx.doi.org/10.1103/PhysRevLett.73.2150} {\path{doi}}.

\bibitem[Kha05]{Khare-05}
A. Khare, \emph{{Fractional Statistics and Quantum Theory}}, 2nd ed., World
  Scientific, Singapore, 2005.

\bibitem[Kr{\"o}94]{Kroeger-94}
P. Kr{\"o}ger, \emph{Estimates for sums of eigenvalues of the {L}aplacian}, J.
  Funct. Anal. \textbf{126} (1994), no.~1, 217--227, \href
  {http://dx.doi.org/10.1006/jfan.1994.1146} {\path{doi}}. \MR{1305068}

\bibitem[Lap12]{Laptev-12}
A. Laptev, \emph{Spectral inequalities for partial differential equations and
  their applications}, Fifth International Congress of Chinese Mathematicians,
  AMS/IP Studies in Advanced Mathematics, vol.~51, Amer. Math. Soc.,
  Providence, RI, 2012, pp.~629--643.

\bibitem[LD68]{DysLen-68}
A. Lenard and F.~J. Dyson, \emph{Stability of matter. {II}}, J. Math. Phys.
  \textbf{9} (1968), no.~5, 698--711, \href
  {http://dx.doi.org/10.1063/1.1664631} {\path{doi}}. \MR{2408897
  (2010e:81267)}

\bibitem[Len73]{Lenard-73}
A. Lenard, \emph{Lectures on the {C}oulomb stability problem}, Statistical
  Mechanics and Mathematical Problems, Battelle Rencontres, Seattle, WA, 1971,
  Lecture Notes in Phys., vol.~20, Springer-Verlag, Berlin, Heidelberg, 1973,
  pp.~114--135.

\bibitem[Len13]{Lenzmann-13}
E. Lenzmann, \emph{Interpolation inequalities and applications to nonlinear
  {PDE}}, Lecture notes at the IHP Program ``Variational and Spectral Methods
  in Quantum Mechanics'', unpublished, 2013.

\bibitem[Lev14]{Levitt-14}
A. Levitt, \emph{{Best constants in Lieb-Thirring inequalities: a numerical
  investigation}}, J. Spectr. Theory \textbf{4} (2014), no.~1, 153--175, \href
  {http://dx.doi.org/10.4171/JST/65} {\path{doi}}.

\bibitem[Lew15]{Lewin-15}
M. Lewin, \emph{Mean-field limit of {B}ose systems: rigorous results},
  Proceedings from the International Congress of Mathematical Physics at
  Santiago de Chile, July 2015, 2015, \href {http://arxiv.org/abs/1510.04407}
  {\path{arXiv:1510.04407}}.

\bibitem[Lie76]{Lieb-76}
E.~H. Lieb, \emph{The stability of matter}, Rev. Mod. Phys. \textbf{48} (1976),
  no.~4, 553--569, \href {http://dx.doi.org/10.1103/RevModPhys.48.553}
  {\path{doi}}.

\bibitem[Lie79]{Lieb-79}
E.~H. Lieb, \emph{{A lower bound for Coulomb energies}}, Phys. Lett. A
  \textbf{70} (1979), 444--446, \href
  {http://dx.doi.org/10.1016/0375-9601(79)90358-X} {\path{doi}}.

\bibitem[LL01]{LieLos-01}
E.~H. Lieb and M. Loss, \emph{Analysis}, 2nd ed., Graduate Studies in
  Mathematics, vol.~14, American Mathematical Society, Providence, RI, 2001.
  \MR{MR1817225 (2001i:00001)}

\bibitem[LL18]{LarLun-16}
S. Larson and D. Lundholm, \emph{Exclusion bounds for extended anyons}, Arch.
  Ration. Mech. Anal. \textbf{227} (2018), 309--365, \href
  {http://dx.doi.org/10.1007/s00205-017-1161-9} {\path{doi}}.

\bibitem[LM77]{LeiMyr-77}
J.~M. {Leinaas} and J. {Myrheim}, \emph{{On the theory of identical
  particles}}, Nuovo Cimento B \textbf{37} (1977), 1--23, \href
  {http://dx.doi.org/10.1007/BF02727953} {\path{doi}}.

\bibitem[LNP16]{LunNamPor-16}
D. Lundholm, P.~T. Nam, and F. Portmann, \emph{Fractional
  {H}ardy-{L}ieb-{T}hirring and related inequalities for interacting systems},
  Arch. Ration. Mech. Anal. \textbf{219} (2016), no.~3, 1343--1382, \href
  {http://dx.doi.org/10.1007/s00205-015-0923-5} {\path{doi}}.

\bibitem[LO94]{LeeOh-94}
T. Lee and P. Oh, \emph{Non-abelian {C}hern-{S}imons quantum mechanics and
  non-abelian {A}haronov-{B}ohm effect}, Ann. Physics \textbf{235} (1994),
  no.~2, 413--434, \href {http://dx.doi.org/10.1006/aphy.1994.1103}
  {\path{doi}}. \MR{1297823}

\bibitem[LPS15]{LunPorSol-15}
D. Lundholm, F. Portmann, and J.~P. Solovej, \emph{Lieb-{T}hirring bounds for
  interacting {B}ose gases}, Comm. Math. Phys. \textbf{335} (2015), no.~2,
  1019--1056, \href {http://dx.doi.org/10.1007/s00220-014-2278-4} {\path{doi}}.

\bibitem[LQ18]{LunQva-18}
D. Lundholm and V. Qvarfordt, \emph{Many-body exclusion properties of
  non-abelian anyons}, in preparation, 2018.

\bibitem[LS04]{LieSol-04}
E.~H. Lieb and J.~P. Solovej, \emph{{Ground state energy of the two-component
  charged Bose gas.}}, Commun. Math. Phys. \textbf{252} (2004), no.~1-3,
  485--534.

\bibitem[LS10]{LieSei-09}
E.~H. Lieb and R. Seiringer, \emph{The {S}tability of {M}atter in {Q}uantum
  {M}echanics}, Cambridge Univ. Press, 2010.

\bibitem[LS13a]{LunSol-13a}
D. Lundholm and J.~P. Solovej, \emph{{Hardy and Lieb-Thirring inequalities for
  anyons}}, Comm. Math. Phys. \textbf{322} (2013), 883--908, \href
  {http://dx.doi.org/10.1007/s00220-013-1748-4} {\path{doi}}.

\bibitem[LS13b]{LunSol-13b}
D. Lundholm and J.~P. Solovej, \emph{Local exclusion principle for identical
  particles obeying intermediate and fractional statistics}, Phys. Rev. A
  \textbf{88} (2013), 062106, \href
  {http://dx.doi.org/10.1103/PhysRevA.88.062106} {\path{doi}}.

\bibitem[LS14]{LunSol-14}
D. Lundholm and J.~P. Solovej, \emph{{Local exclusion and Lieb-Thirring
  inequalities for intermediate and fractional statistics}}, Ann. Henri
  Poincar\'e \textbf{15} (2014), 1061--1107, \href
  {http://dx.doi.org/10.1007/s00023-013-0273-5} {\path{doi}}.

\bibitem[LS18]{LunSei-17}
D. Lundholm and R. Seiringer, \emph{Fermionic behavior of ideal anyons}, to
  appear in Lett. Math. Phys., 2018, \href {http://arxiv.org/abs/1712.06218}
  {\path{arXiv:1712.06218}}.

\bibitem[LSSY05]{LieSeiSolYng-05}
E.~H. Lieb, R. Seiringer, J.~P. Solovej, and J. Yngvason, \emph{The mathematics
  of the {B}ose gas and its condensation}, Oberwolfach {S}eminars,
  Birkh{\"a}user, 2005, \href {http://arxiv.org/abs/cond-mat/0610117}
  {\path{arXiv:cond-mat/0610117}}.

\bibitem[LT75]{LieThi-75}
E.~H. Lieb and W.~E. Thirring, \emph{Bound for the kinetic energy of fermions
  which proves the stability of matter}, Phys. Rev. Lett. \textbf{35} (1975),
  687--689, \href {http://dx.doi.org/10.1103/PhysRevLett.35.687} {\path{doi}}.

\bibitem[LT76]{LieThi-76}
E.~H. Lieb and W.~E. Thirring, \emph{Inequalities for the moments of the
  eigenvalues of the {S}chr{\"o}dinger hamiltonian and their relation to
  {S}obolev inequalities}, Studies in Mathematical Physics, pp.~269--303,
  Princeton Univ. Press, 1976.

\bibitem[Lun08]{Lundholm-08}
D. Lundholm, \emph{On the geometry of supersymmetric quantum mechanical
  systems}, J. Math. Phys. \textbf{49} (2008), no.~6, 062101, \href
  {http://dx.doi.org/10.1063/1.2937096} {\path{doi}}. \MR{2431770}

\bibitem[Lun15]{Lundholm-15}
D. Lundholm, \emph{Geometric extensions of many-particle {H}ardy inequalities},
  J. Phys. A: Math. Theor. \textbf{48} (2015), 175203, \href
  {http://dx.doi.org/10.1088/1751-8113/48/17/175203} {\path{doi}}.

\bibitem[Lun17]{Lundholm-16}
D. Lundholm, \emph{Many-anyon trial states}, Phys. Rev. A \textbf{96} (2017),
  012116, \href {http://dx.doi.org/10.1103/PhysRevA.96.012116} {\path{doi}}.

\bibitem[MB03]{MulBla-03}
W.~J. Mullin and G. Blaylock, \emph{Quantum statistics: Is there an effective
  fermion repulsion or boson attraction?}, Am. J. Phys. \textbf{71} (2003),
  no.~12, 1223--1231, \href {http://dx.doi.org/10.1119/1.1590658} {\path{doi}}.

\bibitem[MD93]{MueDoe-93}
U.~A. Mueller and H.-D. Doebner, \emph{Borel quantum kinematics of rank {$k$}
  on smooth manifolds}, J. Phys. A \textbf{26} (1993), no.~3, 719--730, \href
  {http://dx.doi.org/10.1088/0305-4470/26/3/029} {\path{doi}}. \MR{1210931}

\bibitem[Mic13]{Michiels-13}
D. Michiels, \emph{Moduli spaces of flat connections}, MSc thesis, KU Leuven,
  2013, \url{https://faculty.math.illinois.edu/~michiel2/docs/thesis.pdf}.

\bibitem[MS95]{MunSch-95}
J. Mund and R. Schrader, \emph{Hilbert spaces for nonrelativistic and
  relativistic ``free'' plektons (particles with braid group statistics)},
  Advances in dynamical systems and quantum physics ({C}apri, 1993), World Sci.
  Publ., River Edge, NJ, 1995, pp.~235--259, \href
  {http://arxiv.org/abs/hep-th/9310054} {\path{arXiv:hep-th/9310054}}.
  \MR{1414702}

\bibitem[MS06]{ManSir-06}
E.~B. Manoukian and S. Sirininlakul, \emph{Stability of matter in 2{D}}, Rep.
  Math. Phys. \textbf{58} (2006), no.~2, 263--274, \href
  {http://dx.doi.org/10.1016/S0034-4877(06)80052-2} {\path{doi}}. \MR{2281540}

\bibitem[Myr99]{Myrheim-99}
J. Myrheim, \emph{Anyons}, Topological aspects of low dimensional systems (A.
  Comtet, T. Jolic{\oe}ur, S. Ouvry, and F. David, eds.), Les Houches - Ecole
  d'Ete de Physique Theorique, vol.~69, (Springer-Verlag, Berlin, Germany),
  1999, pp.~265--413, \href {http://dx.doi.org/10.1007/3-540-46637-1_4}
  {\path{doi}}.

\bibitem[Nak03]{Nakahara-03}
M. Nakahara, \emph{Geometry, topology and physics}, Graduate Student Series in
  Physics, Institute of Physics Publishing, Second Edition 2003.

\bibitem[Nam18]{Nam-18}
P.~T. Nam, \emph{Lieb-{T}hirring inequality with semiclassical constant and
  gradient error term}, J. Funct. Anal. \textbf{274} (2018), no.~6, 1739--1746,
  \href {http://dx.doi.org/10.1016/j.jfa.2017.08.007} {\path{doi}}.
  \MR{3758547}

\bibitem[Pau47]{Pauli-47}
W. Pauli, \emph{{N}obel {L}ecture: {E}xclusion principle and quantum
  mechanics}, Editions du Griffon, Neuchatel, 1947,
  \url{http://www.nobelprize.org/nobel_prizes/physics/laureates/1945/pauli-lecture.pdf}.

\bibitem[Pol99]{Polychronakos-99}
A.~P. Polychronakos, \emph{Generalized statistics in one dimension},
  Topological aspects of low dimensional systems (A. Comtet, T. Jolic{\oe}ur,
  S. Ouvry, and F. David, eds.), Les Houches - Ecole d'Ete de Physique
  Theorique, vol.~69, (Springer-Verlag, Berlin, Germany), 1999, pp.~415--471,
  \href {http://dx.doi.org/10.1007/3-540-46637-1_5} {\path{doi}}.

\bibitem[Qva17]{Qvarfordt-17}
V. Qvarfordt, \emph{Non-abelian anyons: Statistical repulsion and topological
  quantum computation}, MSc thesis, KTH, 2017,
  \url{http://urn.kb.se/resolve?urn=urn%3Anbn%3Ase%3Akth%3Adiva-207177}.

\bibitem[Rov04]{Rovelli-04}
C. Rovelli, \emph{Quantum gravity}, Cambridge Monographs on Mathematical
  Physics, Cambridge University Press, Cambridge, 2004, With a foreword by
  James Bjorken, \href {http://dx.doi.org/10.1017/CBO9780511755804}
  {\path{doi}}. \MR{2106565}

\bibitem[RS72]{ReeSim1}
M. Reed and B. Simon, \emph{Methods of {M}odern {M}athematical {P}hysics. {I}.
  functional analysis}, Academic Press, 1972.

\bibitem[RS75]{ReeSim2}
M. Reed and B. Simon, \emph{Methods of {M}odern {M}athematical {P}hysics. {II}.
  {F}ourier analysis, self-adjointness}, Academic Press, New York, 1975.
  \MR{MR0493421 (58 \#12429c)}

\bibitem[Rum10]{Rumin-10}
M. Rumin, \emph{Spectral density and {S}obolev inequalities for pure and mixed
  states}, Geom. Funct. Anal. \textbf{20} (2010), no.~3, 817--844, \href
  {http://dx.doi.org/10.1007/s00039-010-0075-6} {\path{doi}}. \MR{2720233
  (2011m:31014)}

\bibitem[Rum11]{Rumin-11}
M. Rumin, \emph{Balanced distribution-energy inequalities and related entropy
  bounds}, Duke Math. J. \textbf{160} (2011), no.~3, 567--597, \href
  {http://dx.doi.org/10.1215/00127094-1444305} {\path{doi}}.

\bibitem[Sei10]{Seiringer-10}
R. Seiringer, \emph{Inequalities for {S}chr\"odinger operators and applications
  to the stability of matter problem}, Entropy and the quantum, Contemp. Math.,
  vol. 529, Amer. Math. Soc., Providence, RI, 2010, pp.~53--72, \href
  {http://dx.doi.org/10.1090/conm/529/10427} {\path{doi}}. \MR{2681768}

\bibitem[Shu01]{Shubin-01}
M.~A. Shubin, \emph{Pseudodifferential operators and spectral theory}, Springer
  Berlin Heidelberg, Second Edition 2001, \href
  {http://dx.doi.org/10.1007/978-3-642-56579-3} {\path{doi}}.

\bibitem[Sol06a]{Solovej-06b}
J.~P. Solovej, \emph{Stability of {M}atter}, Encyclopedia of Mathematical
  Physics (J.-P. Francoise, G.~L. Naber, and S.~T. Tsou, eds.), vol.~5,
  Elsevier, 2006, pp.~8--14.

\bibitem[Sol06b]{Solovej-06}
J.~P. Solovej, \emph{{Upper bounds to the ground state energies of the one- and
  two-component charged Bose gases}}, Commun. Math. Phys. \textbf{266} (2006),
  no.~3, 797--818.

\bibitem[Sol11]{Solovej-11}
J.~P. Solovej, \emph{Examples: Non-interacting systems and the
  {L}ieb-{T}hirring inequality, {C}harged systems}, Lecture 2 of the course
  ``Spectral Theory of N-body Schr\"odinger operators'' given at the University
  College London, Spring 2011, 2011,
  \url{http://www.ucl.ac.uk/~ucahipe/solovej-lt.pdf}.

\bibitem[Sou70]{Souriau-70}
J.-M. Souriau, \emph{Structure des syst\`emes dynamiques}, Ma{\^{i}}trises de
  math\'ematiques, Dunod, Paris, 1970, English translation by R. H. Cushman and
  G. M. Tuynman, Progress in Mathematics, 149, Birkh\"auser Boston Inc.,
  Boston, MA, 1997,
  \url{http://www.jmsouriau.com/structure_des_systemes_dynamiques.htm}.
  \MR{0260238}

\bibitem[{Sut}71]{Sutherland-71}
B. {Sutherland}, \emph{{Quantum Many-Body Problem in One Dimension: Ground
  State}}, J. Mathematical Phys. \textbf{12} (1971), 246--250, \href
  {http://dx.doi.org/10.1063/1.1665584} {\path{doi}}.

\bibitem[Tal76]{Talenti-76}
G. Talenti, \emph{Best constant in {S}obolev inequality}, Annali di Matematica
  Pura ed Applicata \textbf{110} (1976), no.~1, 353--372, \href
  {http://dx.doi.org/10.1007/BF02418013} {\path{doi}}.

\bibitem[Tes14]{Teschl-14}
G. Teschl, \emph{Mathematical methods in quantum mechanics; with applications
  to {S}chr{\"o}dinger operators}, Graduate Studies in Mathematics, vol.~99,
  Amer. Math. Soc, Providence, RI, Second Edition 2014.

\bibitem[Thi02]{Thirring2}
W.~E. Thirring, \emph{Quantum mathematical physics}, vol. Atoms, Molecules and
  Large Systems, Springer, Second Edition 2002.

\bibitem[Thi03]{Thirring1}
W.~E. Thirring, \emph{Classical mathematical physics}, vol. Dynamical Systems
  and Field Theories, Springer, Third Edition 2003.

\bibitem[Thi07]{Thiemann-07}
T. Thiemann, \emph{Modern canonical quantum general relativity}, Cambridge
  Monographs on Mathematical Physics, Cambridge University Press, Cambridge,
  2007, With a foreword by Chris Isham, \href
  {http://dx.doi.org/10.1017/CBO9780511755682} {\path{doi}}. \MR{2374859}

\bibitem[Tid05]{Tidblom-05}
J. Tidblom, \emph{{Improved $L^p$ Hardy Inequalities}}, Ph.D. thesis, Stockholm
  University, Department of Mathematics, 2005,
  \url{http://urn.kb.se/resolve?urn=urn%3Anbn%3Ase%3Asu%3Adiva-615}.

\bibitem[Ver91]{Verlinde-91}
E.~P. Verlinde, \emph{A note on braid statistics and the non-{A}belian
  {A}haronov-{B}ohm effect}, MODERN QUANTUM FIELD THEORY: proceedings. (S. Das,
  A. Dhar, S. Mukhi, A. Raina, and A. Sen., eds.), World Scientific, 1991,
  Proc. of Conference: International Colloquium on Modern Quantum Field Theory,
  Bombay, 1990, pp.~450--461,
  \url{https://lib-extopc.kek.jp/preprints/PDF/1991/9106/9106160.pdf}.

\bibitem[Wil82]{Wilczek-82b}
F. Wilczek, \emph{Quantum mechanics of fractional-spin particles}, Phys. Rev.
  Lett. \textbf{49} (1982), 957--959, \href
  {http://dx.doi.org/10.1103/PhysRevLett.49.957} {\path{doi}}.

\bibitem[Wu94]{Wu-94}
Y.-S. Wu, \emph{Statistical distribution for generalized ideal gas of
  fractional-statistics particles}, Phys. Rev. Lett. \textbf{73} (1994), 922,
  \href {http://dx.doi.org/10.1103/PhysRevLett.73.922} {\path{doi}}.

\end{thebibliography}

% methmmp.bbl :

\providecommand{\bysame}{\leavevmode\hbox to3em{\hrulefill}\thinspace}
\providecommand{\MR}{\relax\ifhmode\unskip\space\fi MR }
% \MRhref is called by the amsart/book/proc definition of \MR.
\providecommand{\MRhref}[2]{%
  \href{http://www.ams.org/mathscinet-getitem?mr=#1}{#2}
}
\providecommand{\href}[2]{#2}

% ------------------  Index  --------------------

%\section*{Index}
%\addcontentsline{toc}{section}{Index}
	\printindex
%	\newpage

\end{document}